\documentclass[11pt,a4paper,nosumlimits,reqno]{amsart}
\usepackage{booktabs}% http://ctan.org/pkg/booktabs

\usepackage{amsthm}
\usepackage{amsmath}
\usepackage[sort&compress,square]{natbib}
\usepackage{latexsym,amssymb,enumerate,bbm,amscd,graphicx,color,fullpage}
\usepackage[english]{babel}
\usepackage{tikz}
\usepackage{tikz-cd}
\usepackage{subcaption}

\usepackage{booktabs,array}
\usepackage{graphicx} 
\usepackage{textcomp}
\usepackage{simplewick}
\usepackage{float}
\usepackage{xifthen}
\usepackage{xcolor}
\usepackage{sansmath}
\usepackage{marvosym}
\usepackage{tensor}
\usepackage{xparse}

\newcolumntype{L}{>{\arraybackslash}m{3cm}}

\usepackage{relsize}
\NewDocumentCommand{\evalat}{sO{\big}mm}{%
  \IfBooleanTF{#1}
   {\mleft. #3 \mright|_{#4}}
   {#3#2|_{#4}}%
} 
 
\allowdisplaybreaks

 \newcommand{\naranja}[1]{\color{black}#1 \color{black}}
\usetikzlibrary{calc,decorations.pathmorphing,patterns}
\usetikzlibrary{decorations.text}
\usetikzlibrary{arrows}
\usetikzlibrary{decorations.pathmorphing,backgrounds,positioning,fit,petri}
\usetikzlibrary{automata}
\usetikzlibrary{graphs}
\usetikzlibrary{shadows}
\usetikzlibrary{trees}
\usetikzlibrary{decorations.pathmorphing}
\usetikzlibrary{decorations.markings} 

\colorlet{texto}{black!50!gray} 
 
\definecolor{turqueza}{HTML}{40e0d0}%{2B233B}

\usepackage[bookmarks=false]{hyperref}
\hypersetup{
     colorlinks   = true,
     citecolor    = blue!50!lightgray
}
\hypersetup{linkcolor=blue!50!gray}
\newcommand{\titulote}[1]{%
  \ifodd\value{page}%
    \protect\parbox{0.97\linewidth}{#1}\hfill%
  \else%
    \hfill\protect\parbox{0.97\linewidth}{#1}%
  \fi%
}

\numberwithin{equation}{section}
\newtheorem{thm}{Theorem}[section]
\newtheorem{cor}[thm]{Corollary}
\newtheorem{lem}[thm]{Lemma}
\newtheorem{prop}[thm]{Proposition}

\theoremstyle{remark}
\newtheorem{rem}[thm]{Remark}
\newtheorem{example}[thm]{Example}

\theoremstyle{definition}
\newtheorem{defn}[thm]{Definition}

\usepackage{enumitem}

\newcommand{\kc}{K_{\mathrm{c}}(3,3)}
\newcommand{\Aut}{\mathrm{Aut}}
\newcommand{\Autc}{\Aut_{\mathrm{c}}}
\renewcommand{\and}{\mbox{and}}

\newcommand{\T}{\mathbb{T}}
\newcommand{\Tr}{\mathrm{Tr}}

\newcommand{\mtr}[1]{\mathrm{#1}}
\newcommand{\mtf}[1]{\mathfrak{#1}}

\newcommand{\dif}[1]{\mathrm{d}#1}

\newcommand{\re}{\mathbb{R}}

\newcommand{\dervpar}[2]{\frac{\partial #1}{\partial #2}}

\renewcommand{\thefootnote}{\arabic{footnote}}

\newcommand{\si}{\sigma}
\newcommand{\G}{\mathcal{G}}
\renewcommand{\H}{\mathcal{H}}
\newcommand{\B}{\mathcal{B}}
\newcommand{\C}{\mathbb{C}}

\newcommand{\ii}{\mathrm{i}}
\newcommand{\ee}{\mathrm{e}}
 
\newcommand{\inv}{^{-1}} 
\newcommand{\mtc}[1]{\mathcal{#1}} 
\newcommand{\Z}{\mathbb{Z}}
\newcommand{\N}{\mathbb{N}}

\newcommand{\V}{\mtc{V}}

\newcommand{\im}{\mtr{im}}

\newcommand{\where}{\mbox{where}\,\,}

\newcommand{\mtb}[1]{\mathbb{#1}}

\newcommand{\R}{\mathcal R}

\newcommand{\Df}{\mathbb{\mathcal{D}}}

\newcommand{\Sym}{\mathfrak{S}}
 
\newcommand{\hp}[1]{^{(#1)}}

\renewcommand{\phi}{\varphi}

\newcommand{\fey}{\mathsf{Feyn}}
\newcommand{\vfey}{\mathsf{Feyn}^{\mathrm{v}}}
\newcommand{\feyn}{\fey_3(\phi^4)}

  \newcommand{\strictdvGrph}[1]{\Xi^{\mtr{cl}}_{#1}}
\newcommand{\Grph}[1]{\textsf{Grph}_{#1}}

 \newcommand{\dGrph}[1]{\Grph{#1}^{\amalg}}
 \newcommand{\dvGrph}[1]{\Grph{#1}^{\amalg,\mathrm{cl}}}
 \newcommand{\vGrph}[1]{\Grph{#1}^\mathrm{cl}}

\newcommand{\suml}{\sum\limits}

\newcommand{\fder}[2]{\frac{\delta #1}{\delta #2}}

\newcommand{\bJ}{\bar J}
\newcommand{\kthree}{\raisebox{-.2\height}{\includegraphics[height=2ex]{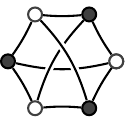}}}

\newcommand{\GDmelon}{G\hp{2}_{\raisebox{-.33\height}{\includegraphics[height=2.2ex]{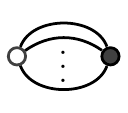}}}}
\newcommand{\Gmelon}{G\hp{2}_{\raisebox{-.33\height}{\includegraphics[height=1.8ex]{graphs/3/Item2_Melon}}}}
\newcommand{\Guno}{G\hp{4}_{\! \raisebox{-.2\height}{\includegraphics[height=1.8ex]{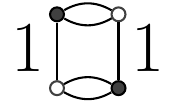}}}}
\newcommand{\Gdos}{G\hp{4}_{\! \raisebox{-.2\height}{\includegraphics[height=1.8ex]{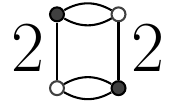}}}}
\newcommand{\Gtres}{G\hp{4}_{\! \raisebox{-.2\height}{\includegraphics[height=1.8ex]{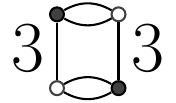}}}}

\newcommand{\vuno}{\raisebox{-.322\height}{\includegraphics[height=2.3ex]{graphs/3/Item4_V1v.pdf}}}
\newcommand{\vdos}{\raisebox{-.322\height}{\includegraphics[height=2.3ex]{graphs/3/Item4_V2v.pdf}}}
\newcommand{\vtres}{\raisebox{-.322\height}{\includegraphics[height=2.3ex]{graphs/3/Item4_V3v.pdf}}}
\newcommand{\vunito}{\raisebox{-.322\height}{\includegraphics[height=1.66ex]{graphs/3/Item4_V1v.pdf}}}

\newcommand{\Gcmm}{G\hp{4}_{|\raisebox{-.33\height}{\includegraphics[height=1.8ex]{graphs/3/Item2_Melon}}
|\raisebox{-.33\height}{\includegraphics[height=1.8ex]{graphs/3/Item2_Melon}}|}}

\newcommand{\Gsmmm}{G\hp{6}_{|\raisebox{-.33\height}{\includegraphics[height=1.8ex]{graphs/3/Item2_Melon}}
|\raisebox{-.33\height}{\includegraphics[height=1.8ex]{graphs/3/Item2_Melon}}
|\raisebox{-.33\height}{\includegraphics[height=1.8ex]{graphs/3/Item2_Melon}}|}}
\newcommand{\Gsma}{G\hp{6}_{|\raisebox{-.33\height}{\includegraphics[height=1.8ex]{graphs/3/Item2_Melon}}
|\raisebox{-.34\height}{\includegraphics[height=1.8ex]{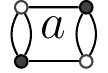}}|}}
\newcommand{\Gsmb}{G\hp{6}_{|\raisebox{-.33\height}{\includegraphics[height=1.8ex]{graphs/3/Item2_Melon}}
|\raisebox{-.34\height}{\includegraphics[height=1.8ex]{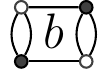}}|}}
\newcommand{\Gsmc}{G\hp{6}_{|\raisebox{-.33\height}{\includegraphics[height=1.8ex]{graphs/3/Item2_Melon}}
|\raisebox{-.34\height}{\includegraphics[height=1.8ex]{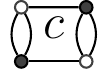}}|}}
\newcommand{\Gsmu}{G\hp{6}_{|\raisebox{-.33\height}{\includegraphics[height=1.8ex]{graphs/3/Item2_Melon}}
|\raisebox{-.34\height}{\includegraphics[height=1.8ex]{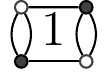}}|}}

\newcommand{\Gsmt}{G\hp{6}_{|\raisebox{-.33\height}{\includegraphics[height=1.8ex]{graphs/3/Item2_Melon}}
|\raisebox{-.34\height}{\includegraphics[height=1.8ex]{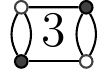}}|}}
\newcommand{\Gsmi}{G\hp{6}_{|\raisebox{-.33\height}{\includegraphics[height=1.8ex]{graphs/3/Item2_Melon}}
|\raisebox{-.34\height}{\includegraphics[height=1.8ex]{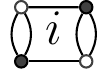}}|}}

\newcommand{\Gseu}{G\hp{6}\raisebox{-.8\height}{ \!\!\!\!\!\!\!\!\!\includegraphics[height=2.7ex]{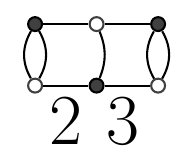}}}
\newcommand{\Gsed}{G\hp{6}\raisebox{-.8\height}{\!\! \!\!\!\!\!\!\!\includegraphics[height=2.7ex]{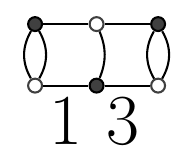}}}
\newcommand{\Gset}{G\hp{6}\raisebox{-.8\height}{ \!\!\!\!\!\!\!\!\!\includegraphics[height=2.7ex]{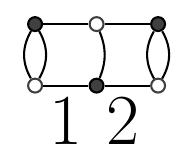}}}
\newcommand{\Gseabc}{G\hp{6}\raisebox{-.8\height}{ \!\!\!\!\!\!\!\!\!\includegraphics[height=2.7ex]{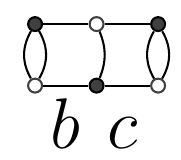}}}
\newcommand{\Gsebac}{G\hp{6}\raisebox{-.8\height}{ \!\!\!\!\!\!\!\!\!\includegraphics[height=2.7ex]{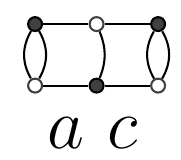}}}
\newcommand{\Gsecab}{G\hp{6}\raisebox{-.8\height}{ \!\!\!\!\!\!\!\!\!\includegraphics[height=2.7ex]{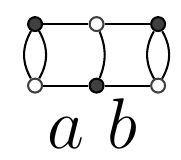}}}

\newcommand{\Gskthree}{G\hp{6}_{\raisebox{-.5\height}{\includegraphics[height=1.9ex]{graphs/3/Item6_K33.pdf}}}}
\newcommand{\Gsqa}{G\hp{6}\raisebox{-.8\height}{ \!\!\!\!\!\!\!\includegraphics[height=2.ex]{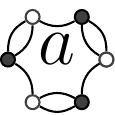}}}
\newcommand{\Gsqb}{G\hp{6}\raisebox{-.8\height}{ \!\!\!\!\!\!\!\includegraphics[height=2.ex]{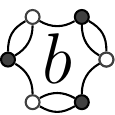}}}
\newcommand{\Gsqc}{G\hp{6}\raisebox{-.8\height}{ \!\!\!\!\!\!\!\includegraphics[height=2.ex]{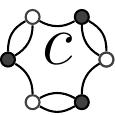}}} 
 
\newcommand{\Gsqu}{G\hp{6}\raisebox{-.8\height}{ \!\!\!\!\!\!\!\includegraphics[height=2.ex]{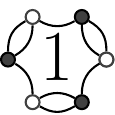}}} 
\newcommand{\Gsqt}{G\hp{6}\raisebox{-.8\height}{ \!\!\!\!\!\!\!\includegraphics[height=2.ex]{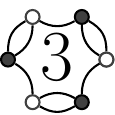}}} 
\newcommand{\Gsqd}{G\hp{6}\raisebox{-.8\height}{ \!\!\!\!\!\!\!\includegraphics[height=2.ex]{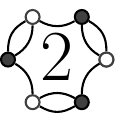}}}

\newcommand{\Gomeu}{G\hp{8}_{|\raisebox{-.33\height}{\includegraphics[height=1.8ex]{graphs/3/Item2_Melon}}|\raisebox{-.2\height}{\includegraphics[height=2.7ex]{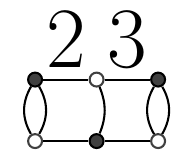}}|}}

\newcommand{\Gomea}{G\hp{8}_{|\raisebox{-.33\height}{\includegraphics[height=1.8ex]{graphs/3/Item2_Melon}}|\raisebox{-.55\height}{\includegraphics[height=3.2ex]{graphs/3/Item6_Eabcs.pdf}}|}}
\newcommand{\Gomeb}{G\hp{8}_{|\raisebox{-.33\height}{\includegraphics[height=1.8ex]{graphs/3/Item2_Melon}}|\raisebox{-.55\height}{\includegraphics[height=3.2ex]{graphs/3/Item6_Ebacs.pdf}}|}}
\newcommand{\Gomec}{G\hp{8}_{|\raisebox{-.33\height}{\includegraphics[height=1.8ex]{graphs/3/Item2_Melon}}|\raisebox{-.55\height}{\includegraphics[height=3.2ex]{graphs/3/Item6_Ecabs.pdf}}|}}
\newcommand{\Gomkthree}{G\hp{8}_{|\raisebox{-.33\height}{\includegraphics[height=1.8ex]{graphs/3/Item2_Melon}}|\raisebox{-.2\height}{\hspace{.7pt}\includegraphics[height=1.85ex]{graphs/3/Item6_K33.pdf}}|}}
\newcommand{\Gomqa}{G\hp{8}_{|\raisebox{-.23\height}{\includegraphics[height=1.8ex]{graphs/3/Item2_Melon}}|\raisebox{-.2\height}{\includegraphics[height=2.1ex]{graphs/3/Item6_Qas.pdf}}|}}

\newcommand{\Gomqb}{G\hp{8}_{|\raisebox{-.23\height}{\includegraphics[height=1.8ex]{graphs/3/Item2_Melon}}|\raisebox{-.2\height}{\includegraphics[height=2.1ex]{graphs/3/Item6_Qbs.pdf}}|} }
\newcommand{\Gomqc}{G\hp{8}_{|\raisebox{-.23\height}{\includegraphics[height=1.8ex]{graphs/3/Item2_Melon}}|\raisebox{-.2\height}{\includegraphics[height=2.1ex]{graphs/3/Item6_Qcs.pdf}}|} }
\newcommand{\Gomqi}{G\hp{8}_{|\raisebox{-.23\height}{\includegraphics[height=1.8ex]{graphs/3/Item2_Melon}}|\raisebox{-.2\height}{\includegraphics[height=2.1ex]{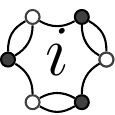}}|} }
\newcommand{\Gomqu}{G\hp{8}_{|\raisebox{-.23\height}{\includegraphics[height=1.8ex]{graphs/3/Item2_Melon}}|\raisebox{-.2\height}{\includegraphics[height=2.1ex]{graphs/3/Item6_Q1s.pdf}}|} }

\newcommand{\Gommb}{G\hp{8}_{|\raisebox{-.33\height}{\includegraphics[height=1.8ex]{graphs/3/Item2_Melon}}|
\raisebox{-.33\height}{\includegraphics[height=1.8ex]{graphs/3/Item2_Melon}}|
\raisebox{-.33\height}{\includegraphics[height=1.85ex]{graphs/3/Item4_Vb.pdf}}|}}
\newcommand{\Gomma}{G\hp{8}_{|\raisebox{-.33\height}{\includegraphics[height=1.8ex]{graphs/3/Item2_Melon}}|
\raisebox{-.33\height}{\includegraphics[height=1.8ex]{graphs/3/Item2_Melon}}|
\raisebox{-.33\height}{\includegraphics[height=1.85ex]{graphs/3/Item4_Va.pdf}}|}}
\newcommand{\Gommc}{G\hp{8}_{|\raisebox{-.33\height}{\includegraphics[height=1.8ex]{graphs/3/Item2_Melon}}|
\raisebox{-.33\height}{\includegraphics[height=1.8ex]{graphs/3/Item2_Melon}}|
\raisebox{-.33\height}{\includegraphics[height=1.85ex]{graphs/3/Item4_Vc.pdf}}|}}
\newcommand{\Gommi}{G\hp{8}_{|\raisebox{-.33\height}{\includegraphics[height=1.8ex]{graphs/3/Item2_Melon}}|
\raisebox{-.33\height}{\includegraphics[height=1.8ex]{graphs/3/Item2_Melon}}|
\raisebox{-.33\height}{\includegraphics[height=1.85ex]{graphs/3/Item4_Vi.pdf}}|}}

\newcommand{\Gommmm}{G\hp{8}_{|\raisebox{-.33\height}{\includegraphics[height=1.8ex]{graphs/3/Item2_Melon}}|
\raisebox{-.33\height}{\includegraphics[height=1.8ex]{graphs/3/Item2_Melon}}|
\raisebox{-.33\height}{\includegraphics[height=1.8ex]{graphs/3/Item2_Melon}}
|\raisebox{-.33\height}{\includegraphics[height=1.8ex]{graphs/3/Item2_Melon}}|}}

\newcommand{\Goaa}{G\hp{8}_{|
\raisebox{-.33\height}{\includegraphics[height=1.85ex]{graphs/3/Item4_Va.pdf}}|
\raisebox{-.33\height}{\includegraphics[height=1.85ex]{graphs/3/Item4_Va.pdf}}|}}

\newcommand{\Goca}{G\hp{8}_{|
\raisebox{-.33\height}{\includegraphics[height=1.85ex]{graphs/3/Item4_Vc.pdf}}|
\raisebox{-.33\height}{\includegraphics[height=1.85ex]{graphs/3/Item4_Va.pdf}}|}}

\newcommand{\Gobc}{G\hp{8}_{|
\raisebox{-.33\height}{\includegraphics[height=1.85ex]{graphs/3/Item4_Vb.pdf}}|
\raisebox{-.33\height}{\includegraphics[height=1.85ex]{graphs/3/Item4_Vc.pdf}}|}}

\newcommand{\Goia}{G\hp{8}_{|
\raisebox{-.33\height}{\includegraphics[height=1.85ex]{graphs/3/Item4_Vi.pdf}}|
\raisebox{-.33\height}{\includegraphics[height=1.85ex]{graphs/3/Item4_Va.pdf}}|}}
\newcommand{\Goii}{G\hp{8}_{|
\raisebox{-.33\height}{\includegraphics[height=1.85ex]{graphs/3/Item4_Vi.pdf}}|
\raisebox{-.33\height}{\includegraphics[height=1.85ex]{graphs/3/Item4_Vi.pdf}}|}}
\newcommand{\Gobb}{G\hp{8}_{|
\raisebox{-.33\height}{\includegraphics[height=1.85ex]{graphs/3/Item4_Vb.pdf}}|
\raisebox{-.33\height}{\includegraphics[height=1.85ex]{graphs/3/Item4_Vb.pdf}}|}}

\newcommand{\Gocc}{G\hp{8}_{|
\raisebox{-.33\height}{\includegraphics[height=1.85ex]{graphs/3/Item4_Vc.pdf}}|
\raisebox{-.33\height}{\includegraphics[height=1.85ex]{graphs/3/Item4_Vc.pdf}}|}}

\newcommand{\GoAu}{G\hp{8}\raisebox{-.9\height}{\!\!\!\!\!\!\!\!\includegraphics[height=2.7ex]{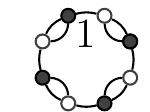}}}

\newcommand{\GoAc}{G\hp{8}\raisebox{-.9\height}{\!\!\!\!\!\!\!\!\includegraphics[height=2.7ex]{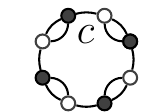}}}

\newcommand{\GoAb}{G\hp{8}\raisebox{-.9\height}{\!\!\!\!\!\!\!\!\includegraphics[height=2.7ex]{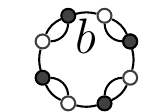}}}
\newcommand{\GoAa}{G\hp{8}\raisebox{-.9\height}{\!\!\!\!\!\!\!\!\includegraphics[height=2.7ex]{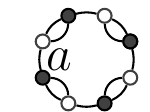}}}
\newcommand{\GoAi}{G\hp{8}\raisebox{-.9\height}{\!\!\!\!\!\!\!\!\includegraphics[height=2.7ex]{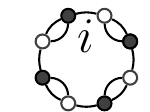}}}
\newcommand{\GoAj}{G\hp{8}_{\raisebox{-.2\height}{\includegraphics[height=3ex]{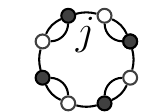}}}} 
\newcommand{\GoWa}{G\hp{8}\raisebox{-.82\height}{\!\!\!\!\!\!\!\!\includegraphics[height=2ex]{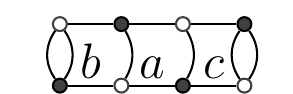}}} 
\newcommand{\GoWb}{G\hp{8}\raisebox{-.82\height}{\!\!\!\!\!\!\!\!\includegraphics[height=2ex]{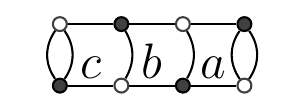}}} 
\newcommand{\GoWc}{G\hp{8}\raisebox{-.82\height}{\!\!\!\!\!\!\!\!\includegraphics[height=2ex]{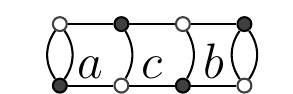}}}
\newcommand{\GoWlji}{G\hp{8}\raisebox{-.82\height}{\!\!\!\!\!\!\!\includegraphics[height=2ex]{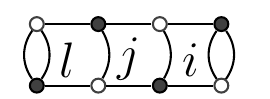}}}

\newcommand{\GoWu}{G\hp{8}\raisebox{-.82\height}{\!\!\!\!\!\!\!\!\includegraphics[height=2ex]{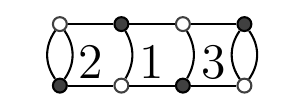}}} 
\newcommand{\GoWd}{G\hp{8}\raisebox{-.82\height}{\!\!\!\!\!\!\!\!\includegraphics[height=2ex]{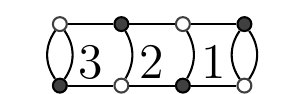}}} 
\newcommand{\GoWt}{G\hp{8}\raisebox{-.82\height}{\!\!\!\!\!\!\!\!\includegraphics[height=2ex]{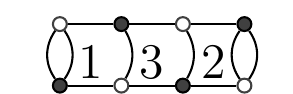}}}

\newcommand{\GoS}{G\hp{8}_{\raisebox{-.2\height}{\!\includegraphics[height=3.7ex]{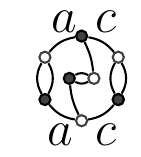}}}} 
\newcommand{\GoSs}{G\hp{8}_{\raisebox{-.2\height}{\!\includegraphics[height=3.1ex]{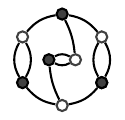}}}} 

\newcommand{\GoSN}{G\hp{8}_{\raisebox{-.2\height}{\!\includegraphics[height=3.7ex]{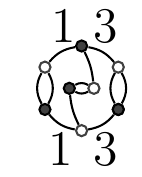}}}} 

\newcommand{\GoCubo}{G\hp{8}\raisebox{-.8\height}{\!\!\!\!\!\!\!\!\includegraphics[height=4ex]{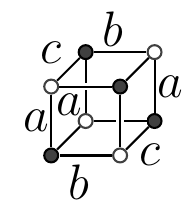}}} 
\newcommand{\GoCuboN}{G\hp{8}\raisebox{-.8\height}{\!\!\!\!\!\!\!\!\includegraphics[height=4ex]{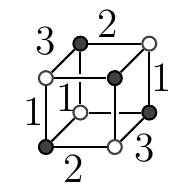}}}

\newcommand{\GoYa}{G\hp{8} \raisebox{-.82\height}{\!\!\!\!\!\!\!\!\includegraphics[height=3ex]{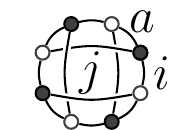}}}
\newcommand{\GoYabc}{G\hp{8} \raisebox{-.82\height}{\!\!\!\!\!\!\!\!\includegraphics[height=3ex]{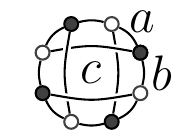}}}
\newcommand{\GoYb}{G\hp{8} \raisebox{-.82\height}{\!\!\!\!\!\!\!\!\includegraphics[height=3ex]{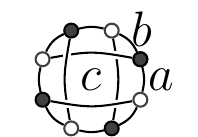}}}
\newcommand{\GoYc}{G\hp{8} \raisebox{-.82\height}{\!\!\!\!\!\!\!\!\includegraphics[height=3ex]{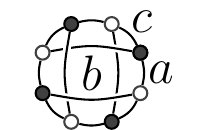}}}
\newcommand{\GoYs}{G\hp{8} \raisebox{-.82\height}{\!\!\!\!\!\!\!\!\includegraphics[height=3ex]{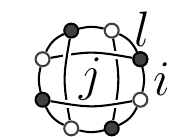}}}
\newcommand{\GoYi}{G\hp{8} \raisebox{-.82\height}{\!\!\!\!\!\!\!\!\includegraphics[height=3ex]{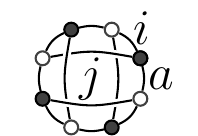}}}

\newcommand{\GoYu}{G\hp{8} \raisebox{-.82\height}{\!\!\!\!\!\!\!\!\includegraphics[height=3ex]{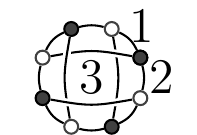}}}
\newcommand{\GoYd}{G\hp{8} \raisebox{-.82\height}{\!\!\!\!\!\!\!\!\includegraphics[height=3ex]{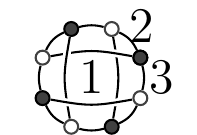}}}
\newcommand{\GoYt}{G\hp{8} \raisebox{-.82\height}{\!\!\!\!\!\!\!\!\includegraphics[height=3ex]{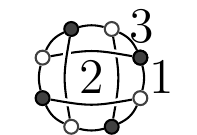}}}

\newcommand{\GoPijs}{G\hp{8} \raisebox{-.77\height}{\!\!\!\!\!\!\!\includegraphics[height=2.73ex]{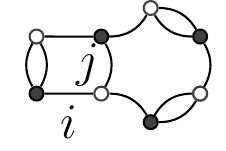}}}
\newcommand{\GoPia}{G\hp{8} \raisebox{-.77\height}{\!\!\!\!\!\!\!\includegraphics[height=2.553ex]{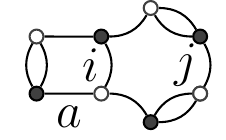}}}
\newcommand{\GoPai}{G\hp{8} \raisebox{-.77\height}{\!\!\!\!\!\!\!\includegraphics[height=2.73ex]{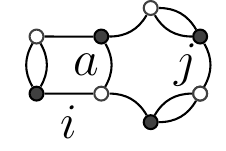}}}
\newcommand{\GoPij}{G\hp{8} \raisebox{-.77\height}{\!\!\!\!\!\!\!\includegraphics[height=2.73ex]{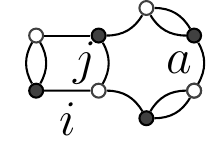}}}
\newcommand{\GoPab}{G\hp{8} \raisebox{-.77\height}{\!\!\!\!\!\!\!\includegraphics[height=2.73ex]{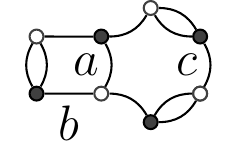}}}
\newcommand{\GoPba}{G\hp{8} \raisebox{-.77\height}{\!\!\!\!\!\!\!\includegraphics[height=2.73ex]{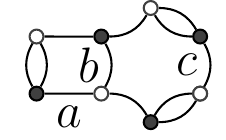}}}

\newcommand{\GoPca}{G\hp{8} \raisebox{-.77\height}{\!\!\!\!\!\!\!\includegraphics[height=2.73ex]{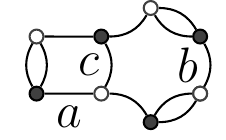}}}
\newcommand{\GoPac}{G\hp{8} \raisebox{-.77\height}{\!\!\!\!\!\!\!\includegraphics[height=2.73ex]{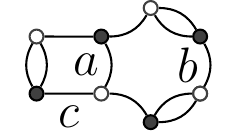}}}
\newcommand{\GoPbc}{G\hp{8} \raisebox{-.77\height}{\!\!\!\!\!\!\!\includegraphics[height=2.73ex]{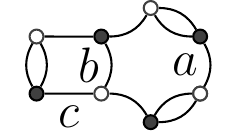}}}
\newcommand{\GoPcb}{G\hp{8} \raisebox{-.77\height}{\!\!\!\!\!\!\!\includegraphics[height=2.73ex]{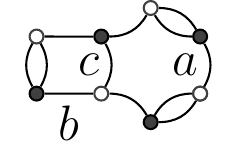}}}

\newcommand{\GoPud}{G\hp{8} \raisebox{-.77\height}{\!\!\!\!\!\!\!\includegraphics[height=2.73ex]{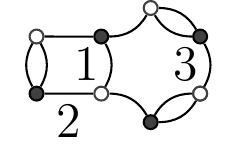}}}
\newcommand{\GoPut}{G\hp{8} \raisebox{-.77\height}{\!\!\!\!\!\!\!\includegraphics[height=2.73ex]{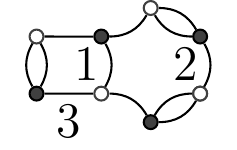}}}

\newcommand{\GoPdt}{G\hp{8} \raisebox{-.77\height}{\!\!\!\!\!\!\!\includegraphics[height=2.73ex]{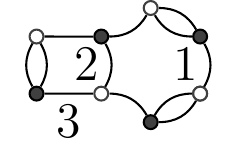}}}
\newcommand{\GoPtd}{G\hp{8} \raisebox{-.77\height}{\!\!\!\!\!\!\!\includegraphics[height=2.73ex]{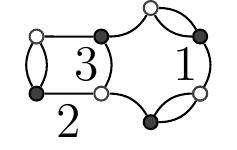}}}

\newcommand{\GoRijl}{G\hp{8}\raisebox{-.73\height}{\!\!\!\!\!\!\!\!\includegraphics[height=3.5ex]{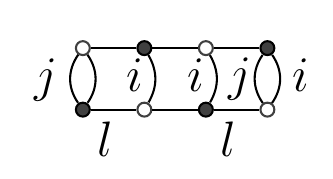}}} 
\newcommand{\GoRij}{G\hp{8}_{\!\!\!\raisebox{-.1\height}{\includegraphics[height=3.3ex]{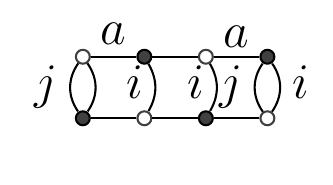}}}} 
\newcommand{\GoRia}{G\hp{8}_{\!\!\!\raisebox{-.1\height}{\includegraphics[height=1.9ex]{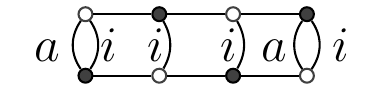}}}} 
 
 \newcommand{\GoRud}{G\hp{8}_{\!\!\!\raisebox{-.1\height}{\includegraphics[height=2.ex]{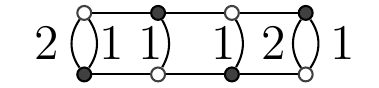}}}} 
  
 \newcommand{\GoRut}{G\hp{8}_{\!\!\!\raisebox{-.1\height}{\includegraphics[height=2.ex]{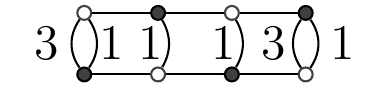}}}} 
  
 \newcommand{\GoRdt}{G\hp{8}_{\!\!\!\raisebox{-.1\height}{\includegraphics[height=2.ex]{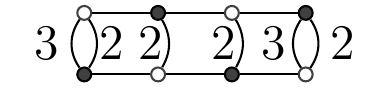}}}} 
 \newcommand{\GoRtd}{G\hp{8}_{\!\!\!\raisebox{-.1\height}{\includegraphics[height=2.ex]{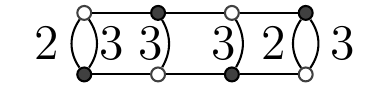}}}}

\newcommand{\GoRcb}{G\hp{8}_{\!\!\!\raisebox{-.1\height}{\includegraphics[height=2.ex]{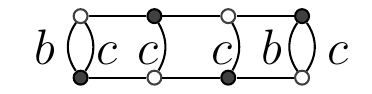}}}} 
\newcommand{\GoRbc}{G\hp{8}_{\!\!\!\raisebox{-.1\height}{\includegraphics[height=2.ex]{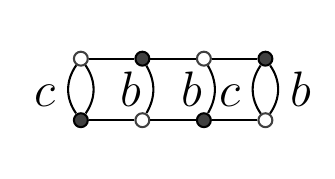}}}} 
\newcommand{\GoRca}{G\hp{8}_{\!\!\!\raisebox{-.1\height}{\includegraphics[height=2.ex]{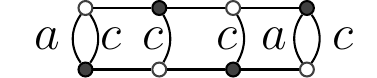}}}} 
\newcommand{\GoRba}{G\hp{8}_{\!\!\!\raisebox{-.1\height}{\includegraphics[height=2.ex]{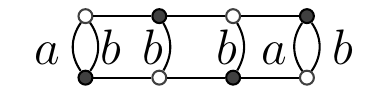}}}} 
\newcommand{\GoRac}{G\hp{8}_{\!\!\!\raisebox{-.1\height}{\includegraphics[height=1.9ex]{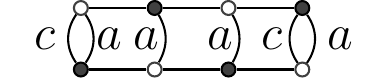}}}} 
\newcommand{\GoRab}{G\hp{8}_{\!\!\!\raisebox{-.1\height}{\includegraphics[height=1.9ex]{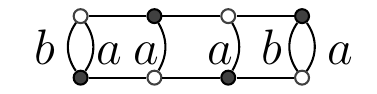}}}} 
\newcommand{\GoRai}{G\hp{8}_{\!\!\!\raisebox{-.1\height}{\includegraphics[height=2ex]{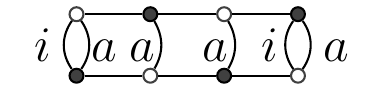}}}} 
\newcommand{\GoXis}{G\hp{8}\raisebox{-.79\height}{\!\!\!\!\!\!\!\!\!  \includegraphics[height=2.99ex]{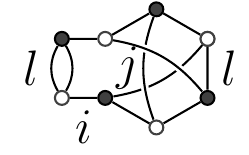}}}
\newcommand{\GoXi}{G\hp{8}\raisebox{-.79\height}{\!\!\!\!\!\!\!\!\!  \includegraphics[height=2.99ex]{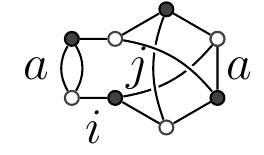}}} 
\newcommand{\GoXu}{G\hp{8}\raisebox{-.79\height}{\!\!\!\!\!\!\!\!\!  \includegraphics[height=2.99ex]{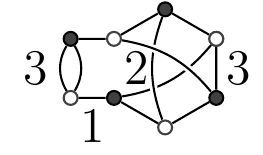}}}
\newcommand{\GoXd}{G\hp{8}\raisebox{-.79\height}{\!\!\!\!\!\!\!\!\!  \includegraphics[height=2.99ex]{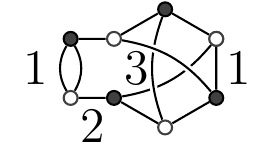}}}
\newcommand{\GoXt}{G\hp{8}\raisebox{-.79\height}{\!\!\!\!\!\!\!\!\!  \includegraphics[height=2.99ex]{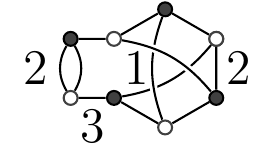}}}
\newcommand{\GoXa}{G\hp{8}\raisebox{-.79\height}{\!\!\!\!\!\!\!\!\!  \includegraphics[height=2.99ex]{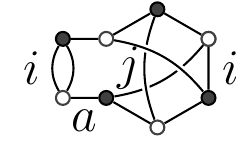}}}
\newcommand{\GoXb}{G\hp{8}\raisebox{-.79\height}{\!\!\!\!\!\!\!\!\!  \includegraphics[height=2.99ex]{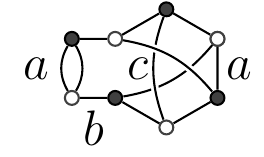}}}
\newcommand{\GoXc}{G\hp{8}\raisebox{-.79\height}{\!\!\!\!\!\!\!\!\!  \includegraphics[height=2.99ex]{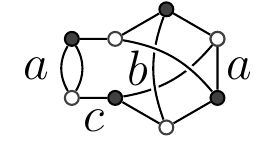}}}
\newcommand{\GoXabc}{G\hp{8}\raisebox{-.79\height}{\!\!\!\!\!\!\!\!\!  \includegraphics[height=2.99ex]{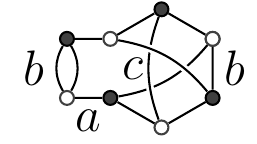}}}
\newcommand{\edgebracket}[1]{\langle \hspace{-2.19pt}\langle #1 \rangle \hspace{-2.19pt}\rangle_{s_a}}
\newcommand{\edgebrackett}[1]{\Big\langle  \hspace{-4.19pt} \Big\langle #1 \Big\rangle \hspace{-4.29pt} \Big\rangle_{s_a}}
\newcommand{\Edgebracket}[1]{\big\langle \hspace{-3.14pt} \big\langle #1 \big\rangle \hspace{-3.1pt} \big\rangle_{s_a}}

\newcommand{\GGmelon}{G\hp{2}_{\raisebox{-.33\height}{\includegraphics[height=2.3ex]{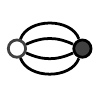}}}}
\newcommand{\GGuno}{G\hp{4}_{\! \raisebox{-.2\height}{\includegraphics[height=2.5ex]{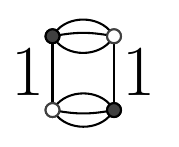}}}}
\newcommand{\GGdos}{G\hp{4}_{\! \raisebox{-.2\height}{\includegraphics[height=2.5ex]{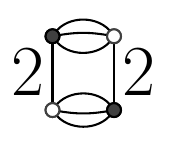}}}}
\newcommand{\GGtres}{G\hp{4}_{\! \raisebox{-.2\height}{\includegraphics[height=2.5ex]{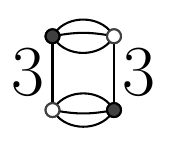}}}}

\newcommand{\GGcuatro}{G\hp{4}_{\! \raisebox{-.2\height}{\includegraphics[height=2.5ex]{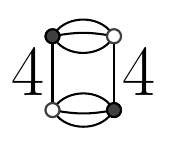}}}}

\newcommand{\Vuno}{\raisebox{-.322\height}{\includegraphics[height=2.3ex]{graphs/4/Icono4_V1v.pdf}}}

\newcommand{\GGcmm}{G\hp{4}_{|\raisebox{-.33\height}{\includegraphics[height=2.2ex]{graphs/4/Icono2_Melon}}
|\raisebox{-.33\height}{\includegraphics[height=2.2ex]{graphs/4/Icono2_Melon}}|}}
\newcommand{\GGcnud}{G\hp{4}_{ \raisebox{-.33\height}{\includegraphics[height=2.2ex]{graphs/4/Icono4_N12}} }}

\newcommand{\sumud}{\suml_{r=1,2}}
\newcommand{\sumtc}{\suml_{p=3,4}}

\newcommand{\Sint}{V}
\newcommand{\fderJ}[1]{\fder{}{J{\phantom{a}}}_{\!\!\!\!\!{#1}}}
\newcommand{\fderbJ}[1]{\fder{}{\bJ{\phantom{a}}}_{\!\!\!\!\!{#1}}}
\newcommand{\zjj}{Z[J,\bJ]}
\newcommand{\ysa}{Y\hp{a}_{s_a}[J,\bJ]}
\newcommand{\jj}{[J,\bJ]}

\newcommand{\Cc}{\mathcal{C}}
\newcommand{\Qc}{\mathcal{Q}}

\newcommand{\wh}{_{\mathrm{w}\vphantom{b}}}
\newcommand{\bl}{_{\mathrm{b}}}
\newcommand{\xb}{\mathbf{x}}
\newcommand{\yb}{\mathbf{y}}
\newcommand{\zb}{\mathbf{z}}
\newcommand{\Xb}{\mathbf{X}}
\newcommand{\Yb}{\mathbf{Y}}
\renewcommand{\sb}{\mathbf{s}}
\newcommand{\tb}{\mathbf{t}}

\usepackage{multicol,array}
\makeatletter
\def\moverlay{\mathpalette\mov@rlay}
\def\mov@rlay#1#2{\leavevmode\vtop{%
   \baselineskip\z@skip \lineskiplimit-\maxdimen
   \ialign{\hfil$\m@th#1##$\hfil\cr#2\crcr}}}
\newcommand{\charfusion}[3][\mathord]{
    #1{\ifx#1\mathop\vphantom{#2}\fi
        \mathpalette\mov@rlay{#2\cr#3}
      }
    \ifx#1\mathop\expandafter\displaylimits\fi}
\makeatother

\newcommand{\cupdot}{\charfusion[\mathbin]{\cup}{\cdot}}

\newcommand{\Dsa}[1]{\Delta_{s_a,#1}}
\newcommand{\Dma}[1]{\Delta_{m_a,#1}}

\usepackage{rotating}

\usepackage{multirow}
\newcommand{\logo}[4]{\raisebox{-.#4\height}{\includegraphics[height=#3 ex]{graphs/3/Logo#1_#2.pdf}}} 
\newcommand{\icono}[4]{\raisebox{-.#4\height}{\includegraphics[height=#3 ex]{graphs/3/Item#1_#2.pdf}}} 
\newcommand{\Logo}[4]{\raisebox{-.#4\height}{\includegraphics[height=#3 ex]{graphs/4/Logo#1_#2.pdf}}} 
\newcommand{\Icono}[4]{\raisebox{-.#4\height}{\includegraphics[height=#3 ex]{graphs/4/Icono#1_#2.pdf}}} 

\newcommand{\meloncito}{\logo{2}{Melon}{2}{24}}
\newcommand{\melon}{\logo{2}{Melon}{3.1}{3}}
\newcommand{\J}{\mathbb{J}}
\makeatletter
\def\leqno{\tagsleft@false}

\usepackage{authblk}
\def\reqno{\tagsleft@false}
\def\fleqn{\@fleqnfalse}
\def\cneqn{\@fleqnfalse}
\makeatother 

\makeatletter
\g@addto@macro{\endabstract}{\@setabstract}
\newcommand{\authorfootnotes}{\renewcommand\thefootnote{\@fnsymbol\c@footnote}}%
\makeatother

\usepackage{fancyhdr} 
 \RequirePackage{fix-cm}
 \DeclareRobustCommand{\gobblefive}[5]{}
\newcommand*{\SkipTocEntry}{\addtocontents{toc}{\gobblefive}}
\usepackage{verbatim}
\begin{document}\vspace{-1cm}
\begin{center}
 \Large \textbf{
Correlation functions of $\mathrm{U}(N)$-tensor models \\ and their Schwinger-Dyson equations}

\vspace{1cm}

 \normalsize
  \authorfootnotes
  
    \footnotetext[0]{\textit{E-mails:} \texttt{romain.pascalie@u-bordeaux.fr},  \texttt{perezsan@uni-muenster.de}$^*$, \texttt{raimar@math.uni-muenster.de}} 
\textsc{Romain Pascalie}\textsuperscript{1,2},  
 \textsc{Carlos I. P\'erez-S\'anchez}\textsuperscript{1,}\footnote{Corresponding author}
  and  \textsc{Raimar Wulkenhaar}\textsuperscript{1} \par \bigskip \bigskip 
% \vspace{1cm}
 \textsuperscript{1}{Mathematisches Institut der Westf\"alischen
  Wilhelms-Universit\"at\\
Einsteinstra\ss{}e 62, D-48149 M\"unster, Germany}
\vspace{.2cm}

\textsuperscript{2}{Universit\'e Bordeaux, LaBRI,
UMR 5800, 33400 Talence, France}

% \vspace{-1.2cm}
\fontsize{11.5}{14.3}\selectfont  
% \vspace{-.41cm}

%v2. 50 pages, 7 TikZ-figures and TikZ graph theory. R. Pascalie added as co-author (see page 17). 
%v3: If generated, v3 coincides with v2 modulo the abstract (the accent in "Gur\{u}au" has been removed, 
%    to make it searchable, and "graph calculus" has been added) and metadata-title: "coloured" replaced by the modern terminology.

\end{center}
\begin{abstract}
 We analyse the correlation functions of $\mathrm{U}(N)$-tensor models (or complex tensor models), which turn out to be classified by boundary graphs, 
 and use the Ward-Takahashi identity and the graph calculus developed in [\textit{Commun. Math. Phys.} (2018) 358: 589] in order to derive the complete tower of exact, analytic Schwinger-Dyson equations for correlation functions with connected boundary graphs.  We write them explicitly for ranks $D=3$ and $D=4$. Throughout, we follow a non-perturbative approach to Tensor (Group) Field Theories.  We propose the extension of this program to the Gur\u{a}u-Witten model, a holographic tensor model based on the Sachdev-Ye-Kitaev model (SYK model).
 \end{abstract}
%  \maketitle
%  \vspace{-.21cm}
\section[sec]{\for{toc}{Introduction}\except{toc}{Introduction}}
% \section[s]{Introduction}
Inspired by matrix models, and being initially in fact very similar to
them, \textit{tensors models} are a natural way to extend the two
dimensional random geometry of matrix models \cite{dFGZ} to higher
dimensions.  As an important application, tensor models aim at
evaluating the partition function of simplicial quantum gravity
\cite{RTM_QG,Koslowski,Ambjorn,track4,GurauRyan,OritiGFT,On,su2,gurauReview} and could be
seen, under mild assumptions, as a generator of a family of
graph-encoded discretisations of the Einstein-Hilbert action, in whose
continuum limit smooth geometries are expected to emerge. \par One of
the initial differences between matrix and tensor models, which turned
out to impair the natural development of the latter, was the
impossibility to derive their $1/N$-expansion.  This problem was
solved by Gur\u au \cite{Nexpansion_coloured}, who introduced a
unitary symmetry, the `colouring' the tensors, that forbids some
unwished contributions in the perturbative expansion (see
Sec. \ref{sec:recap} here and \cite{Gurau:2009tw}). He showed that the
$1/N$-expansion of rank-$D$ tensor models is controlled by an integer
called Gur\u au's degree, associated to each Feynman graph.  This
integer happens to be related to the value of the Einstein-Hilbert
action \cite{critical} on a $D$-dimensional equilateral triangulation
associated to that graph.  The discrete spectrum of Gur\u au's degree
is then the set of values that the Regge discretisation of the general
relativity action can take.  Additional to the initial motivations of
tensor models, new applications to AdS/CFT (also admitting a
$1/N$-expansion \cite{GurauSYKinvN}) via the Sachdev-Ye--Kitaev (SYK)
models have been found in \cite{WittenSYK}; along these lines the
Gur\u au-Witten model \cite{GurauSYK} has been newly proposed.  This
sets the foundations for the so-called holographic tensors.  \par

All these new results enliven the physics of random tensors.  Yet, the
quantum theory of these objects itself deserves a more thorough
mathematical scrutiny, and, in this vein, the present paper is a study
of the correlation functions of complex tensors models (CTM), already
begun in \cite{fullward}, and of the equations they obey (see
\cite{us} as well).  These models \cite{Gurau:2009tw} are sometimes
referred to as \textit{uncoloured tensor models} \cite{uncoloring},
terminology which we avoid here.
% Although the tensor models we work with are not ``coloured'' in the historical sense\footnote{
% We rather work with models addressed in \cite{GurauVirasoro} and
% named in \cite{uncoloring} \textit{uncoloured} tensors. 
% We use the terminology  \textit{coloured} tensor models and remark
% that this should not confused with models with multiple tensor fields
% introduced in \cite{Gurau:2009tw}; in other words, in our models
% the tensor indices are coloured, not the tensors themselves. An another possible name would be 
% $\mathrm U(N)$-\textit{tensor models}, but this terminology would not 
% evoke (as `coloured' tensor models do) the independence of the $\mathrm U(N)$-action
% on the tensor indices.}, 
Rather, since the tensor fields retain some colouring in their
indices, which is a byproduct of an independent $\mathrm
U(N)$-symmetry for each tensor index, we call our models $\mathrm
U(N)$-\textit{tensor models}.  For each symmetry, the partition
function $\zjj$ of a CTM satisfies a full version \cite{fullward} of
the Ward-Takahashi identity \cite{DineWard}. It has been anticipated
\cite{fullward} that this constraint would allow to derive an equation
for each correlation function of complex tensor models, and the aim of
this paper is to obtain those for arbitrary rank.  \par These are the
analytic Schwinger-Dyson equations (SDE).  Their derivation is
independent from existent Schwinger-Dyson equations for tensor models
(e.g. those obtained by Gur\u au \cite{GurauVirasoro} or Krajewski and
Toriumi \cite{krajewskireiko}), which, crucially, differ from our SDE
in that those SDE reported before are algebraic.  That is to say, one
can see the partition function of a CTM, $Z(
\{\lambda_\alpha\}_\alpha) $, as a function of all (possible) coupling
constants $ \{\lambda_\alpha\}_\alpha$.  Whilst
\cite{GurauVirasoro,krajewskireiko} derive recursions for (numerical)
expectation values $ \log Z(
\{\lambda_\alpha\}_\alpha)/\partial\lambda_\gamma $, the framework we
offer here, on the other hand, allows to derive equations for
functional derivatives of $\log Z [J,\bar J]$ with respect to the
sources $J$ and $\bar J$, thus leading to integro-differential
Schwinger-Dyson equations in a quantum field theory context.

\par The connected correlation $2k$-point function of rank-$D$ Tensor
Field theories are usually defined by
\begin{equation}\label{eq:usualDefCorrFunc}
 G\hp{2k}(\xb_1,\dots,\xb_k; \yb_1,\ldots,\yb_k)=
\big(\prod\limits_{i=1}^k \fderJ{\xb_i}\fderbJ{\yb_i} \big) \log(\zjj) 
\bigg|_{J=\bar J=0} \, \quad (\xb_i,\yb_i\in\Z^D). \tag*{$(*)$}
\end{equation}
For CTM, this definition is redundant, when not equivocal (e.g. $G\hp 2 (\xb;\yb)$
identically vanishes outside the diagonal $\xb=\yb$).
In \cite{fullward} we proposed to split each function $G\hp{2k}$ 
in sectors $G_\B\hp{2k}$ that encompass all Feynman graphs indexed by so-called
boundary graphs $\B$ (see Sec. \ref{sec:boundarygraphsexpansion}).
Here $2k$ denotes 
the number of vertices of $\B$ and this integer coincides with the number of 
external legs of the graphs summed in $G_\B\hp{2k}$. 

% \begin{itemize} \item 
There are two reasons to classify correlation functions by boundary
graphs: First, by using these correlation functions one gains a clear
geometric interpretation in terms of bordisms.  Feynman diagrams in
CTMs are coloured graphs, and these represent graph-encoded
triangulations of PL-manifolds. The momentum flux between external
legs of an open graph $\G$ determines its so-called boundary,
$\B=\partial \G$.  Boundary graphs are important because they also
triangulate a manifold, and this manifold coincides with the boundary,
in the usual sense, of the manifold that the original graph
triangulates \cite{GurauRyan}.  Also, by fixing a boundary graph $\B$,
one can sum all connected Feynman graphs that contribute to
$G\hp{2k}_\B$, and these are interpreted as bordisms whose boundary is
triangulated by $\B$; for instance, the connected components,
$\meloncito$ and $\kthree$ of (the graph indexing) $\Gomkthree$
triangulate a sphere and a torus, respectively, and (connected
Feynman) graphs contributing to this correlation function are
triangulations of bordisms $\mathbb S^2\to \T^2$ that are compatible
with their boundary being `triangulated by' $\meloncito\sqcup
\kthree$. \par

Secondly, one must do the splitting of the correlations 
in boundary graphs, otherwise the momenta of the sources
interfere with one another. The correlation functions 
that we propose here 
need only half the arguments of the functions from definition \ref{eq:usualDefCorrFunc}. 
For $k=1,2,3,4$, the connected $2k$-correlation functions indexed by connected \textit{boundary} graphs are
\allowdisplaybreaks[0]
\begin{subequations}
 \begin{align}
 &\Gmelon,\,\\
 & \Guno ,\,\,
 \Gdos ,\,\,
 \Gtres, \\
 &
 \Gskthree,\,
 \Gsqu,\, \Gsqd ,\, \Gsqt ,\, 
 \Gseu ,\,
 \Gsed ,\,
 \Gset,\, %\Gsei
 \\
 &
\GoRijl,\, \GoWlji,\, \GoPijs ,\, \GoXis,\, \GoYs,\, \GoSs,\,\GoAi\,, \GoCubo\,,  \mbox{ and $\Sym_3$-orbits thereof}. 
 \end{align}
 \end{subequations}  \allowdisplaybreaks 
 Also, functions like $\Gsmi$ and $\Gomkthree$, indexed by
 disconnected graphs, need to be considered. None of these graphs is a
 Feynman graph: in fact we will not deal with them
 here\footnote{Except two Feynman diagram examples appearing in
   Sec. \ref{sec:recap} and in Fig. \ref{fig:8pt}.}, since we proceed
 non-perturbatively. \par
% \item 
%  \end{itemize}
 To these two reasons, we add as motivation the success that 
 this treatment gave for matrix models \cite{GW12}. There, by  
 splitting in boundary components,  
 the matricial Ward identity was exploited and combined with the Schwinger
 Dyson equations. This allowed to derive an integral equation for the  
 quartic matrix models and, in the planar 
 sector, finally solve for all correlation functions in terms of the two point function \cite{GW12} 
 via algebraic recursions. Here, we import these techniques to the CTM setting.\par
   
 In this article we derive the full tower of equations that correspond
 to connected boundary graphs.  We also obtain the 2-point and some
 higher-point Schwinger-Dyson equations (SDE) in an explicit form
 rank-$3$ and rank-$4$ theories.  Section
 \ref{sec:boundarygraphsexpansion} recalls the setting of complex
 tensor models in a condensed fashion, and the expansion of the free
 energy in boundary graphs. The Ward-Takahashi Identity (WTI)
 \cite{fullward} for complex tensors, which we recall in Section
 \ref{sec:functionals}, is a fundamental auxiliary and bases on this
 boundary graph expansion.  There, we also introduce language to deal
 with the proper derivation of the full SDE-tower in Section
 \ref{sec:SDE}. We continue with the derivation of the SDE-equations
 for quartic rank-$3$ theories (Sec. \ref{sec:rank3SDE}) and rank-$4$
 theories (Sec. \ref{sec:rankfour}; moreover, rank-$5$ are shortly
 addressed in App. \ref{sec:rankfive}). \par
 
 In order to derive the SDE for a certain $2k$-point function it is
 necessary to know, also to order $2k$ in the sources, the form of
 certain generating functional (for rank $3$, Lemma
 \ref{thm:ordersix}, with proof located in
 App. \ref{sec:proofoflemma}) which appears in the Ward-Identity.
 This requires knowledge of the free energy to order $2(k+1)$ (in the
 sources), which in turn needs information about all the graphs with
 this number of vertices and their coloured automorphism groups. Later
 on, in Section \ref{sec:rankfour} we find the SDE for rank-$4$
 theories with melonic quartic vertices. Explicitly, only the
 two-point functions and $4$-point functions are obtained, since the
 graph theory in four colours is much more complicated. Section
 \ref{sec:simple} presents a model that has simpler SDEs and looks
 solvable, since, as shown there, it posses a very similar expansion
 in boundary graphs.  It is a tensor model that can be used to study
 the random geometry of $3$-spheres.  \par A short section before the
 conclusions analyzes the boundary graphs of the Gur\u au-Witten
 (SYK-like) model and sets a plan to extend to it the CTM-methods
 developed in previous sections.
 
%  \clearpage 
We motivate a non-linear reading of this article. 
The dependence of the sections is sketched 
by means of lines in the next diagram, the dashed ones 
meaning weak dependence.
 \vspace{1cm}
% \bigskip

\[
\hspace{-3cm}
\includegraphics[width=8cm]{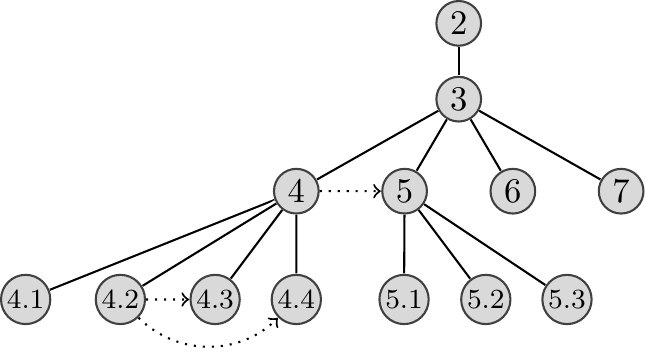}
\]

\bigskip
\SkipTocEntry\section*{Acknowledgments}  
  C.I.P.S and R.W. thank the SFB 878  (DFG Sonderforschungsbereich \textit{Groups, Geometry \& Actions}) 
  for financial support; so does R.P. for hospitality. 
  R.P. is supported by the `CNRS Infiniti ModTens grant'.
  C.I.P.S. thanks additionally the DAAD  (\textit{Deutscher Akademischer Austauschdienst})  
  for financial support in the beginning of this work, 
  and the 
  \textit{Mainz Institute for Theoretical Physics}
  (MITP), where this paper was concluded,
  for support, office hours and hospitality during
  the workshop \textit{Foundational and structural aspects of gauge theories}. 
  We thank two referees for extremely useful remarks. 
%   R.P is acknowledged for careful reading of the manuscript and for 
%   pointing out a mistake in (a previous incomplete version of) Theorem \ref{thm:SDEs} in the first preprint version.
%   C.I.P.S and R.W. think this correction is enough to add R.P. as a coauthor. 
% \clearpage
% \addtocontents{toc}{\setcounter{tocdepth}{-1}}
\newpage
\fontsize{11.5}{18}\selectfont  
 \tableofcontents
\fontsize{11.5}{14.5}\selectfont

% \addtocontents{toc}{\setcounter{tocdepth}{2}}

%   \vspace{1cm}
\clearpage
\section{Boundary graph expansions} \label{sec:boundarygraphsexpansion}

This section rapidly introduces the notation in graph theory and
recapitulates previous results that are relevant in our present study.
There are few examples in Fig. \ref{fig:graphexamples} that are
intended as support to rapidly grasp the next definitions.  Also the
rather panoramic Table \ref{table:graphs} organizes the concepts
introduced below.

\subsection{Coloured tensors and coloured graphs} \label{sec:recap}
Let $N$ be a (large) integer, thought of as an energy scale, and 
consider $D$ distinguished representations, $(\H_1,\rho_1),\ldots,(\H_D,\rho_D)$ of $\mathrm U(N)$. 
A complex tensor model is concerned with the quantum theory of tensor fields 
$\phi,\bar\phi: \H_1\otimes \H_2\otimes\ldots\otimes\H_D\to \C$
whose components transform under said $D$ representations as 
\begin{align*}
\phi_{x_1\ldots x_D}& \mapsto \phi'_{x_1\ldots x_D}= \sum_{y_a}[\rho_a(W_a)]_{x_ay_a} \phi_{x_1\ldots y_a\ldots x_D}\,, \\
\bar\phi_{x_1\ldots x_D}& \mapsto \bar\phi'_{x_1\ldots x_D}= \sum_{y_a}[\overline{\rho_a(W_a)}]_{x_ay_a} \bar\phi_{x_1\ldots y_a\ldots x_D} \,, 
\end{align*}
for all $W_a\in\mtr{U}(N)$ and being each $x_a$ and $y_a$ in suitable
index-sets $I_a\subset \Z$, for each integer (or \textit{colour})
$a=1,\ldots,D$. Usually one sets $\H_a=\C^N$ or $\H_a=\ell^2[-n,n]$
for suitable $n=n(N)$, and $\rho_a=\mathrm{id}_{\H_a}$ for each colour
$a$. However, at the same time, one insists that the representations
are distinguished, so that indices are anchored to a spot assigned by
its colour. Thus, the indices of the tensors have no symmetries
(e.g. $\phi_{ijk}=\phi_{ikj}$ is forbidden) and only indices of the
same colour can be contracted. \\

A particular tensor model is specified by two additional data: a finite subset  
of interaction vertices given by real monomials in $\phi$ and $\bar\phi$ that 
are $\mtr{U}(N)$-invariant under the chosen $D$ representations; 
 the second data is a quadratic form 
\[S_0(\phi,\bar\phi)=\Tr_2(\bar\phi,E\phi)=\sum_{\xb} \bar\phi_\xb E_\xb \phi_\xb,\,\, 
\mbox{ for certain function }E: I_1\times\ldots \times I_D \to
\re^+\,, \] determining the kinetic term $S_0$ in the classical
action.  Sums are (implicitly) over the finite lattice $
I_1\times\ldots \times I_D\subset \Z^D$.  These $I_a$ sets depend
usually on a cutoff scale related to $N$ and we will assume, also
implicitly, that throughout they are all $\Z$, keeping in mind that
one needs to regularize.

In order to characterize the interaction vertices, one uses
vertex-bipartite regularly edge-$D$-coloured graphs, or, in the
sequel, just `$D$-\textit{coloured graphs}'.  A graph $\G$ being
\textit{vertex-bipartite} means that its vertex-set $\G\hp0$ splits
into two disjoint sets $\G\hp0=\G\hp0\wh \cupdot \G\hp0\bl$.  The set
$\G\hp0\wh$ (resp. $\G\hp0\bl$) consists of \textit{white}
(resp. \textit{black}) vertices.  The set of edges, denoted by
$\G\hp1$ is split as $\G\hp1 = \cupdot_a \G\hp1_a$ into $D$ disjoint
sets $\G\hp1_a$ of $a$-coloured edges, $a=1,\ldots,D$.  Given any edge
$e$, the white and black vertices $e$ is attached at, are denoted by
$s(e) \in\G\wh\hp 0$ and $t(e)\in\G\hp 0 \bl$, respectively.  This
defines the maps $s,t:\G\hp 1 \to \G\hp0 $.  Regularity of the
colouring means that, for each $v\in\G\hp 0\wh$ and each $w\in\G\hp 0
\bl$, both preimages $s\inv (v)$ and $t\inv(w)$ consist precisely of
$D$ edges of different colours.  By regularity, the number of white
and black vertices is the same and is equal to $k(\G):=\# \, \G \hp 0
/2$.  The set of (closed) $D$-coloured graphs is denoted by $
\vGrph{D}$.  \par

The only way to obtain monomials in the fields $\phi$ and $\bar\phi$
that are also invariants, is contracting each coordinate index
$\phi_{\ldots x_c\ldots }$ by a delta $\delta_{x_c y_c}$ with the
coordinate $\bar\phi_{\ldots y_c\ldots }$ of the respective colour of
the field $ \bar\phi$. The imposed $\mtr{U}(N)$-invariance requires
then $D\cdot k(\G)$ such coloured deltas. One thus associates to each
occurrence of $\phi$ a white vertex $v$ and to each occurrence of
$\bar\phi$ a black vertex $w$. For each colour $c$, to each
$\delta_{x_c y_c}$ contracting $\phi_{\ldots x_c\ldots }$ and
$\bar\phi_{\ldots y_c\ldots }$ one draws a $c$-coloured edge which
starts at $v$, $v=s(e)$, and ends at $w$, $w=t(e)$.  Thus, any
invariant monomial $ \Tr_\B$ is fully determined by a coloured graph
$\B$, and vice versa.  For instance, the trace
$\sum_{\mathbf{a,b,c,p,q,r}} (\bar
\phi_{r_1r_2r_3}\bar\phi_{q_1q_2q_3} \bar\phi_{p_1p_2p_3}) \cdot
(\delta_{a_1p_1} \delta_{a_2r_2} \delta_{a_3q_3} \delta_{b_1q_1}
\delta_{b_2p_2} \delta_{b_3r_3} \delta_{c_1r_1} \delta_{c_2q_2}
\delta_{c_3p_3}) \cdot (\phi_{a_1a_2a_3} \phi_{b_1b_2b_3}
\phi_{c_1c_2c_3})$ is depicted in \ref{fig:k33}.
                                                
\begin{figure} 
    \centering
    \begin{subfigure}[b]{0.25\textwidth} 
    \centering
    \includegraphics[width=.86\textwidth]{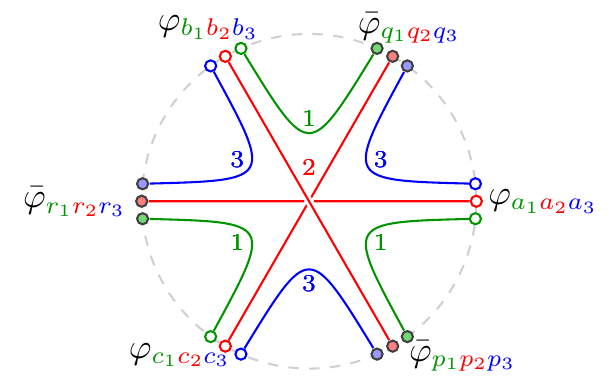}
        \includegraphics[width=.45\textwidth]{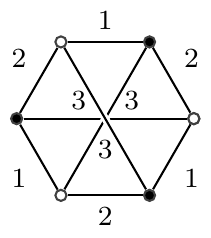}
        \caption{Example of the relation between traces and monomials. This graph
        is denoted by $K_{\mtr c}(3,3)$}
        \label{fig:k33}
    \end{subfigure}
    \qquad %add desired spacing between images, e. g. ~, \quad, \qquad, \hfill etc. 
      %(or a blank line to force the subfigure onto a new line)
    \begin{subfigure}[b]{0.3\textwidth}\centering
     \includegraphics[width=.79\textwidth]{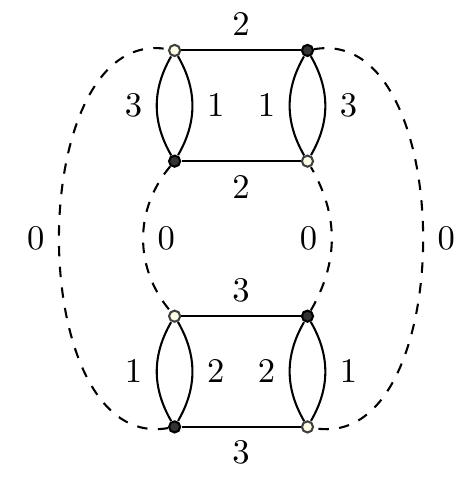}
        \caption{This graph is in $\vfey_3(\phi^4)$, i.e.
        it is a vacuum Feynman graph of the $\phi^4_3$-model, 
         $V(\phi,\bar\phi)= \lambda (\vuno + \vdos  + \vtres )$ }
        \label{fig:feynmansimple}
    \end{subfigure}
    ~ \,\,%add desired spacing between images, e. g. ~, \quad, \qquad, \hfill etc. 
    %(or a blank line to force the subfigure onto a new line)
    \begin{subfigure}[b]{0.3\textwidth}
        \includegraphics[width=\textwidth]{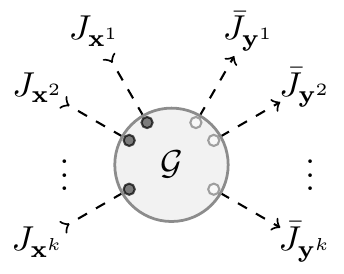}
        \caption{Anatomy of a Feynman graph and how it determines boundary graph $\B$, which 
        induces the map $\B_*:(\xb^1,\ldots,\xb^k)\mapsto (\yb^1,\ldots,\yb^k)$}
        \label{fig:combinatorics}
    \end{subfigure}

  \bigskip
  
      \begin{subfigure}[b]{0.25\textwidth}

     \centering
     \includegraphics[width=.85\textwidth]{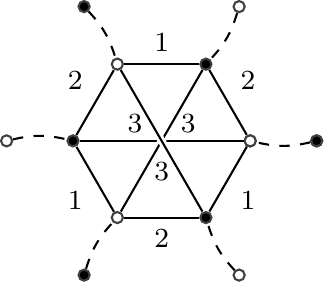}
             \vspace{.5cm}
       \caption{Open graph with boundary $\kc$ (in fact it is the 
        cone of $\kc$) but is not in $\fey_3(\phi^4)$}
        \label{fig:coned}
    \end{subfigure}
    \quad %add desired spacing between images, e. g. ~, \quad, \qquad, \hfill etc. 
      %(or a blank line to force the subfigure onto a new line)
    \begin{subfigure}[b]{0.3\textwidth}
        \includegraphics[width=.9\textwidth]{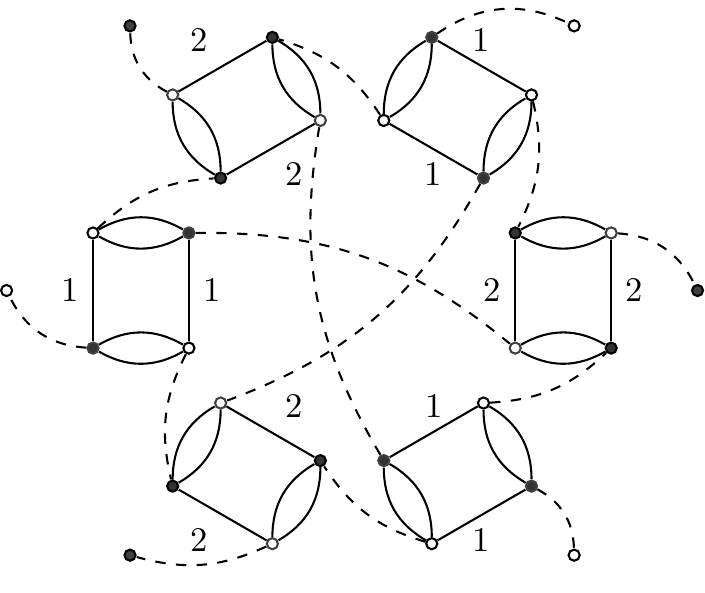}
        \caption{ This graph $\R$ is an open graph in $\fey_3(\phi^4)\subset  \Grph{3+1}\hp 6 \subset \Grph{3+1}$ 
        with $\partial\R=\kc$ (see explanation in \ref{fig:amputated})}
        \label{fig:integratedk33}
    \end{subfigure}
    ~\,\, %add desired spacing between images, e. g. ~, \quad, \qquad, \hfill etc. 
    %(or a blank line to force the subfigure onto a new line)
    \begin{subfigure}[b]{0.3\textwidth} \centering
        \includegraphics[width=.85\textwidth]{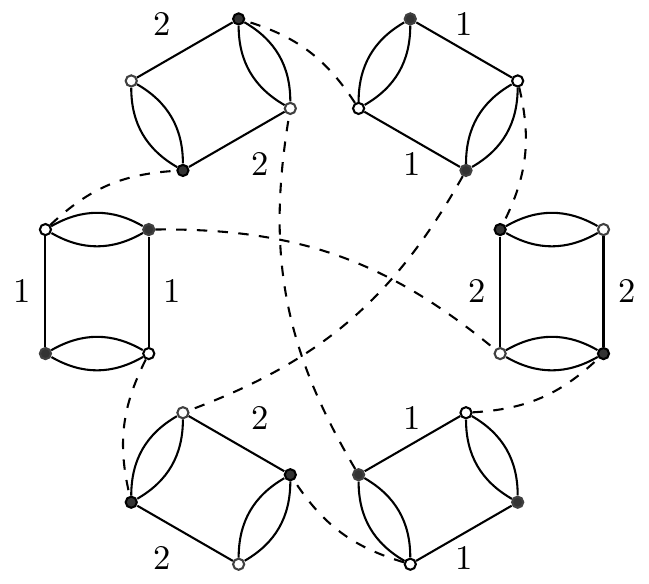}
        \caption{The amputation of $\R$. If one 
        erases the $0$-coloured (or dashed) edges, one gets connected
        components in $\{\vuno,\vdos,\vtres\}$ }
        \label{fig:amputated}
    \end{subfigure}
    \caption{Graph terminology of Sec. \protect \ref{sec:recap} and concerning examples \label{fig:graphexamples}}
\end{figure}

Any model is then given by an interaction potential $V(\phi,\bar\phi)=
\sum_{\B\in\Lambda} \lambda_\B\Tr_{\B}(\phi,\bar\phi)$, for $\Lambda$
a finite subset of $\vGrph{D}$. For a fixed model $S=S_0+V$, one can
write down the corresponding partition function:
\begin{equation} \label{eq:partfunction}
 Z[J,\bar J]=Z_0 \int\Df [\phi, \bar\phi] \,\ee^{\Tr_2{(\bar J,\phi)}+
\Tr_2{(\bar \phi , J)}-N^{D-1}S[\phi,\bar\phi]} 
%  {\int\Df [\phi, \bar\phi]\, \ee^{-N^{D-1}S[\phi,\bar\phi]}} 
\,,  \,\, 
% \mbox{ with }
\Df [\phi, \bar\phi]:=\!\!\!\prod\limits_{\,\,\,\,\mathbf {x} \in
\Z^D} \!\!\! N^{D-1}\frac{\dif\varphi_{\mathbf x} 
\dif\bar\phi_{\mathbf{x}}}{2\pi \ii}
% \ee^{-\Tr_{2}(\varphi,\bar\phi)}
\,. 
\end{equation} 
Here $S_0=\Tr_2(\bar\phi,\phi)$ is the only quadratic invariant,
namely $\melon$.  Later on, at the level of propagator, we will allow
this invariance to be broken (see
Sec. \ref{sec:boundarygraphsexpansion}). \par Using Wick's theorem one
evaluates the contributions to the generating functional. Wick's
contractions (\textit{propagators}) are assigned a new colour, $0$,
which one commonly draws as dashed line.  For (complex) matrix models
($D=2$), this $0$ colour would be the ribbon line propagator, thus,
for tensors, this colour $0$ substitutes a cumbersome notation of $D$
parallel lines.  It is easy to see that Feynman vacuum graphs of
rank-$D$ complex tensors are vertex-bipartite regularly
edge-$(D+1)$-coloured graphs, now the colours being the integers from
$0$ to $D$. Vacuum graphs can be connected or disconnected. The set of
strictly disconnected graphs is denoted by $\strictdvGrph{D+1}$ and
$\dvGrph{D+1}$ denotes the set of possibly disconnected graphs.  We
assume that any Feynman graph is connected and get rid of
\textit{Feynman} graphs in $\strictdvGrph{D+1}$ by working with the
free energy, $W\jj=\log (\zjj)$, rather than with the partition
function. \par Since we are mainly interested in the
\textit{connected} correlation functions we have to consider
\textit{open} Feynman graphs, i.e. graphs with $n$ external legs, each
of which is attached to a tensorial source, $J$ or $\bar J$, that
obeys the same transformation rules of the field $\phi$ or $\bar\phi$,
respectively. The external legs are exceptional edges of valence-1
white (for the source $J$) or black (for $\bJ$) vertices. All external
legs' edges have colour $0$.  Clearly, because of bipartiteness, this
number has to be even, $n=2k$. We denote by $\Grph{D+1}\hp{2k}$ the
set of Feynman diagrams with $2k$ external legs and further set
\[\Grph{D+1}=\cup_{k=1}^\infty \Grph{D+1}\hp{2k} \cup \vGrph{D+1} \, \]
generically for open or closed ($D+1$)-coloured graphs. \par 
Importantly, not every graph in $\Grph{D+1}$ is a Feynman graph. The set of 
Feynman graphs of a model 
$V(\phi,\bar\phi)= \sum_{\B\in\Lambda} \lambda_\B\Tr_{\B}(\phi,\bar\phi)$ is denoted by $\fey_D(V(\phi,\bar\phi))$ or $\fey_D(V)$.
This set consists of the graphs in $\Grph{D+1}$ that satisfy the following condition: 
 after amputating all external legs and removing all the $0$-coloured 
 edges, the remaining graph has connected components in 
 the set of interaction-vertices $ \Lambda \subset  \vGrph{D}$ 
 (see Figs. \ref{fig:integratedk33}, \ref{fig:amputated}).

\begin{table}\centering
 \includegraphics[width=1.03\textwidth]{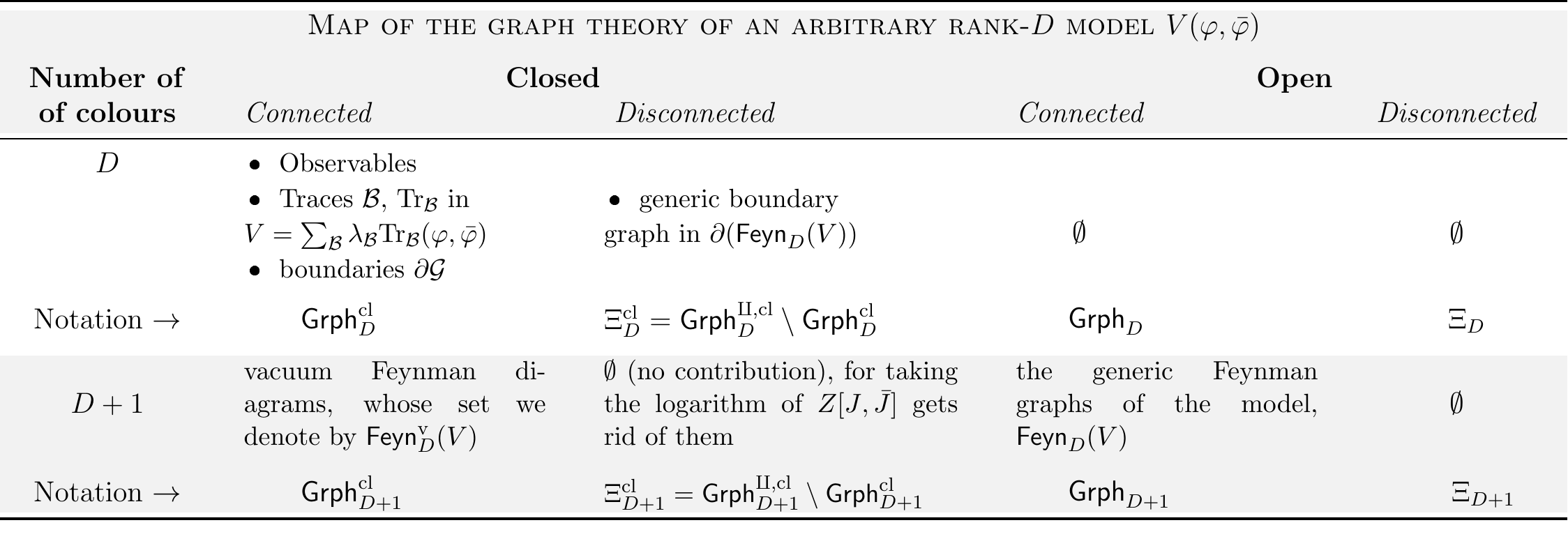}
  \caption{\label{table:graphs} Terminology of
  Sec. \ref{sec:boundarygraphsexpansion}, and why both  
  graphs in $D$ and $D+1$ number of colours appear in a rank-$D$ models,
  and which their respective roles are. Here `disconnected' strictly means 
  `not connected'. The notation $\emptyset$
  stands for `no contributions to/no role in the rank-$D$ theory'.}
\end{table}

\subsection{Boundary graphs}\label{sec:boundary_graphs_subsection}
There is a boundary map $\partial: \Grph{D+1} \to  \dvGrph{D} $, 
which for all $\G \in \Grph{D+1}$ is given by
\begin{align*}
(\partial\G)\hp 0&:= \{\mbox{external legs of } \G\}\,, \\
(\partial\G)\hp 1_a&:=\{(0a)\mbox{-bicoloured paths between external legs in }\G \} \,.
\end{align*}
The vertex set inherits the bipartiteness from $\G$, to wit a vertex
in $(\partial\G)\hp 0$ is \textit{black} if it corresponds from an
external line attached to a white vertex, and \textit{white} if it is
attached to a black vertex in $\G$.  The edge set is regularly
$D$-coloured $(\partial\G)\hp 1= \cupdot_a (\partial\G)\hp 1_a$.  \par
For a fixed model $V(\phi,\bar\phi)$, the image of the restriction
$\partial_V:=\partial|_{\fey_D(V) }$ of $\partial$ to $\fey_D(V) $ is
deemed \textit{boundary sector}, and this set is, of course, model
dependent. A graph in the boundary sector is a \textit{boundary
  graph}.  For melonic quartic theories, as a matter of fact
\cite{fullward}, this boundary map is surjective, so all (possibly
disconnected) $D$-coloured graphs are boundaries. Thus, all the
correlation functions we propose have non-trivial
contributions. Incidentally, this means that quartic coloured random
tensor models are able to ponder probabilities of triangulation of all
bordisms, provided they exist, as in dimension $d$ ($d=D-1=2,3$) as
classical objects (oriented manifolds); in presence of obstructions,
there are pseudo-manifolds yielding those bordisms.  \par

Given a closed coloured graph $\B$, $\Autc(\B)$ denotes the set of its
\textit{coloured automorphisms}.  These are graph maps $\B\to \B$ that
preserve adjacency, the bipartiteness of $\B\hp 0$ and also its
edge-colouring. Each automorphism of $\B$ arises from a lifting of an
element $\pi$ of $\mathrm{Sym} (\B\hp0 \wh) =\Sym_{k(\B)}$ to a unique
map $\hat\pi:\B\to\B$, as one can easily see, determined by the
preservation of said structure.  Figures \ref{fig:graphs3} and
\ref{fig:graphs4} show all the automorphism groups for graphs having
up to $8$ vertices in $D=3$ and up to $6$ vertices for $D=4$,
respectively.\par

\begin{figure} 
\includegraphics[width=.96\textwidth]{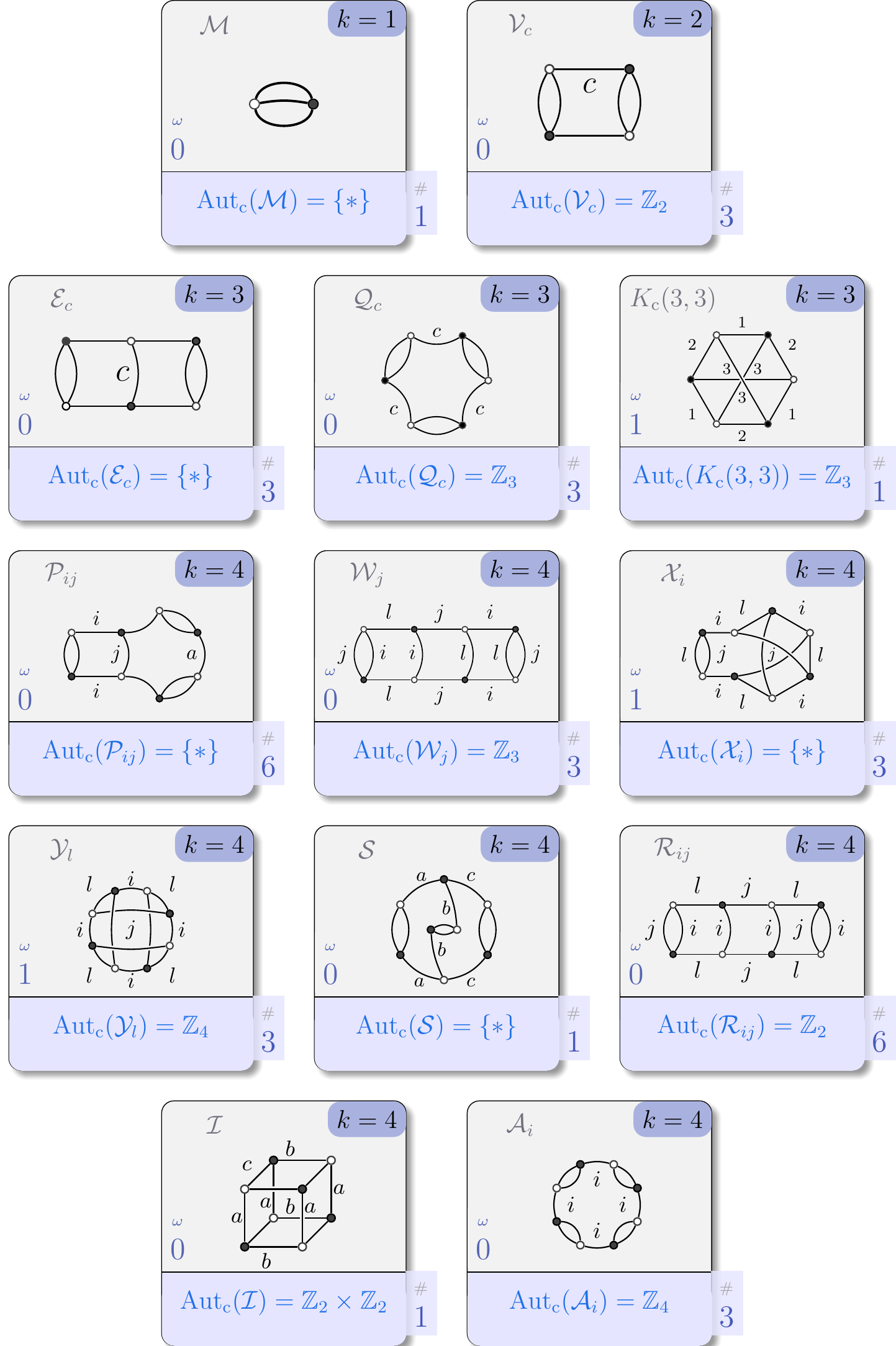}
\caption{ 
Enumeration of $3$-coloured graphs with $2,4,6$ and $8$ vertices and their 
Gur\u au-degree $\omega$ and coloured automorphism group. 
\label{fig:graphs3}}
\end{figure}
\begin{figure} 
\includegraphics[width=.84\textwidth]{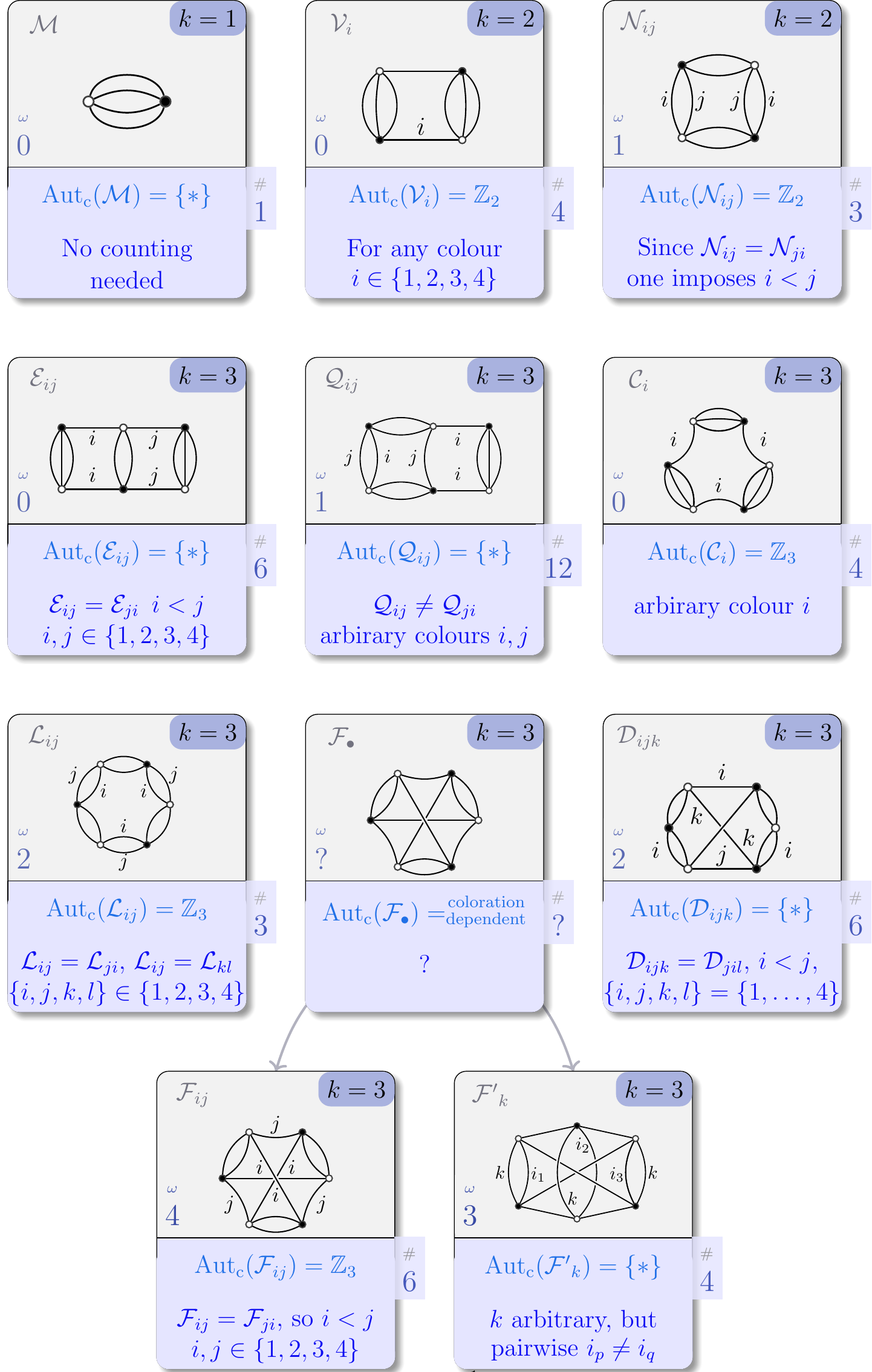}
\caption{\small This table shows the rank-$4$ graphs until 6
  vertices. As before, $\#$ is the number of graphs that are obtained
  by the action of $\Sym_4$ in the edge-colouring and $\omega$ is
  Gur\u au's degree of the graph in question.  For graphs with $k=4$
  see \protect \cite[Fig. 8]{KlebanovTarnopolsky} (there, only those marked
  with a $B$ are bipartite). The graphs displayed there are neither
  given a colouration nor classified by $\Sym_4$-orbits, though
  (Klebanov and Tarnopolsky treat them as vacuum Feynman graphs; here
  our graphs are boundaries and we need to count them and their
  $\Autc$-groups)
\label{fig:graphs4}}
\end{figure}

We shall assume that both the white vertex-set
$\B\hp0\wh=(v^1,\ldots,v^{k(\B)})$ as well as the black vertex-set
$\B\hp0\bl=(w^1,\ldots,w^{k(\B)})$ of a boundary graph $\B$ are given
an ordering.  Then $e_a^{v^\mu}$, the edge of colour $a$ attached to a
white vertex $v\in \B\hp 0\wh$, i.e. $s(e^{v^\mu}_a)=v^\mu$, is
denoted by $e^\mu_a$.  \par Let $\B$ be a boundary graph and
$k=k(\B)$. Then $\B$ induces a map\footnote{What this paper concerns,
  the use of matrices $ M_{D\times k}(\Z)$, instead of plainly $
  \Z^{kD}$, merely eases the definition of $\mathcal F_{D,k}$ below.
  No matrix multiplication is so far needed.}  $\B_*: M_{D\times k}
(\Z) \to M_{D\times k} (\Z)$ by $\Xb=(\xb^1,\ldots,\xb^{k})\mapsto
\B_*(\Xb)=(\yb^1,\ldots,\yb^{k})$, where $y^\alpha_a =x^\nu_a$ (for
$\alpha=1,\ldots,k$) if and only if there exists an $a$-coloured edge
starting at $v^\alpha$ and ending at $w^\nu$.  Regularity of the
colouring and bipartiteness of the vertex set ensure that there is
exactly one such edge, thus rendering $\B_*$ well-defined.  This map
$\B_*$ is deduced by momentum transmission inside any graph $\G$ for
with $\partial \G=\B$ by following the $a0$-coloured paths in $\G$
between its external vertices.  One further associates to $\B$ and
$\Xb$ a cycle of sources
\begin{equation} \label{eq:Jcycle}
\J(
\B) (\Xb)= J_{\xb^1}\cdots J_{\xb^ k} \bJ_{\yb^1}\cdots \bJ_{\yb^k} 
\,,  \qquad\qquad \mbox{where }\B_*(\Xb)=(\yb^1,\ldots,\yb^k)\,,
\end{equation}
which is evidently independent of the ordering given to $\B\hp0\wh$ and $\B\hp0\bl$.  
According to \cite{fullward}, the free energy $W\jj=\log (\zjj)$ can be expanded 
in these cycles indexed by all the boundary graphs of a given model:
 \begin{align} \label{expansion_Wgeneral}
 W[J,\bar J] & = \sum\limits_{l=1}^\infty
  \sum\limits_{\substack{\B \in \, \im \,\partial_V  \\ 
k(\B)=l}}  \frac{1}{|\Autc(\B)|} 
G^{(2 l)}_{\B} \star \mathbb J(\B) 
% = \sum\limits_{k=1}^\infty
%   \sum\limits_{\substack{\B \in \partial\fey_D(V(\phi,\bar\phi))  \\ 
% k=\frac12\#(\B^{(0)})}}  \frac{1}{|\Autc(\B)|}  (G\cup \mathbb{J})
% (\B)\, . 
\end{align}
where $\star$ is a pairing between a function $f:M_{D\times k(\B)} \to
\C $ and a boundary graph $\B\in \im \,\partial_V \subset \dvGrph{D}$
given by $f\star \J(\B)=\sum_{\mathbf \Xb \in M_{D\times k(\B)}(\Z)}
f(\Xb)\cdot\J(\B)(\Xb)$.  To read off the the correlation functions
$G\hp{2l}_\B$ from eq. \eqref{expansion_Wgeneral}, one takes graph
derivatives, introduced in \cite{fullward} and recapitulated in the
next section.

% \comment{Tell why the matrix notation for $\mathcal{F}_{D,k}$ has been chosen. Try to 
% let the colour action and the action of $\Sym_k$ act on it and deduce something from 
% the freeness of the latter. Tell also that the matrix structure (algebraic) is not 
% needed and one could as well consider just $\Z^{k \cdot D}$ directly. }

\subsection{Graph-generated functionals}
\label{sec:functionals} 
We also recall some results from \cite{fullward}.  Let 
\[
% \boldsymbol{\Delta}_{D,k}=
\mathcal{F}_{D,k}:=
\{(\yb^1,\ldots,\yb^k) \in M_{D\times k} (\Z) \,|\, y^\alpha_{c}\neq y^\nu_c \mbox{ for all } 
c=1,\ldots,D \mbox{ and } 
\alpha,\nu=1,\ldots, k, \alpha\neq \nu\}\,. \] 
 
Thus $\mathcal{F}_{D,k}$ is the set of matrices $M_{D\times k} (\Z) $
having all different entries on any fixed row. 
We define the graph derivative of any functional $X[J,\bJ]$ with respect to 
$\B$ at $\Xb \in \mathcal F_{D,k}$ as 
\[
\dervpar{ X[J,\bJ]}{\B(\Xb)}:=\fder{^{2k(\B)}X[J,\bJ]}{(\J(\B))(\Xb)}\bigg|_{J=0=\bJ}=
\prod\limits_{\alpha=1}^k \fder{}{J_{\xb^\alpha}} \fder{}{\bJ_{\yb^\alpha}} X\jj \bigg|_{J=0=\bJ}\,.
\]
Let $(\yb^1,\ldots,\yb^l) \in M_{D\times l}(\Z)$. 
For closed, coloured graphs $\mathcal Q,\Cc \in \vGrph{D}$ one has \cite{fullward}:
\begin{align}\label{graph_calc_disc}
\frac{\partial  \mathcal Q (\yb^1,\ldots,\yb^ l)  }
{\partial\mathcal \Cc (\xb^1,\ldots,\xb^ k)}=\left.
\begin{cases} 
\sum\limits_{\hat\sigma\in\Autc(\mathcal C)} \delta_{lk}\cdot 
\delta^{\mathbf{y}^{\sigma(1)},\ldots,\,\yb^{\sigma (k)}} 
      _{\xb^1,\xb^2, \ldots,\xb^k} & 
      \,\, \mbox{if } \,\, \Cc \cong   \mathcal{Q} \, 
 \\
 \!\! \qquad 0 & \,\, \mbox{otherwise}\, 
\end{cases} \right\}=
% \begin{cases} 
\sum\limits_{\sigma\in \Sym_k} 
\delta^{\mathbf{y}^{\sigma(1)},\ldots,\,\yb^{\sigma (k)}}
      _{\xb^1,\xb^2, \ldots,\xb^k} \delta (\Qc,\Cc)
% \end{cases}
\end{align}
where $\delta (\Qc,\Cc)=1 $ if the graphs $\Qc$ and $\Cc$ are
isomorphic, and $0$ otherwise. We consider functionals generated by a
given family of closed $D$-coloured (non-isomorphic) graphs, $\Upsilon
\subset \dvGrph{D}$.  That means that if
\begin{equation}\label{eq:expansionX}
X\jj = \suml_{\Cc \in \Upsilon}
\mathfrak{l\hspace{.5pt}}_{\Cc} 
\star
\mathbb J (\Cc),\, \qquad \qquad  \mbox{for }\mathfrak{l\hspace{.5pt}}_{\Cc}  : (\Z^D)^{\times k(\Cc)}  \to \C, \,\,\Cc\in \Upsilon\,,
\end{equation}
is known, we want to know the graph derivatives of $X\jj$ with 
respect to connected graphs. Here $k(\Cc)$ denotes  
the number $\#(\Cc\hp0\wh)$  of white (or black) vertices of $\Cc$.
\begin{prop} \label{thm:promedio_grupo}
Let $X$ be as in eq. \eqref{eq:expansionX}. Then, for all $\mtc C\in
\Upsilon \cap \vGrph{D}$, the functions
$\mathfrak{l\hspace{.4pt}}_{\Cc} $ satisfy
\[
\dervpar{X\jj}{\, \Cc(\mathbf X) }= 
\suml_{\hat\sigma \in \Autc(\Cc)} (\sigma^*\mathfrak{l\hspace{.5pt}}_{\Cc} )(\mathbf{X}),
~\where~~ (\sigma^*\mathfrak{l\hspace{.5pt}}_{\Cc} )(\xb^1,\ldots,\xb^{k(\Cc)}):= 
\mathfrak{l\hspace{.6pt}}_{\Cc}(\xb^{\sigma\inv(1)},\ldots,\xb^{\si\inv(k(\Cc))})\,,\quad
\]
for all $\Xb=(\xb^1,\ldots,\xb^{k(\Cc)}) \in \mathcal F_{D,k(\Cc)}  $. 
\end{prop}

\begin{proof}
From formula \eqref{graph_calc_disc}, one has
\begin{align*}
\dervpar{X\jj}{(\Cc(\Xb))} 
& = 
\dervpar{}{\,\Cc(\Xb)}
\suml_{\mathcal Q \in \Upsilon}
\mathfrak{l\hspace{.5pt}}_{\mathcal Q} 
\star
\mathbb J (\mathcal Q) 
=
\suml_{\mathcal Q \in \Upsilon}
\suml_{\Yb\in(\Z^D)^k}
\mathfrak{l\hspace{.5pt}}_{\mathcal Q}(\Yb) 
\dervpar{\,\mathcal{Q}(\Yb)}{\,\Cc(\Xb)}
\\
& =
\suml_{\mathcal Q \in \Upsilon}\,
\suml_{\Yb\in(\Z^D)^{\times k(\mathcal{Q})}}
\mathfrak{l\hspace{.5pt}}_{\mathcal Q}(\Yb) 
\suml_{\sigma \in \Sym_{ k(\mathcal{Q})}}
\delta(\mathcal{Q},\Cc)
\prod\limits_{i=1}^{k(\mathcal{Q})} 
\delta(\yb^{\sigma(i)},\xb^{ i })
      \\
& =
\suml_{\mathcal Q \in \Upsilon}\,
\suml_{\Yb\in(\Z^D)^{\times k(\Cc)}}
\mathfrak{l\hspace{.5pt}}_{\mathcal Q}(\Yb) 
\suml_{\sigma \in \Sym_{k(\mathcal Q)}}
\delta(\mathcal{Q},\Cc)
\prod\limits_{i=1}^{k(\mathcal Q)} 
\delta(\yb^{i},\xb^{\sigma\inv(i)})
\\
& =
\suml_{\mathcal Q \in \Upsilon}\,
% \suml_{\Yb\in(\Z^D)^{\times k(\mathcal{Q})}}
% \mathfrak{f}_{\mathcal Q}(\Yb) d
\suml_{\sigma \in \Sym_{k(\mathcal Q)}}
\mathfrak{l\hspace{.5pt}}_\Cc(\xb^{\sigma\inv(1)},\ldots,\xb^{\si\inv(k)})
\delta(\mathcal{Q},\Cc)\,.
% \prod\limits_{i=1}^{k(\Cc)} 
% \delta(\yb^{i},\xb^{\sigma\inv(i)})
\end{align*}
Since $\Upsilon$ consists only of graphs that are not isomorphic,
the sum over $\Qc$ yields, because of the delta $\delta(\mathcal{Q},\Cc)$,  only one term. Hence,
the last expression is precisely the sum over automorphisms of $\Cc$.
\end{proof}
As a consequence of this, one can recover the correlation functions via 
\[
G_\B\hp{2k(\B)}(\Xb)= 
\dervpar{W\jj}{\B(\Xb)}\,.  
\]
Notice that $\Xb=(\xb^1,\ldots,\xb^ k)\in \mathcal{F}_{D,k}$ if and only if
$\B_*(\Xb) \in \mathcal{F}_{D,k}$. 
Since $W\jj$ is real-valued, one has the relation 
\begin{align}
 \overline{G_\B\hp{2k(\B)}}(\Xb)= \prod\limits_{\alpha=1}^k   \fder{}{\bJ_{\xb^\alpha}} 
 \fder{}{J_{\yb^\alpha}} \overline{W\jj} \Big|_{J=0=\bJ} = G_{\bar\B}\hp{2k(\B)} (\B_*(\Xb))\,  \qquad (\Xb\in\mathcal{F}_{D,k(\B)})\,,  
\end{align}
where $\bar\B$ is essentially the graph $\B$ after inverting vertex-colouration, $\bar \B \hp0 \wh= \B \hp0 \bl$ and 
$\bar \B \hp0 \bl= \B \hp0 \wh$, but otherwise with the same adjacency and edge-colouration.

We now explain how this graph derivatives are relevant in the
WTI. The WTI is rather a set of equations, one for each 
 colour $a=1,2,\ldots, D$, in which a new generating functional %(a ``correction'')
%  $Y\hp {a}\jj$ 
 of the form 
 \begin{equation} \label{eq:Yterm_expansionOmega}
Y\hp a_{s_a}\jj = \sum_{ \Cc \in \Omega_V} \mtf{f}_{\Cc,s_a}\hp a \star \mathbb J (\mathcal C)  \qquad\qquad (s_a\in I_a\subset \Z)
 \end{equation}
appears. Here, $\partial_V: \fey_D(V) \to \dvGrph{D}$ denotes the
boundary map in terms of which we describe the graph family $\Omega_V$
as follows: If $e^{v}_a$ is the $a$-coloured edge at the white vertex
$v\in\B\hp0 \wh$, then the graph $\B\ominus e_a^v$ denotes the graph
that is obtained by the next steps: first, remove the two
end-vertices, $v=s(e_a^v)$ and $t(e_a^v)$, of $e_a^v$; then, remove
all their common edges $I(e_a^v):={s\inv ( s(e_a^v) ) \cap
  t\inv(t(e^v_a))}$; finally, glue colourwise the broken edges,
i.e. the each broken edge of the set $s\inv(v) \setminus I(e_a^v) $
with the respective broken edge in $t\inv (t(e^v_a)) \setminus
I(e_a^v)$. Then $\Omega$ is defined by
\[
 \Omega_V :=  \{  \B\ominus e_a^v   \,|\, \B\in \mathrm{im}\, \partial_V ,  v\in \B\hp 0\wh  \}\,. 
\] 

\begin{defn} 
\begin{figure}\centering\centering
\begin{subfigure}{0.5\textwidth}
\includegraphics[width=.87\linewidth]{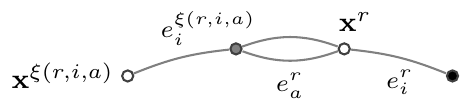}
\caption{ Locally, edge and vertex labeling 
in $\B$ before forming $\B\ominus e_a^r$} \label{fig:colapso_2a}
\end{subfigure}
\hspace*{\fill} 
\begin{subfigure}{0.34\textwidth}
\includegraphics[width=.9\linewidth]{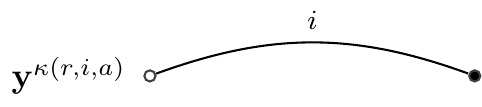} 
% \vspace{.0cm}
\caption{ Locally $\B\ominus e_a^r$} \label{fig:colapso_2b}
\end{subfigure}
\caption{ Some notation concerning the definition of 
  $\Delta_{s_a,r}^\B$.   \label{fig:colapso_2}}
\end{figure}

Let $a$ be a colour, $F:(\Z^D)^k\to \C$ a function
  and $\B\in \dvGrph{D}$. For any integer $r$, $1\leq r\leq
  k(\B)$, we define the function $\Delta_{s_a,r}^\B F:
  (\Z^D)^{k-1}\to \C$ by
\[
(\Delta_{s_a,r}^\B F)(\Yb)=
% \begin{cases}
 \sum\limits_{\substack{q_h  }}
F(\mathbf y^1,\ldots,\mathbf y ^ {r-1}, {\mathbf z}^{r}(s_a,\mathbf q,\Yb),\mathbf y^r,
\ldots,\mathbf y^{k-1} )\,,
% \end{cases}
\]
for each $\Yb=(\mathbf y^1,\ldots, \mathbf y^{k-1}) \in (\Z^D)^{k-1}$, 
where the sum is over a dummy variable $q_h$ for each 
element of the set $h\in I(e_a^r)\setminus\{a\}$. Before 
specifying ${\mathbf z}^{r}(s_a,\mathbf q,\Yb)$, we stress
that this sum can be empty, in which case 
\[(\Delta_{s_a,r}^\B F)(\Yb)=F(\mathbf y^1,\ldots,\mathbf y ^ {r-1}, {\mathbf z}^{r}(s_a,\Yb),\mathbf y^r,
\ldots,\mathbf y^ k )\,. \]
The momentum ${\mathbf z}^ r\in\Z^D$ has entries defined by:
\[
z_i^ r(s_a,\mathbf q,\Yb)=
\begin{cases}
s_a & \qquad \mbox{ if } \,\, i=a\,,\\
q_i & \qquad \mbox{ if }\,\, i\in I(e_a^r)\setminus \{a\}\,,\\
y_i^{\kappa(r,i,a)} & \qquad \mbox{ if }\,\, i \in \mbox{colours of } A_{t(e_a^r)} =
\{1,\ldots, D\}\setminus I(e_a^r) \,,
 \end{cases}
\]
where $\mathbf y^{\kappa(r,i,a)}$ ($1\leq \kappa(r,i,a)< k$) is the
white vertex $\B\ominus e^r_a$ defined by
\begin{equation}
\kappa(r,i,a)= \begin{cases}
                \xi(r,i,a)  & 
                \mbox{ if } \xi(r,i,a) < r\,, \\
                \xi(r,i,a)-1   & 
                \mbox{ if } \xi(r,i,a) > r\,.
               \end{cases} \label{def:kappa}
\end{equation}
(see also
Fig. \ref{fig:colapso_2}).
This definition depends on the labeling of the vertices. However, the
pairing $\langle\!\langle G\hp{2k}_{\B} , {\B} \rangle\!\rangle_{s_a}$
defined as follows does not, for it is a sum over graphs after 
removal of \textit{all} $a$-coloured edges:
\begin{equation}
 \langle\!\langle 
G\hp{2k}_{\B} , 
{\B} \rangle\!\rangle_{s_a} :=  
\sum\limits_{r=1}^k
\Big(\Delta_{s_a,r}^\B G\hp{2k}_{\B}\Big)\star \mathbb J
{(\B\ominus e_a^r)}\label{parejas} \,.
\end{equation}  
\end{defn}

\begin{rem} \label{rem:ordering_vertices}
Unless otherwise stated, we set the convention 
of ordering the white-vertex-set $\B\hp 0\wh$ 
in appearance from left to right.
\end{rem}

\begin{example}
Let $\{a,b,c,d\}=\{1,2,3,4\}$ and  \vspace{-.3cm} 
\[ \vspace{-.1cm}
\mathcal F_c'=
\raisebox{-.4\height}{
 \includegraphics[width=2.1cm]{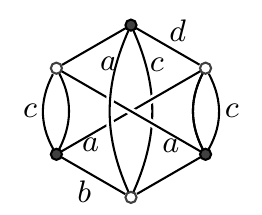}}
\]
 (see also Fig. \ref{fig:graphs3}).  For a fixed colour $a$ and
$s_a\in I_a\subset \Z$, we obtain $\langle\!\langle G\hp{2k}_{\mtc
  F'_c} , {\mtc F'_c} \rangle\!\rangle_{s_a}$.  According to remark
\ref{rem:ordering_vertices}, the first white vertex is the left upper
left white vertex, the second is the lowermost, the third is the upper
right. This orders the $a$-coloured edges
$\{e_a^1,e_a^2,e_a^3\}$. Explicitly, 
\[\mathcal F_c'=
\raisebox{-.4\height}{
 \includegraphics[width=2.1cm]{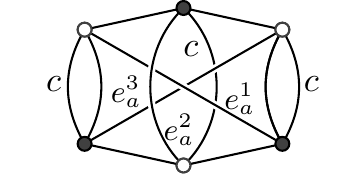}} ,
\mbox{ thus, }
\mathcal F_c'\ominus e^1_a
=\mathcal F_c'\ominus e^3_a
=
\raisebox{-.4\height}{\!\! \includegraphics[width=1.4cm]{graphs/4/Logo_N}} , 
\quad 
\mathcal F_c'\ominus e^2_a=
\raisebox{-.34\height}{\! \includegraphics[width=1.4cm]{graphs/4//LogoVa}} 
\]
hence 
 \begin{align*}
% \hspace{-1cm}
\big\langle\hspace{-.60ex} \big\langle 
G\hp 6_{\mathcal F'_c}, 
\mathcal F'_c
\big\rangle \hspace{-.60ex}\big\rangle _{s_a} =
& \,\,
\Dsa{1}G\hp 6_{\mathcal F'_c} \star \mathbb J \big(
{ \raisebox{-.4\height}{\includegraphics[width=1.1cm]{graphs/4/Logo_N}}} 
\big)
% \\
% & \,\,
+
\Dsa{2}G\hp 6_{\mathcal F'_c} \star \mathbb J \big(
{ \raisebox{-.34\height}{\includegraphics[width=1.0cm]{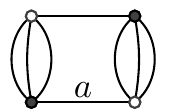}}} 
\big)
% \\
% &\,\,
+
\Dsa{3}G\hp 6_{\mathcal F'_c} \star \mathbb J \big(
{ \raisebox{-.4\height}{\includegraphics[width=1.1cm]{graphs/4/Logo_N}}} 
\big)
\end{align*}
which, in turn, equals 
\allowdisplaybreaks[0]
 \begin{align*}
 & \suml_{\yb,\mathbf z}  \big\{
  \Dsa{1} G\hp 6_{\mathcal F'_c}(\mathbf y,\mathbf z) 
  \J \big(
{ \raisebox{-.34\height}{\includegraphics[width=1.0cm]{graphs/4/Logo_N}}} \!
\big) (\yb,\mathbf z)
+ \Dsa{2} G\hp 6_{\mathcal F'_c}(\mathbf y,\mathbf z) \,\J \big(
{ \raisebox{-.28\height}{\includegraphics[width=.88cm]{graphs/4/LogoVa}}} 
\big) (\mathbf y,\mathbf z)
\\
& \qquad  + 
 \Dsa{3} G\hp 6_{\mathcal F'_c}(\mathbf y,\mathbf z)\, \J
  \big(\!
{ \raisebox{-.34\height}{\includegraphics[width=1.0cm]{graphs/4/Logo_N}}} 
\!\big) (\yb,\mathbf z) \big\}
\\
 & =  \suml_{\yb,\mathbf z}  \big\{
  G\hp 6_{\mathcal F'_c}(s_a,z_b,z_c,y_d , \yb,\mathbf z)  \bJ_{y_a z_b y_c z_d} \bJ_{z_a y_b z_c y_d} J_{\yb} J_{\mathbf z}
  +\big(\suml_{q_c} 
  G\hp 6_{\mathcal F'_c}(\mathbf y,  s_a,y_b,q_c,z_d ,\mathbf z) 
 \\ 
 & \qquad \times
 \bJ_{y_a z_bz_cz_d} \bJ_{z_a y_b y_c y_d} J_{\yb} J_{\mathbf z} \big)+
 G\hp 6_{\mathcal F'_c}(\yb,\mathbf z, s_a,z_b,y_c,y_d )
  \bJ_{y_a z_b y_c z_d} \bJ_{z_a y_b z_c y_d} J_{\yb} J_{\mathbf z}
 \big\} \,.
\end{align*}
\allowdisplaybreaks
We assume all the entries of momenta in $\Z^4$ are 
ordered by colour, e.g. $(z_1 y_4 z_3 y_2)$
really means $(z_1y_2z_3y_4)$.
\end{example}

We now recall the full Ward-Takahashi Identity, proven in \cite{fullward}.
\begin{thm}\label{thm:full_Ward}
Consider a rank-$D$ tensor model, $S=S_0+V$, with a kinetic form  $\Tr_2(\bar\phi,E\phi)$ such that the difference of propagators
$
E_{p_1\ldots  p_{a-1}m_ap_{a+1}\ldots p_D}-
E_{p_1\ldots  p_{a-1}n_ap_{a+1}\ldots p_D}
=E(m_a,n_a) % \delta_{m_an_a}  \
$  is independent of the momenta  
$\mathbf p_{\hat a}=(p_1,\ldots, \widehat p_a,\ldots ,p_D)$. Then that model has a partition function
$Z[J,\bJ]$ that satisfies
\begin{align} \label{ward2}
 &\sum\limits_{\mathbf p_{\hat a}} \fder{^2 Z[J,\bar J]}{J_{{p_1\ldots  p_{a-1}m_ap_{a+1} \ldots p_D}} 
\delta\bar J_{{p_1\ldots  p_{a-1}n_ap_{a+1}\ldots p_D}}} 
 -  \left(\delta_{m_an_a} Y\hp a_{m_a}[J,\bar J]\right) \cdot Z[J,\bar J]
\\
=& \sum\limits_{\mathbf p_{\hat a} }
\frac1{ 
E_{ {p_1\ldots  m_a \ldots p_D}}
-E_{ {p_1\ldots   n_a  \ldots p_D}}  } 
 \left( 
\bar J_{p_1\ldots  m_a \ldots p_D} 
\fder{}{\bar J}_{p_1\ldots   n_a \ldots p_D}  
-J_{p_1\ldots   n_a  \ldots p_D}\nonumber 
\fder{ }{J}_{p_1\ldots   m_a  \ldots p_D}\right)Z[J,\bar J]
\end{align}   
where \vspace{-9pt}
\begin{align} \label{eq:termY}
Y\hp a_{m_a}[J,\bar J] &:= \sum\limits_{l=1}^\infty
  \sum\limits_{\substack{\B \in \im \,\partial_V  \\ 
k(\B)=l}}  \langle\!\langle 
G\hp{2l}_{\B} , 
{\B} \rangle\!\rangle_{m_a}  \,.
% & =
% \sum\limits_{l=1}^\infty
%   \sum\limits_{\substack{\B \in \im \,\partial  \\ 
% k(\B)=l}}   
% \sum\limits_{r=1}^k
% \left(\Delta_{m_a,r}^\B G\hp{2k}_{\B}\right)\star \mathbb{J}
% {(\B\ominus e_a^r)}\,. \nonumber
\end{align} 
\end{thm} 
There is a subtlety regarding the ordering of the vertices. We
associate an ordering of the white vertices of a graph $\B$ in $
G_\B\hp{2k}$.  The $k$ arguments (in $\Z^D$) of this function match
this vertex-ordering.  But the edge-removal sometimes will yield a
graph which should be reoriented. To illustrate this, for $D=3$,
consider for instance the next graph $\mtc S$.  The edge contraction
yields, for any $i=1,2,3$, the following:
\begin{equation}\label{eq:orderexample}
  \mbox{ If } 
\mathcal S =\!\!\!\!\logo{8}{S}{12}{42} \mbox{ then } \mathcal S\ominus e_i^1=\!\!\!\!\logo{8}{S}{12}{42} \ominus e_i^{1}=\!\!\!\!\raisebox{-5.0ex}{\includegraphics[height=11ex]{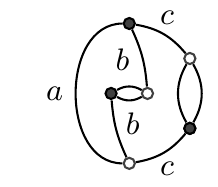} } 
\end{equation}

As a graph, $\mathcal S\ominus e_i^1$ is just $\icono{6}{Eabc}{3.9}{2}$, but when one 
considers $f \,\star\, \mathbb J( \mathcal S\ominus e_i^1)$, for some
function $f:(\Z^3)^{\times 3} \to \C$, 
the order of the vertices does matter:
\[
f\star
\J  \Big(\!\!\!\!\!\!\raisebox{-3.30ex}{\includegraphics[height=8ex]{graphs/3/Logo4S} }\!\!\Big)=((12)^*f) \star \J (\logo{6}{Eabc}{5.2}{3})
\]
In going from the graph \ref{eq:orderexample}
to $\icono{6}{Eabc}{3.9}{2}$, one permuted the first and second white vertices. 
Accordingly, one `corrects' $f$ and replaces it by $(12)^*(f)$. 
Notice that the cycle $(12)\in\Sym_3$ does not lift to a
coloured automorphism. If this was the case, we could just as well 
ignore the correction.

% $\mathfrak{a b c d e f g h i j k l m n o p q r s t u v w x y z}$\par 
% $\mathfrak{A B C D E F G H I J K L M N O P Q R S T U V W X Y Z}$
The next definition is needed in order to describe 
some terms appearing in the SDEs.
\begin{defn}
Let $\B\in \vGrph{D}$ and let $v,w$ be vertices of the same colour
(either both black $v,w\in \B\bl\hp0$ or both white
$v,w\in\B\wh\hp0$).  We define the graph $\varsigma_a(\B;v,w)$ as the
coloured graph obtained from $\B$ by swapping the $a$-coloured edges
at $v$ and $w$.  Usually, vertices in boundary graphs are indexed by
numbered momenta $v=\xb^\alpha,w=\xb^\gamma \in\Z^D$, in which case we
write $\varsigma_a(\B;\xb^\alpha,\xb^\gamma)$ or just
$\varsigma_a(\B;\alpha,\gamma)$. These graphs are, generally,
disconnected.
% That is $\varsigma_a(\B;v,w)$ has the same vertices as $\B$, with the same bipartiteness 
\end{defn}

\begin{example} For any colour $a=1,2,3$, one has 
$\varsigma_a(\kthree;u,v) = \logo{6}{Eabc}{4.5}{25}=:\mathcal E_a$ for two black (or white) vertices $u,v$ of $\kthree$. 
If $x$ and $y$ are the leftmost black vertices of $\mathcal E_a$, then
$\varsigma_b(\mathcal E_a;x,y)=\melon  \sqcup \logo{4}{Vc}{3.3}{33}$.  
\end{example}

\section{The Schwinger-Dyson equation tower in arbitrary rank}\label{sec:SDE}
% \subsection{A useful formula for the derivative of $\Sint(\phi,\bar\phi)$}\label{sec:usefulformula}
We pick the following quartic model $S=S_0+\Sint$, with interaction vertices 
$\Sint[\phi,\bar\phi]=\lambda \sum_{a=1}^D{\Tr_{\mathcal V_a}}(\phi,\bar\phi)$,
being each vertex $\V_a$ the melonic vertex of colour $a$, 
\begin{equation} \label{eq:DefVa}
\V_a=\logo{4}{VaD}{7.9}{4}
 \end{equation}
Moreover, assume that the propagator obeys that, for each colour $a$,
the following difference
\[
E(t_a,s_a):=E_{p_1\ldots t_a\ldots p_D} - E_{p_1\ldots s_a\ldots p_D}\, 
\]
does not depend on $p_i$, for each $i\neq a$. Such is the case for
Tensor Group Field theories, say with group $\mtr U(1)$, being the
origin of $E$ is the Laplacian operator on $\mtr U(1)^D$ after taking
Fourier transform, and the tensors the Fourier modes. We
call this model the $\phi_{D,\mathrm{m}}^4$-theory\footnote{For $D=3$
  all quartic invariants are melonic, so we refer to it only as
  $\phi_{3}^4$-theory.}.  Here, the subindex `$\mathrm{m}$' denotes
melonicity. For specific choices of propagators and 
renormalizability, see \cite{Geloun:2016bhh} 
and references therein (additionally \cite{wgauge,Geloun:2014ema,Carrozza:2012uv}).

\par One observes that, if $\delta(\V_a)(\mathbf
b,\mathbf c,\xb,\yb)$ is the invariant of the trace, that
is \[\Tr_{\V_a}(\phi,\bar\phi)= \lambda\sum\limits_{\mathbf b,\mathbf
  c,\xb,\yb} \bar\phi_{\mathbf b} \bar \phi_{\mathbf c} \delta(\V_a)
(\mathbf b,\mathbf c,\xb,\yb) \phi_{\mathbf y}\phi_{\mathbf x} \,,
\]
one gets for $\sb=(s_1,\ldots,s_D)\in\Z^D$, the following expression:
\begin{align} 
 \label{eq:Sint}
\left( \dervpar{\Sint(\phi,\bar\phi)}{\bar \phi_\sb} \right) _{\phi^\flat\to\, \delta/\delta{J^\sharp}}
 &=
 2\lambda \bigg\{ \sum\limits_a \Big(\sum\limits_{b_a}
 \fder{}{\bJ\phantom{a}}_{\!\!\!\!\!s_1\ldots s_{a-1}b_as_{a+1}\ldots s_D} \sum\limits_{\mathbf b_{\hat a}}
  \fderJ{b_1\ldots b_D}
  \fderbJ{b_1\ldots b_{a-1} s_a b_{a+1} \ldots b_D}
  \Big)\bigg\}\,,   
\end{align}
where $\mathbf b_{\hat a}=(b_1,\dots \widehat{b}_a,
\ldots,b_D)=(b_1\ldots,b_{a-1},b_{a+1}\ldots,b_D)\in\Z^{D-1}$ and
$\sharp$ can either act trivially on a variable or be complex
conjugation, and $\phi^\flat=\bar\phi$ or $\phi^\flat= \phi$ according
to whether $J^\sharp=\bar J$ or $J^\sharp=J$, respectively.  The term
$\evalat{\big(\partial \Sint(\phi,\bar\phi)/\partial\bar \phi_\sb}
{\phi^\flat\to\, \delta/\delta{J^\sharp}}\big) Z[J,\bJ]$ can be
computed with aid of the WTI.  We depart from the formally integrated
form of the partition function
\[
Z\jj \propto \evalat{\exp{( -\Sint(\phi,\bar\phi) )}}{\phi^\flat \to \delta/\delta J^\sharp} \exp\Big(\sum_{\mathbf q\in \Z^D} \,\bar J_{\mathbf q} E\inv _{\mathbf q}J_{\mathbf q} \Big)\,,
\]
where we will ignore a (possibly infinite) constant and write equality
and derive its logarithm: 
\begin{align} \label{eq:WnachJ_s}
\fder{ W\jj }{\bJ_\sb}
& 
= \frac{1}{\zjj}
 \exp{( -\Sint(\delta/\delta \bar J,\delta/\delta J) )}  J_\sb E\inv_\sb \mathrm{e}^{ \sum_{\mathbf q\in \Z^D} \,\bar J_{\mathbf q} E\inv _{\mathbf q}J_{\mathbf q} }\\
& = \nonumber
\frac{1}{\zjj}  
\bigg\{
\frac{1}{E_{\sb}} J_{\sb}  \exp{( -\Sint(\delta/\delta \bar J,\delta/\delta J) )}  \mathrm{e}^{ \sum_{\mathbf q\in \Z^D} \,\bar J_{\mathbf q} E\inv _{\mathbf q}J_{\mathbf q} }
\\
\nonumber 
& \quad\qquad\qquad +
\bigg(\dervpar{\Sint(\phi,\bar\phi)}{\bar\phi_\sb}\bigg)\bigg|_{\phi^\flat \to \delta/\delta J^\sharp} 
 \evalat{\exp{( -\Sint(\phi,\bar\phi) )}}{\phi^\flat \to \delta/\delta J^\sharp} \mathrm{e}^{ \sum_{\mathbf q\in \Z^D} \,\bar J_{\mathbf q} E\inv _{\mathbf q}J_{\mathbf q} }
\bigg\}\, \\
\nonumber
& 
=
 \frac{1}{E_{\sb}} 
 \Bigg\{ J_\sb - \frac{1}{\zjj} \bigg(\dervpar{\Sint(\phi,\bar\phi)}{\bar\phi_\sb}\bigg)\bigg|_{\phi^\flat \to \delta/\delta J^\sharp} 
 Z[J,\bar J] \Bigg\} \,.
\end{align}
For sake of notation, we introduce the shorthands 
$\mathbf b_{\hat a} s_a =(b_1\ldots,b_{a-1},s_a,b_{a+1}\ldots,b_D) $
and, similarly,  $
\mathbf s_{\hat a} b_a =(s_1\ldots,s_{a-1},b_a,s_{a+1}\ldots,s_D)$, for any $a=1,\ldots,D$.  
By applying the colour-$a$-WTI to the rightmost double derivative term 
appearing in \eqref{eq:Sint}, the following:
\begin{align} \label{eq:SintPostWard}
\bigg(\dervpar{\Sint(\phi,\bar\phi)}{\bar\phi_\sb}\bigg)\Bigg|_{\phi^\flat\to\, \delta/\delta{J^\sharp}} Z[J,\bJ]&=
 2\lambda 
 \sum\limits_a
 \bigg\{ 
 \sum\limits_{b_a}
 \fderbJ{\sb_{\hat a } b_a}
 \Big(
 \delta_{s_ab_a}Y_{s_a}\hp a [J,\bJ] \cdot
 \\ \nonumber
 &\qquad\quad  +
 \sum_{\mathbf b_{\hat a}}
 \frac{1}{E(b_a,s_a)}
 \big(
 \bJ_{\mathbf b}
 \fderbJ{\mathbf b_{\hat a} s_a} 
 -
 J_{\mathbf b_{\hat a} s_a}
 \fderJ{\mathbf b}
 \big)
 \Big)\zjj 
 \bigg\} 
 \\ \nonumber
 &= \nonumber
 2\lambda 
 \sum\limits_a
 \bigg\{
 \fder{Y\hp a_{s_a} [J,\bJ]}{\bJ_\sb} 
 \cdot
 Z[J,\bJ]
 +
 Y\hp a_{s_a} [J,\bJ] 
 \cdot
 \fder{Z[J,\bJ]}{\bJ_\sb}
 \\ \nonumber
 & \quad\qquad  +
 \sum_{\mathbf b}
 \frac{1}{E(b_a,s_a)}
 \fderbJ{\mathbf s_{\hat a}b_a}\!\!
 \big(
 \bJ_{\mathbf b}
 \fderbJ{\mathbf b_{\hat a} s_a} 
 -
 J_{\mathbf b_{\hat a} s_a}
 \fderJ{\mathbf b}
 \big)
 Z[J,\bJ]
 \bigg\} \\  \nonumber 
 & =
  2\lambda 
 \sum\limits_a
 \bigg\{
 \fder{Y\hp a_{s_a} [J,\bJ]}{\bJ_\sb} 
 \cdot
 Z[J,\bJ]\nonumber
 +
 Y\hp a_{s_a} [J,\bJ] 
 \cdot
 \fder{Z[J,\bJ]}{\bJ_\sb}
 \\ \nonumber
 & \qquad\quad +
 \sum\limits_{b_a}
 \frac{1}{E(b_a,s_a)}   
 \fder{Z[J,\bJ]}{\bJ_{\sb}}
 +  \sum\limits_{\mathbf b}
 \frac{\bJ_{\mathbf b}}{E(b_a,s_a)} 
 \frac{\delta^2Z [J,\bJ]}{\delta \bJ_{\sb_{\hat a}b_a} \delta \bJ_{\mathbf b_{\hat a}s_a}}
%  \end{align}  
%  \begin{align}
\\
 \nonumber
 &\qquad \quad  - 
 \sum_{\mathbf b}
 \frac{1}{E(b_a,s_a)}
 J_{\mathbf b_{\hat a} s_a}
 \fder{^2Z[J,\bJ]}{\bJ_{\mathbf s_{\hat a}b_a} \delta J_\mathbf b}
 \bigg\}  
 \\ \nonumber 
 & = 2\lambda \sum_a(A_a(\sb) -  B_a(\sb)  + C_a(\sb) + D_a(\sb) +F_a(\sb) )\,,
 \end{align}
%  \vspace{-.4cm}
with
\begin{align*}
% \mbox{with }\qquad 
A_a(\sb) & = Y\hp a_{s_a} [J,\bJ] 
 \cdot
 \fder{Z[J,\bJ]}{\bJ_\sb} \,,
 &  B_a(\sb) & = \sum_{\mathbf b}
 \frac{1}{E(b_a,s_a)}
 J_{\mathbf b_{\hat a} s_a}
 \fder{^2Z[J,\bJ]}{\bJ_{\mathbf s_{\hat a}b_a} \delta J_\mathbf b}\,,
\\
C_a(\sb) & =  \sum\limits_{b_a}
 \frac{1}{E(b_a,s_a)}   
 \fder{Z[J,\bJ]}{\bJ_{\sb}} \,, &  D_a(\sb) & =  \sum\limits_{\mathbf b}
 \frac{\bJ_{\mathbf b}}{E(b_a,s_a)} 
 \frac{\delta^2Z [J,\bJ]}{\delta \bJ_{\sb_{\hat a}b_a} \delta \bJ_{\mathbf b_{\hat a}s_a}} \,,
\\  F_a(\sb) & =     
 \fder{Y\hp a_{s_a} [J,\bJ]}{\bJ_\sb} 
 \cdot
 Z[J,\bJ]\,.
\end{align*}

% \subsection{Derivation of the SDE}
One shall be interested in derivatives of $W\jj$ of the following form:  
\begin{equation} \label{eq:redundant}
\big(\prod\limits_{i=1}^k \fderJ{\xb_i}\fderbJ{\yb_i} \big) W\jj \bigg|_{J=\bar J=0} 
\, , \quad  \Xb=(\xb^1,\ldots,\xb^k)
\in \mathcal{F}_{D,k},\,\,\B_*(\Xb)=(\yb^1,\ldots,\yb^k), 
\end{equation}
for $\B\in \vGrph{D} $, and then use formula \eqref{eq:SintPostWard}
with, say, the vertex $\sb=\yb^{1}$.  As said in the introduction,
deltas of the interaction vertices and the propagators (proportional
to deltas) inside each Feynman diagrams render the definition of the
$2k$-multi-point function based on \eqref{eq:redundant} redundant, if
one treats the $\xb$-variables and the $\yb$-variables as independent.
In fact, all the $\yb$'s can be expressed in terms of coordinates of
$\Xb=(\xb^1,\ldots,\xb^k)$ of the same colour, the combinatorics of
which uniquely determines a so-called boundary graph $\B$ with $2k$
vertices; moreover, non-vanishing terms in the formula above are
precisely a \textit{graph derivative of} $W\jj$ with respect to $\B$
at $\Xb$. \par For the time being, we pick only a \textit{connected}
boundary graph $\B$ and we want to know what the rest of the
derivatives $\delta/\delta J_{\xb^\alpha}, \delta/\delta
\bJ_{\yb^{\alpha}(\{\xb^\nu\}_\nu)}$, ($\alpha=2,\ldots, D$) do to the
expression \eqref{eq:WnachJ_s}.  By using \eqref{eq:SintPostWard} with
$\sb=\yb^1$ we analyze the five summands in the (lowermost) RHS:
\[
\mathfrak m_a{(\Xb;\sb;\B)}:= \frac{1}{Z_0} \!\!\!
\prod\limits_{\substack{\alpha >1 \\ \nu=1,\ldots,k} }\fder{}{\bJ_{\yb^\alpha}} \fder{}{J_{\xb^\nu}}   
M_a(\sb) \bigg|_{J=\bJ=0}\, \vspace{-.4cm}
\]
for
\[
(\mathfrak{m},M)\in\{ (\mathfrak{a},A),(\mathfrak{b},B), (\mathfrak{c},C), (\mathfrak{d},D),(\mathfrak{f},F)\}\,. 
\]
Actually $\sb$ is a function of $\Xb$ ---and so is any other
$\yb^\alpha$--- but the dependence of $\mathfrak m_a$ on it only shows
that $\sb$ is the variable respect to which we firstly\footnote{The
  order should play no role when one obtains closed equations for a
  single correlation function in each sector of common Gur\u
  au-degree} derived $W\jj$. Each $\mathfrak m_a$ depends on the
boundary graph $\B$ through $\{\yb^\alpha\}_{\alpha=1}^k$ given by
$\eqref{eq:Jcycle}$.  Ignoring the common $(-2\lambda/E_{\sb})$
prefactor:
\begin{itemize}
 \item $\mathfrak a_a(\Xb;\sb;\B)$ is easily seen to yield $ Y\hp{a}_{s_a}[0,0]\cdot G_{\B}\hp{2k}(\Xb)$ 
 \item also, derivatives on $C_a(\sb)$, $\mathfrak c_a(\Xb;\sb;\B)$, 
 readily give $\sum_{b_a} E(b_a,s_a)\inv G_\B\hp{2k}(\Xb) $ 
 \item the term $\mtf{f}_a(\Xb;\sb;\B)$ is, according to Proposition \ref{thm:promedio_grupo}, 
$\sum_{\hat\pi \in \Autc(\B)} \pi^*\mathfrak f _{\B} \hp{a}(\Xb)$  
  \end{itemize}
The remaining two terms, $\mtf{b}_a$ and $\mathfrak{d}_a$, need a more detailed 
inspection, though: 
\begin{align} \nonumber 
\mathcal{O}(\bar J)+ 
\prod\limits_{\substack{\alpha >1 \\ \nu=1,\ldots,k} } \fder{}{\bJ_{\yb^\alpha}} \fder{}{J_{\xb^\nu}}D_a (\sb) & =
 \prod\limits_{\substack{\alpha >1; \nu}} \fder{}{\bJ_{\yb^\alpha}}\fder{}{J_{\xb^\nu}}\bigg[
 \sum_{\mathbf b}
 \frac{1}{E(b_a,s_a)}
 \bJ_{\mathbf b}
  \frac{\delta^2Z [J,\bJ]}{\delta \bJ_{\sb_{\hat a}b_a} \delta \bJ_{\mathbf b_{\hat a}s_a}} 
 \bigg] 
\\
& = 
\suml_{\rho=2}^k
\prod\limits_{\substack{\alpha>1 (\alpha\neq \rho)\\ \nu=1,\ldots, k}} \fder{}{\bJ_{\yb^\alpha}}\fder{}{J_{\xb^\nu}}\bigg[
 \sum_{\mathbf b}
 \frac{ \delta^{\mathbf b}_{\yb^\rho} }{E(b_a,s_a)}
  \frac{\delta^2Z [J,\bJ]}{\delta \bJ_{\sb_{\hat a}b_a} \delta \bJ_{\mathbf b_{\hat a}s_a}}
 \bigg]\, \nonumber \\
 & = 
\suml_{\rho=2}^k \,
\prod\limits_{\substack{\alpha; ( 1 \neq \alpha\neq \rho) \\ \nu=1,\ldots, k}} \fder{}{\bJ_{\yb^\alpha}}\fder{}{J_{\xb^\nu}}\bigg[
 \frac{ 1 }{E(y^\rho_a,s_a)}
  \frac{\delta^2Z [J,\bJ]}{\delta \bJ_{\sb_{\hat a}y_a^\rho} \delta \bJ_{\mathbf y_{\hat a}^\rho s_a}}
 \bigg]\, . \label{eq:swap_black}
\end{align}
As for the derivatives on $B_a(\sb)$, 
\begin{align} \nonumber 
\mathcal{O}(J)+\prod\limits_{\substack{\alpha >1 \\ \nu=1,\ldots,k} } \fder{}{\bJ_{\yb^\alpha}} \fder{}{J_{\xb^\nu}} B_a (\sb)  & =
 \prod\limits_{\substack{\alpha >1; \nu}} \fder{}{\bJ_{\yb^\alpha}}\fder{}{J_{\xb^\nu}}\bigg[
 \sum_{\mathbf b}
 \frac{1}{E(b_a,s_a)}
  J_{\mathbf b_{\hat a} s_a}
 \fder{^2Z[J,\bJ]}{\bJ_{\mathbf s_{\hat a}b_a} \delta J_\mathbf b} 
\bigg]
\\
& = 
\suml_{\beta=1}^k
\prod\limits_{\substack{\alpha >1; \nu\neq \beta}} \fder{}{\bJ_{\yb^\alpha}}\fder{}{J_{\xb^\nu}}\bigg[
 \sum_{\mathbf b}
 \frac{1}{E(b_a,s_a)}\delta^{s_a }_{x_a^\beta}
 \delta^{ {\mathbf b}_{\hat a}}_{ \xb^\beta_{\hat a}}
 \fder{^2Z[J,\bJ]}{\bJ_{\mathbf s_{\hat a}b_a} \delta J_\mathbf b} 
  \bigg]\, \nonumber
    \\ &=
\nonumber
\suml_{\beta=1}^k
\prod\limits_{\substack{\alpha >1; \nu\neq \beta}} \fder{}{\bJ_{\yb^\alpha}}\fder{}{J_{\xb^\nu}}\bigg[ 
 \suml_{b_a}
 \frac{1}{E(b_a,s_a)}
 \delta^{s_a }_{x_a^\beta}
 \fder{^2Z[J,\bJ]}{\bJ_{\mathbf s_{\hat a}b_a} \delta J_{\xb^\beta_{\hat a}b_a }} 
  \bigg]
  \\ &=
\prod\limits_{\substack{\alpha >1; \nu\neq \gamma}} \fder{}{\bJ_{\yb^\alpha}}\fder{}{J_{\xb^\nu}}\bigg[ 
 \suml_{b_a}
 \frac{1}{E(b_a,x_a^\gamma)}
 \fder{^2Z[J,\bJ]}{\bJ_{\mathbf s_{\hat a}b_a} \delta J_{\xb^\gamma_{\hat a}b_a }} 
  \bigg] \,. \label{eq:runningba}
\end{align}
For the last equality, one uses the fact that $\B$ is regular. Thus,
there exists precisely one white vertex
$\xb^{\gamma}$, $\gamma=\gamma(a)$, such that 
$x_a^\gamma=s_a$. In turn, this means that $\delta_{x_a^\beta}^{s_a}=\delta_{x_a^\beta}^{s_a}\delta_{\beta }^\gamma$.
\begin{itemize}
 \item as evident in eq. \eqref{eq:swap_black}, the derivatives on the $D_a$-term give,
 after setting the sources to zero, all the (coloured) graphs obtained 
 from $\B$ by a swapping 
 of the following form (only $a$-colour and only the four implied vertices visible):
\begin{align} \label{swappingpicture}
\includegraphics[height=3.7cm]{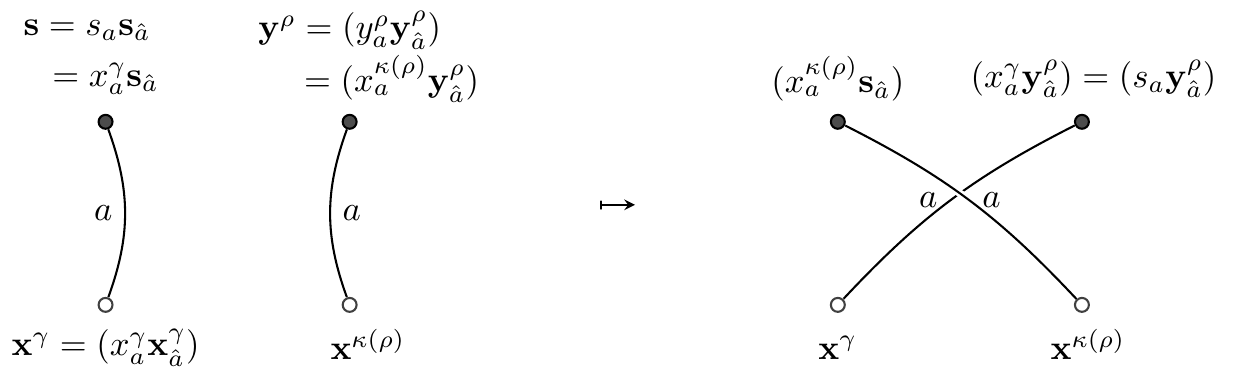}
\end{align}
for $\rho$ running over the black vertices which are not $\bJ_\sb=\bJ_{\yb^1}$. Hence the contribution of this term is
\[\suml_{\rho>1 }  \frac{1}{  E(y^\rho_a,s_a)} Z_0\inv\dervpar{Z[J,J]}{\varsigma_a(\B ;1,\rho)(\Xb) } 
\mbox{ for }\rho=2,\ldots,k.\] Since $\yb^1=\sb$, 
we also write $\varsigma_a(\B;\yb^1,\yb^\rho)=\varsigma_a(\B;1,\rho)$ 
for this new indexing graph ($\rho>1$).
 
 \item concerning the derivatives of $B_a$ above in eq. \eqref{eq:runningba}, the only surviving term is 
 $\delta^2 Z[J,\bJ]/\delta \bJ_{\mathbf s_{\hat a}b_a} \delta J_{\xb^\gamma_{\hat a}b_a } $,
 and is selected by $\delta_{x_a^\beta}^{s_a}$, 
 which is just, after taking into account the rest of the derivatives, 
 the graph derivative $\partial Z/\partial \B(\Xb)|_{x^\gamma_a\to b_a}$, with 
 the single coordinate $x^\gamma_a$ being substituted by (the running) $b_a$.
 If $\B$ is connected (as we assumed), after setting the sources to zero,
 this accounts for the \textit{first}\footnote{One of the authors (R.P.) in this new version  
 detected the incompleteness of eq. \eqref{eq:new_bterm} and 
 found the second line. This corrects the previous 
 (only-preprint-)version \href{https://arxiv.org/abs/1706.07358v1}{arxiv.org/abs/1706.07358v1}, where
 only C.I.P.S. and R.W. appear as authors.} 
 line in the next term:
 \begin{align} \nonumber
 \mathfrak b_a(\Xb;\sb;\B)& =\suml_{b_a}\frac{1}{E(b_a,x^\gamma_a)} 
 G_\B\hp{2k}(\xb^1,\ldots,\xb^{\gamma-1}, 
 (x_1^\gamma,\ldots,x_{a-1}^\gamma,b_a,x_{a+1}^\gamma,\ldots x^\gamma_D), 
 \xb^{\gamma+1},
 \ldots,\xb^k)  \\
 & + \suml_{\rho>1} \frac{1}{E(x_a^{\kappa(\rho)},x^\gamma_a)} 
 \frac{\partial Z[J,\bar J]}{\partial \,\varsigma_a(\B;1,\rho)}  
 (\xb^1,\ldots, (x_1^\gamma,\ldots,x_{a-1}^\gamma,x_a^{\kappa(\rho)},x_{a+1}^\gamma,\ldots x^\gamma_D), 
 \ldots,\xb^k)  \,. \label{eq:new_bterm}
 \end{align}
 \end{itemize}
 The second line is found by noticing that one also gets a contribution from the 
 graph derivative 
 $\partial Z[J,\bar J]/\partial\varsigma_a(\B;1,\rho)$ if this is evaluated at 
 $\Xb|_{x_a^\gamma\to x_a^{\kappa(\rho)}}$, where $\kappa(\rho)$ 
 is determined by \eqref{swappingpicture} (i.e. $x_a^{\kappa(\rho)}=y_a^{\rho}$).
From eqs. \eqref{eq:SintPostWard} one has
\begin{align*}
\dervpar{W\jj}{\B(\Xb)} & = \!\!\prod\limits_{\substack{\alpha >1 \\ 
\nu=1,\ldots,k} } \fder{}{\bJ_{\yb^\alpha}} \fder{}{J_{\xb^\nu}} \bigg(
\frac{-2\lambda  E_\sb\inv}{\zjj} \sum_a(A_a(\sb) +  C_a(\sb)  
+ D_a(\sb) + F_a(\sb) -  B_a(\sb) )\bigg)\!\bigg|_{\substack{\bJ=0 \\ J=0}} \\
&= \frac{(-2\lambda)}{E_{\sb}}
 \sum_a(\mathfrak{a}_a(\Xb;\sb;\B) +  \mathfrak{c}_a(\Xb;\sb;\B)   +
 \mathfrak{d}_a(\Xb;\sb;\B)  + \mathfrak{f}_a(\Xb;\sb;\B)-  \mathfrak{b}_a(\Xb;\sb;\B) ) \,,
\end{align*}
where each summand is now known. 
Because $Y\hp a_{s_a}[0,0]=\Delta_{s_a,1}\GDmelon$ we have proven:
\begin{thm}[Schwinger-Dyson equations] \label{thm:SDEs} Let $D\geq 3$ and 
let $\B$ be a connected boundary graph of the quartic melonic model, 
$\B\in \,\fey_D(\phi^4_{\mtr{m},D})= \dvGrph{D}$. 
Let $2k$ denote the number of vertices of $\B$. Pick a $\bar J$-external line,
that is, $\bar J_{\yb^\alpha}$, for $1\leq \alpha \leq k$, where $\B_*(\Xb)=(\yb^1,\dots,\yb^\alpha,\ldots \yb^k)$
for  $\Xb \in \mathcal{F}_{k(\B),D}$, and set $\sb=\yb^\alpha$ for sake of notation. 
The $(2k)$-point Schwinger-Dyson equation 
corresponding to $\B$ is 
% \leqnomode
\allowdisplaybreaks[0]
\begin{align}
&\!\!\!\bigg(1+\frac{2\lambda}{E_{\sb}} \suml_{a=1}^D \suml_{\mathbf q_{\hat a} }\GDmelon (s_a,\mathbf q_{\hat a})\bigg)
 \label{eq:SDEs}
G_\B\hp{2k}(\Xb) \\
& =\frac{\delta_{1,k}}{E_\sb}+\frac{(-2\lambda)}{E_{\sb}} \suml_{a=1}^D \Bigg\{ 
\suml_{\hat \sigma \in \Autc(\B)} \sigma^* \mathfrak{f}\hp{a}_{\B,s_a}(\Xb) \nonumber
%   & \qquad\qquad \qquad 
%  \,\,\,\,\,
 \\
 &\qquad \quad   +
 \suml_{\rho \neq \alpha }  \frac{Z_0\inv }{  E(y^\rho_a,s_a)} \bigg[\dervpar{Z[J,J]}{\varsigma_a(\B ; \alpha,\rho) }  (\Xb)
- \dervpar{Z[J,J]}{\varsigma_a(\B ; \alpha,\rho) } (\Xb|_{s_a \to y_a^{\rho}})\bigg] \nonumber
 \\
& \nonumber 
\,\,\,\,\qquad \qquad\qquad\qquad - \suml_{b_a}\frac{1}{E(s_a,b_a)} \big[ G_\B\hp{2k}(\Xb) - 
G_{\B}\hp{2k}(\Xb|_{s_a \to b_a}) \big] \Bigg\}   
\end{align}
\allowdisplaybreaks
% \reqnomode
% \begin{align}\label{eq:SDEs}
% \bigg(1+\frac{E_{\sb}}{2\lambda} \suml_{a=1}^D \suml_{\mathbf q_{\hat a} }\GDmelon (s_a,\mathbf q_{\hat a})\bigg)
% G_\B(\Xb) & =\frac{(-2\lambda)}{E_{\sb}} \sum_{a=1}^D \Bigg\{ 
% \suml_{\hat \sigma \in \Autc(\B)} \sigma^* F\hp{a}_\B(\Xb) 
% + \big(\suml_{\mathbf q_{\hat a} }\GDmelon (s_a,\mathbf q_{\hat a})\big) \cdot 
% G_\B(\Xb)
% \\
% & \Bigg\} \nonumber 
% \end{align}
 for all $\Xb \in \mathcal{F}_{D,k(\B)}$. 
%  Here $\gamma\in \{1,\ldots,k\}$ 
% is uniquely determined by $s_a=x^\gamma_a$ and 
Notice that $s_a=x^\gamma_a$ for certain $\gamma$, 
$\Xb = (\xb^1, \ldots, \xb^\gamma, \ldots, \xb^k)$,
and in this sense $\Xb|_{s_a \to b_a}$ means the replacement in $\Xb$
of $x^\gamma_a$ by the summed index $b_a$ (see eq. \eqref{eq:new_bterm}, with 
a similar situation for $\Xb|_{s_a \to y_a^\rho}$). Here $(s_a,\mathbf q_{\hat a})$ 
is abuse of notation for $(q_1,q_2,\ldots,q_{a-1},s_a,q_{a+1},\ldots,q_D)$.  Also 
recall that 
$E(u_a,v_a)=E_{u_a \mathbf q_{\hat{a}}} -E_{v_a \mathbf q_{\hat{a}}}$.
\end{thm} 

\begin{proof}
The 2-point equation has the same structure, but the propagator is
added (accounting for the delta $\delta_{1,k}$ term. This equation has
been proven in \cite{fullward}. For the higher-point functions, the
proof of this theorem (with $\alpha=1$) precedes the statement.  We
remark that, since $\Xb \in \mathcal{F}_{k(\B),D}$, the denominators
of the form $E(y^\rho_a,s_a)$ are well defined. In the limit $b_a\to
s_a $, $ E(s_a,b_a)$ becomes singular, but also the numerator, and a
derivative term arises.
\end{proof}

A graphical interpretation of this theorem shall be given in 
a future work; therein, in particular, the perturbative expansion of this
equation will be addressed in simple cases. 

We will ease the notation $
\mathfrak{f}\hp{a}_{\B,s_a}=\mathfrak{f}\hp{a}_{\B}$, when no risk of
confusion arises, keeping in mind the dependence of this function on
$s_a$.  Notice that if the graph $\varsigma_a(\B ;1,\rho) $ is
connected, then the respective derivative on $\zjj$ is just
\begin{equation*} % \label{eq:disconnected_Green}
\frac{1}{Z_0}\dervpar{Z[J,J]}{\varsigma_a(\B ;1,\rho)(\Xb)}= 
G_{\varsigma_a(\B ;1,\rho)}\hp{2k} (\Xb)\,;
\end{equation*}
otherwise, the RHS of this expression contains, on top of
$G\hp{2k}_{\varsigma_a(\B ;1,\rho)}$, also a product of correlation
functions indexed by the connected components of $\varsigma_a(\B
;1,\rho)$ with a number of points which add up to $2k$ (see
Sec. \ref{sec:fourpoint}).  Observe that the equation still depends at
this stage on the choice of the vertex $\bJ_\sb$, with respect to
which we first derived.  Thus, one has $k$ independent SDE for
$G_{\B}\hp{2k}$, when $\B$ has no symmetries.

\section{Schwinger-Dyson equations for rank-3 theories}\label{sec:rank3SDE}

% \subsection{The expansion of the free energy to $\mathcal O(J^4,\bJ^4)$}
According to \cite{fullward}, the boundary sector $\im\, \partial$ of the
$\phi^4_3$-theory is all of $ \partial (\feyn ) = \dvGrph{3}  $.
Therefore $W\jj=\log \zjj$ can be expanded in boundary graphs as: \\

% \fontsize{11.5}{14.3}\selectfont  
\begin{math}
\!\!\!\!\!\!\!
W_{D=3}[J,\bar J]=G\hp 2 _{\includegraphics[height=1.9ex]{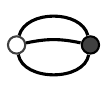}} \star \mathbb{J} (
\raisebox{-.3\height}{\includegraphics[height=3ex]{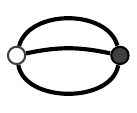}} ) + 
\dfrac1{2!}
G\hp 4 _{|\raisebox{-.33\height}{\includegraphics[height=2ex]{graphs/3/Item2_Melon.pdf}}|
\raisebox{-.33\height}{\includegraphics[height=2ex]{graphs/3/Item2_Melon.pdf}}|} \star \mathbb{J} (
\raisebox{-.3\height}{\includegraphics[height=3ex]{graphs/3/Logo2_Melon.pdf}}
^{\sqcup 2} 
% \raisebox{-.3\height}{\includegraphics[height=3ex]{graphs/3/Logo2_Melon.pdf}} 
) + 
\dfrac1{2}\suml_cG\hp4_{\raisebox{-.4\height}{\includegraphics[width=3ex]{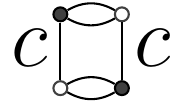}}} \star \mathbb{J}
\big( \,
\raisebox{-.34\height}{\includegraphics[width=0.7cm]{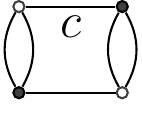}} \,\big)  
% +\dfrac1{2}G\hp4_{\raisebox{-.4\height}{\includegraphics[width=3ex]{graphs/3/Item4_V2v.pdf}}}  
% \star \mathbb{J}
% \big(\, 
% \raisebox{-.34\height}{\includegraphics[width=0.7cm]{graphs/3/Logo4_V2.pdf}}\, \big)
% +\dfrac1{2}G\hp4_{\raisebox{-.4\height}{\includegraphics[width=3ex]{graphs/3/Item4_V3v.pdf}}} 
% \star \mathbb{J}
% \big( \,
% \raisebox{-.34\height}{\includegraphics[width=0.7cm]{graphs/3/Logo4_V3.pdf}}\, \big)
+ \dfrac1{3}\suml_{c}G\hp 6_{\raisebox{-.4\height}{\includegraphics[height=2.3ex]{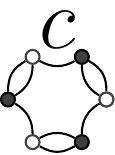}}} 
\star \mathbb{J}
\Big(
\raisebox{-.34\height}{\includegraphics[width=.8cm]{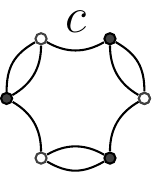}}\Big)
 \nonumber
\,\,\,\,+
\dfrac1{3} G\hp6 _{\raisebox{-.4\height}{\includegraphics[height=1.8ex]{graphs/3/Item6_K33.pdf}}} \star \mathbb{J}
\Big(
\raisebox{-.3\height}{\includegraphics[width=.66cm]{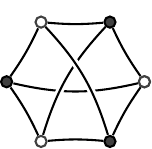}}\Big) 
+ \suml_{i} G\hp 6_{
\raisebox{-.34\height}{\includegraphics[width=0.65cm]{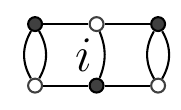}}} \star \mathbb{J}
\Big(
\raisebox{-.34\height}{\includegraphics[width=1.2cm]{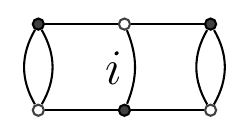}} \Big)
+ \dfrac1{ 3!}
G\hp 6 _{|\raisebox{-.33\height}{\includegraphics[height=2ex]{graphs/3/Item2_Melon.pdf}}|
\raisebox{-.33\height}{\includegraphics[height=2ex]{graphs/3/Item2_Melon.pdf}}|  
\raisebox{-.33\height}{\includegraphics[height=2ex]{graphs/3/Item2_Melon.pdf}}|}  \star 
 \mathbb{J} (
\raisebox{-.3\height}{\includegraphics[height=3ex]{graphs/3/Logo2_Melon.pdf}}
^{\sqcup 3}) 
+\dfrac12\suml_{c}
 G\hp 6 _{|\raisebox{-.33\height}{\includegraphics[height=2ex]{graphs/3/Item2_Melon.pdf}}|
\raisebox{-.4\height}{\includegraphics[height=2ex]{graphs/3/Item4_Vcv.pdf}}|} 
\star  \mathbb{J}
\big(
\raisebox{-.3\height}{\includegraphics[height=3ex]{graphs/3/Logo2_Melon.pdf}}
\sqcup 
\raisebox{-.34\height}{\includegraphics[width=.8cm]{graphs/3/Logo4_Vc.pdf}}\,
\big)   +
\dfrac1{2!\cdot 2^2} 
\suml_{c}
G\hp 8 _{|\raisebox{-.4\height}{\includegraphics[height=2ex]{graphs/3/Item4_Vcv.pdf}}|
\raisebox{-.4\height}{\includegraphics[height=2ex]{graphs/3/Item4_Vcv.pdf}}|} 
\star  \mathbb{J}
\big(\raisebox{-.34\height}{\includegraphics[width=.7cm]{graphs/3/Logo4_Vc.pdf}}
\sqcup 
\raisebox{-.34\height}{\includegraphics[width=.7cm]{graphs/3/Logo4_Vc.pdf}}\,
\big)  
+\dfrac1{2^2}\suml_{c<i} 
 G\hp 8 _{|\raisebox{-.4\height}{\includegraphics[height=2ex]{graphs/3/Item4_Vcv.pdf}}|
\raisebox{-.4\height}{\includegraphics[height=2ex]{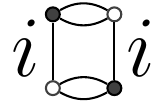}}|} 
\star  \mathbb{J}
\big(\raisebox{-.34\height}{\includegraphics[width=.8cm]{graphs/3/Logo4_Vc.pdf}}
\sqcup 
\raisebox{-.34\height}{\includegraphics[width=.8cm]{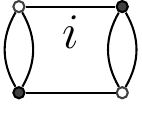}}\,
\big)   
+ \dfrac1{ 4!}
G\hp 8 _{|\raisebox{-.33\height}{\includegraphics[height=2ex]{graphs/3/Item2_Melon.pdf}}|
\raisebox{-.33\height}{\includegraphics[height=2ex]{graphs/3/Item2_Melon.pdf}}|  
\raisebox{-.33\height}{\includegraphics[height=2ex]{graphs/3/Item2_Melon.pdf}}|
\raisebox{-.33\height}{\includegraphics[height=2ex]{graphs/3/Item2_Melon.pdf}}|}  \star 
 \mathbb{J} (
\raisebox{-.3\height}{\includegraphics[height=3ex]{graphs/3/Logo2_Melon.pdf}}
^{\sqcup 4}) 
+
\dfrac1{2\cdot 2!}
\suml_{c} 
 G\hp 8 _{|
\raisebox{-.33\height}{\includegraphics[height=2ex]{graphs/3/Item2_Melon.pdf}}|
\raisebox{-.33\height}{\includegraphics[height=2ex]{graphs/3/Item2_Melon.pdf}}|
\raisebox{-.4\height}{\includegraphics[height=2ex]{graphs/3/Item4_Vcv.pdf}}|} 
\star  \mathbb{J}
\big(\raisebox{-.3\height}{\includegraphics[height=3ex]{graphs/3/Logo2_Melon.pdf}}
\sqcup 
\raisebox{-.3\height}{\includegraphics[height=3ex]{graphs/3/Logo2_Melon.pdf}}
\sqcup 
\raisebox{-.34\height}{\includegraphics[width=.8cm]{graphs/3/Logo4_Vc.pdf}}\,
\big)  
 \nonumber
 +
\dfrac1{3} G\hp8 _{ 
|\raisebox{-.33\height}{\includegraphics[height=2ex]{graphs/3/Item2_Melon.pdf}}|
\raisebox{-.34\height}{\includegraphics[width=.36cm]{graphs/3/Item6_K33.pdf}}|} \star \mathbb{J}
\Big(
\raisebox{-.3\height}{\includegraphics[height=3ex]{graphs/3/Logo2_Melon.pdf}}
\sqcup
\raisebox{-.3\height}{\includegraphics[width=.66cm]{graphs/3/Logo6_K33}}\Big) +
\dfrac1{3}
 \suml_{c}
 G\hp 8_{|\raisebox{-.23\height}{\includegraphics[height=2ex]{graphs/3/Item2_Melon.pdf}}|\,
\raisebox{-.2\height}{\includegraphics[width=.3cm]{graphs/3/Item6_Qc.pdf}}|} 
\star \mathbb{J}
\Big(
\raisebox{-.3\height}{\includegraphics[height=3ex]{graphs/3/Logo2_Melon.pdf}}
\sqcup
\raisebox{-.34\height}{\includegraphics[width=.8cm]{graphs/3/Logo6_Qc}}\Big)
 + \suml_{i} G\hp 8_{
|\raisebox{-.34\height}{\includegraphics[height=1.9ex]{graphs/3/Item2_Melon.pdf}}|
\raisebox{-.4\height}{\includegraphics[width=0.7cm]{graphs/3/Item6_Eipure.pdf}}|} \star \mathbb{J}
\Big(
\raisebox{-.3\height}{\includegraphics[height=3ex]{graphs/3/Logo2_Melon.pdf}}
\sqcup
\raisebox{-.34\height}{\includegraphics[width=1.2cm]{graphs/3/Logo6_Ei.pdf}} \Big)
+ \suml_{\substack{j \,;\, l< i}}  \GoWlji 
\star
\mathbb J
\big(\!\!
\raisebox{-.42\height}{\includegraphics[width=17ex]{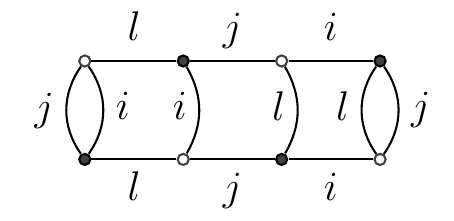}}
\!\big)+
 \suml_{j\neq i}\GoRijl
\star
\mathbb J
\big( \!\!
\raisebox{-.4\height}{\includegraphics[width=17ex]{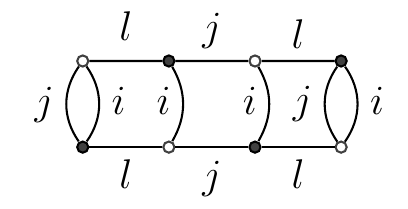}}
\!\big)
 +\dfrac14 \suml_{j} \GoAj
\star
\mathbb J
\big(
\raisebox{-.42\height}{\includegraphics[width=6ex]{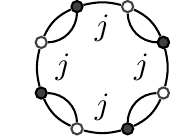}}
\big)
 +
 \suml_{j\neq i}\GoPijs 
\star
\mathbb J
\big(
\raisebox{-.4\height}{\includegraphics[width=10ex]{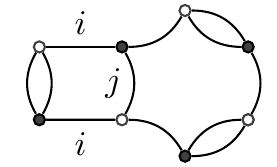}}
\big)+\suml_{i}
\GoXis
\star
\mathbb J
\big(\!\!\!\!\!
\raisebox{-.4\height}{\includegraphics[width=12ex]{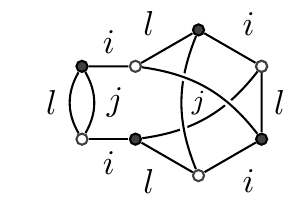}}
\big) 
+\suml_{\substack{l\neq i\neq j}}
\GoYs
\star
\mathbb J
\big(\!\! 
\raisebox{-.4\height}{\includegraphics[width=10ex]{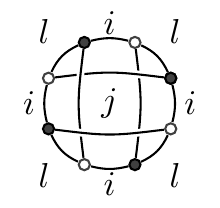}}
\big) +
\GoS
\star
\mathbb J
\big(
\raisebox{-.4\height}{\includegraphics[width=9ex]{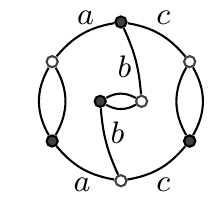}}
\big) 
+\GoCubo \star
\mathbb J
\big( 
\raisebox{-.4\height}{\includegraphics[width=9ex]{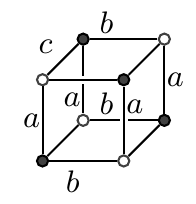}}
\big)   +
\mathcal O(10)\,.
\end{math} \par
% \fontsize{11.5}{14.3}\selectfont  
The WTI will be used for each colour $a=1,2,3$, and it will be convenient to single
out $a$ in this last expression. From here on\footnote{Beware this is only a notation for rank-$3$
theories; for rank $4$ another notation shall be used} $b=b(a)=\min(\{1,2,3\}\setminus\{a\})$ and 
$c=c(a)=\max(\{1,2,3\}\setminus \{a\})$: 
\vspace{-.3cm}
%   { \large 
%  \vspace{-1cm}
% \allowdisplaybreaks[0]
\begin{align*}
% \label{expansion_congraficas}
% correct factors
& \hspace{6.6cm}W_{D=3}[J,\bar J] \,=
\\
&
\quad\,\,\,\, G\hp 2 _{\includegraphics[height=1.9ex]{graphs/3/Item2_Melon.pdf}} \star \mathbb{J} (
\raisebox{-.3\height}{\includegraphics[height=3ex]{graphs/3/Logo2_Melon.pdf}} ) + 
\frac1{2!}
G\hp 4 _{|\raisebox{-.33\height}{\includegraphics[height=2ex]{graphs/3/Item2_Melon.pdf}}|
\raisebox{-.33\height}{\includegraphics[height=2ex]{graphs/3/Item2_Melon.pdf}}|} \star \mathbb{J} (
\raisebox{-.3\height}{\includegraphics[height=3ex]{graphs/3/Logo2_Melon.pdf}}
\sqcup 
\raisebox{-.3\height}{\includegraphics[height=3ex]{graphs/3/Logo2_Melon.pdf}} ) + 
\frac1{2}G\hp4_{\raisebox{-.4\height}{\includegraphics[height=1.7ex]{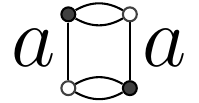}}} \star \mathbb{J}
\big( \,
\raisebox{-.4\height}{\includegraphics[height=4ex]{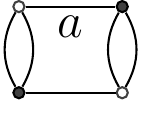}} \,\big) 
+\frac1{2} \suml_{i\neq a }G\hp4_{\raisebox{-.4\height}{\includegraphics[width=.5cm]{graphs/3/Item4_Viv.pdf}}}  
\star \mathbb{J}
\big(\, 
\raisebox{-.4\height}{\includegraphics[height=4ex]{graphs/3/Logo4_Vi.pdf}}\, \big)
\\  
\nonumber 
& \,\, \,\,
 + \frac1{3}G\hp 6_{\raisebox{-.4\height}{\includegraphics[height=3ex]{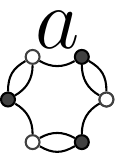}}} 
\star \mathbb{J}
\Big(
\raisebox{-.4\height}{\includegraphics[height=6ex]{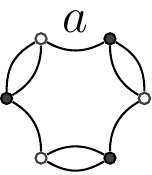}}\Big)
+ \frac1{3}\suml_{i\neq a}G\hp 6_{\raisebox{-.4\height}{\includegraphics[height=3.5ex]{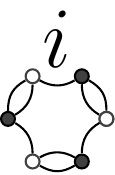}}} 
\star \mathbb{J}
\Big(
\raisebox{-.34\height}{\includegraphics[height=6ex]{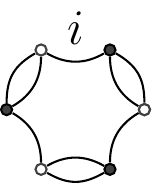}}\Big)
+ \frac1{ 3!}
G\hp 6 _{|\raisebox{-.33\height}{\includegraphics[height=2ex]{graphs/3/Item2_Melon.pdf}}|
\raisebox{-.33\height}{\includegraphics[height=2ex]{graphs/3/Item2_Melon.pdf}}|  
\raisebox{-.33\height}{\includegraphics[height=2ex]{graphs/3/Item2_Melon.pdf}}|}  \star 
 \mathbb{J} (
\raisebox{-.3\height}{\includegraphics[height=3ex]{graphs/3/Logo2_Melon.pdf}}
\sqcup 
\raisebox{-.3\height}{\includegraphics[height=3ex]{graphs/3/Logo2_Melon.pdf}} 
\sqcup 
\raisebox{-.3\height}{\includegraphics[height=3ex]{graphs/3/Logo2_Melon.pdf}})
\\  \nonumber
&\,\,\,\,+
\frac1{3} G\hp6 _{\raisebox{-.4\height}{\includegraphics[height=2ex]{graphs/3/Item6_K33.pdf}}} \star \mathbb{J}
\Big(
\raisebox{-.3\height}{\includegraphics[width=.76cm]{graphs/3/Logo6_K33}}\Big) +
G\hp 6_{
\raisebox{-.3\height}{\includegraphics[width=0.5cm]{graphs/3/Item6_Eabcs.pdf}}} \star \mathbb{J}
\big(
\raisebox{-.3\height}{\includegraphics[width=1.3cm]{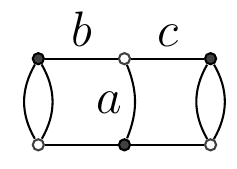}} \big)
+  G\hp 6_{
\raisebox{-.34\height}{\includegraphics[width=0.5cm]{graphs/3/Item6_Ebacs.pdf}}} \star \mathbb{J}
\big(
\raisebox{-.3\height}{\includegraphics[width=1.3cm]{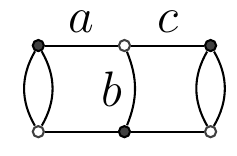}}\, \big)
+  G\hp 6_{
\raisebox{-.34\height}{\includegraphics[width=0.5cm]{graphs/3/Item6_Ecabs.pdf}}} \star \mathbb{J}
\big(
\raisebox{-.3\height}{\includegraphics[width=1.3cm]{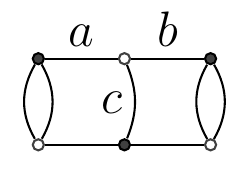}}\, \big)
\\   
\nonumber
& \,\, \,\,
+\frac12 
 G\hp 6 _{|\raisebox{-.33\height}{\includegraphics[height=2ex]{graphs/3/Item2_Melon.pdf}}|
\raisebox{-.34\height}{\includegraphics[height=1.7ex]{graphs/3/Item4_Va.pdf}}|} 
\star  \mathbb{J}
\big(
\raisebox{-.3\height}{\includegraphics[height=3ex]{graphs/3/Logo2_Melon.pdf}}
\sqcup 
\raisebox{-.4\height}{\includegraphics[height=4ex]{graphs/3/Logo4_Va.pdf}}\,
\big) 
+\frac12\suml_{i\neq a}
 G\hp 6 _{|\raisebox{-.33\height}{\includegraphics[height=2ex]{graphs/3/Item2_Melon.pdf}}|
\raisebox{-.34\height}{\includegraphics[height=1.7ex]{graphs/3/Item4_Vi.pdf}}|} 
\star  \mathbb{J}
\big(
\raisebox{-.3\height}{\includegraphics[height=3ex]{graphs/3/Logo2_Melon.pdf}}
\sqcup 
\raisebox{-.4\height}{\includegraphics[height=4ex]{graphs/3/Logo4_Vi.pdf}}\,
\big)   
\\ & \,\,\,\, +
\frac1{2!\cdot 2^2} 
G\hp 8 _{|\raisebox{-.34\height}{\includegraphics[height=1.7ex]{graphs/3/Item4_Va.pdf}}|
\raisebox{-.34\height}{\includegraphics[height=1.7ex]{graphs/3/Item4_Va.pdf}}|} 
\star  \mathbb{J}
\big(\raisebox{-.4\height}{\includegraphics[height=4ex]{graphs/3/Logo4_Va.pdf}}
\sqcup 
\raisebox{-.4\height}{\includegraphics[height=4ex]{graphs/3/Logo4_Va.pdf}}\,
\big)
+
\frac1{2!\cdot 2^2} 
\suml_{i\neq a}
G\hp 8 _{|\raisebox{-.4\height}{\includegraphics[height=1.7ex]{graphs/3/Item4_Vi.pdf}}|
\raisebox{-.4\height}{\includegraphics[height=1.7ex]{graphs/3/Item4_Vi.pdf}}|} 
\star  \mathbb{J}
\big(\raisebox{-.4\height}{\includegraphics[height=4ex]{graphs/3/Logo4_Vi.pdf}}
\sqcup 
\raisebox{-.4\height}{\includegraphics[height=4ex]{graphs/3/Logo4_Vi.pdf}}\,
\big)
\\
&\,\,\,\,
+\frac1{2^2}\suml_{i\neq a} 
 G\hp 8 _{|\raisebox{-.4\height}{\includegraphics[height=1.7ex]{graphs/3/Item4_Vi.pdf}}|
\raisebox{-.4\height}{\includegraphics[height=1.7ex]{graphs/3/Item4_Va.pdf}}|} 
\star  \mathbb{J}
\big(\raisebox{-.4\height}{\includegraphics[height=4ex]{graphs/3/Logo4_Vi.pdf}}
\sqcup 
\raisebox{-.4\height}{\includegraphics[height=4ex]{graphs/3/Logo4_Va.pdf}}\,
\big)  
+
\frac1{2^2} 
 G\hp 8 _{|\raisebox{-.34\height}{\includegraphics[height=1.7ex]{graphs/3/Item4_Vb.pdf}}|
\raisebox{-.34\height}{\includegraphics[height=1.7ex]{graphs/3/Item4_Vc.pdf}}|} 
\star  \mathbb{J}
\big(\raisebox{-.4\height}{\includegraphics[height=4ex]{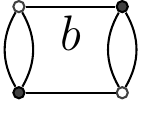}}
\sqcup 
\raisebox{-.4\height}{\includegraphics[height=4ex]{graphs/3/Logo4_Vc.pdf}}\,
\big)
\\  \nonumber
&\,\,\,\,+
\frac1{3} G\hp8 _{ 
|\raisebox{-.33\height}{\includegraphics[height=2ex]{graphs/3/Item2_Melon.pdf}}|
\raisebox{-.34\height}{\includegraphics[width=.36cm]{graphs/3/Item6_K33.pdf}}|} \star \mathbb{J}
\Big(
\raisebox{-.3\height}{\includegraphics[height=3ex]{graphs/3/Logo2_Melon.pdf}}
\sqcup
\raisebox{-.3\height}{\includegraphics[width=.76cm]{graphs/3/Logo6_K33}}\Big) +
\frac1{3}
 \suml_{i\neq a}
 \Gomqi
\star \mathbb{J}
\Big(
\raisebox{-.3\height}{\includegraphics[height=3ex]{graphs/3/Logo2_Melon.pdf}}
\sqcup
\raisebox{-.34\height}{\includegraphics[height=6ex]{graphs/3/Logo6_Qi.pdf}}\Big)
\\ \nonumber
&  \,\,\,\, +
\frac1{3}
\Gomqa
%  G\hp 8_{|\raisebox{-.33\height}{\includegraphics[height=2ex]{graphs/3/Item2_Melon.pdf}}|\,
% \raisebox{-.23\height}{\includegraphics[height=3.5ex]{graphs/3/Item6_Qa}}\,|} 
\star \mathbb{J}
\Big(
\raisebox{-.3\height}{\includegraphics[height=3ex]{graphs/3/Logo2_Melon.pdf}}
\sqcup
\raisebox{-.34\height}{\includegraphics[height=6ex]{graphs/3/Logo6_Qa.pdf}}\Big)
% \\
% &
\,\,\,\,+ \suml_{i\neq a} 
G\hp 8_{
|\raisebox{-.24\height}{\includegraphics[height=1.9ex]{graphs/3/Item2_Melon.pdf}}|
\raisebox{-.25\height}{\includegraphics[height=2.7ex]{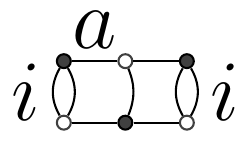}}|} 
\star \mathbb{J}
\Big(
\raisebox{-.3\height}{\includegraphics[height=3ex]{graphs/3/Logo2_Melon.pdf}}
\sqcup
\raisebox{-.34\height}{\includegraphics[height=5.5ex]{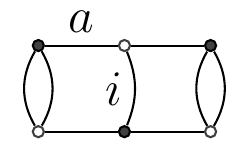}} \Big)
\\ 
&  \,\,\,\,
+
G\hp 8_{
|\raisebox{-.34\height}{\includegraphics[height=1.9ex]{graphs/3/Item2_Melon.pdf}}|
\raisebox{-.24\height}{\includegraphics[height=3ex]{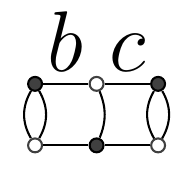}}|} \star \mathbb{J}
\Big(
\raisebox{-.3\height}{\includegraphics[height=3ex]{graphs/3/Logo2_Melon.pdf}}
\sqcup
\raisebox{-.3\height}{\includegraphics[height=5.5ex]{graphs/3/Logo6_Eabc.pdf}} \Big) 
% \\
% & 
% \,\,\,\,
+
\frac1{2\cdot 2!}
\suml_{i\neq a} 
 G\hp 8 _{|
\raisebox{-.33\height}{\includegraphics[height=2ex]{graphs/3/Item2_Melon.pdf}}|
\raisebox{-.33\height}{\includegraphics[height=2ex]{graphs/3/Item2_Melon.pdf}}|
\raisebox{-.34\height}{\includegraphics[height=1.7ex]{graphs/3/Item4_Vi.pdf}}|} 
\star  \mathbb{J}
\big( \raisebox{-.3\height}{\includegraphics[height=3ex]{graphs/3/Logo2_Melon.pdf}}
\sqcup 
\raisebox{-.3\height}{\includegraphics[height=3ex]{graphs/3/Logo2_Melon.pdf}}
\sqcup 
\raisebox{-.4\height}{\includegraphics[height=4ex]{graphs/3/Logo4_Vi.pdf}}\,
\big)  
\nonumber
\\ 
&  \,\,\,\,
+\frac1{2\cdot 2!}
 G\hp 8 _{|
\raisebox{-.33\height}{\includegraphics[height=2ex]{graphs/3/Item2_Melon.pdf}}|
\raisebox{-.33\height}{\includegraphics[height=2ex]{graphs/3/Item2_Melon.pdf}}|
\raisebox{-.34\height}{\includegraphics[height=1.7ex]{graphs/3/Item4_Va.pdf}}|} 
\star  \mathbb{J}
\big(\raisebox{-.3\height}{\includegraphics[height=3ex]{graphs/3/Logo2_Melon.pdf}}
\sqcup 
\raisebox{-.3\height}{\includegraphics[height=3ex]{graphs/3/Logo2_Melon.pdf}}\sqcup 
\raisebox{-.4\height}{\includegraphics[height=4ex]{graphs/3/Logo4_Va.pdf}}\,
\big) 
% \\
% &
% \,\,\,\, 
+ \frac1{ 4!}
G\hp 8 _{|\raisebox{-.33\height}{\includegraphics[height=2ex]{graphs/3/Item2_Melon.pdf}}|
\raisebox{-.33\height}{\includegraphics[height=2ex]{graphs/3/Item2_Melon.pdf}}|  
\raisebox{-.33\height}{\includegraphics[height=2ex]{graphs/3/Item2_Melon.pdf}}|
\raisebox{-.33\height}{\includegraphics[height=2ex]{graphs/3/Item2_Melon.pdf}}|}  \star 
 \mathbb{J} (
\raisebox{-.3\height}{\includegraphics[height=3ex]{graphs/3/Logo2_Melon.pdf}}
\sqcup 
\raisebox{-.3\height}{\includegraphics[height=3ex]{graphs/3/Logo2_Melon.pdf}}
\sqcup 
\raisebox{-.3\height}{\includegraphics[height=3ex]{graphs/3/Logo2_Melon.pdf}}
\sqcup 
\raisebox{-.3\height}{\includegraphics[height=3ex]{graphs/3/Logo2_Melon.pdf}}
) 
\nonumber
\\ 
&  \,\,\,\,
+\GoWa
\star
\mathbb J
\big(\!\!
\raisebox{-.42\height}{\includegraphics[width=17ex]{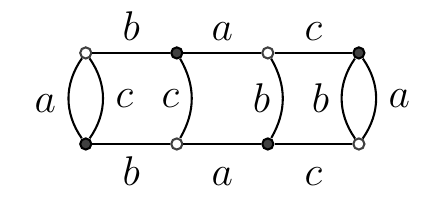}}
\big)
%  \\  
% &\,\, \,\,
% 
 + 
% \suml_{\substack{j, l< i}}  
\GoWb
\star
\mathbb J
\big(\!
\raisebox{-.44\height}{\!\!\includegraphics[width=17ex]{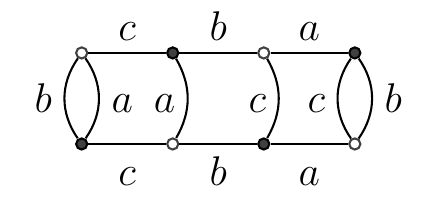}}
\!\!\big)
% \\
% &\,\,\,\,
\nonumber
\\ 
&  \,\,\,\,
+
\GoWc
\star
\mathbb J
\big(\!\!
\raisebox{-.42\height}{\includegraphics[width=17ex]{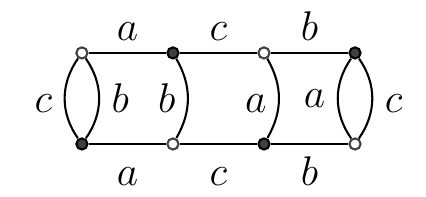}}
\big)
+ \frac12\suml_{\substack{i=b,c \\ (a\neq j\neq i)}}\GoRij
\star
\mathbb J
\Big(\!\!
\raisebox{-.42\height}{\includegraphics[width=17ex]{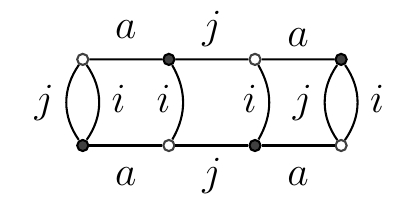}}
\!\!\Big)
\\  
&
\,\,\,\,
+
 \frac12\suml_{i=b,c}\GoRai
\star
\mathbb J
\Big(\!\!
\raisebox{-.42\height}{\includegraphics[width=17ex]{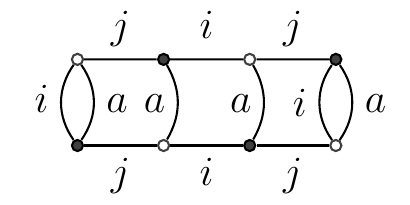}}
\!\!\Big)
+
 \frac12\suml_{i=b,c}\GoRia
\star
\mathbb J
\Big( \!\!
\raisebox{-.42\height}{\includegraphics[width=17ex]{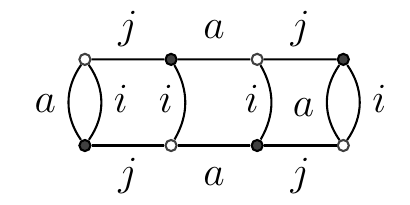}}
\!\!
\Big)
\\
& \,\, \,\,+\frac14 \suml_{j\neq a} \GoAj
\star
\mathbb J
\Big(\!\!
\raisebox{-.4\height}{\includegraphics[width=7ex]{graphs/3/Logo8_Aj.pdf}}
\Big)
+
\frac14
\GoAa
\star
\mathbb J
\Big(\!\! 
\raisebox{-.4\height}{\includegraphics[width=7ex]{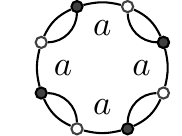}}
\Big)
 +
 \suml_{i\neq a }\GoPai
\star
\mathbb J
\Big(\!\!
\raisebox{-.4\height}{\includegraphics[width=10ex]{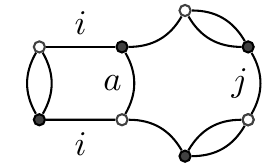}}
\Big)
\\
&  \,\, \,\,+
 \suml_{i\neq a }\GoPia
\star
\mathbb J
\Big( \!
\raisebox{-.4\height}{\includegraphics[width=10ex]{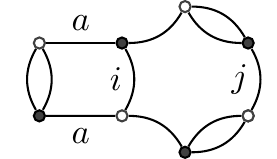}}
\Big)
+
 \suml_{i\neq a\neq j }\GoPij
\star
\mathbb J
\Big(\!
\raisebox{-.4\height}{\includegraphics[width=10ex]{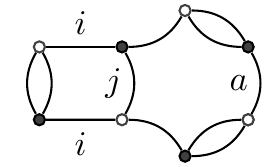}}
\Big)
+\suml_{i\neq a }
\GoXi
\star
\mathbb J
\Big(\!\!\!\!\!
\raisebox{-.42\height}{\includegraphics[width=12ex]{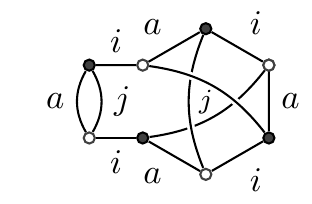}}
\Big) 
\\
&  \,\, \,\,
+
\suml_{i\neq a }
\GoXa
\star
\mathbb J
\Big( \!\!\!\!
\raisebox{-.42\height}{\includegraphics[width=12ex]{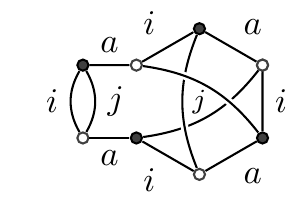}}
\Big) 
+
\GoS
\star
\mathbb J
\Big(\!\!
\raisebox{-.4\height}{\includegraphics[width=9ex]{graphs/3/Logo8_S.pdf}}
\Big) 
+\frac14\GoCubo
\star
\mathbb J
\Big[\!\! 
\raisebox{-.4\height}{\includegraphics[width=9ex]{graphs/3/Logo8_cubo.pdf}}
\Big]  
\\
&\,\,\,\,
+\frac14\suml_{\substack{ i\neq j}}
\GoYi
\star
\mathbb J
\Big(\!\!
\raisebox{-.44\height}{\includegraphics[width=11ex]{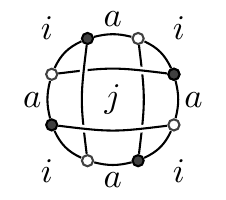}}\!
\!\Big) 
+\frac14 
% \suml_{\substack{ i\neq j}}
\GoYa
\star
\mathbb J
\Big(\!\! 
\raisebox{-.44\height}{\includegraphics[width=11ex]{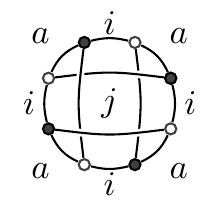}}
\!
\Big)
+ \mathcal O(10)
\end{align*}
% } 
\allowdisplaybreaks

% \subsection{Derivation of the $\ysa$-term for quartic theories}

In \cite{fullward}, the term $\ysa$ for the $\phi^4_3$-theory 
to $\mathcal{O}(4)$ has been found.  
This expansion is enough for deriving any of the 4-point
SDEs. However, since we want the explicit 6-point SDEs, we need
to compute $\ysa$ to $\mathcal{O}(6)$ in the sources, i.e. consider the
free energy to order $\mathcal{O}(J^4,\bar J^4)$, to be precise.  

\begin{lem} \label{thm:ordersix}
To order-$6$, $\ysa$ is given by: 
\begin{subequations}
\label{eq:laangeY}
\begin{align*}  
&\quad \ysa  
\\
& = 
\suml_{q_b,q_c}  \Gmelon(s_a,q_b,q_c) +
 \frac12
\sum\limits_{r=1}^2\big( \Delta_{s_a,r}
% ^{{|\raisebox{-.4\height}{\includegraphics[width=.42cm]{gfx_Ward/ward_icono_melon.pdf}}|
% \raisebox{-.4\height}{\includegraphics[width=.42cm]{gfx_Ward/ward_icono_melon.pdf}}|}}
G\hp 4 _{|\raisebox{-.33\height}{\includegraphics[height=1.8ex]{graphs/3/Item2_Melon}}|
\raisebox{-.33\height}{\includegraphics[height=1.8ex]{graphs/3/Item2_Melon}}|}+
\Delta_{s_a,r}
% ^{{\raisebox{-.4\height}{\includegraphics[width=.6cm]{gfx_Ward/ward_logo_vone.pdf}}}}
G\hp4_{\raisebox{-.4\height}{\includegraphics[width=.48cm]{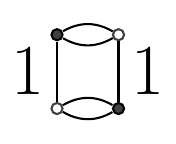}}}
+ 
\Delta_{s_a,r}
% ^{{\raisebox{-.4\height}{\includegraphics[width=.6cm]{gfx_Ward/ward_logo_vtwo.pdf}}}}
G\hp4_{\raisebox{-.4\height}{\includegraphics[width=.48cm]{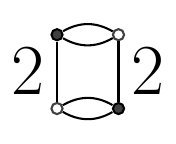}}}
+ \nonumber
\Delta_{s_a,r}
% ^{{\raisebox{-.4\height}{\includegraphics[width=.6cm]{gfx_Ward/ward_logo_vthree.pdf}}}}
G\hp4_{\raisebox{-.4\height}{\includegraphics[width=.48cm]{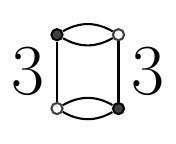}}}\big)
\star \mathbb{J}  \big(
\raisebox{-.0\height}{\logo{2}{Melon}{3.1}{3}}\big)
\\
& \,\,\,\,+ 
\Big(
\frac13
\sum_{r=1}^3 
( \Dsa{r}\Gsqa
+ \Dsa{r}\Gskthree )
+ \Dsa{3}\Gsebac +
\Dsa{3}\Gsecab +\frac12 \Dsa{1}\Gsma
\Big)
\star  \nonumber 
\mathbb J(
\,\logo{4}{Va}{4}{38} \,)
\\
&
\,\,\,\,
+
\Big(
\frac13
\sum_{r=1}^3 
\Dsa{r}\Gsqb
+
\Dsa{3}\Gseabc
+
\sum_{r=1}^2
\Dsa{r} \Gsecab
+\frac12 \Dsa{1}\Gsmb
\Big) 
\star  \nonumber
\mathbb J(
\,\logo{4}{Vb}{4}{38} \,)
\\
&
\,\,\,\,
+
\Big(
\frac13
\sum_{r=1}^3 
\Dsa{r}\Gsqc
+
\Dsa{2}\Gseabc
+
\sum_{r=1}^2
\Dsa{r} \Gsebac
+\frac12 \Dsa{1}\Gsmc
\Big) 
\star   
\mathbb J(
\,\logo{4}{Vc}{4}{38} \,) \nonumber
\\
& \,\,\,\,+\nonumber
\Big(
\frac{1}{3!} \sum_{r=1}^3
\Dsa{r}\Gsmmm 
+ \Dsa{2} \Gseabc 
+\sum_{r=2,3}\sum_{i=1}^3
\Dsa{r} \Gsmi
\Big)
\star  \nonumber
\mathbb J(
\,\logo{2}{Melon}{3.1}{3}\sqcup 
\logo{2}{Melon}{3.1}{3} \,)
\\
{}
&  \,\,\,\,+																										
\Big\{
\frac14
 \sumtc (\Dsa{p} \Gomma+\Dsa{p}  \Gommi) \\
 &  \,\,\,\,+
\frac{1}{4!}
 \suml_{u=1}^4 (\Dsa{u} \Gommmm )   %\\  &   \qquad
 +  \Dsa{3} \Gomea   \bigg\}\star \J (\melon \sqcup  \melon \sqcup \melon) %\melon\sqcup \melon)
\\
&+
\Big\{ 
\frac18 \big( 
\suml_{r=1,2} \Dsa{r} \Goaa +
\suml_{r=3,4} (123)^*(\Dsa{r} \Goaa) 
\big)  +\frac14 \suml_{i\neq a} 
 (
\sumud \Dsa{r} \Goia )
\\ &
+
\frac13\suml_{q=2,3,4} \Dsa{q} \Gomkthree 
+\frac13\suml_{q=2,3,4}
 \Dsa{q} \Gomqa + \Dsa{4} \Gomec +  
\Dsa{4} \Gomeb 
\\
 &  + \frac14 \sumud \Dsa{r} \Gomma +   \Dsa{2} \GoWb  
 +(123)^*(\Dsa{3} \GoWc) 
\Big\}\star 
 \mathbb J ( \melon \sqcup \,\logo{4}{Va}{3.5}{38} )
  \\ 
&
 +\Big\{
 \frac{1}{8} \big(
\suml_{r=1,2} \Dsa{r} \Gobb
+\suml_{r=3,4} (123)^*\Dsa{r} \Gobb \big) 
+\frac14\sumtc (123)^*\Dsa{p}\Goia  
 \\ & +
 \frac{1}{4}\sumtc (123)^*\Dsa{p}\Gobc
 +\suml_{h=2,3} \Dsa{h} \Gomec +\Dsa{4} \Gomea +
   \frac14 \sumud \Dsa{r} \Gommb  
 \\ &  +\frac12\big( \Dsa{2} \GoRac + 
  (123)^* \big( \Dsa{3} \GoRac) \big)+ \Dsa{2} \GoPac 
 \Big\}\star 
 \mathbb J ( \melon \sqcup \,\logo{4}{Vb}{3.5}{38} )
  \\ & +\Big\{ \frac{1}{8} \big(
\suml_{r=1,2} \Dsa{r} \Gocc
+\suml_{r=3,4} (123)^*\Dsa{r} \Gocc +\frac14\sumtc (123)^*\Dsa{p}\Goca  \big)
  \\ & +\frac14(
\sumud \Dsa{r} \Gobc )
 +  
\suml_{h=2,3} \Dsa{h} \Gomeb + \Dsa{2} \Gomea +  
  \frac14 \sumud \Dsa{r} \Gommc   
 \\ &  +\frac12\big(\Dsa{2} \GoRab +(123)^*(\Dsa{3} \GoRab) \big)
 +\Dsa{2} \GoPab
 \Big\}\star 
 \mathbb J ( \melon \sqcup \,\logo{4}{Vc}{3.5}{38} ) 
\\ \nonumber
 &  
 +
 \Big\{ 
 \frac{1}{3} \Dsa{1} \Gomkthree + \Dsa{1} \GoXb + \Dsa{1} \GoXc \\
 & \qquad+
 \sumtc \Dsa{p} \GoXa +\frac{1}{4} \suml_{u=1}^4 \Dsa{u} \GoYa \!
 \Big\}  \star \J(\logo{6}{K33}{3.6}{3})  
\\
 & +
 \Big\{
 \frac13   
 \Dsa{1} \Gomqa +
  \Dsa{1} \GoPbc + \Dsa{1} \GoPcb +  
  \\
 & \qquad +
  \frac14 \suml_{r=1}^4 \Dsa{r}    \GoAa+ \sumud \Dsa{r} \GoXabc \Big\} \star \J (\logo{6}{Qa}{4.6}{25}) \\ 
 & +\Big\{ 
 \frac13
 \Dsa{1} \Gomqb + \frac12\suml_{h=2,3} 
  \Dsa{h} \GoRca\\
 & \quad +\frac14 \suml_{r=1}^4 \Dsa{r}  \GoAb  +  \Dsa{1} \GoPac + \sumud \Dsa{r} \GoPca\Big\} \star \J (\logo{6}{Qb}{4.6}{25}) \\
  & +\Big\{ 
 \frac13
 \Dsa{1} \Gomqc + \frac12\suml_{h=2,3} 
  \Dsa{h} \GoRba \\
 & \quad +\frac14 \suml_{r=1}^4 \Dsa{r}  \GoAc + \Dsa{1} \GoPab + \sumud \Dsa{r} \GoPba \Big\} \star \J (\logo{6}{Qc}{4.6}{25}) 
%     \end{align*} % \\
%  \begin{align*} 
\\
  &+\Big\{ \Dsa{1} \Gomea + \suml_{h=2,3}\Dsa{h} \GoWa  + (13)^*(\Dsa{3} \GoWb) +
  (13)^*(\Dsa{4} \GoWb) + 
  \\
 & \quad +
  \sumtc \big( \Dsa{p} \GoPab  
  + (13)^* \Dsa{p} \GoPac  \big)+ (13)^*(\Dsa{1} \GoWc) + (13)^*(\Dsa{2} \GoWc)  
  \\
 & \quad +\frac12\Big( 
    \Dsa{1} \GoRab  
     + (13)^* (\Dsa{1} \GoRac) + (13)^*( \Dsa{4} \GoRab) +\Dsa{4} \GoRac \Big)  
     \\
 & \quad +\sumud (13)^* \big (\Dsa{r} \GoS \big) +
  \frac14\big( (23)^*\Dsa{1} \GoCubo 
+ 
\Dsa{2} \GoCubo \\
 & \quad +(13)^* \Dsa{3} \GoCubo 
+(123)^* \Dsa{4} \GoCubo\big)
  \Big\} \star \J (\logo{6}{Eabc}{5}{3})\\ 
  &+ \Big\{ \Dsa{1} \Gomeb + \Dsa{1} \GoWa +  \Dsa{4} \GoWc  +  
  \frac12 \big( \sumud (13)^*(\Dsa{r} \GoRbc) 
  \\
 & \quad +
\sumtc \Dsa{p}\GoRbc  
 +  \Dsa{1}\GoRba +(13)^* \Dsa{4} \GoRba  \big) + \sumtc   \Dsa{p} \GoPba   
 \\
 & \quad 
  + \suml_{q=2,3,4}(13)^* \Dsa{q} \GoPbc   + (13)^*\Dsa{q} \GoXc  +\Dsa{3} \GoS 
  \\ 
 & \quad+ \frac14\big( \Dsa{1} \GoYc + (123)^*\Dsa{2}\GoYc + (23)^*\Dsa{3}\GoYc + (13)^*\Dsa{4}\GoYc \big)\big\}\star \J (\logo{6}{Ebac}{4.6}{25})
%   \end{align*}
%   \begin{align*}
\\   &+ \Big\{  \Dsa{1} \Gomec  +(13)^*( \Dsa{4} \GoWa)  +(13)^*(\Dsa{1} \GoWb) 
  \\
 & \quad 
  +  
  \frac12   \big( \sumud (13)^*(\Dsa{r} \GoRcb )  
 + 
\sumtc \Dsa{p}\GoRcb +  \Dsa{1}\GoRca 
 \\
 & \quad 
  +(13)^* \Dsa{4} \GoRca  \big) +\sumtc  \Dsa{p} \GoPca + \suml_{q=2,3,4}(13)^* \Dsa{q} \GoPcb 
  \\ &\quad+\sum_{q=2,3,4} (13)^*\Dsa{q} \GoXb +\sumtc  (13)^* \big( \Dsa{p} \GoS)\nonumber
  \\
 &\quad+\frac14\big(\Dsa{1} \GoYb 
   +(123)^* \Dsa{2}\GoYb + (23)^*\Dsa{3}\GoYb + (13)^*\Dsa{4}\GoYb \big) \big\} \star \J (\logo{6}{Ecab}{4.8}{25})\,.   \nonumber
\end{align*}  
\end{subequations}
\end{lem}

\begin{proof}
 See Appendix \ref{sec:proofoflemma}.
\end{proof}

Since we already derived the 2-point equation in 
\cite{fullward}, we immediately proceed with the 
higher-point functions. 

 \subsection{Four-point function SDEs for the $\phi^4_{3}$-theory} \label{sec:fourpoint}
 
We can use the colour symmetry in order to 
write down the equations for $\Gdos$ and $\Gtres$ 
from that for $\Guno$, which we now compute.
We will obtain, as stated by the theorem of previous section, the SDE for $\Guno$.
\par 
% \section{Six-point function equations for rank-$3$ theories}
We need first, to compute the functions $\mathfrak{f}_{\vunito}\hp a$
for each colour $a$. To this end, Lemma \ref{thm:ordersix} 
is used:
\begin{align*}
\mathfrak{f}_{\vunito}\hp 1 &=\frac13 \sum_{r=1}^3  (\Delta_{x_1,r} \Gsqu)+ 
\frac13 \sum_{r=1}^3  (\Delta_{x_1,r}  \Gskthree)  
+ (\Delta_{x_1,3} \Gsed + \Delta_{x_1,3} \Gset ) +
\Delta_{x_1,1} G\hp 6 _{|\raisebox{-.33\height}{\includegraphics[height=2ex]{graphs/3/Item2_Melon.pdf}}|
\raisebox{-.4\height}{\includegraphics[width=3ex]{graphs/3/Item4_V1v.pdf}}|}  \,,
\\
\mathfrak{f}_{\vunito}\hp 2&= \frac13\sum_{r=1}^3(\Delta_{y_2,r} \Gsqu)  
+\Delta_{y_2,3} \Gsed 
+\sum_{r=1}^2 (\Delta_{y_2,r} \Gset) 
+\frac{1}{2} (\Delta_{y_2,1} 
G\hp 6 _{|\raisebox{-.33\height}{\includegraphics[height=2ex]{graphs/3/Item2_Melon.pdf}}|
\raisebox{-.4\height}{\includegraphics[width=3ex]{graphs/3/Item4_V1v.pdf}}|} )\,,
\\
\mathfrak{f}_{\vunito}\hp 3&=
 \frac13\sum_{r=1}^3(\Delta_{y_3,r} \Gsqu)
 +\Delta_{y_2,3} \Gset 
+\sum_{r=1}^2 (\Delta_{y_3,r} \Gsed) 
+\frac{1}{2} (\Delta_{y_3,1} 
G\hp 6 _{|\raisebox{-.33\height}{\includegraphics[height=2ex]{graphs/3/Item2_Melon.pdf}}|
\raisebox{-.4\height}{\includegraphics[width=3ex]{graphs/3/Item4_V1v.pdf}}|} )\,.
\end{align*}
Also notice that 
\begin{align*}
\varsigma_1(\vuno;1,2)= \meloncito\sqcup \meloncito, \quad  \varsigma_2(\vuno;1,2)= \vtres, \quad \varsigma_3(\vuno;1,2)=\vdos\,.
\end{align*}
The derivatives with respect to these, evaluated in $\Xb=(\xb,\yb)$ are then
\[
\Gmelon(\xb) \cdot \Gmelon(\yb) + \Gcmm (\xb,\yb)\,,\quad
\Gtres  (\xb,\yb) 
\,,\quad \mbox{and } 
\Gdos  (\xb,\yb) 
\,.
\]
respectively. Letting $\sb=(x_1,y_2,y_3)$
and $\tb=(y_1,x_2,x_3)$ and using Theorem \ref{thm:SDEs}, one obtains
\begin{align} \label{eq:fourSDE}
&\bigg(1+\frac{2\lambda}{E_{\sb}}\sum_a\sum_{\mathbf q_{\hat a}} \Gmelon(s_a,\mathbf q_{\hat a})\bigg)\cdot \Guno(\xb,\yb) 
  \\
& =\frac{(-2\lambda)}{E_{\sb}} \suml_{a=1}^3 \Bigg\{  \nonumber
\suml_{\hat \sigma \in \Autc(\vunito)} \sigma^* \mathfrak{f}\hp{a}_{\vunito}(\Xb) 
\\
\nonumber 
& \qquad\qquad \qquad \quad+ 
\suml_{\rho>1 }  \frac{Z_0\inv}{  E(y^\rho_a,s_a)} \dervpar{Z[J,J]}{\varsigma_a(\vunito ;1,\rho) }  (\Xb)-
\suml_{\rho>1 }  \frac{Z_0\inv}{  E(y^\rho_a,s_a)} \dervpar{Z[J,J]}{\varsigma_a(\vunito ;1,\rho) }  (\Xb|_{s_a\to y_a^\rho})
% \label{eq:sixSDEs}
\\
& \nonumber
\qquad\qquad \qquad- \suml_{b_a}\frac{1}{E(s_a,b_a)} \big[ \Guno(\Xb) - \Guno(\Xb|_{s_a\to b_a}) \big] \Bigg\}  
\\
& \nonumber
=\frac{(-2\lambda)}{E_{x_1y_2y_3}} \bigg\{   ( \delta_{\mathbf u}^{\mathbf x}\delta_{\mathbf v}^{\mathbf y}+
\delta_{\mathbf u}^{\mathbf y}\delta_{\mathbf v}^{\mathbf x}) \cdot
\Big(
  \frac13 \suml_{r=1}^3  (\Delta_{x_1,r} \Gsqu)+ 
\frac13 \suml_{r=1}^3  (\Delta_{x_1,r}  \Gskthree)  
\\ &\qquad\qquad\qquad \nonumber 
+ (\Delta_{x_1,3} \Gsed + \Delta_{x_1,3} \Gset )  +
\Delta_{x_1,1} G\hp 6 _{|\raisebox{-.33\height}{\includegraphics[height=2ex]{graphs/3/Item2_Melon.pdf}}|
\raisebox{-.4\height}{\includegraphics[width=3ex]{graphs/3/Item4_V1v.pdf}}|} 
\\
 & \qquad\qquad \qquad+\frac13\suml_{r=1}^3(\Delta_{y_2,r} \Gsqu)   \nonumber 
+\Delta_{y_2,3} \Gsed 
+\suml_{r=1}^2 (\Delta_{y_2,r} \Gset) 
+\frac{1}{2} (\Delta_{y_2,1} 
G\hp 6 _{|\raisebox{-.33\height}{\includegraphics[height=2ex]{graphs/3/Item2_Melon.pdf}}|
\raisebox{-.4\height}{\includegraphics[width=3ex]{graphs/3/Item4_V1v.pdf}}|} ) 
\\
 & \qquad\qquad \qquad   +
 \frac13\suml_{r=1}^3(\Delta_{y_3,r} \Gsqu)
 +\Delta_{y_2,3} \Gset 
+\suml_{r=1}^2 (\Delta_{y_3,r} \Gsed) 
+\frac{1}{2} (\Delta_{y_3,1} 
G\hp 6 _{|\raisebox{-.33\height}{\includegraphics[height=2ex]{graphs/3/Item2_Melon.pdf}}|
\raisebox{-.4\height}{\includegraphics[width=3ex]{graphs/3/Item4_V1v.pdf}}|} \Big) (\mathbf u,\mathbf v)
\nonumber 
\\ &\qquad\qquad \qquad+\frac{1}{E(y_1,x_1)} \big[ (\Gmelon(\xb) - \Gmelon(y_1,x_2,x_3) )\cdot \Gmelon(\yb) \nonumber 
\\ 
& \qquad\qquad \qquad \qquad \qquad \qquad \qquad+ \Gcmm (\xb,\yb) -\Gcmm (y_1,x_2,x_3,\yb) \big] \nonumber \\
&\nonumber 
\qquad\qquad \qquad+\frac{ \Gtres(\xb,\yb)- \Gtres(\xb,y_1,x_2,y_3)}{E(x_2,y_2)} + \frac{ \Gdos(\xb,\yb)-
 \Gdos(\xb,y_1,y_2,x_3)}{E(x_3,y_3)} 
\\
& \qquad\qquad \qquad-\sum_{b_1} \frac{1}{E(x_1,b_1)} \big(\Guno(\xb,\yb)-\Guno(b_1,x_2,x_3,\yb)\big)\nonumber\\
&\qquad\qquad \qquad-\sum_{b_2} \frac{1}{E(y_2,b_2)} \big(\Guno(\xb,\yb)-\Guno(\xb,y_1,b_2,y_3)\big)\nonumber\\
&\qquad\qquad \qquad-\sum_{b_3} \frac{1}{E(y_3,b_3)} \big(\Guno(\xb,\yb)-\Guno(\xb,y_1,y_2,b_3)\big)\bigg\}\nonumber
\end{align}

\subsection[subsec]{\for{toc}{The Schwinger-Dyson equation for $G\hp6_{K_{\mtr c}{(3,3)}}$}\except{toc}{The Schwinger-Dyson equation for $G\hp6_{\protect \kthree}$}}
We now derive the whole set of six-point function equations for the $\phi_3^4$-theory. 
They hold for any model whose boundary sector is the whole of $ \dvGrph{3}$.
From Prop. \ref{thm:ordersix},  one can read off the $\mathfrak{f}\hp a_{\B}$ functions. \\

For the boundary graph  $\kthree$, one has, for each colour $a=1,2,3,$
\[
\mathfrak{f}\hp a_{\kthree}= \frac{1}{3} \Dsa{1} \Gomkthree + \Dsa{1} \GoXb + \Dsa{1} \GoXc + \sumtc \Dsa{p} \GoXa +\frac{1}{4} \sum_{u=1}^4 \Dsa{u} \GoYa 
.\] 
Departing from Theorem \ref{thm:SDEs}, this last very expression 
allows now for an explicit derivation of the 
equation for $\Gskthree$. Namely, for $\Xb=(\xb^1,\xb^2,\xb^3)=(\xb,\yb,\zb)$, and choosing 
$\sb=(x_1,y_2,z_3)$,
\begin{align} \nonumber
&\qquad \bigg(1+\frac{2\lambda}{E_{x_1y_2z_3}} \suml_{a=1}^3 \suml_{\mathbf q_{\hat a} }
\Gmelon (s_a,\mathbf q_{\hat a})\bigg)
\Gskthree(\Xb)
\\
& =\frac{(-2\lambda)}{E_{x_1y_2z_3}} \suml_{a=1}^3 \Bigg\{ 
\suml_{\hat \sigma \in \Autc(\kthree)} \sigma^* \mathfrak{f}\hp{a}_{\kthree}(\Xb) 
+ \suml_{\rho>1 }  \frac{Z_0\inv}{  E(y^\rho_a,s_a)}\bigg[ \dervpar{Z[J,J]}{\varsigma_a(\kthree ;1,\rho) }  (\Xb)
-
\dervpar{Z[J,J]}{\varsigma_a(\kthree ;1,\rho) }  (\Xb|_{x_a\to s_a}) \bigg]
\nonumber
\\
& \label{eq:sixSDEs}
\qquad\qquad \qquad- \suml_{b_a}\frac{1}{E(s_a,b_a)} \big[ \Gskthree(\Xb) - \Gskthree(\Xb|_{s_a\to b_a}) \big] \Bigg\}   
\end{align}
One finds:
\begin{align*}
&\quad \suml_{a=1}^3  
\suml_{\rho>1 }  \frac{Z_0\inv}{  E(y^\rho_a,s_a)} \dervpar{Z[J,J]}{\varsigma_a(\kthree ;1,\rho)(\Xb) }  \nonumber
\\
& = \bigg(
\frac{1}{E(y_1,x_1)}
(23)^*\Gseu+
\frac{1}{E(z_1,x_1)}(13)^*\Gseu+
\frac{1}{E(z_2,y_2)}
(123)^*\Gsed \\ & \,\,+ 
\frac{1}{E(x_2,y_2)}
(132)^*\Gsed + 
\frac{1}{E(x_3,z_3)}
(13)^*\Gset +
\frac{1}{E(y_3,z_3)}
(123)^* \Gset 
\bigg)(\Xb) \,,
 \end{align*}
where we recall that for a function of three arguments and 
$\sigma\in\Sym_3$, $\sigma^* f$ is given by Prop. \ref{thm:promedio_grupo}. \par
The meaning of the $ \mathfrak f\hp a_{\kthree}$ summed over colours $a$ and over the automorphism group
is 
\begin{align*} 
\phantom{=}& \suml_a \suml_{\hat\pi\in \Z_3} \pi^*\Big\{
\frac{1}{3}  
\Gomkthree + \Dsa{1} \GoXb + \Dsa{1} \GoXc 
\\
&\quad \qquad \quad +\sumtc \Dsa{p} \GoXa +\frac{1}{4} \sum_{u=1}^4 \Dsa{u} \GoYa \Big\} \\
=&  \suml_{\hat\pi\in \Z_3} \pi^*\bigg\{
 \frac13\sum_a\Dsa{1}\Gomkthree + \big( \Delta_{x_1,1} \GoXd+  \Delta_{x_1,1} \GoXt  +   \Delta_{y_2,1}  \GoXu +  \Delta_{y_2,1} \GoXt 
 \\& \qquad \qquad + \Delta_{z_3,1} \GoXd 
  +\Delta_{z_3,1} \GoXt +\sumtc \Delta_{x_1,p} \GoXu +  \Delta_{y_2,p} \GoXd  + \Delta_{z_3,p} \GoXt \big) \\
& \qquad\qquad + \frac14\suml_{u=1}^4  \big(\Delta_{x_1,u} \GoYu
 + \Delta_{y_2,u} \GoYd+  \Delta_{z_3,u} \GoYt\big) \bigg\}
\end{align*}
where $\Z_3$ is generated rotation of $\kthree$ by $2\pi/3$, that is $\hat \pi$
is the liftings of the identity, of $(123)$ and $(132)$ in $\Sym_3$. Finally, the 
difference-term is
\begin{align*}
\suml_a
\suml_{b_a}\frac{1}{E(s_a,b_a)} \big[ \Gskthree(\Xb) - \Gskthree(\Xb\big|_{s_a\to b_a}) \big]& = 
\suml_{b_1}\frac{1}{E(x_1,b_1)} \big[ \Gskthree(\Xb) - \Gskthree(b_1,x_2,x_3;\yb;\zb) \big] \\
& + \suml_{b_2}\frac{1}{E(y_2,b_2)} \big[ \Gskthree(\Xb) - \Gskthree(\xb; y_1,b_2,y_3;\zb)  \big]\\
& + \suml_{b_3}\frac{1}{E(z_3,b_3)} \big[ \Gskthree(\Xb) - \Gskthree(\xb;\yb; z_1,z_2,b_3) \big]
\end{align*}
Explicitly,
\begin{align}  \nonumber
\phantom{=}&\,\,-\bigg(\frac{E_{x_1y_2z_3}}{2\lambda}+ \suml_{m,n}\big[\Gmelon (x_1,m,n)+ \Gmelon (m,y_2,n)+\Gmelon(m,n,z_3)\big]\bigg)
\cdot
\Gskthree(\Xb)    \\
& =   \nonumber 
\bigg( 
\frac{1}{E(y_1,x_1)}
\big[\Gseu(\xb,\zb,\yb)-\Gseu(y_1,x_2,x_3,\zb,\yb)  \big]
  \\ & \quad \nonumber
+\frac{1}{E(z_1,x_1)}\big[\Gseu (\zb,\yb,\xb)- \Gseu (\zb,\yb,z_1,x_2,x_3) \big] 
\\ 
& \quad \nonumber +
\frac{1}{E(z_2,y_2)} \big[\Gsed (\zb,\xb,\yb)-\Gsed(\zb,\xb,y_1,z_2,y_3) \big] 
\\ & \,\,\,\,\,\,+ \nonumber 
\frac{1}{E(x_2,y_2)} 
\big[
\Gsed (\yb,\zb,\xb)- \Gsed (y_1,x_2,y_3,\zb,\xb)\big]
\\ & \,\,\,\,\,\,+ 
\frac{1}{E(x_3,z_3)}\nonumber  
\big[ \Gset (\zb,\yb,\xb) - \Gset(z_1,z_2,x_3,\yb,\xb)    \big]
\\ & \,\,\,\,\,\,+ 
\frac{1}{E(y_3,z_3)} \big[\Gset(\yb,\xb,\zb)-\Gset (\yb,\xb,z_1,z_2,y_3) \big]
\bigg) 
\\  & \quad +\suml_{\hat\pi\in \Z_3} \pi^*\Big[
 \frac13 \sum_a\Dsa{1}\Gomkthree    \label{eq:sixSDEexplicit}
\\ \nonumber
&
\quad  + \big( \Delta_{x_1,1} \GoXd+  \Delta_{x_1,1} \GoXt  +   \Delta_{y_2,1}  \GoXu +  \Delta_{y_2,1} \GoXt 
 + \Delta_{z_3,1} \GoXd +\Delta_{z_3,1} \GoXt
  \nonumber\\&  \,\,\,\, \qquad  
   +\sumtc \Delta_{x_1,p} \GoXu +  \Delta_{y_2,p} \GoXd  + \Delta_{z_3,p} \GoXt \big) 
  + \frac14 \suml_{u=1}^4  \big(\Delta_{x_1,u} \GoYu   \nonumber  \\ \nonumber
& \qquad    \,\,\,\, 
 + \Delta_{y_2,u} \GoYd+  \Delta_{z_3,u} \GoYt\big) \Big](\Xb) 
 - \suml_{b_1}\frac{1}{E(x_1,b_1)} \big[ \Gskthree(\Xb) - \Gskthree(b_1,x_2,x_3;\yb;\zb) \big] \\ \nonumber
& \hphantom{(\Xb)}\quad\qquad  \qquad  \qquad\qquad\qquad \qquad \qquad \,\,\,  
- \suml_{b_2}\frac{1}{E(y_2,b_2)} \big[ \Gskthree(\Xb) - \Gskthree(\xb; y_1,b_2,y_3;\zb)  \big]\\  \nonumber
& \hphantom{(\Xb)}\quad\qquad  \qquad \qquad\qquad\qquad \qquad \qquad \,\,\, 
- \suml_{b_3}\frac{1}{E(z_3,b_3)} \big[ \Gskthree(\Xb) - \Gskthree(\xb;\yb; z_1,z_2,b_3) \big] \,. 
\end{align}

\subsection[subsec]{\for{toc}{The Schwinger-Dyson equation for $G\hp6_{\mtc{Q}_a}$}\except{toc}{The Schwinger-Dyson equation for $G\hp 6_{\protect \icono{6}{Qa}{2.5}{2}}$}}

First, we compute $\mathfrak{b}_a$-terms for $\mathcal{Q}_a$, one by one are:
\begin{subequations} \label{eq:switchQ}
\begin{align}
 \frac{1}{Z_0}\dervpar{Z[J,J]}{\varsigma_1(\icono{6}{Q1s}{3}{2} ;1,2)}& =  \Guno (\xb,\zb) \cdot \Gmelon(\yb) + (12)^*\Gsmu (\Xb)\,, & 
  &  \tag{\ref{eq:switchQ}a}   \\ 
 \frac{1}{Z_0}\dervpar{Z[J,J]}{\varsigma_1(\icono{6}{Q1s}{3}{2} ;1,3)}&=  \Guno (\yb,\zb) \cdot \Gmelon(\xb) + \Gsmu(\Xb)  \,,&  \tag{\ref{eq:switchQ}b} \\
 \frac{1}{Z_0}\dervpar{Z[J,J]}{\varsigma_2(\icono{6}{Q1s}{3}{2} ;1,2)}& = (23)^*\Gsed (\Xb)\,, &   \hspace{-1cm}
 \hspace{-4cm} \frac{1}{Z_0}\dervpar{Z[J,J]}{\varsigma_2(\icono{6}{Q1s}{3}{2} ;1,2)} &= (13)^*\Gsed (\Xb) \,, \tag{\ref{eq:switchQ}c,d} \\
 \frac{1}{Z_0}\dervpar{Z[J,J]}{\varsigma_3(\icono{6}{Q1s}{3}{2} ;1,3)}& =  (23)^* \Gset (\Xb)\,,  \hspace{-1cm}
 & \hspace{-4cm}    \frac{1}{Z_0}\dervpar{Z[J,J]}{\varsigma_3(\icono{6}{Q1s}{3}{2} ;1,3)}&= (13)^* \Gset (\Xb)\,.  \tag{\ref{eq:switchQ}e,f}
 \end{align}
\end{subequations}
 
We stepwise collect the $\mathfrak{f}\hp a_{\icono{6}{Q1s}{2}{2}}$-terms from the expansion in Prop. \ref{thm:ordersix}: 
% Notice that $\mathfrak{f}\hp 1_{\icono{6}{Q1s}{2}{2}}$ is the coefficient
% (function) of $\J(\icono{6}{Q1s}{2}{2})$ in the functional $Y_{x_1}\hp 1 [J,\bJ]$; 
% $\mathfrak{f}\hp 2_{\icono{6}{Q1s}{2}{2}}$ and setting $a=2,b=1,c=3$ there
% and thus reading the off the coefficient of $ \mathfrak{f}\hp 1_{\icono{6}{Qbs}{2}{2}}$;
% similarly, $\mathfrak{f}\hp 2_{\icono{6}{Q1s}{2}{2}}$ is read by setting $a=3,b=1,c=2$.

\begin{align}
  \mathfrak{f}\hp 1_{\icono{6}{Q1s}{2}{2}}\mbox{ is the coefficient of} & \mbox{ $\J(\icono{6}{Qa}{3.4}{25})$ in $ Y_{x_1}\hp 1 [J,\bJ] $ with $a=1,b=2,c=2$, } \\
  \mathfrak{f}\hp 2_{\icono{6}{Q1s}{2}{2}}\mbox{ is the coefficient of} & \mbox{ $\J(\icono{6}{Qb}{3.6}{25})$ in $ Y_{y_2}\hp 2 [J,\bJ]$ with $a=2,b=1,c=3$, and } \\
  \mathfrak{f}\hp 3_{\icono{6}{Q1s}{2}{2}}\mbox{ is the coefficient of} & \mbox{ $\J(\icono{6}{Qb}{3.6}{25})$ in $Y_{y_3}\hp 3 [J,\bJ]$ with $a=3,b=1,c=2$,} 
\end{align}
namely 
\begin{align*}
  \mathfrak{f}\hp 1_{\icono{6}{Q1s}{2}{2}} & =  \frac13   
 \Delta_{x_1,1} \Gomqu +
  \Delta_{x_1,1} \GoPdt + \Delta_{x_1,1} \GoPtd +  \frac14 \suml_{r=1}^4 \Delta_{x_1,r}    \GoAu+ \sumud \Delta_{x_1,r} \GoXu  
  \,,
  \\
  \mathfrak{f}\hp 2_{\icono{6}{Q1s}{2}{2}} & =   
 \frac13
 \Delta_{y_2,1} \Gomqu + \frac12\suml_{h=2,3} 
  \Delta_{y_2,h} \GoRtd +\frac14 \suml_{r=1}^4 \Delta_{y_2,r}  \GoAu +\Delta_{y_2,1} \GoPdt  \\
  & \qquad \qquad \hspace{8.9cm}+  \sumud \Delta_{y_2,r} \GoPtd \,,\\
  \mathfrak{f}\hp 3_{\icono{6}{Q1s}{2}{2}} & =  
 \frac13
 \Delta_{y_3,1} \Gomqu + \frac12\suml_{h=2,3} 
  \Delta_{y_3,h} \GoRdt +\frac14 \suml_{r=1}^4 \Delta_{y_3,r}  \GoAu + \Delta_{y_3,1} \GoPtd   \\
   & \qquad \qquad \hspace{8.9cm} +\sumud \Delta_{y_3,r} \GoPdt\,.
 \end{align*}
Explicitly,
\begin{align}  
\phantom{=}&\,\,-\bigg(\frac{E_{x_1y_2y_3}}{2\lambda}+ \suml_{m,n}\big[\Gmelon (x_1,m,n)+ \Gmelon (m,y_2,n)+\Gmelon(m,n,y_3)\big]\bigg)
\cdot
\Gsqu (\Xb)  \nonumber\\
& =  \frac{1}{E(y_1,x_1)} \big[ \Guno (\xb,\zb)-  \Guno (y_1,x_2,x_3,\zb) \big] \cdot \Gmelon(\yb) 
\nonumber 
\\ & 
+   \nonumber
 \frac{1}{E(z_1,x_1)}\big[ \Gmelon(\xb) - \Gmelon(z_1,x_2,x_3)\big] \cdot \Guno(\yb,\zb)
 \\ \nonumber
&  +   \frac{1}{E(y_1,x_1)} \big[ \Gsmu (\yb,\xb,\zb) -  \Gsmu (\yb,y_1,x_2,x_3,\zb)\big] 
  \\ \nonumber
&  
+ \frac{1}{E(z_1,x_1)}\big[ \Gsmu(\xb,\yb,\zb)-\Gsmu(z_1,x_2,x_3,\yb,\zb) \big] 
  \\ \nonumber
&  +  
\frac{1}{E(z_2,y_2)} \big[ \Gsed(\xb,\zb,\yb) - \Gsed(\xb,\zb,y_1,z_2,y_3)  \big]\nonumber \\
&  +\frac{1}{E(x_2,y_2)}\big[\Gsed (\zb,\yb,\xb)-\Gsed (\zb,y_1,x_2,y_3,\xb)\big]   \label{eq:sixSDEQ} 
  \\ \nonumber
&   
+ \frac{1}{E(z_3,y_3)} \big[ \Gset (\xb,\zb,\yb)-\Gset(\xb,\zb,y_1,y_2,z_3)\big]   \\ \nonumber
&  +   \frac{1}{E(x_3,y_3)}\big[  \Gset (\zb,\yb,\xb)-  \Gset (\zb,y_1,y_2,x_3,\xb)
\big]   \nonumber
 \nonumber  \\ \nonumber
&  +\suml_{\hat\pi\in \Z_3} \pi^*\Big[ 
\frac13   
 \Delta_{x_1,1} \Gomqu +
  \Delta_{x_1,1} \GoPdt + \Delta_{x_1,1} \GoPtd 
  \\ \nonumber
& + 
  \frac14 \suml_{r=1}^4 \Delta_{x_1,r}    \GoAu+ \sumud \Delta_{x_1,r} \GoXu  
%   \\ &
+
  \frac13\nonumber
 \Delta_{y_2,1} \Gomqu + \frac12\suml_{h=2,3} 
  \Delta_{y_2,h} \GoRtd  
  \\ & 
  \nonumber +
  \frac14 \suml_{r=1}^4 \Delta_{y_2,r}  \GoAu  +  \Delta_{y_2,1} \GoPdt + \sumud \Delta_{y_2,r} \GoPtd
%  \\ & 
 + 
\frac13
 \Delta_{y_3,1} \Gomqu 
 \\ & 
 +\frac12\suml_{h=2,3} 
  \Delta_{y_3,h} \GoRdt +\frac14 \suml_{r=1}^4 \Delta_{y_3,r}  \GoAu  +  \Delta_{y_3,1} \GoPtd + \sumud \Delta_{y_3,r} \GoPdt
\Big](\Xb) \nonumber
\\ \nonumber 
 & - \suml_{b_1}\frac{1}{E(x_1,b_1)} \big[ \Gsqu(\Xb) - \Gsqu(b_1,x_2,x_3;\yb;\zb) \big] \\ \nonumber
& - \suml_{b_2}\frac{1}{E(y_2,b_2)} \big[ \Gsqu(\Xb) - \Gsqu(\xb; y_1,b_2,y_3;\zb)  \big]\\  \nonumber
&- \suml_{b_3}\frac{1}{E(y_3,b_3)} \big[ \Gsqu(\Xb) - \Gsqu(\xb;y_1,y_2,b_3;\zb) \big] \,. \nonumber
\end{align}

 \subsection[subsec]{\for{toc}{The Schwinger-Dyson equation for $G\hp6_{\mathcal E_a}$} \except{toc}
 {The Schwinger-Dyson equation for $G\hp6_{\protect\icono{6}{Eipure}{2.2}{20}}$}}
Concerning the correlation function $\Gseu$, the terms with swapping 
black vertices are
\begin{align}
 \Gsqt,  \Gskthree , \Gsmt, \Gmelon\cdot \Gtres,  (132)^*\Gset, \Gsed ,(12)^* \Gsed  \,,
 \end{align}
 which need to be divided by differences of propagators. We now 
find the rest of the terms. 
Since $\Autc(\icono{6}{E123}{3.2}{23})$ is trivial,
the contribution of the $\icono{6}{E123}{3.2}{23}$-derivative on $\sum_a\ysa$ is given by 
the sum $\sum_a\mathfrak f \hp a_{\icono{6}{E123}{2}{38}}$
where, for each colour $a$:
% \begin{subequations}
\begin{align*}
\mathfrak f \hp 1_{\icono{6}{E123}{2}{38}} & =  \Delta_{x_1,1} \Gomeu + \suml_{h=2,3}\Delta_{x_1,h} \GoWu  + (13)^*(\Delta_{x_1,3} \GoWd) +
  (13)^*(\Delta_{x_1,4} \GoWd)  \\  & +\sumtc \big( \Delta_{x_1,p} \GoPud  
   + (13)^* \Delta_{x_1,p} \GoPut  \big)+ (13)^*(\Delta_{x_1,1} \GoWt) + (13)^*(\Delta_{x_1,2} \GoWt)  
   \\  &
   + \frac12\Big( 
    \Delta_{x_1,1} \GoRud  
   + (13)^* (\Delta_{x_1,1} \GoRut) + (13)^*( \Delta_{x_1,4} \GoRut) 
  +\Delta_{x_1,4} \GoRut \Big)   
  \\ & +\sumud (13)^* \big (\Delta_{x_1,r} \GoSN \big)
   +  \frac14\Big( (23)^*\Delta_{x_1,1} \GoCuboN 
+ 
\Delta_{x_1,2} \GoCuboN
+(13)^* \Delta_{x_1,3} \GoCuboN  \\
& \hspace{10.1cm}
+(123)^* \Delta_{x_1,4} \GoCuboN \Big)  \\
\mathfrak f \hp 2 _{\icono{6}{E123}{2}{38}} & = 
\Delta_{y_2,1} \Gomeu + \Delta_{y_2,1} \GoWt +  \Delta_{y_2,4} \GoWt  +  
  \frac12 \big( \sumud (13)^*(\Delta_{y_2,r} \GoRut)   
\\  
  &
  +\sumtc \Delta_{y_2,p}\GoRut  
 +  \Delta_{y_2,1}\GoRud +(13)^* \Delta_{y_2,4} \GoRud  \big)   \\
 & + \sumtc   \Delta_{y_2,p} \GoPud
  + \suml_{q=2,3,4}(13)^* \Delta_{y_2,q} \GoPut   + (13)^*\Delta_{y_2,q} \GoXt  
  \\
 &+\Delta_{y_2,3} \GoSN  + \frac14\big( \Delta_{y_2,1} \GoYt + (123)^*\Delta_{y_2,2}\GoYt + (23)^*\Delta_{y_2,3}\GoYt
 + (13)^*\Delta_{y_2,4}\GoYt \big)  \\
 \mathfrak f \hp 3 _{\icono{6}{E132}{2}{38}}&=   (13)^*\Delta_{x_3,1} \Gomeu 
     + \Delta_{x_3,1} \GoWt +  \Delta_{x_3,4} \GoWd  +  
  \frac12 \big( \sumud (13)^*(\Delta_{x_3,r} \GoRud)
 \\  
  &  +
\sumtc \Delta_{x_3,p}\GoRud  +  \Delta_{x_3,1}\GoRut +(13)^* \Delta_{x_3,4} \GoRut  \big) + \sumtc   \Delta_{x_3,p} \GoPut   
  \\ & + \suml_{q=2,3,4}(13)^* \Delta_{x_3,q} \GoPud   + (13)^*\Delta_{x_3,q} \GoXd  
  +\Delta_{x_3,3} \GoSN  
  \\
 &+ \frac14\big( \Delta_{x_3,1} \GoYd + (123)^*\Delta_{x_3,2}\GoYd + (23)^*\Delta_{x_3,3}\GoYd 
 + (13)^*\Delta_{x_3,4}\GoYd \big)  
 \end{align*}
 Here $\sb=(x_1,y_2,x_3)$.  Therefore, the explicit equation is 
\begin{align}  
\phantom{=}&\,-\bigg(\frac{E_{x_1y_2x_3}}{2\lambda}+ 
\suml_{m,n}\big[\Gmelon (x_1,m,n)+ \Gmelon (m,y_2,n)+\Gmelon(m,n,x_3)\big]\bigg)
\cdot
\Gseu (\Xb) \label{eq:sixSDEE}     \\
& =   \nonumber
\frac{1}{E(y_1,x_1)} \big[
\Gsqt(\xb,\yb,\zb)- \Gsqt(y_1,x_2,x_3,\yb,\zb)\big] 
\nonumber   \\ \nonumber
&
 +\frac{1}{E(z_1,x_1)}\big[
 \Gskthree (\xb,\yb,\zb)-\Gskthree (z_1,x_2,x_3,\yb,\zb) \big] 
\nonumber   \\ \nonumber
&
 +
 \frac{1}{E(x_2,y_2)}
  \big[ \Gsmt (\xb,\yb,\zb)  - \Gsmt (\xb,y_1,x_2,y_3,\zb) \big] \nonumber   \\ \nonumber
&
 +
 \frac{1}{E(z_2,y_2)} \big[\Gset  (\yb,\zb,\xb)- \Gset  (y_1,z_2,y_3,\zb,\xb) \big]
 \\ &  + \nonumber
 \frac{1}{E(z_3,x_3)} \big[
 \Gsed (\xb,\yb,\zb)  - \Gsed (x_1,x_2,z_3,\yb,\zb)   \big] \\ &  + \nonumber 
 \frac{1}{E(y_3,x_3)}  \big[\Gsed  (\yb,\xb,\zb)-\Gsed(\yb,x_1,x_2,y_3,\zb)   \big]
 \\ &  + \nonumber 
 \frac{1}{E(x_2,y_2)}
\big(\Gtres (\yb,\mathbf z) -\Gtres (y_1,x_2,y_3,\mathbf z) \big) \cdot \Gmelon (\xb) 
\nonumber   \\ \nonumber
&
 + \bigg\{ 
 \Delta_{x_1,1} \Gomeu + \suml_{h=2,3}\Delta_{x_1,h} \GoWu  + (13)^*(\Delta_{x_1,3} \GoWd) +
  (13)^*(\Delta_{x_1,4} \GoWd)  \\  \nonumber
  & + \sumtc \big( \Delta_{x_1,p} \GoPud  + (13)^* \Delta_{x_1,p} \GoPut  \big)+ (13)^*(\Delta_{x_1,1} \GoWt) + (13)^*(\Delta_{x_1,2} \GoWt)   \\  \nonumber
  & + \frac12\Big( 
    \Delta_{x_1,1} \GoRud  
    + (13)^* (\Delta_{x_1,1} \GoRut) + (13)^*( \Delta_{x_1,4} \GoRut) +\Delta_{x_1,4} \GoRut \Big)   
   \\  \nonumber
  & 
  +\sumud (13)^* \big (\Delta_{x_1,r} \GoSN \big) + \Delta_{y_2,1} \Gomeu + \Delta_{y_2,1} \GoWt +  \Delta_{y_2,4} \GoWt   
  \\ \nonumber& +  \frac14\big( (23)^*\Delta_{x_1,1} \GoCuboN 
+ 
\Delta_{x_1,2} \GoCuboN
+(13)^* \Delta_{x_1,3} \GoCuboN 
+(123)^* \Delta_{x_1,4} \GoCuboN \big)
  \nonumber  \\
  \nonumber  &+   
  \frac12 \big( \sumud (13)^*(\Delta_{y_2,r} \GoRut)  +
\sumtc \Delta_{y_2,p}\GoRut + \Delta_{y_2,1}\GoRud 
\\
 \nonumber 
  &+  (13)^* \Delta_{y_2,4} \GoRud  \big) + \sumtc   \Delta_{y_2,p} \GoPud
  + \suml_{q=2,3,4}(13)^* \Delta_{y_2,q} \GoPut   + (13)^*\Delta_{y_2,q} \GoXt  
\\ \nonumber
 &+\Delta_{y_2,3} \GoSN  + \frac14\big( \Delta_{y_2,1} \GoYt + (123)^*\Delta_{y_2,2}\GoYt + (23)^*\Delta_{y_2,3}\GoYt
 + (13)^*\Delta_{y_2,4}\GoYt \big)  \\ \nonumber
 &   + (13)^* \Big[ (13)^*\Delta_{x_3,1} \Gomeu 
     + \Delta_{x_3,1} \GoWt +  \Delta_{x_3,4} \GoWd  +  
  \frac12 \big( \sumud (13)^*(\Delta_{x_3,r} \GoRud)
%   \end{align} 
%  \begin{align}   
 \\ \nonumber
  &  +
\sumtc \Delta_{x_3,p}\GoRud  +  \Delta_{x_3,1}\GoRut +(13)^* \Delta_{x_3,4} \GoRut  \big)
  \\  \nonumber &  + \sumtc   \Delta_{x_3,p} \GoPut   + \suml_{q=2,3,4}(13)^* \Delta_{x_3,q} \GoPud   + (13)^*\Delta_{x_3,q} \GoXd  
  +\Delta_{x_3,3} \GoSN  
  \\ \nonumber
 &+ \frac14\big( \Delta_{x_3,1} \GoYd + (123)^*\Delta_{x_3,2}\GoYd + (23)^*\Delta_{x_3,3}\GoYd 
 + (13)^*\Delta_{x_3,4}\GoYd \big) \Big] 
\bigg\}(\Xb) \nonumber
\\
 & - \suml_{b_1}\frac{1}{E(x_1,b_1)} \big[ \Gseu(\Xb) - \Gseu(b_1,x_2,x_3;\yb;\zb) \big] \nonumber
 \\
& - \suml_{b_2}\frac{1}{E(y_2,b_2)} \big[ \Gseu(\Xb) - \Gseu(\xb; y_1,b_2,y_3;\zb)  \big] \nonumber
\\  
&
- \suml_{b_3}\frac{1}{E(x_3,b_3)} \big[ \Gseu(\Xb) - \Gseu(x_1,x_2,b_3;\yb; \zb) \big] \,. \nonumber
\end{align}
% \vspace{12cm}

\section{Four-coloured graphs and melonic quartic rank-$4$ theories}\label{sec:rankfour}
We count the graphs with $2k$ vertices for $k<4$ in order to obtain 
free energy expansion until  $\mathcal O(J^4,\bJ^4)$.

 Figure \ref{fig:graphs4} summarizes 
 some properties of these graphs that
 the free energy expansion depends 
 on. Although they had been enumerated,
 neither had they been identified nor
 their symmetry factors (the order of the
 coloured automorphism groups) found.
 We can now expand the free energy 
 until sixth oder, which would in theory allow 
 the computation of 4-point function's equations 
starting from \eqref{expansion_Wgeneral}. 
For the $\phi^4_\mathsf{m}$-theory, the sum is over all
$\partial\,\fey_D(\phi^4_\mathsf{m})=\dvGrph{D}$, as shown in \cite{fullward}.  
For that model (and also for any other model 
containing those interaction vertices and thus the same boundary sector), 
the free energy $W_{D=4}[J,\bar J]$ 
to $\mathcal O(6)$ is then given by the following 
expansion, where $b=b_{a,c}=\min(\{1,2,3,4\}\setminus \{a,c\})$
and $d=d_{a,c}=\max(\{1,2,3,4\}\setminus \{a,c\})$ and
$(i_1(a),i_2(a),i_3(a))$ is the ordered set of $\{1,2,3,4\}
\setminus \{a\}$:  

\begin{align*}
% \label{expansion_congraficas4}
% correct factors
W_{D=4}[J,\bar J] & = 
G\hp 2 _{\includegraphics[height=2ex]{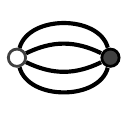}} \star \mathbb{J}\big(
\raisebox{-.34\height}{\includegraphics[height=3.77ex]{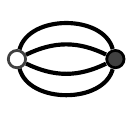}}\big) + 
\frac1{2!}
G\hp 4 _{|\raisebox{-.34\height}{\includegraphics[height=2ex]{graphs/4/Icono_melon4w.pdf}}|
\raisebox{-.34\height}{\includegraphics[height=2ex]{graphs/4/Icono_melon4w.pdf}}|} \star \mathbb{J}\big(
\raisebox{-.34\height}{\includegraphics[height=3.75ex]{graphs/4/Logo_melon4w.pdf}}|
\raisebox{-.34\height}{\includegraphics[height=3.75ex]{graphs/4/Logo_melon4w.pdf}}\big) 
+
\sum\limits_{j=1}^4 \frac1{2}G\hp4_{\raisebox{-.4\height}{\includegraphics[height=2.2ex]{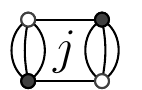}}} \star \mathbb{J}
\big(  
\raisebox{-.34\height}{\includegraphics[height=3.9ex]{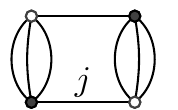}} \! \big)
\\
&\quad
+\sum\limits_{i<j} \frac1{2}G\hp4_{\raisebox{-.4\height}{\includegraphics[height=2.5ex]{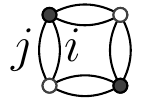}}} 
\! 
\star \mathbb{J}
\bigg(\!\!
\raisebox{-.4\height}{\includegraphics[height=5.7ex]{graphs/4/Logo_NG}}\!\!\!\bigg) \, 
+\frac{1}{3!}
G\hp{6} _{|
\raisebox{-.4\height}{\includegraphics[height=2ex]{graphs/4/Icono_melon4w.pdf}}|
\raisebox{-.4\height}{\includegraphics[height=2ex]{graphs/4/Icono_melon4w.pdf}}|
\raisebox{-.4\height}{\includegraphics[height=2ex]{graphs/4/Icono_melon4w.pdf}}|} 
\star
\mathbb{J}\big(
\raisebox{-.34\height}{\includegraphics[height=3.75ex]{graphs/4/Logo_melon4w.pdf}}|
\raisebox{-.34\height}{\includegraphics[height=3.75ex]{graphs/4/Logo_melon4w.pdf}}|
\raisebox{-.34\height}{\includegraphics[height=3.75ex]{graphs/4/Logo_melon4w.pdf}} 
\big)
\\
&  \quad +
\frac{1}{2} \sum\limits_{j=1}^4
G\hp{6} _{|
\raisebox{-.34\height}{\includegraphics[height=2ex]{graphs/4/Icono_melon4w.pdf}}|
\raisebox{-.24\height}{\includegraphics[height=2ex]{graphs/4/Icono_Vj.pdf}}|} 
\star
\mathbb{J}\Big(
\raisebox{-.34\height}{\includegraphics[height=3.5ex]{graphs/4/Logo_melon4w.pdf}}  \,|\,
\raisebox{-.3\height}{\includegraphics[height=4.25ex]{graphs/4/Logo_Vj.pdf}} 
\Big)
+
\frac{1}{2} \sum\limits_{i<j}
G\hp{6} _{|
\raisebox{-.4\height}{\includegraphics[height=2ex]{graphs/4/Icono_melon4w.pdf}}|
\raisebox{-.4\height}{\includegraphics[height=2ex]{graphs/4/Icono_Nj.pdf}}|} 
\star
\mathbb{J}\Big(
\raisebox{-.34\height}{\includegraphics[height=3.5ex]{graphs/4/Logo_melon4w.pdf}}|
\raisebox{-.4\height}{\includegraphics[height=5.5ex]{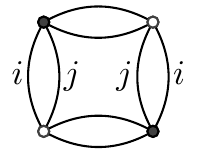}} \!\!
\Big)
 \\
& \quad
+  \sum\limits_{i<j}
G\hp{6}_{\raisebox{-.4\height}{\includegraphics[height=1.42ex]{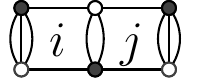}}
}
\star
\mathbb{J}\Big(\raisebox{-.34\height}{\includegraphics[height=3.9ex]{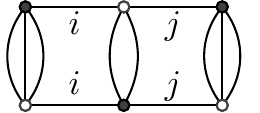}}
\!
\Big)
+  \sum\limits_{ i\neq j}  
G\hp{6}_{\raisebox{-.4\height}{\includegraphics[height=2.6ex]{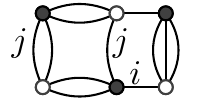}}
}
\star
\mathbb{J}\Big(\raisebox{-.4\height}{\includegraphics[height=6.5ex]{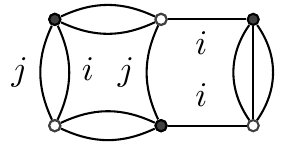}}
\Big)
\\
& \quad
+ \frac13 \sum\limits_{i<j}
G\hp{6}_{\raisebox{-.4\height}{\includegraphics[height=3.3ex]{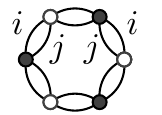}}
}
\star
\mathbb{J}\Big(\raisebox{-.5\height}{\includegraphics[height=8ex]{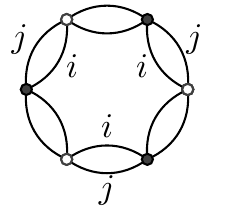}}\Big)
+ 
\frac13
\sum\limits_{i=1}^4
G\hp{6}_{\raisebox{-.4\height}{\includegraphics[height=2.6ex]{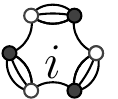}}
}
\star
\mathbb{J}\Big(\raisebox{-.4\height}{\includegraphics[height=6.5ex]{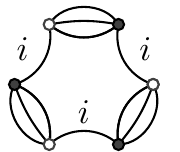}}\Big)
\\
& \quad
+  \sum\limits_{i<j} 
\sum\limits_{\substack{k\neq i \\ k \neq j\hphantom{}}}
G\hp{6}_{\raisebox{-.4\height}{\includegraphics[height=3.5ex]{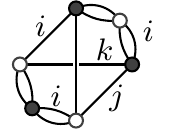}}
}
\star
\mathbb{J}\Big(\raisebox{-.4\height}{\includegraphics[height=7.3ex]{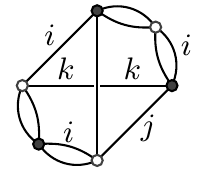}}\! \Big)
% \\
% & \quad
+  \frac13 \sum\limits_{i<j}
G\hp{6}_{\raisebox{-.4\height}{\includegraphics[height=3ex]{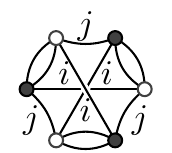}}
}
\star
\mathbb{J}\Big(\raisebox{-.4\height}{\includegraphics[height=7.1ex]{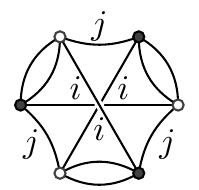}}\Big)
\\
& \quad
% +  \sum\limits_{c\neq a}
+ \sum\limits_{k=1}^4
\bigg\{
G\hp{6}_{\raisebox{-.4\height}{\!\!\includegraphics[height=3.6ex]{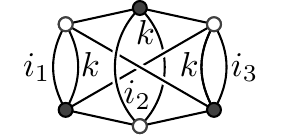}\!\!}
}
\star
\left.
\mathbb{J}\Big(\raisebox{-.43\height}{
\!\!\!\!\!\includegraphics[height=6.9ex]{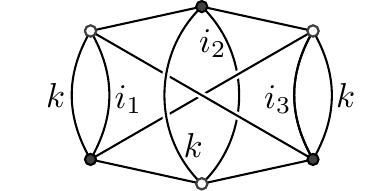}}\!\!\!\Big)
\bigg\}\right|_{\{i_1,i_2,i_3\}=\{k\}^{\mathrm c}} 
+O(J^4,\bar J^4)
\end{align*}  
It is convenient to single a particular colour $a$
we want to use the WTI for. Care has been 
taken in order to colour the graph's edges 
in non-redundant,  but univocal way. In
particular, edges are labeled strictly
by the closest letter next to them. 
\begin{align*}
% \label{expansion_congraficas4}
% correct factors
 & \quad W_{D=4}[J,\bar J]\\ 
 & = G\hp 2 _{\includegraphics[height=2ex]{graphs/4/Icono_melon4w.pdf}} \star \mathbb{J}\big(
\raisebox{-.34\height}{\includegraphics[height=3.77ex]{graphs/4/Logo_melon4w.pdf}}\big) + 
\frac1{2!}
G\hp 4 _{|\raisebox{-.34\height}{\includegraphics[height=2ex]{graphs/4/Icono_melon4w.pdf}}|
\raisebox{-.34\height}{\includegraphics[height=2ex]{graphs/4/Icono_melon4w.pdf}}|} \star \mathbb{J}\big(
\raisebox{-.34\height}{\includegraphics[height=3.75ex]{graphs/4/Logo_melon4w.pdf}}|
\raisebox{-.34\height}{\includegraphics[height=3.75ex]{graphs/4/Logo_melon4w.pdf}}\big) 
% \\
% \\  
% \nonumber
% & \,\,
%  \quad  
+
\sum\limits_{c\neq a} \frac1{2}G\hp4_{\raisebox{-.4\height}{\includegraphics[height=2.5ex]{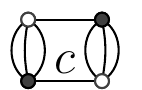}}} \star \mathbb{J}
\big(  
\raisebox{-.34\height}{\includegraphics[height=3ex]{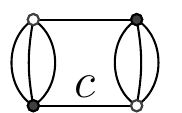}} \! \big)
+
\frac1{2}G\hp4_{\raisebox{-.4\height}{\includegraphics[height=2.5ex]{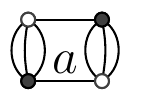}}} \star \mathbb{J}
\big(  
\raisebox{-.34\height}{\includegraphics[height=3ex]{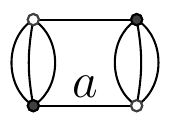}} \! \big)
\\
&\quad
+\sum\limits_{c\neq a} \frac1{2}G\hp4_{\raisebox{-.4\height}{\includegraphics[height=2.5ex]{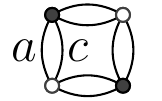}}} 
\! 
\star \mathbb{J}
\bigg(\!\!
\raisebox{-.4\height}{\includegraphics[height=5.7ex]{graphs/4/Logo_N}}\!\!\!\bigg) \, 
+\frac{1}{3!}
G\hp{6} _{|
\raisebox{-.4\height}{\includegraphics[height=2ex]{graphs/4/Icono_melon4w.pdf}}|
\raisebox{-.4\height}{\includegraphics[height=2ex]{graphs/4/Icono_melon4w.pdf}}|
\raisebox{-.4\height}{\includegraphics[height=2ex]{graphs/4/Icono_melon4w.pdf}}|} 
\star
\mathbb{J}\big(
\raisebox{-.34\height}{\includegraphics[height=3.75ex]{graphs/4/Logo_melon4w.pdf}}|
\raisebox{-.34\height}{\includegraphics[height=3.75ex]{graphs/4/Logo_melon4w.pdf}}|
\raisebox{-.34\height}{\includegraphics[height=3.75ex]{graphs/4/Logo_melon4w.pdf}} 
\big)
\\
& \quad + \frac12 
G\hp{6} _{|
\raisebox{-.34\height}{\includegraphics[height=2ex]{graphs/4/Icono_melon4w.pdf}}|
\raisebox{-.24\height}{\includegraphics[height=2ex]{graphs/4/Icono_Va.pdf}}|} 
\star
\mathbb{J}\Big(
\raisebox{-.34\height}{\includegraphics[height=3.75ex]{graphs/4/Logo_melon4w.pdf}} \,|\,
\raisebox{-.3\height}{\includegraphics[height=4.25ex]{graphs/4/Logo_Va.pdf}}  
\Big)+
\frac{1}{2} \sum\limits_{c\neq a}
G\hp{6} _{|
\raisebox{-.34\height}{\includegraphics[height=2ex]{graphs/4/Icono_melon4w.pdf}}|
\raisebox{-.24\height}{\includegraphics[height=2ex]{graphs/4/Icono_V.pdf}}|} 
\star
\mathbb{J}\Big(
\raisebox{-.34\height}{\includegraphics[height=3.5ex]{graphs/4/Logo_melon4w.pdf}}  \,|\,
\raisebox{-.3\height}{\includegraphics[height=4.25ex]{graphs/4/Logo_V.pdf}} 
\Big)
\\
& 
\quad +
\frac{1}{2} \sum\limits_{c\neq a}
G\hp{6} _{|
\raisebox{-.4\height}{\includegraphics[height=2ex]{graphs/4/Icono_melon4w.pdf}}|
\raisebox{-.4\height}{\includegraphics[height=2ex]{graphs/4/Icono_N.pdf}}|} 
\star
\mathbb{J}\Big(
\raisebox{-.34\height}{\includegraphics[height=3.5ex]{graphs/4/Logo_melon4w.pdf}}|
\raisebox{-.4\height}{\includegraphics[height=5.5ex]{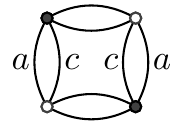}} \!\!
\Big)
+ \sum\limits_{c\neq a}
G\hp{6}_{\raisebox{-.4\height}{\includegraphics[height=2.3ex]{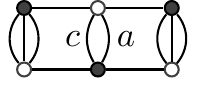}}
}\star
\mathbb{J}\Big(\raisebox{-.43\height}{\includegraphics[height=5ex]{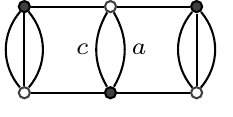}}\Big) \\
& \quad
+  \sum\limits_{c\neq a}
G\hp{6}_{\raisebox{-.4\height}{\includegraphics[height=2.3ex]{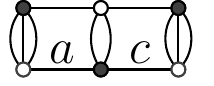}}
}
\star
\mathbb{J}\Big(\raisebox{-.37\height}{\includegraphics[height=4.15ex]{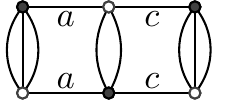}}
\!
\Big)
+  \sum\limits_{c\neq a}
G\hp{6}_{\raisebox{-.4\height}{\includegraphics[height=2.6ex]{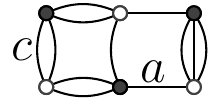}}
}
\star
\mathbb{J}\Big(\raisebox{-.4\height}{\includegraphics[height=6.5ex]{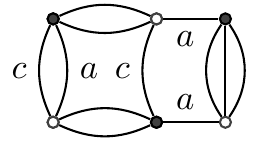}}
\Big)
\\
& \quad
+  \sum\limits_{c\neq a}
G\hp{6}_{\raisebox{-.4\height}{\includegraphics[height=3ex]{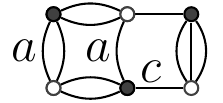}}
}
\star
\mathbb{J}\Big(\raisebox{-.4\height}{\includegraphics[height=6.5ex]{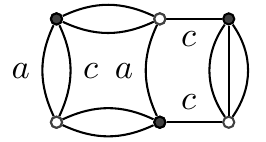}}\Big)
+  \sum\limits_{c\neq a}
 \sum\limits_{f=b,d}
G\hp{6}_{\raisebox{-.4\height}{\includegraphics[height=2.6ex]{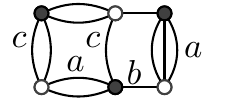}}
}
\star
\mathbb{J}\Big(\raisebox{-.4\height}{\includegraphics[height=6.3ex]{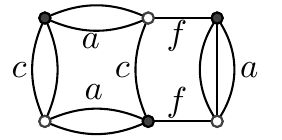}}\Big)
\\
& \quad
+ \frac13 \sum\limits_{c\neq a}
G\hp{6}_{\raisebox{-.4\height}{\includegraphics[height=3.3ex]{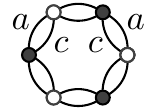}}
}
\star
\mathbb{J}\Big(\raisebox{-.5\height}{\includegraphics[height=8ex]{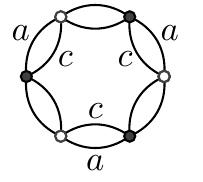}}\Big)
+ 
\frac13
G\hp{6}_{\raisebox{-.4\height}{\includegraphics[height=2.6ex]{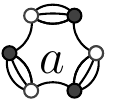}}
}
\star
\mathbb{J}\Big(\raisebox{-.4\height}{\includegraphics[height=6.5ex]{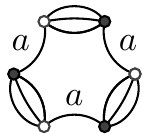}}\Big)+
\frac13
\sum\limits_{c\neq a}
G\hp{6}_{\raisebox{-.4\height}{\includegraphics[height=2.6ex]{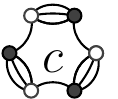}}
}
\star
\mathbb{J}\Big(\raisebox{-.4\height}{\includegraphics[height=6.5ex]{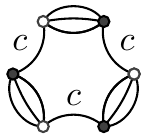}}\Big)
\\
& \quad
+  \sum\limits_{c\neq a} \sum\limits_{f=b,d}	
G\hp{6}_{\raisebox{-.4\height}{\includegraphics[height=3.5ex]{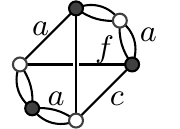}}
}
\star
\mathbb{J}\Big(\raisebox{-.4\height}{\includegraphics[height=7.3ex]{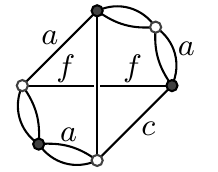}}\! \Big)
% \\
% & \quad
+  \frac13\sum\limits_{c\neq a}	
G\hp{6}_{\raisebox{-.4\height}{\includegraphics[height=3ex]{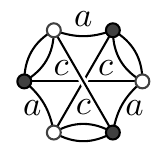}}
}
\star
\mathbb{J}\Big(\raisebox{-.4\height}{\includegraphics[height=7.1ex]{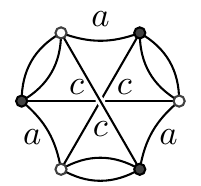}}\Big)
+ \frac13\sum\limits_{c\neq a}
G\hp{6}_{\raisebox{-.4\height}{\includegraphics[height=3.4ex]{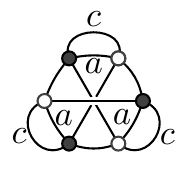}}
}
\star
\mathbb{J}\Big(\raisebox{-.5\height}{\includegraphics[height=7.1ex]{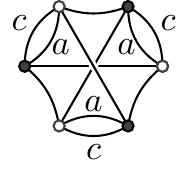}}\Big)
\\
& \quad
% +  \sum\limits_{c\neq a}
+ \sum\limits_{c\neq a}
\bigg\{
G\hp{6}_{\raisebox{-.4\height}{\!\!\includegraphics[height=3.6ex]{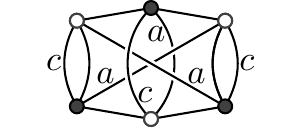}\!\!}
}
\star
\mathbb{J}\Big(\raisebox{-.43\height}{
\!\!\!\!\!\includegraphics[height=6.9ex]{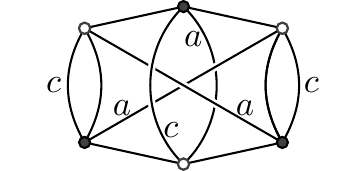}}\!\!\!\Big)
\bigg\}
+
G\hp{6}_{\!\!\!\raisebox{.31\height}{\includegraphics[height=3.7ex]{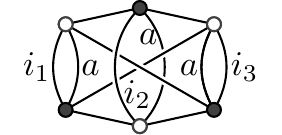}}} 
\star
\mathbb{J}\Big(\!\!\!\raisebox{-.4\height}{\includegraphics[height=6.9ex]{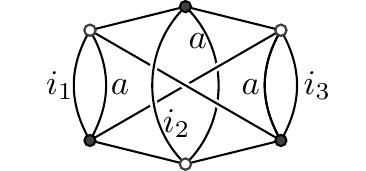}}\!\!\Big)
+O(J^4,\bar J^4)
\end{align*}  
 
For the sequel, we adopt the notation of writing entries of $\Z^D$ as unordered sets,
even though we mean them having a colour-ordering (by the subindices). Hence, the $D$-tuple
$(q_{p_1},q_{p_2},\ldots,q_{p_{D-1}},q_{p_D})$ actually means $(q_{r},q_{s},\ldots, q_{t},q_{u})$
where $r<s<\ldots<t<u$, being  $\{r,s,\ldots,t,u\}=\{p_i\}_{i=1}^D$ as sets. 

The simplified  $Y\hp a_{m_a}$-term 
given by \eqref{eq:Yterm_expansionOmega}
and by the Ward Takahashi identity after taking the $(m_an_a)$-entry 
of a generator of the $a$-th summand of 
$\mathsf{Lie}(\mathrm{U}(N)^D)$,
% for the melonic $\phi^4_4$-theory 
reads then: {
\begin{align}
%  Y\hp a_{m_a}[J,\bar J]= 
& 
\sum_{q_{i_1},q_{i_2},q_{i_3}} \nonumber
G\hp 2 _{\includegraphics[height=2ex]{graphs/4/Icono_melon4w.pdf}} (m_a,q_{i_1},q_{i_2},q_{i_3})
\\ 
& +   \label{eq:Yrank4}
\frac{1}{2}
\bigg\{
\sum\limits_{s=1,2}
\Dma{s}
G\hp 4 _{|\raisebox{-.34\height}{\includegraphics[height=2ex]{graphs/4/Icono_melon4w.pdf}}|
\raisebox{-.34\height}{\includegraphics[height=2ex]{graphs/4/Icono_melon4w.pdf}}|}+
\sum\limits_{i=1}^4 \sum\limits_{s=1,2}
\Dma{s}G\hp4_{\raisebox{-.4\height}{\includegraphics[height=2.5ex]{graphs/4/Icono_Vj.pdf}}} 
+
\sum\limits_{c\neq a}
\sum\limits_{s=1,2}
\Dma{s}
G\hp4_{\raisebox{-.4\height}{\includegraphics[height=2.5ex]{graphs/4/Icono_N.pdf}}} 
\bigg\} \star \mathbb{J}(\raisebox{-.3\height}{\includegraphics[height=3ex]{graphs/4/Logo_melon4w.pdf}})
\\ 
& +
\bigg\{
\frac{1}{3!}
\sum\limits_{r=1}^3
\Dma{r}
G\hp{6} _{|
\raisebox{-.4\height}{\includegraphics[height=2ex]{graphs/4/Icono_melon4w.pdf}}|
\raisebox{-.4\height}{\includegraphics[height=2ex]{graphs/4/Icono_melon4w.pdf}}|
\raisebox{-.4\height}{\includegraphics[height=2ex]{graphs/4/Icono_melon4w.pdf}}|} 
+
\frac{1}{2}
\sum\limits_{i=1}^4
\sum\limits_{s=2,3} 
\Dma{s}
G\hp{6}_{|
\raisebox{-.34\height}{\includegraphics[height=2ex]{graphs/4/Icono_melon4w.pdf}}|
\raisebox{-.24\height}{\includegraphics[height=2ex]{graphs/4/Icono_Vj.pdf}}|}
+\frac12 
\sum\limits_{c\neq a}\Big[
\sum\limits_{s=2,3}
\Dma{s}
G\hp{6} _{|
\raisebox{-.4\height}{\includegraphics[height=2ex]{graphs/4/Icono_melon4w.pdf}}|
\raisebox{-.4\height}{\includegraphics[height=2ex]{graphs/4/Icono_N.pdf}}|} 
\\ \nonumber
& 
\qquad +
\Dma{2}
G\hp{6}_{\raisebox{-.4\height}{\includegraphics[height=3ex]{graphs/4/Icono_E.pdf}}
}
+
\Dma{2}
G\hp{6}_{\raisebox{-.4\height}{\includegraphics[height=3ex]{graphs/4/Icono_Qp.pdf}}}
\Big]
\bigg\}
\star 
\mathbb{J}\big(
\raisebox{-.34\height}{\includegraphics[height=3.75ex]{graphs/4/Logo_melon4w.pdf}}|
\raisebox{-.34\height}{\includegraphics[height=3.75ex]{graphs/4/Logo_melon4w.pdf}}\big) 
\\ \nonumber
&
+
\bigg\{
\frac12 
\Dma{1}
G\hp{6} _{|
\raisebox{-.34\height}{\includegraphics[height=2ex]{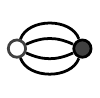}}|
\raisebox{-.24\height}{\includegraphics[height=2ex]{graphs/4/Icono_Va.pdf}}|} 
+ 
\frac{1}{3}
\sum\limits_{r=1}^3
\Dma{r }G\hp{6}_{\raisebox{-.4\height}{\includegraphics[height=2.6ex]{graphs/4/Icono_Ca.pdf}}}
+
\sum\limits_{c\neq a}
\bigg[
\Dma{1}
G\hp{6}_{\raisebox{-.4\height}{\includegraphics[height=2.6ex]{graphs/4/Icono_Q.pdf}}}
+
\frac13
\sum\limits_{r=1}^3
\Dma{r}
G\hp{6}_{\raisebox{-.4\height}{\includegraphics[height=3ex]{graphs/4/Icono_M.pdf}}}
\\ \nonumber
&
\qquad
+
\Dma{2}
G\hp{6}_{\raisebox{-.4\height}{\includegraphics[height=3.6ex]{graphs/4/Icono_Pp.pdf}}}
+
\Dma{3}
G\hp{6}_{\raisebox{-.4\height}{\includegraphics[height=2.3ex]{graphs/4/Icono_Ep.pdf}}}
\bigg]
\bigg\}
\star \mathbb{J}
\big(  
\raisebox{-.34\height}{\includegraphics[height=3ex]{graphs/4/Logo_Va.pdf}} \! \big)
\\ \nonumber
& 
+
\sum\limits_{\alpha=1}^3
\Big\{ 
\frac12 
\Dma{1}
G\hp{6} _{|
\raisebox{-.34\height}{\includegraphics[height=2ex]{graphs/4/Icono_melon4.pdf}}|
\raisebox{-.17\height}{\includegraphics[height=3ex]{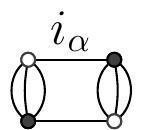}}\!|}
+
\Dma{\alpha}
G\hp{6}_{\!\!\!\raisebox{.31\height}{\includegraphics[height=3.7ex]{graphs/4/Icono_Pnew.pdf}}} 
\Big\}
\star \J 
\big(  
\raisebox{-.34\height}{\includegraphics[height=3ex]{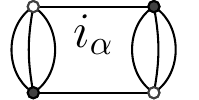}} \! 
\big)
\\ \nonumber
&
+
\sum\limits_{c\neq a}
\bigg\{ \Big[
\sum\limits_{s=1,2}
\Dma{s}
G\hp{6}_{\raisebox{-.4\height}{\includegraphics[height=2.3ex]{graphs/4/Icono_Ep.pdf}}}
+
\Dma{1}
G\hp{6}_{\raisebox{-.4\height}{\includegraphics[height=3ex]{graphs/4/Icono_Qp.pdf}}
}
+
\frac13
\sum\limits_{r=1}^3
\Dma{r}
G\hp{6}_{\raisebox{-.4\height}{\includegraphics[height=2.6ex]{graphs/4/Icono_C.pdf}}}
\\ \nonumber
&
\qquad \quad+
\Dma{1}
G\hp{6}_{\raisebox{-.4\height}{\includegraphics[height=3.5ex]{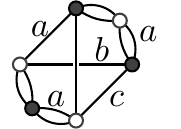}}
}
+
\Dma{1}
G\hp{6}_{\raisebox{-.4\height}{\includegraphics[height=3.5ex]{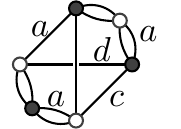}}
}
\Big]
\star
\mathbb{J}
\big(  
\raisebox{-.34\height}{\includegraphics[height=3ex]{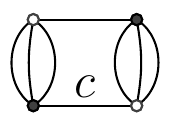}} \! 
\big)
\\ \nonumber
&
\qquad \quad 
+\Big[
\Dma{3} 
G\hp{6}_{\raisebox{-.4\height}{\includegraphics[height=2.3ex]{graphs/4/Icono_E.pdf}}}
+
\sum\limits_{s=1,2}
\Dma{s}
G\hp{6}_{\raisebox{-.4\height}{\includegraphics[height=3ex]{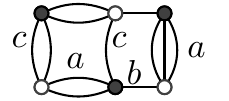}}}
\Big]
\star
\mathbb{J}
\big(  
\raisebox{-.34\height}{\includegraphics[height=3ex]{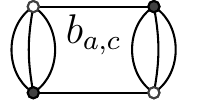}} \! 
\big)
\\ \nonumber
&
\qquad \quad 
+\Big[
\Dma{1} 
G\hp{6}_{\raisebox{-.4\height}{\includegraphics[height=2.3ex]{graphs/4/Icono_E.pdf}}}
+
\sum\limits_{s=1,2}
\Dma{s}
G\hp{6}_{\raisebox{-.4\height}{\includegraphics[height=3ex]{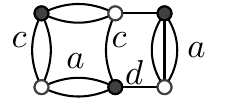}}}
\Big]
\star
\mathbb{J}
\big(  
\raisebox{-.34\height}{\includegraphics[height=4.0ex]{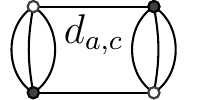}} \! 
\big)
\bigg\}
\\ \nonumber
&
+
\sum\limits_{c\neq a}
\bigg\{ \Big[
\frac12
\Dma{1}
G\hp{6} _{|
\raisebox{-.4\height}{\includegraphics[height=2ex]{graphs/4/Icono_melon4.pdf}}|
\raisebox{-.4\height}{\includegraphics[height=2ex]{graphs/4/Icono_N.pdf}}|} 
+
\sum\limits_{s=2,3}
\Dma{s}
G\hp{6}_{\raisebox{-.4\height}{\includegraphics[height=2.6ex]{graphs/4/Icono_Q.pdf}}}
+
\Dma{3}
G\hp{6}_{\raisebox{-.4\height}{\includegraphics[height=3ex]{graphs/4/Icono_Qp.pdf}}
}
\\ \nonumber
& \qquad\qquad 
+
\frac13
\sum\limits_{r=1}^3
\Dma{r}
G\hp{6}_{\raisebox{-.4\height}{\includegraphics[height=3ex]{graphs/4/Icono_L.pdf}}
}
+
\frac13
\sum\limits_{r=1}^3
\Dma{r}
G\hp{6}_{\raisebox{-.4\height}{\includegraphics[height=3.4ex]{graphs/4/Icono_Mp.pdf}}
}
+
\sum\limits_{\ell=1,3}
\Dma{\ell}
G\hp{6}_{\raisebox{-.4\height}{\includegraphics[height=3.6ex]{graphs/4/Icono_Pp.pdf}}
}
\Big]
\star
\mathbb{J}
\big(  
\raisebox{-.4\height}{\includegraphics[height=4.7ex]{graphs/4/Logo_N.pdf}} \! 
\big)
\\ \nonumber
&
\qquad \quad 
+\Big[
\Dma{3}
G\hp{6}_{\raisebox{-.4\height}{\includegraphics[height=3ex]{graphs/4/Icono_Qppd.pdf}}}
+
\Dma{2}
G\hp{6}_{\raisebox{-.4\height}{\includegraphics[height=3.5ex]{graphs/4/Icono_Dd.pdf}}
}
+
\Dma{3}
G\hp{6}_{\raisebox{-.4\height}{\includegraphics[height=3.5ex]{graphs/4/Icono_Db.pdf}}
}
\Big]
\star
\mathbb{J}
\big(  
\raisebox{-.4\height}{\includegraphics[height=5ex]{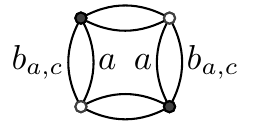}} \! 
\big)
\\ \nonumber
&
\qquad \quad 
+\Big[ 
\Dma{3}
G\hp{6}_{\raisebox{-.4\height}{\includegraphics[height=3ex]{graphs/4/Icono_Qppb.pdf}}}
+
\Dma{3}
G\hp{6}_{\raisebox{-.4\height}{\includegraphics[height=3.5ex]{graphs/4/Icono_Dd.pdf}}
}
+
\Dma{2}
G\hp{6}_{\raisebox{-.4\height}{\includegraphics[height=3.5ex]{graphs/4/Icono_Db.pdf}}
}
\Big]
\star
\mathbb{J}
\big(  
\raisebox{-.4\height}{\includegraphics[height=5.0ex]{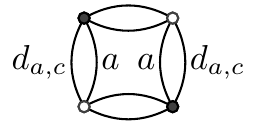}} \! 
\big)
\bigg\}+O(J^3,\bar J^3)\,\,.
\end{align}
}

 \subsection{Two-point equation for rank-$4$ theories}
 Also rank-$4$ theories are also an active topic \cite{LionniJohannes}.
We think it is instructive to derive directly, without using
 the theorem the SDE for the 2-point function:
\begin{align} 
G\hp 2_{\includegraphics[height=2.5ex]{graphs/4/Icono_melon4w.pdf}}(\mathbf a)
& 
% =\frac{1}{Z_0}\left\{\fder{^2 }{ \nonumber\bar 
% J_{\mathbf a} \delta J_{\mathbf a}} 
% \exp\left(-\Sint(\delta/\delta \bar J,\delta/\delta J)\right)
% \ee^{\sum_{\mathbf{q}} \bar J_{\mathbf q}
% E _{\mathbf q}\inv J_{\mathbf q}}\right\}_{J=\bJ=0}\\
% & 
=  \frac{1}{Z_0}\left\{\fder{ }{ \nonumber
J_{\mathbf a} } \left[
\exp\left(-\Sint(\delta/\delta \bar J,\delta/\delta J)\right)
\frac{1}{E_{\mathbf a}} J_{\mathbf a}\ee^{\sum_{\mathbf{q}} 
\bar J_{\mathbf q}
E _{\mathbf q}\inv J_{\mathbf q}}\right]\right\}_{J=\bJ=0}\\  
&=\frac{1}{Z_0E_{\mathbf a}}\left[ 
\exp\left(-\Sint(\delta/\delta \bar J,\delta/\delta J)\right)
\ee^{\sum_{\mathbf{q}} 
\bar J_{\mathbf q}
E _{\mathbf q}\inv J_{\mathbf q}}
\right]_{J=\bJ=0} \\
  & \hphantom{=} + \frac{1}{Z_0E_{\mathbf a}}\left(
\exp\left(-\Sint(\delta/\delta \bar J,\delta/\delta J)\right) 
J_{\mathbf a}
\fder{ }{J_{\mathbf a}}
\ee^{\sum_{\mathbf{q}} 
\bar J_{\mathbf q}
E _{\mathbf q}\inv J_{\mathbf q}}\nonumber
\right)_{J=\bJ=0} \\
& = \frac{1}{E_{\mathbf a}} + \frac{1}{Z_0} \frac{1}{E_{\mathbf a}} 
\left.\left(
\bar\phi_{\mathbf a} \dervpar{}{\bar \phi_{\mathbf a} 
} \left( -
\Sint(\phi,\bar\phi) 
\right)
\right)_{\phi^\flat\to\, \delta/\delta{J^\sharp}} Z[J,\bar J ]  \right|_{J=\bJ=0}\,, \nonumber
\end{align}

\begin{align*}
&\,\,\quad\left.\bar\phi_{\mathbf x}\dervpar{(- \Sint(\phi,\bar\phi))}{\bar\phi_{\mathbf x}}\right|_{{\phi^\flat\to \,  \delta/\delta J^\sharp }} & &
\\ 
 & =-2\lambda\Bigg\{    
\fder{ }{ J_{x_1x_2x_3x_4}} \sum\limits_{y_1}\fder{ }{\bJ_{y_1x_2x_3x_4}}
\sum\limits_{y_2,y_3,y_4}\fder{}{J_{y_1y_2y_3y_4}}\fder{ }{\bJ_{x_1y_2y_3y_4}}  \qquad ( \mathbf x\in\Z^4)
\\ &  \qquad \,\,\,\,+ \fder{ }{ J_{x_1x_2x_3x_4}} \sum\limits_{y_2}\fder{ }{\bJ_{x_1y_2x_3x_4}}
\sum\limits_{y_1,y_3,y_4}\fder{}{J_{y_1y_2y_3y_4}}\fder{ }{\bJ_{y_1x_2y_3y_4}}  
\\
&    \qquad \,\,\,\, +\fder{ }{ J_{x_1x_2x_3x_4}} \sum\limits_{y_3}\fder{ }{\bJ_{x_1x_2y_3x_4}}
\sum\limits_{y_1,y_2,y_4}\fder{}{J_{y_1y_2y_3y_4}}\fder{ }{\bJ_{y_1y_2x_3y_4}}    
\\
 &  \qquad \,\,\,\,+ \fder{ }{ J_{x_1x_2x_3x_4}} \sum\limits_{y_4}\fder{ }{\bJ_{x_1x_2x_3y_4}}
\sum\limits_{y_1,y_2,y_3}\fder{}{J_{y_1y_2y_3y_4}}\fder{ }{\bJ_{y_1y_2y_3x_4}}  
   \!\!      \left. \vphantom{\fder{ }{ J_{x_1x_2x_3x_4yy}} } \!  \Bigg\}Z[J,\bJ] \right|_{J=\bJ=0}\,. 
\end{align*}
One uses the WTI for the double derivatives of the form 
\[
\sum\limits_{y_2,y_3,y_4}\fder{ ^2Z[J,\bar J]}{J_{y_1y_2y_3y_4} \delta \bJ_{x_1y_2y_3y_4}} \,,\,\,\ldots\,\,,
 \sum\limits_{y_1,y_2,y_3}\fder{ ^2 Z[J,\bar J]}{J_{y_1y_2y_3y_4} \delta \bJ_{y_1y_2y_3x_4}}\,. 
\]
Then
\begin{align}
&\nonumber
\qquad\left.\bar\phi_{\mathbf x}\dervpar{(- \Sint(\phi,\bar\phi))}{\bar\phi_{\mathbf x}}\right|_{{\phi^\flat\to \,  \delta/\delta J^\sharp }}
 \\
 &=-2\lambda Z_0
 \Bigg\{
 \sum\limits_{a=1}^4 \bigg[
 \left.
 \fder{^2 Y\hp{a}_{x_a}[J,\bar J]}{J_{\mathbf x} \delta \bJ_{\mathbf x}}\right|_{J=\bJ=0}
  +Y\hp{a}_{x_a}[0,0] \cdot
  G\hp 2_{\includegraphics[height=2.5ex]{graphs/4/Icono_melon4w.pdf}}(\mathbf x)
  \\ \nonumber 
  & \quad\qquad\quad - \sum\limits_{y_a} \frac{1}{|x_a|^2-|y_a|^2} \big( 
   G\hp 2_{\includegraphics[height=2.5ex]{graphs/4/Icono_melon4w.pdf}}(\mathbf x)
   -
    G\hp 2_{\includegraphics[height=2.5ex]{graphs/4/Icono_melon4w.pdf}}(y_a, x_{i_1(a)},x_{i_2(a)},x_{i_3(a)})
  \big)
 \bigg]
 \Bigg\}\,\,.
\end{align}
Recall that $(q_i, q_j,q_{k},q_{l})$ implies an ordering of the entries,
that is, reordering so that 
$ q_s$ appears to the left of $q_r$
if and only if $s<r$, $s,r\in\{i,j,k,l\}=\{1,2,3,4\}$. 
Twice the double derivative appearing there, 
$ 2\delta^2 Y\hp a_{x_a}[J,\bJ]/\delta{J_{\mathbf x} \delta \bJ_{\mathbf x}}$, 
is given by

\begin{align*}
&
\sum\limits_{q_{i_1(a)},q_{i_2(a)},q_{i_3(a)}}  
\Big(
G\hp 4 _{|\raisebox{-.34\height}{\includegraphics[height=2ex]{graphs/4/Icono_melon4w.pdf}}|
\raisebox{-.34\height}{\includegraphics[height=2ex]{graphs/4/Icono_melon4w.pdf}}|}
(x_a,q_{i_1(a)},q_{i_2(a)},q_{i_3(a)}; \mathbf{x})
+
G\hp 4 _{|\raisebox{-.34\height}{\includegraphics[height=2ex]{graphs/4/Icono_melon4w.pdf}}|
\raisebox{-.34\height}{\includegraphics[height=2ex]{graphs/4/Icono_melon4w.pdf}}|}
( \mathbf{x};x_a,q_{i_1(a)},q_{i_2(a)},q_{i_3(a)}) 
\Big)
\\ &+
\sum\limits_{c\neq a}
\sum\limits_{q_{b(a,c)},q_{d(a,c)}}
\Big(
G\hp4_{\raisebox{-.4\height}{\includegraphics[height=2.5ex]{graphs/4/Icono_V.pdf}}}(x_a,x_c,q_b,q_d;\mathbf{x})
+
G\hp4_{\raisebox{-.4\height}{\includegraphics[height=2.5ex]{graphs/4/Icono_V.pdf}}}(\mathbf{x};x_a,x_c,q_b,q_d)
\Big)
\\ &+
2 G\hp4_{\raisebox{-.4\height}{\includegraphics[height=2.5ex]{graphs/4/Icono_Va.pdf}}}(\mathbf{x};\mathbf{x})+
\sum\limits_{c\neq a}
\sum\limits_{q_c}
\Big(
G\hp4_{\raisebox{-.4\height}{\includegraphics[height=2.5ex]{graphs/4/Icono_N.pdf}}} (x_a,x_b,q_c,x_d;\mathbf x)
+
G\hp4_{\raisebox{-.4\height}{\includegraphics[height=2.5ex]{graphs/4/Icono_N.pdf}}} (\mathbf x;x_a,x_b,q_c,x_d)
\Big)\,\,.
\end{align*}
Thus, since
% \[
$Y\hp a_{m_a}[0,\bar 0]
=  \sum_{q_{i_1},q_{i_2},q_{i_3}}
G\hp 2 _{\includegraphics[height=2ex]{graphs/4/Icono_melon4w.pdf}} (m_a,q_{i_1},q_{i_2},q_{i_3}), $
% \] 
one has
\begin{align*}
G\hp 2_{\includegraphics[height=2.5ex]{graphs/4/Icono_melon4w.pdf}}(\mathbf x)
&=\frac{1}{E_{\mathbf x}} + \frac{1}{Z_0} \frac{1}{E_{\mathbf x}} 
\left.\left(
\bar\phi_{\mathbf x} \dervpar{}{\bar \phi_{\mathbf x} 
} \left( -
\Sint(\phi,\bar\phi) 
\right)
\right)_{\phi^\flat\to\, \delta/\delta{J^\sharp}} Z[J,\bar J ]  \right|_{J=\bJ=0}  \,  \nonumber
\\
% & = 
% \frac{1}{E_{\mathbf x}} + \frac{1}{Z_0} \frac{1}{E_{\mathbf x}} 
% (-2\lambda Z_0	)
%  \Bigg\{
%   \sum\limits_{a=1}^4 \bigg[
%  \left.
%  Y\hp{a}_{x_a}[0,0] \cdot
%   G\hp 2_{\includegraphics[height=2.5ex]{graphs/4/Icono_melon4w.pdf}}(\mathbf x)
%   +
%  \fder{^2 Y\hp{a}_{x_a}[J,\bar J]}{J_{\mathbf x} \delta \bJ_{\mathbf x}}\right|_{J=\bJ=0}
%   \\ \nonumber 
%   & \!\quad\qquad+ \sum\limits_{y_a} \frac{1}{|x_a|^2-|y_a|^2} \big( 
%    G\hp 2_{\includegraphics[height=2.5ex]{graphs/4/Icono_melon4w.pdf}}(\mathbf x)
%    -
%     G\hp 2_{\includegraphics[height=2.5ex]{graphs/4/Icono_melon4w.pdf}}( y_a, x_{i_1(a)},x_{i_2(a)},x_{i_3(a)})
%   \big)
%  \bigg]
%  \Bigg\} 
%  \\
& = 
\frac{1}{E_{\mathbf x}} +  \frac{(-\lambda  )}{E_{\mathbf x}} 
 \Bigg\{
  \sum\limits_{a=1}^4 \bigg[
  2 \cdot G\hp 2_{\includegraphics[height=2.5ex]{graphs/4/Icono_melon4w.pdf}}(\mathbf x)
  \cdot
  \big(\sum\limits_{q_{i_1(a)} }
  \sum\limits_{q_{i_2(a)} }
  \sum\limits_{q_{i_3(a)}}
G\hp 2 _{\includegraphics[height=2ex]{graphs/4/Icono_melon4w.pdf}} (x_a,q_{i_1(a)},q_{i_2(a)},q_{i_3(a)})\big)
  \\
  &
  \quad
+\sum\limits_{q_{i_1(a)}} 
\sum\limits_{q_{i_2(a)}}
\sum\limits_{q_{i_3(a)}} 
\Big(
G\hp 4 _{|\raisebox{-.34\height}{\includegraphics[height=2ex]{graphs/4/Icono_melon4w.pdf}}|
\raisebox{-.34\height}{\includegraphics[height=2ex]{graphs/4/Icono_melon4w.pdf}}|}
(x_a,q_{i_1(a)},q_{i_2(a)},q_{i_3(a)}; \mathbf{x})
\\ &\qquad\qquad\qquad \qquad +
G\hp 4 _{|\raisebox{-.34\height}{\includegraphics[height=2ex]{graphs/4/Icono_melon4w.pdf}}|
\raisebox{-.34\height}{\includegraphics[height=2ex]{graphs/4/Icono_melon4w.pdf}}|}
( \mathbf{x};x_a,q_{i_1(a)},q_{i_2(a)},q_{i_3(a)}) 
\Big)
\\& 
\quad 
+
\sum\limits_{c\neq a}
\sum\limits_{q_{b(a,c)} }
\sum\limits_{q_{d(a,c)}}
\Big(
G\hp4_{\raisebox{-.4\height}{\includegraphics[height=2.5ex]{graphs/4/Icono_V.pdf}}}(x_a,x_c,q_b,q_d;\mathbf{x})
+
G\hp4_{\raisebox{-.4\height}{\includegraphics[height=2.5ex]{graphs/4/Icono_V.pdf}}}(\mathbf{x};x_a,x_c,q_b,q_d)
\Big)
\\&
\quad+
\sum\limits_{c\neq a}
\sum\limits_{q_c}
\Big(
G\hp4_{\raisebox{-.4\height}{\includegraphics[height=2.5ex]{graphs/4/Icono_N.pdf}}} (x_a,x_b,q_c,x_d;\mathbf x)
+
G\hp4_{\raisebox{-.4\height}{\includegraphics[height=2.5ex]{graphs/4/Icono_N.pdf}}} (\mathbf x;x_a,x_b,q_c,x_d)
\Big)+
2 G\hp4_{\raisebox{-.4\height}{\includegraphics[height=2.5ex]{graphs/4/Icono_Va.pdf}}}(\mathbf{x};\mathbf{x})
  \\ \nonumber 
  & \quad - \sum\limits_{y_a} \frac{2}{|x_a|^2-|y_a|^2} \big( 
   G\hp 2_{\includegraphics[height=2.5ex]{graphs/4/Icono_melon4w.pdf}}(\mathbf x)
   -
    G\hp 2_{\includegraphics[height=2.5ex]{graphs/4/Icono_melon4w.pdf}}(y_a, x_{i_1(a)},x_{i_2(a)},x_{i_3(a)})
  \big)
 \bigg]
 \Bigg\} 
\end{align*}  

 \subsection[subsec]{\for{toc}{Four-point equation for $G\hp4_{\V_i}$ in rank-$4$ theories}\except{toc}{Four-point equation for $G\hp4_{\protect \Icono{4}{V1v}{2.6}{53}}$ in rank-$4$ theories}}

Since $\V_1$ has $\Z_2$ as automorphism
group, according to Theorem \ref{thm:SDEs}, the equation 
satisfied by $\GGuno$ is the following: 
\begin{align}\label{eq:SDErank4melonic4pt}
&\qquad
\bigg(1+\frac{2\lambda}{E_{x_1,y_2,y_3,y_4}} \suml_{a=1}^4 \suml_{\mathbf q_{\hat a} }
\GGmelon (s_a,\mathbf q_{\hat a})\bigg)
\GGuno (\Xb) 
\\ \nonumber
& =\frac{(-2\lambda)}{E_{\sb}} \suml_{a=1}^4 \Bigg\{ 
\suml_{\hat\sigma\in\Z_2} \sigma^* \mathfrak{f}\hp{a}_{\Icono{4}{V1v}{2.6}{53}}(\Xb)+
\suml_{\rho>1 }  \frac{Z_0\inv}{ E(y^\rho_a,s_a)}
\bigg[
\dervpar{Z[J,J]}{\varsigma_a(\Vuno ;1,\rho) }(\Xb) -\dervpar{Z[J,J]}{\varsigma_a(\Vuno ;1,\rho) }(\Xb|_{s_a\to y^\rho_a})  \bigg]   
 \\ &\qquad\qquad \qquad
- \suml_{b_a}\frac{1}{E(s_a,b_a)} 
\big[ \GGuno (\Xb) - \GGuno (\Xb|_{s_a \to b_a}) \big] \Bigg\}   \nonumber
\end{align}
for $\xb,\yb\in\Z^4$, $\Xb=(\xb,\yb)$, and $\sb=(x_1,y_2,y_3,y_4)$. 
We write down first the term in square brackets in the RHS, which one finds trivially: 
\begin{align}
\suml_{a=1}^4 
\frac{Z_0\inv}{E(y^2_a,s_a)}\dervpar{Z[J,J]}{\varsigma_a(\Vuno ;1,2)(\Xb) } & =
\frac{1}{E(y_1,x_1)}\big(\GGcmm(\Xb)+ \GGmelon (\xb) \cdot \GGmelon (\yb) \big) \nonumber 
  \\ & \qquad
+\suml_{c\neq 1}\frac{1}{E(x_c,y_c)} G\hp4_{\Icono{4}{N1c}{2.5}{28}}(\Xb) \nonumber
\end{align}
Less so is to find $\sum_a \mathfrak{f}\hp{a}_{\Icono{4}{V1v}{2.6}{53}}$.
The contributions to $\mathfrak{f}\hp{a}_{\Icono{4}{V1v}{2.6}{53}}$, for 
fixed colour $a$, are all functions occurring in front of a $\J(\Logo{4}{V1}{2.9}{22})$-source
term. These functions come from coefficients of the following source terms in  \eqref{eq:Yrank4} 
 $\J 
\big(\raisebox{-.28\height}{\includegraphics[height=2.6ex]{graphs/4/Logo_Va.pdf}} \! \big)
$ and $\J \big(\raisebox{-.28\height}{\includegraphics[height=2.6ex]{graphs/4/Logo_Vc.pdf}} \! \big)$ 
for the values\footnote{Recall that $c$ ($c\neq a$) in that expansion \eqref{eq:Yrank4} is seen as running variable,
while $b<d$ are defined in terms of $a$ and $c$ by $\{a,b,c,d\}=\{1,2,3,4\}$. Also $i_1(a)<i_2(a)<i_3(a)$, and $\{i_1,i_2,i_3,a\}=\{1,2,3,4\}$.} 
when $a=1$ or $c=1$ (in the sum over $c$), but also form the following values:
\begin{align*} 
 \J 
\big(  
\raisebox{-.3 \height}{\includegraphics[height=3ex]{graphs/4/Logo_Vii.pdf}} \! 
\big), \mbox{ for }a \in \{2,3,4\}; \,\,
\mathbb{J}
\big(  
\raisebox{-.28\height}{\includegraphics[height=3ex]{graphs/4/Logo_Vb.pdf}} \! 
\big), \mbox{ for }(a,c)\in\{(2,3),(2,4),(3,2),(3,4),(4,2),(4,3)\}    
;\\
\mbox{none from }\mathbb{J}
\big(  
\raisebox{-.28\height}{\includegraphics[height=3ex]{graphs/4/Logo_Vd.pdf}} \! 
\big)\,.\vspace{-.4cm}
\end{align*} 
Hence 
\begin{align}
 \mathfrak{f}\hp{1}_{\Icono{4}{V1v}{2.6}{53}}    & =
 \bigg\{
\frac12 
\Delta_{x_1,1}
G\hp{6} _{|
\raisebox{-.34\height}{\includegraphics[height=2ex]{graphs/4/Icono_melon4.pdf}}|
\raisebox{-.24\height}{\includegraphics[height=2ex]{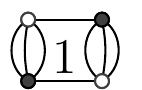}}|} 
+ 
\frac{1}{3}
\sum\limits_{r=1}^3
\Delta_{x_1,r }G\hp{6}_{\raisebox{-.4\height}{\includegraphics[height=2.6ex]{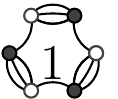}}}
+
\sum\limits_{c\neq a}
\bigg[
\Delta_{x_1,1}
G\hp{6}_{\raisebox{-.4\height}{\includegraphics[height=2.6ex]{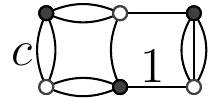}}}
+
\frac13
\sum\limits_{r=1}^3
\Delta_{x_1,r}
G\hp{6}_{\raisebox{-.4\height}{\includegraphics[height=3ex]{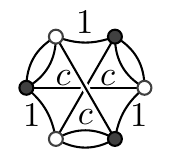}}}
\\ \nonumber
&
\qquad
+
\Delta_{x_1,2}
G\hp{6}_{\raisebox{-.4\height}{\includegraphics[height=3.6ex]{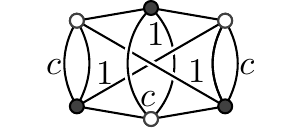}}}
+
\Delta_{x_1,3}
G\hp{6}_{\raisebox{-.4\height}{\includegraphics[height=2.22ex]{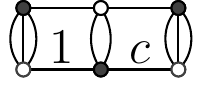}}}
\bigg]
\bigg\} \,, 
 % % % % %% % % % % % % %
\nonumber \\
 \mathfrak{f}\hp{2}_{\Icono{4}{V1v}{2.6}{53}}  & =  
\sum\limits_{s=1,2}
\Delta_{y_2,s}
G\hp{6}_{\raisebox{-.4\height}{\includegraphics[height=2.22ex]{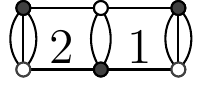}}}
+
\Delta_{y_2,1}
G\hp{6}_{\raisebox{-.4\height}{\includegraphics[height=3ex]{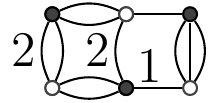}}
}
+
\frac13
\sum\limits_{r=1}^3
\Delta_{y_2,r}
G\hp{6}_{\raisebox{-.4\height}{\includegraphics[height=2.6ex]{graphs/4/Icono6_C1.pdf}}}
+
\Delta_{y_2,1}
G\hp{6}_{\raisebox{-.4\height}{\includegraphics[height=3.5ex]{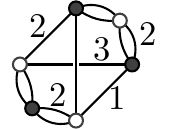}}
}
+
\Delta_{y_2,1}
G\hp{6}_{\raisebox{-.4\height}{\includegraphics[height=3.5ex]{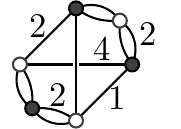}}
}  \\
&\quad +  
\frac12 
\Delta_{y_2,1}
G\hp{6} _{|
\raisebox{-.34\height}{\includegraphics[height=2.2ex]{graphs/4/Icono_melon4.pdf}}|
\raisebox{-.17\height}{\includegraphics[height=2.2ex]{graphs/4/Icono_V1.pdf}}\!|}
+
\Delta_{y_2,\alpha}
G\hp{6}_{\!\!\!\raisebox{.31\height}{\includegraphics[height=3.7ex]{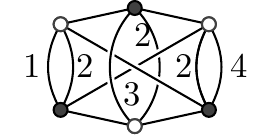}}}  
+\suml_{c=3,4}\Big[
\Delta_{y_2,3} 
G\hp{6}_{\raisebox{-.4\height}{\includegraphics[height=3ex]{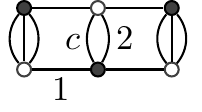}}}
+
\sum\limits_{s=1,2}
\Delta_{y_2,s}
G\hp{6}_{\raisebox{-.4\height}{\includegraphics[height=3ex]{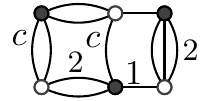}}}
\Big]\,, \nonumber 
% \\
% %% %%%%%%%%%%%%%%%%%%%%%%%%%55
\\
 \mathfrak{f}\hp{3}_{\Icono{4}{V1v}{2.6}{53}}  & =  
\sum\limits_{s=1,2}
\Delta_{y_3,s}
G\hp{6}_{\raisebox{-.4\height}{\includegraphics[height=2.22ex]{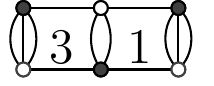}}}
+
\Delta_{y_3,1}
G\hp{6}_{\raisebox{-.4\height}{\includegraphics[height=3ex]{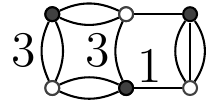}}
}
+
\frac13
\sum\limits_{r=1}^3
\Delta_{y_3,r}
G\hp{6}_{\raisebox{-.4\height}{\includegraphics[height=2.6ex]{graphs/4/Icono6_C1.pdf}}}
+
\Delta_{y_3,1}
G\hp{6}_{\raisebox{-.4\height}{\includegraphics[height=3.5ex]{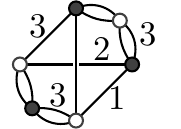}}
}
+
\Delta_{y_3,1}
G\hp{6}_{\raisebox{-.4\height}{\includegraphics[height=3.5ex]{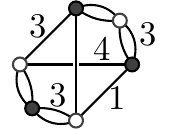}}
}  \\
&\quad + 
\frac12 
\Delta_{y_3,1}
G\hp{6} _{|
\raisebox{-.34\height}{\includegraphics[height=2.2ex]{graphs/4/Icono_melon4.pdf}}|
\raisebox{-.17\height}{\includegraphics[height=2.2ex]{graphs/4/Icono_V1.pdf}}\!|}
+
\Delta_{y_3,\alpha}
G\hp{6}_{\!\!\!\raisebox{.31\height}{\includegraphics[height=3.7ex]{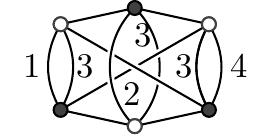}}}  
+\suml_{c=2,4}\Big[
\Delta_{y_3,3} 
G\hp{6}_{\raisebox{-.4\height}{\includegraphics[height=3ex]{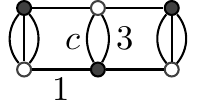}}}
+
\sum\limits_{s=1,2}
\Delta_{y_3,s}
G\hp{6}_{\raisebox{-.4\height}{\includegraphics[height=3ex]{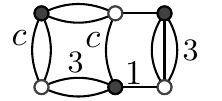}}}
\Big]\,,  \nonumber
\\
% %%%%%%%%%%%%%%%%%%%%
 \mathfrak{f}\hp{4}_{\Icono{4}{V1v}{2.6}{53}}  & =  
\sum\limits_{s=1,2}
\Delta_{y_4,s}
G\hp{6}_{\raisebox{-.4\height}{\includegraphics[height=2.22ex]{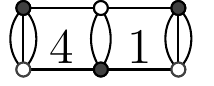}}}
+
\Delta_{y_4,1}
G\hp{6}_{\raisebox{-.4\height}{\includegraphics[height=3ex]{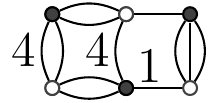}}
}
+
\frac13
\sum\limits_{r=1}^3
\Delta_{y_4,r}
G\hp{6}_{\raisebox{-.4\height}{\includegraphics[height=2.6ex]{graphs/4/Icono6_C1.pdf}}}
+
\Delta_{y_4,1}
G\hp{6}_{\raisebox{-.4\height}{\includegraphics[height=3.5ex]{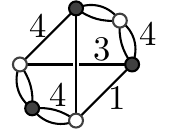}}
}
+
\Delta_{y_4,1}
G\hp{6}_{\raisebox{-.4\height}{\includegraphics[height=3.5ex]{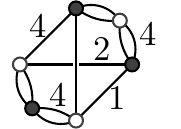}}
}  \\
&\quad +  
\frac12 
\Delta_{y_4,1}
G\hp{6} _{|
\raisebox{-.34\height}{\includegraphics[height=2.2ex]{graphs/4/Icono_melon4.pdf}}|
\raisebox{-.17\height}{\includegraphics[height=2.2ex]{graphs/4/Icono_V1.pdf}}\!|}
+
\Delta_{y_4,\alpha}
G\hp{6}_{\!\!\!\raisebox{.31\height}{\includegraphics[height=3.7ex]{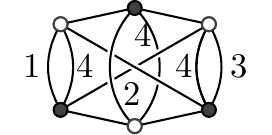}}} 
+\suml_{c=2,3}\Big[
\Delta_{y_4,3} 
G\hp{6}_{\raisebox{-.4\height}{\includegraphics[height=3ex]{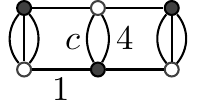}}}
+
\sum\limits_{s=1,2}
\Delta_{y_4,s}
G\hp{6}_{\raisebox{-.4\height}{\includegraphics[height=3ex]{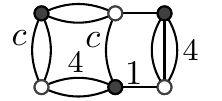}}}
\Big] 
\,. \nonumber
% \\ 
% & \quad
% +\Big[
% \Dma{3} 
% G\hp{6}_{\raisebox{-.4\height}{\includegraphics[height=3ex]{graphs/4/Icono_E.pdf}}}
% +
% \sum\limits_{s=1,2}
% \Dma{s}
% G\hp{6}_{\raisebox{-.4\height}{\includegraphics[height=3ex]{graphs/4/Icono_Qppb.pdf}}}
% \Big]
% \nonumber 
 \end{align}
One inserts the sum of these four terms in equation \eqref{eq:SDErank4melonic4pt}. 
\subsection[subsec]{\for{toc}{Four-point equation for $G\hp4_{\mtc N_{ij}}$ in rank-$4$ theories}\except{toc}
{Four-point equation for $G\hp4_{\protect \Icono{4}{N12}{2.2}{53}}$ in rank-$4$ theories}}

In order to get the equation for $\GGcnud$, we calculate
first 
$\sum_a \mtf{f}\hp{a}_{ \raisebox{-.33\height}{\includegraphics[height=2.2ex]{graphs/4/Icono4_N12}} }$.
\begin{subequations} \label{eq:f_rank4_nonmelonic}
\begin{align} 
 \mtf{f}\hp{1}_{ \raisebox{-.33\height}{\includegraphics[height=2.2ex]{graphs/4/Icono4_N12}} } & =   
 % % % % %
 \Big[
\frac12
\Delta_{x_1,1}
G\hp{6} _{|
\raisebox{-.4\height}{\includegraphics[height=2ex]{graphs/4/Icono_melon4.pdf}}|
\raisebox{-.4\height}{\includegraphics[height=2ex]{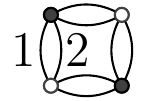}}|} 
+
\sum\limits_{s=2,3}
\Delta_{x_1,s}
G\hp{6}_{\raisebox{-.4\height}{\includegraphics[height=2.6ex]{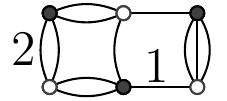}}}
+
\Delta_{x_1,3}
G\hp{6}_{\raisebox{-.4\height}{\includegraphics[height=2.6ex]{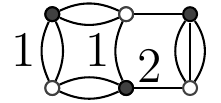}}
}
\\ \nonumber
& \qquad\qquad 
+
\frac13
\sum\limits_{r=1}^3
\Delta_{x_1,r}
G\hp{6}_{\raisebox{-.4\height}{\includegraphics[height=3ex]{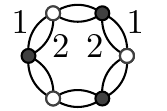}}
}
+
\frac13
\sum\limits_{r=1}^3
\Delta_{x_1,r}
G\hp{6}_{\raisebox{-.4\height}{\includegraphics[height=3.4ex]{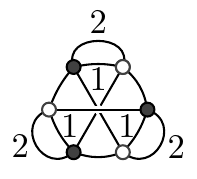}}
}
+
\sum\limits_{\ell=1,3}
\Delta_{x_1,\ell}
G\hp{6}_{\raisebox{-.4\height}{\!\!\!\includegraphics[height=3.2ex]{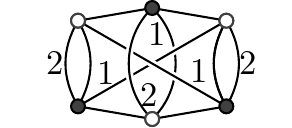}}
}
\Big]
 \\
 & \quad+   \sum_{c=3,4}
 \Big[
\Delta_{x_1,3}
G\hp{6}_{\raisebox{-.4\height}{\includegraphics[height=3ex]{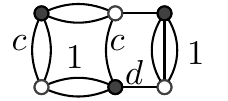}}}
+
\Delta_{x_1,2}
G\hp{6}_{\raisebox{-.4\height}{\includegraphics[height=3.5ex]{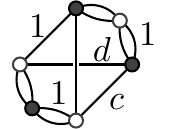}}
}
+
\Delta_{x_1,3}
G\hp{6}_{\raisebox{-.4\height}{\includegraphics[height=3.5ex]{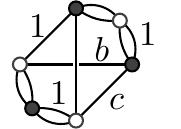}}
}
\Big] \nonumber 
 \,,
% % % %
\\
 \mtf{f}\hp{2}_{ \raisebox{-.33\height}{\includegraphics[height=2.2ex]{graphs/4/Icono4_N12}} }& =  
 % % % %
 \Big[
\frac12
\Delta_{x_2,1}
G\hp{6} _{|
\raisebox{-.4\height}{\includegraphics[height=2ex]{graphs/4/Icono_melon4.pdf}}|
\raisebox{-.4\height}{\includegraphics[height=2ex]{graphs/4/Icono4_N12.pdf}}|} 
+
\sum\limits_{s=2,3}
\Delta_{x_2,s}
G\hp{6}_{\raisebox{-.4\height}{\includegraphics[height=2.6ex]{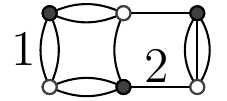}}}
+
\Delta_{x_2,3}
G\hp{6}_{\raisebox{-.4\height}{\includegraphics[height=3ex]{graphs/4/Icono6_Qp21.pdf}}
}
\\ \nonumber
& \qquad\qquad 
+
\frac13
\sum\limits_{r=1}^3
\Delta_{x_2,r}
G\hp{6}_{\raisebox{-.4\height}{\includegraphics[height=3ex]{graphs/4/Icono6_L12.pdf}}
}
+
\frac13
\sum\limits_{r=1}^3
\Delta_{x_2,r}
G\hp{6}_{\raisebox{-.4\height}{\includegraphics[height=3.4ex]{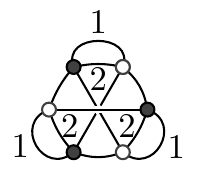}}
}
+
\sum\limits_{\ell=1,3}
\Delta_{x_2,\ell}
G\hp{6}_{\raisebox{-.4\height}{\!\!\!\includegraphics[height=3.2ex]{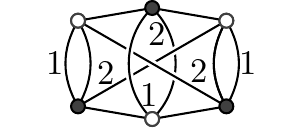}}
}
\Big] \\
&\quad + \sum_{c=3,4}
\Big[
\Delta_{x_2,3}
G\hp{6}_{\raisebox{-.4\height}{\includegraphics[height=3ex]{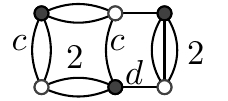}}}
+
\Delta_{x_2,2}
G\hp{6}_{\raisebox{-.4\height}{\includegraphics[height=3.5ex]{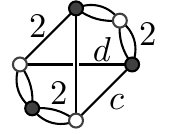}}
}
+
\Delta_{x_2,3}
G\hp{6}_{\raisebox{-.4\height}{\includegraphics[height=3.5ex]{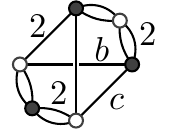}}
}
\Big] \,,
\nonumber     \\
 \mtf{f}\hp{3}_{ \raisebox{-.33\height}{\includegraphics[height=2.2ex]{graphs/4/Icono4_N12}} }& =   
  % % % %
  \Big[
\frac12
\Delta_{y_3, 1}
G\hp{6} _{|
\raisebox{-.4\height}{\includegraphics[height=2ex]{graphs/4/Icono_melon4.pdf}}|
\raisebox{-.4\height}{\includegraphics[height=2ex]{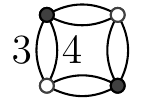}}|} 
+
\sum\limits_{s=2,3}
\Delta_{y_3, s}
G\hp{6}_{\raisebox{-.4\height}{\includegraphics[height=2.6ex]{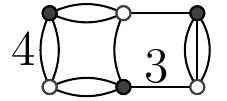}}}
+
\Delta_{y_3, 3}
G\hp{6}_{\raisebox{-.4\height}{\includegraphics[height=3ex]{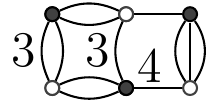}}
}
\\ \nonumber
& \qquad\qquad 
+
\frac13
\sum\limits_{r=1}^3
\Delta_{y_3, r}
G\hp{6}_{\raisebox{-.4\height}{\includegraphics[height=3ex]{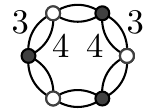}}
}
+
\frac13
\sum\limits_{r=1}^3
\Delta_{y_3, r}
G\hp{6}_{\raisebox{-.4\height}{\includegraphics[height=3.4ex]{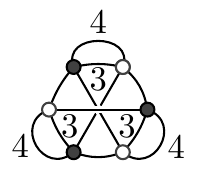}}
}
+
\sum\limits_{\ell=1,3}
\Delta_{y_3, \ell}
G\hp{6}_{\!\!\!\raisebox{-.4\height}{\includegraphics[height=3.3ex]{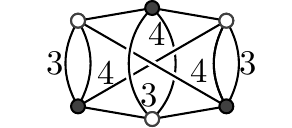}}
}
\Big] \nonumber
\\ 
\nonumber 
& \quad + \sum_{c=2,4}
\Big[ 
\Delta_{y_3, 3}
G\hp{6}_{\raisebox{-.4\height}{\includegraphics[height=3ex]{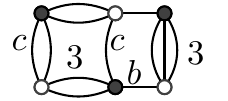}}}
+
\Delta_{y_3, 3}
G\hp{6}_{\raisebox{-.4\height}{\includegraphics[height=3.5ex]{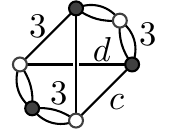}}
}
+
\Delta_{y_3, 2}
G\hp{6}_{\raisebox{-.4\height}{\includegraphics[height=3.5ex]{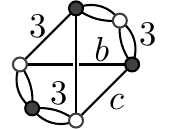}}
}
\Big]
\,,
  % % % %
\nonumber       \\
  \mtf{f}\hp{4}_{ \raisebox{-.33\height}{\includegraphics[height=2.2ex]{graphs/4/Icono4_N12}} }& =   
   % % % %
   \Big[
\frac12
\Delta_{y_4, 1}
G\hp{6} _{|
\raisebox{-.4\height}{\includegraphics[height=2ex]{graphs/4/Icono_melon4.pdf}}|
\raisebox{-.4\height}{\includegraphics[height=2ex]{graphs/4/Icono4_N34.pdf}}|} 
+
\sum\limits_{s=2,3}
\Delta_{y_4, s}
G\hp{6}_{\raisebox{-.4\height}{\includegraphics[height=2.6ex]{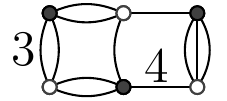}}}
+
\Delta_{y_4, 3}
G\hp{6}_{\raisebox{-.4\height}{\includegraphics[height=3ex]{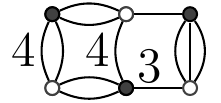}}
}
\\ \nonumber
& \qquad\qquad 
+
\frac13
\sum\limits_{r=1}^3
\Delta_{y_4, r}
G\hp{6}_{\raisebox{-.4\height}{\includegraphics[height=3ex]{graphs/4/Icono6_L34.pdf}}
}
+
\frac13
\sum\limits_{r=1}^3
\Delta_{y_4, r}
G\hp{6}_{\raisebox{-.4\height}{\includegraphics[height=3.4ex]{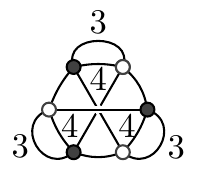}}
}
+
\sum\limits_{\ell=1,3}
\Delta_{y_4, \ell}
G\hp{6}_{\!\!\!\raisebox{-.4\height}{\includegraphics[height=3.3ex]{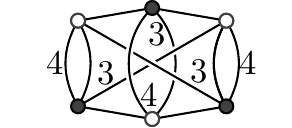}}
}
\Big]\,.
\nonumber 
\\ 
\nonumber 
 & \quad + \sum_{c=2,3}
\Big[ 
\Delta_{y_4, 3}
G\hp{6}_{\raisebox{-.4\height}{\includegraphics[height=3ex]{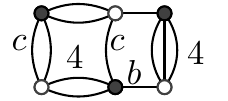}}}
+
\Delta_{y_4, 3}
G\hp{6}_{\raisebox{-.4\height}{\includegraphics[height=3.5ex]{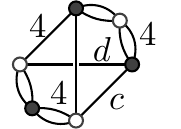}}
}
+
\Delta_{y_4, 2}
G\hp{6}_{\raisebox{-.4\height}{\includegraphics[height=3.5ex]{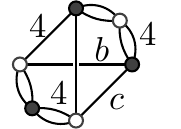}}
}
\Big]
 \end{align}
\end{subequations}
 The remaining terms from swapping edges are:
\begin{align*}
\suml_{a=1}^4 \frac{Z_0\inv}{E(y^2_a,s_a)}\dervpar{Z[J,J]}{\varsigma_a(\Icono{4}{N12}{3.0}{28};1,2) } (\Xb)& =
\frac{1}{E(y_1,x_1)} \GGdos(\xb,\yb) + \frac{1}{E(y_2,x_2)} \GGuno(\xb,\yb) \\
& +
\frac{1}{E(x_3,y_3)}\GGcuatro (\xb,\yb) +\frac{1}{E(x_4,y_4)} \GGtres (\xb,\yb)   \,, 
\end{align*}
whence the SDE for $G\hp4_{\Icono{4}{N12}{2.3}{28}}$ is 
\begin{align}
& \qquad
\bigg(1+\frac{2\lambda}{E_{\mathbf{s}}} \suml_{a=1}^4 \suml_{\mathbf q_{\hat a} } \nonumber
\GGmelon (s_a,\mathbf q_{\hat a})\bigg)
\cdot G\hp4_{\Icono{4}{N12}{2.3}{28}} (\Xb) \nonumber 
\\  \label{eq:SDErank4NONmelonic4pt}
& =\frac{(-2\lambda)}{E_{\sb}} \suml_{a=1}^4 \Bigg\{ 
\suml_{\hat\sigma\in\Z_2} \sigma^* \mathfrak{f}\hp{a}_{\Icono{4}{N12}{2.3}{28}}(\Xb) + 
\frac{1}{E(y_1,x_1)}\big[\GGdos(\xb,\yb)-\GGdos(y_1,x_2,x_3,x_4,\yb) \big]+
\nonumber 
\\
&\qquad +\frac{1}{E(y_2,x_2)}\big[\GGuno(\xb,\yb) - \nonumber 
\GGuno(x_1,y_2,x_3,x_4,\yb)\big] \\ \nonumber 
& \qquad +
\frac{1}{E(x_3,y_3)} \big[ \GGcuatro (\xb,\yb)- \GGcuatro(\xb,y_1,y_2,x_3,y_4) \big]  \\ \nonumber 
& \qquad + \frac{1}{E(x_4,y_4)} 
\big[\GGtres (\xb,\yb) -\GGtres ( \xb , y_1,y_2,y_3,x_4) \big]
 \nonumber  
 \\ \nonumber 
& \qquad
 - \suml_{b_a}\frac{1}{E(s_a,b_a)} 
\big[ G\hp4_{\Icono{4}{N12}{2.3}{28}}  (\Xb) -G\hp4_{\Icono{4}{N12}{2.3}{28}} (\Xb|_{s_a \to b_a}) \big] \Bigg\}   
\end{align}
with $\Xb=(\xb,\yb)$ and $\mathbf{s}=(x_1,x_2,y_3,y_4)$ with the functions $\mathfrak{f}\hp a_{\Icono{4}{N12}{2.3}{28}}$ 
given by eqs. \eqref{eq:f_rank4_nonmelonic}.

\section{A simple quartic model}\label{sec:simple}

In order to obtain a simpler set of SDE, we consider a model which has
less correlation functions. Its probability theory is expected to
ponder only geometries with spherical boundaries. Nevertheless, it is
interesting because its equations are particularly simple.  We
consider the rank-$3$ tensor model with action $S_{\mtr{
}}[\phi,\bar\phi ]= S_0[\phi,\bar\phi]+ \Sint[\phi,\bar\phi] $ where
\begin{equation} \label{eq:toy}
 S_0[\phi,\bar\phi]=\Tr_2(\bar\phi,E \phi)
 =\sum_{\xb\in \Z^3} \bar\phi_{\xb}(m^2+|\xb|^2) \phi_{\xb} \quad\mbox{ and }\quad
  \Sint[\phi,\bar\phi]= %\lambda \cdot \V_1(\phi,\bar\phi)  =
  \lambda \cdot \vuno  \,.
\end{equation} 
Here $|\xb|^2=x_1^2+x_2^2+x_3^2$, $\xb=(x_1,x_2,x_3)\in \Z^3$.  In
particular all the bordisms that this theory triangulates are
null-bordisms and bordisms between spheres.  Notice that the boundary
graphs are all graphs having the following property: two edges are
connected by a $2$-coloured edge, if and only if they are connected by
a $3$-coloured edge. We denote by $\Theta$ ($\Theta \subset \Grph{3}$)
the set of \textit{connected} graphs with this property. Thus
\[
\fey_3(\vuno) = \{ \B \in\dGrph{3} : \B \mbox{ has connected components in }\Theta\}\,,
\]
being  \vspace{-.35cm}
\[\Theta=\big\{\meloncito,~\vuno\,,~\logo{6}{Q1}{4.2}{23}\,,~\icono{8}{Q1}{4.2}{23}\,,~\icono{10}{Q1}{5.2}{23}
\,,~\icono{12}{Q1}{6.5}{23}\,,\ldots\big\}\,.\] Let $\mathcal X_{2k}$
be the graph in $\Theta$ with $2k$ vertices.  That is to say, the set
of (connected) correlation functions with connected boundary is
precisely indexed by $\Theta$ and we set $G\hp{2k}:=G\hp{2k}_{\mtc
  X_{2k}}$, i.e.
\[
G\hp2 =\Gmelon,~G\hp4 =\Guno,~ G\hp6 =\Gsqu\,,~ 
G\hp{8} =G\hp8\raisebox{-.85\height}{ 
\negthickspace\negthickspace\negthickspace\negthickspace\negthickspace \includegraphics[height=3ex]{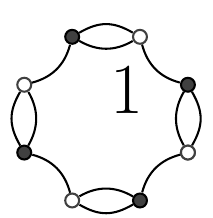}} \,, 
~G\hp{10}=
G\hp{10} \raisebox{-.85\height}{ \negthickspace\negthickspace
\negthickspace\negthickspace\negthickspace \includegraphics[height=4ex]{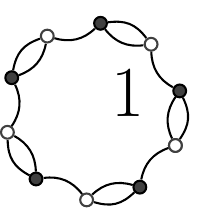}}\,.
\]
Any ($2k$)-point function with disconnected 
components can be labeled by integer partitions 
$(n_1,\ldots,n_{\ell})$ such that 
\begin{equation}
\B=\mtc{X}^{n_1}_2\sqcup\mtc{X}^{n_1}_4 \sqcup \ldots 
\sqcup \mtc X_{2\ell }^{n_\ell}\,, \label{eq:Bindisconnected}
\end{equation}
being $\ell$ the maximum number of vertices 
that a connected component of $\B$ has. 
These numbers $n_i$ satisfy 
\begin{equation}
\label{eq:Toyrestrictions}
k = \sum_{i=1}^\ell i \cdot n_i  
 ~~~~\mbox{and}~~~~ 
 B= \sum_{i=1}^\ell n_i\,,
\end{equation}
where $B$ is the number of connected components of
$\B$. Then, the free energy boils down to the 
expression
\begin{align}
 W\jj & =\suml_{l=1}^\infty ~ \sum\limits_{\substack{\B \in \partial (\fey_3(\vunito) ) \\ 
k(\B)=l}}
 \!\!\!\!\!\!\!\!\!\!\!\!\!  {\vphantom{\sum}}'\,\,\,
\,\,\,\,\,\, G\hp{2l}_\B\star \J(\B) \,,\end{align}
where the prime in the sum means that it is performed 
with the restrictions \eqref{eq:Toyrestrictions}. 
More concretely, writing any graph $\B$ as in eq. \eqref{eq:Bindisconnected},
one can rephrase the sum rather over $\ell$,
the largest number of black (or white) vertices found in 
a connected component of $\B$. This modification readily yields 
\begin{align*}
W\jj & = \suml_{\ell=1}^\infty
 \bigg(\prod\limits_{j=1}^{\ell} \frac{1}{j^{n_j} \cdot n_j!}\bigg)
G\hp{2k}_{|\mathcal X_{2}^{\sqcup n_1}|\ldots| \mathcal X_{2i}^{\sqcup n_i} \ldots | \mathcal X_{2\ell}^{\sqcup n_\ell}|}
\star \J (\mtc{X}^{n_1}_2\sqcup\mtc{X}^{n_1}_4 \sqcup \ldots 
\sqcup \mtc X_{2\ell }^{n_\ell})\,.
\end{align*}
To obtain the last line one observes that $\Autc(\mtc X_{2k})=\langle
\mbox{rotation by $2\pi/k$} \rangle=\Z_k$, and $
|\Autc(\B)|=n_1!\ldots n_\ell! \cdot |\Autc(\mtc X_{2})|^{n_1} \cdots
|\Autc(\mtc X_{2\ell})|^{n_\ell} $.  It should be noticed that this
form has already been found in the free energy expansion of (real)
matrix models, here with twice the number of sources of each monomial
with respect to that \cite[Sec. 2.3]{GW12}.  It is also noteworthy
that the Grosse-Wulkenhar model ($\phi^{\star4}_4$ self-dual theory)
\cite{GW12} was shown to be solvable by using matrix techniques.  Here
we have shown that the $\vuno$-model obeys the very same expansion of
the free energy and that the number of ($2k$)-point functions of both
theories is the same for any $k$.

The growth,
as function of the number of vertices, of the number
 of correlation functions of this model is milder than 
 that of the models with full boundary sector. 
 We further simplify the notation and set 
$\mathfrak f_{2k,s_1}= \mathfrak f_{\mtc X_{2k},s_1}\hp 1$. 
With this notation, the Schwinger-Dyson equations in Section \ref{sec:SDE}
can be derived  
for the connected boundary graphs of the $\vuno$-model.

\begin{prop}[Schwinger-Dyson equations for the $\vuno$-model] \label{thm:SDEToy}  
Let $\B$ be a connected boundary graph of the quartic model with $2k$ vertices ($k\geq 1$), 
$\B\in \,\fey_3(\vuno)$. Let $\sb=\yb^1$, where $(\mtc X_{2k})_*(\Xb)=(\yb^1,\dots,\yb^k)$
for any $\Xb \in \mathcal{F}_{3,k}$.  The $(2k)$-point Schwinger-Dyson equation 
corresponding to $\B$ is 
% \leqnomode
% \allowdisplaybreaks[0]
\begin{align}& \hspace{-.4cm}
\bigg(1+\frac{2\lambda}{m^2+|\mathbf{s}|^2}  \suml_{q,p\,\in\Z }
G\hp2  (s_1,q,p)\bigg) \cdot 
G \hp{2k}(\Xb)  \label{eq:SDEtoy}  
\\ 
&=  \frac{2\lambda}{m^2+|\mathbf{s}|^2}  \Bigg\{ \frac{\delta_{1,k}}{2\lambda} -
\suml_{\hat \sigma \in \Z_	k} \sigma^* \mathfrak{f}_{2k,s_1}(\Xb) \nonumber
 \\
\nonumber
   & \qquad\qquad \qquad 
 -  
 \suml_{\rho>1}  \frac{Z_0\inv}{  [(y_1^\rho)^2-s_1^2]} \cdot \bigg[\dervpar{Z[J,J]}{\varsigma_1(\mathcal{X}_{2k} ;1,\rho) }
 (\Xb)
 -\dervpar{Z[J,J]}{\varsigma_1(\mathcal{X}_{2k} ;1,\rho) } (\Xb|_{s_1\to y_1^\rho})
 \bigg]
\\
& \nonumber 
\,\,\,\,\qquad \qquad\,\,\,\,\,\,\,+ \suml_{q\in \Z}\frac{1}{s_1^2-q^2} \big[ G \hp{2k}(\Xb) - G \hp{2k}(\Xb|_{s_1\to q}) \big] \Bigg\} \,.  
\end{align}
\end{prop} 

\begin{proof}
 For $k>1$, it is immediate by setting $D=3$ and by cutting the sums
 over the number of colours to only $a=1$, since one does no longer
 have the vertices $\vdos$ and $\vtres$ in the action.  After using
 $\Autc(\mtc X_{2k})=\Z_k$, and after inserting the form of the
 difference of propagators, as given by \eqref{eq:toy}, the result
 follows. If $k=1$, one additionally obtains the pure propagator term
 (the $\delta_{1,k}$-term) that would be otherwise annihilated by
 fourth or higher derivatives. For $k=1$, the sum over $\rho$ is empty
 (thus equal to zero).
 \end{proof}

One can still work out the functions $\mathfrak f_{2k}$ and 
give the correlation functions implied in the $\varsigma_1(\mathcal{X}_{2k} ;1,\rho)$-derivatives
in eq. \eqref{eq:SDEtoy}.
Notice that the expansion of the term $Y\hp1_{s_1}\jj$ is  
\begin{equation}
\label{eq:ToyY} 
Y\hp1_{s_1}\jj=
\suml_{k=0}^\infty \mathfrak{f}_{2k,s_1} \star \J(\mathcal X_{2k}) + \suml_{\mathcal C \,\,\mathrm{disconnected}} 
\mathfrak{f}_{\mathcal C,s_1}\hp{1}  \star \J(\mtc C)
\end{equation}
In order to determine $\mathfrak{f}_{2k,s_1}$ 
we find the graphs 
$\B$ such that $ \B \ominus e_1^r=\mathcal X_{2k}$ 
for certain (say, the $r$-th) vertex of $\B$. The restrictions \eqref{eq:Toyrestrictions}
with $B\geq 2$ and the connectedness of $ \B $ after edge-removal
imply that either 
\[n_1=n_k=1 ~~~~\mbox{and} ~~~~~n_i=0,~~ \mbox{ if } i\neq1,k\,,\]
or
\[n_{k+1}=1 ~~~~~\mbox{and}~~~~~n_i=0 ~~~~\mbox{ if }i\neq k+1\,.\]
That is to say, any such $\B$
has $2(k+1)$ vertices and, concretely, they might only be either $\logo{2}{Melon}{2.63}{3}\sqcup \mtc X_{2k}$
or $\mtc X_{2k+2}$, when $k\geq 2$. Adding the obvious case when $k=1$, one has:
\begin{subequations} \label{eq:Toy_f_explicit}
\begin{align}
\mathfrak f_{2,s_1} &=
\frac12\sum_{r=1}^2 \big(\Delta_{s_1,r} \Gcmm  + \Delta_{s_1,r} G\hp4 \big)   \\
\mathfrak f_{2k,s_1}&=
\frac{1}{k}\Delta_{s_1,1} G\hp{2k+2}_{|\icono{2}{Melon}{2.05}{3}| \mathcal X_{2k}|} 
+\frac{1}{k+1}\suml_{r=1}^k \Delta_{s_1,r} G\hp{2k+2},    \mbox{ for } k\geq 2  \,.
\end{align}
\end{subequations}
Notice that $\varsigma_1(\mathcal{X}_{2k} ;1,\rho)=\mathcal X_{2\rho-2} \sqcup \mathcal X_{2k-2\rho+2}$,
whence (see Fig. \ref{fig:simplemodel})
\begin{align} \nonumber
 % \frac{1}{Z_0}\suml_{\rho>1} 
 \frac{1}{Z_0}\dervpar{Z[J,J]}{\varsigma_1(\mathcal{X}_{2k} ;1,\rho)(\Xb) } & = 
%  \suml_{\rho>1} 
  G\hp{2\rho-2}(\xb^1,\ldots,\xb^{\rho-1}) \cdot G\hp{2 k-2\rho+2}(\xb^\rho,\ldots,\xb^k) 
  \\ & \,  
  +G\hp{2k}_{|\mtc X_{2(\rho-1)}|\mtc X_{2k-2(\rho-1)}|}(\Xb) \,.
\label{eq:Toy_varsigmas}
\end{align}
\begin{figure}
\raisebox{-2.3cm}{\includegraphics[height=5cm]{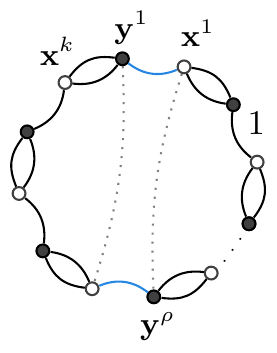}} $\qquad\mapsto \qquad$
\raisebox{-2.3cm}{\includegraphics[height=5cm]{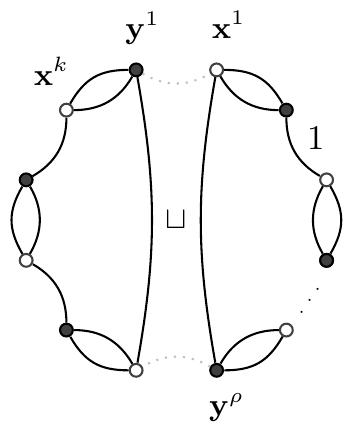}}
\caption{Shows the splitting of $\mathcal X_{2k}$ into the two components of $\varsigma_1(\mathcal X_{2k};1,\rho)$, $\rho>0$\label{fig:simplemodel}}
\end{figure}

Using the last four equations one can easily prove
\begin{cor}
 The exact 2-point equation for the $\vuno$-model is given, for any $\xb=(x_1,x_2,x_3)\in\Z^3$,  by
\begin{align}& \,\,\,\,\,\,
\bigg(1+\frac{2\lambda}{m^2+|\mathbf{x}|^2}  \suml_{q,p\,\in\Z }
G\hp2  (x_1,q,p)\bigg) \cdot 
G \hp{2}(\xb)  \label{eq:SDEtoy2pt}  
\\ 
&=  \frac{ 1}{m^2+|\mathbf{x}|^2} + \frac{(- 2\lambda) }{m^2+|\mathbf{x}|^2} \Bigg\{
\suml_{p,q\in\Z }  \Gcmm (x_1,q,p,\xb) +  G\hp4 (\xb,\xb) \nonumber
% \\
\nonumber
% 
%   & \qquad\qquad \qquad 
%  -  \suml_{\rho>1}  \frac{Z_0\inv}{  [(y_1^\rho)^2-s_1^2]} \cdot\dervpar{Z[J,J]}{\varsigma_1(\mathcal{X}_{2k} ;1,\rho)(\Xb) }  
  \\
 & \nonumber \,\,\,\,\qquad \qquad\qquad \qquad\qquad
- \suml_{q\in \Z}\frac{1}{x_1^2-q^2} \big[ G \hp{2}(x_1,x_2,x_3) - G \hp{2}(q,x_2,x_3) \big] \Bigg\} \,.  
\end{align}
 For $k\geq 2$, the multi-point equation for 
 $G\hp{2k}$, the single correlation function 
 of connected boundary graph, is given by 
 
\begin{align}& \,\,\,\,\,\, \nonumber
\bigg(1+\frac{2\lambda}{m^2+|\mathbf{s}|^2}  \suml_{q,p\,\in\Z }
G\hp2  (x_1^1,q,p)\bigg) \cdot 
G \hp{2k}(\xb^1,\ldots,\xb^k)  \label{eq:SDEtoyMultipoint}  
\\ 
&=  \frac{(-2\lambda)}{m^2+|\mathbf{s}|^2}  \Bigg\{   
\suml_{l=1 }^k  
\bigg[
\frac{1}{k} \suml_{p,q\in\Z}   G\hp{2k+2}_
{|\icono{2}{Melon}{2.05}{3}| \mathcal X_{2k}|} 
 (x_1^1,q,p;\xb^{1+l},\ldots,\xb^{k+l})
\\
&\qquad\qquad 
+\frac{1}{k+1}\suml_{r=1}^k   G\hp{2k+2}(\xb^{1+l},\xb^{2+l},\ldots,\xb^{r+l-1},x^1_1,x_2^{r+l-1},x_2^{r+l-1},\xb^{r+l},\ldots,\xb^{k+l})
\bigg]
\nonumber
\\
\nonumber
   & \qquad\qquad 
 + \suml_{\rho=2}^k \bigg[ \frac{G\hp{2\rho-2}(\xb^1,\ldots,\xb^{\rho-1})
 -G\hp{2\rho-2}(x_1^\rho, x_2^1,x_3^1,\ldots,\xb^{\rho-1})
}{  [(x_1^\rho)^2-(x^1_1)^2]} \cdot
G\hp{2k-2\rho+2}(\xb^\rho,\ldots,\xb^{k})
\\
& \nonumber  \qquad\qquad 
+  \frac{ 
G\hp{2k}_{|\mtc X_{2(\rho-1)}|\mtc X_{2k-2(\rho-1)}|}(\Xb) 
 -
G\hp{2k}_{|\mtc X_{2(\rho-1)}|\mtc X_{2k-2(\rho-1)}|}(x_1^\rho, x_2^1,x_3^1,\xb^2,\ldots,\xb^k)  
}{  [(x_1^\rho)^2-(x^1_1)^2]}\bigg]
\\
& \nonumber 
 \qquad \qquad - \suml_{q\in \Z}\frac{G \hp{2k}(x_1^1,x_2^1,x_3^1,\xb^2,\ldots,\xb^k) 
 - G \hp{2k}(q,x_2^1,x_3^1,\xb^2,\ldots,\xb^k)}{(x_1^1)^2-q^2}\Bigg\} \,.  
\end{align}
for $\Xb=( \xb^1,\ldots,\xb^k)\in \mathcal F_{3,k}$, $\mathbf{s}:=(x_1^1,x_2^r,x_3^r)$,
and $\xb^i=(x_1^i,x_2^i,x_3^i)$ for all $i\in\{1,\ldots,k\} $. Moreover
$\xb^j=\xb^i \,\mod \,k$, for and $j\in \N$ with $i\in\{1,\ldots,k\} $. 
\end{cor}
It is pertinent to stress that 
$\mathbf s=\yb^1$ is a `chosen' black vertex, and this equation
holds for any other choice $\mathbf s=\yb^i$, $i\neq 1$, $(\mathcal X_{2k})_*(\Xb)=(\yb^1,\ldots,\yb^k)$,
after the pertinent changes (e.g. the sum over $\rho$ excludes not $1$ 
but $i$).
\begin{proof}
 One uses the equations \eqref{eq:Toy_f_explicit} and \eqref{eq:Toy_varsigmas},
 the triviality of the automorphisms group $\Autc(\meloncito)$, and 
 the invariance of $\Gcmm$ and $G\hp4$:
 \[
 \Gcmm(\zb,\yb)=\Gcmm(\yb,\zb) \qquad \mbox{and} \qquad G\hp4(\zb,\yb)=G\hp4 (\yb,\zb)\,.
 \]
  This is enough to obtain the $2$-point equation. For $k\geq2$, 
  on top of using \eqref{eq:Toy_f_explicit} and \eqref{eq:Toy_varsigmas}
  one explicitly writes the action of $\sigma\in\Z_k$. This 
  is rotation by $2\pi l/k$, $1\leq l\leq k$, so 
 \[\sigma ^*f(\Xb)=f(\xb^{\sigma \inv(1)},\ldots,\xb^{\sigma \inv (k)})=f(\xb^{1+l},\ldots,\xb^{k+l} )\]
 where $\xb^j=\xb^i$ mod $k$, for $i\in\{1,\ldots,k\} $ and $j\in \N$ and $f$ any
 (appropriate) function.  \qedhere 
 \end{proof}
 
\begin{rem} \label{rem:melonicSDEs}
An analysis on the divergence degree as function of Gur\u au's degree,
and the boundary components was done in \cite{wgauge,4renorm} for
group field theories.  It turns out that graphs with a disconnected
boundary are suppressed and therefore any graph contributing to
$\Gcmm$ is expected to be suppressed at least by $N\inv$, with respect
to those summed in $G\hp4$.  Nevertheless these correlation functions
with disconnected boundary can back react as do their analogous in
matrix models \cite{EynardTopologicalRecursion} in the topological
recursion \cite{BorotNotes}.  Also, by results of matrix theory
\cite{GW12}, the term $G\hp4 (\xb,\xb)$ is expected to be analogously
suppressed.  Hence, conjecturally, for the $\vuno$-model, the leading
order $G_{\mathrm{mel}}\hp2 $ of the two-point function
\eqref{eq:SDEtoy2pt} should satisfy the clearly more simple closed
equation \allowdisplaybreaks[0]
 \begin{align}& \,\,\,\,\,\,
\bigg(m^2+|\mathbf{x}|^2+2\lambda \suml_{q,p\,\in\Z }
G_{\mathrm{mel}}\hp2  (x_1,q,p)\bigg) \cdot 
G _{\mathrm{mel}}\hp{2}(\xb)  \label{eq:SDEtoy2ptmel}  
 \\ 
&=   1+ 2\lambda  
% \suml_{p,q\in\Z }  \Gcmm (x_1,q,p,\xb) +  G\hp4 (\xb,\xb) 
\nonumber
 \suml_{q\in \Z}\frac{1}{x_1^2-q^2} \big[ G_{\mathrm{mel}} \hp{2}(x_1,x_2,x_3) - G_{\mathrm{mel}}\hp{2}(q,x_2,x_3) \big]  \,.  
\end{align}
and by the same token, one could truncate the 
equation for the $2k$-point function
\eqref{eq:SDEtoyMultipoint} to the following one,
where the equally suppressed terms $\mathfrak f_{2k,s_1}$
also are neglected:
\begin{align}& \,\,\,\,\,\,
\bigg(1+\frac{2\lambda}{m^2+|\mathbf{s}|^2}  \suml_{q,p\,\in\Z }
G_{\mathrm{mel}}\hp2  (x_1^1,q,p)\bigg) \cdot 
G_{\mathrm{mel}}\hp{2k}(\xb^1,\ldots,\xb^k)  \label{eq:melonicSDEtoyMultipoint}  
\\ 
&=  \frac{(-2\lambda)}{m^2+|\mathbf{s}|^2}  \bigg[    \nonumber
% \suml_{l=1 }^k  
% \bigg[
% \frac{1}{k} \suml_{p,q\in\Z}   G\hp{2k+2}_
% {|\icono{2}{Melon}{2.05}{3}| \mathcal X_{2k}|} 
%  (x_1^1,q,p;\xb^{1+l},\ldots,\xb^{k+l})
% \\
% &\qquad\qquad\qquad
% +\frac{1}{k+1}\suml_{r=1}^k   G\hp{2k+2}(\xb^{1+l},\xb^{2+l},\ldots,\xb^{r+l-1},x^1_1,x_2^{r+l-1},x_2^{r+l-1},\xb^{r+l},\ldots,\xb^{k+l})
% \bigg]
% \nonumber
% \\
\nonumber
 \suml_{\rho=2}^k  \frac{G\hp{2\rho-2}_{\mathrm{mel}}(\xb^1,\ldots,\xb^{\rho-1})-G\hp{2\rho-2}_{\mathrm{mel}}(x_1^\rho,x_2^1,x_3^1,\xb^2,\ldots,\xb^{\rho-1})}{  [(x_1^\rho)^2-(x^1_1)^2]} 
\cdot
G\hp{2k-2\rho+2}_{\mathrm{mel}}(\xb^\rho,\ldots,\xb^{k})
\\
& \nonumber 
 \qquad \qquad \quad- \suml_{q\in \Z}\frac{G_{\mathrm{mel}} \hp{2k}(x_1^1,x_2^1,x_3^1,\xb^2,\ldots,\xb^k) 
 - G_{\mathrm{mel}} \hp{2k}(q,x_2^1,x_3^1,\xb^2,\ldots,\xb^k)}{(x_1^1)^2-q^2}\bigg] \,.  
\end{align}
We warn the reader that these relations ---what we could call the 
\textit{melonic limit} and corresponds to the planar limit in matrix models \cite{GW12}--- still 
must be carefully proven (see Sec. \ref{sec:conclusions}).   It is 
very encouraging to see, though, that after 
determining $G_{\mathrm{mel}}\hp2$, 
the `melonic $2k$-point SDE'  \eqref{eq:melonicSDEtoyMultipoint} for 
any $k>1$ can now be entirely expressed 
in terms of already known functions $G_{\mathrm{mel}}\hp2, G_{\mathrm{mel}}\hp4,\ldots,G_{\mathrm{mel}}\hp{2k-2}$
and constitutes an equation only for $G_{\mathrm{mel}}\hp{2k}$,
which would decouple the tower.
\allowdisplaybreaks
\end{rem}

\section{Outlook: Gurau-Witten SYK-like model} \label{sec:GW}
We believe that some of the present methods can be extended to the 
so-called Gur\u au-Witten model(s) based on work of Sachdev, Ye and Kitaev. 
We sketch here how.
\par  
Gur\u{a}u-Witten model consists of fermions $\psi^a$ that are tensorial
of rank $3$, transforming in the trifundamental 
representation of $ G_{ab} \times G_{ac}\times G_{ad} $,
where $\{a,b,c,d\}=\{0,1,2,3\}$ and each $G_{ij}=G_{ji}$ is a copy of a 
Lie group, e.g. $\mtr O(n)$. That is, in $(\psi^a)_{bcd}$, each subindex $e$,
independently transforms under the fundamental representation of $G_{ae}$, $e\neq a$.
The new integer $n$ is related to the 
old number of sites $N$ (fermions in the original SYK-model \cite{SachdevYe,Kitaev}) by $N=4n^3$. 
The quartic monomial in Witten's model is  %\correction{$N \neq n$ }
given by the $\mtr O(n)$-\textit{invariant}
\begin{equation}
S_{\mtr{int.}}[\,\{\psi^b\}_{b=0}^3\,]= -
% \ii^{q/2}
\eta(n)   \suml_{\mu_i^j=1}^n 
 \tensor{(\psi^0)}{_{\mu_1^0\mu_2^0\mu_3^0}} (\psi^1)_{\mu_0^1\mu_2^1\mu_3^1}
(\psi^2)_{\mu_0^2\mu_1^2\mu_3^2} (\psi^3)_{\mu_0^3\mu_1^3\mu_2^3}
\prod\limits_{i\neq j} \delta_{\mu^i_j\;\!,\,\mu_j^i} \,,
\label{eq:GuWi_invariant}
%\qquad (\eta \in \re)
\end{equation}
which is also abbreviated as $- %\ii^{q/2}
\eta(n)  
\psi^0\psi^1\psi^2\psi^3$, being $\eta(n)=\eta_0 n^{-3/2}, \eta_0\in \re$. 
With his action, the partition function is 
\[ 
Z^\re_{\mtr{Gu.Wi.}}[\{J^a \}]=Z_0\inv \int \Df{\psi}\exp\bigg(-\int \dif{\tau} \sum_b 
\frac12 \psi^b \frac{\dif{\psi^b}}{\dif \tau} + \naranja{
% \ii^{q/2}
\eta(n)  }
\psi^0\psi^1\psi^2\psi^3 + \sum_b \psi^bJ^b  \bigg)\,.
 \]
 Here the propagator is the sum of the four 
 quadratic $\mtr O(n)$-invariants, thus, for each $b$, 
 $\psi^b \dif\psi^b/
 \dif \tau$ stands actually for $\sum_{\rho, \mu,\sigma=1}^n 
 \tensor{(\psi^b)}{_{\alpha \gamma \sigma}} \tensor{ (\dif\psi^b/\dif\tau)}{_{\alpha\gamma\sigma}} $.  
We use the complex version of the model, whose partition 
function $Z^\C_{\mtr{Gu.Wi.}}[\{J^a,\bar J^a\}]$ reads 
\[
\int  \Df{\bar\psi}\Df{\psi}\exp\bigg(-\int \dif{\tau} \sum_b 
  \bar\psi_b \frac{\dif{\psi^b}}{\dif \tau} + \naranja{ %\ii^{q/2}
  \eta(n)  
  } (
\psi^0\psi^1\psi^2\psi^3 + \bar\psi^3\bar\psi^2\bar\psi^1\bar\psi^0 )+\sum_b \bar\psi^bJ^b + \bar J^b \psi^b \bigg)\,.
 \]
The interaction-vertices of the last theory are graphically represented by 
\[
\psi^0\psi^1\psi^2\psi^3 =\,\,\raisebox{-1.52cm}{
\includegraphics[width=3.2cm]{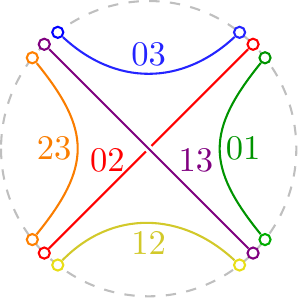}} \qquad \mbox{and} \qquad
\bar\psi^3\bar\psi^2\bar\psi^1\bar\psi^0=\,\, \raisebox{-1.52cm}{\includegraphics[width=3.2cm]{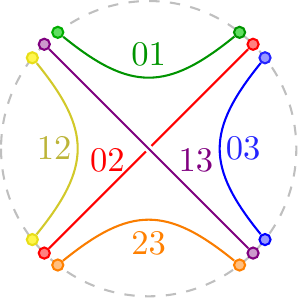}}
\] 
Triples of empty (resp. filled) dots 
marked by edges having the colour $a$
in the bicolouration represent a the field $\psi^a$ (resp. $\bar \psi^a$). These edges
are the deltas in eq. \eqref{eq:GuWi_invariant}. 
Also, there are four terms in the propagator
(the four summands in the quadratic part; see 
the Feynman diagram shown in Fig. \ref{fig:8pt})
\begin{figure}
\includegraphics[width=.57\textwidth]{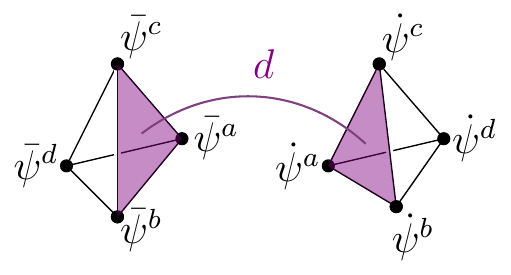}
 \caption{Gluing of tetrahedra caused by the Wick contraction $\contraction{}{\bar\psi^d}{}{\psi^d}\bar\psi^d\dot\psi^d$, $\dot \psi=\dif \bar\psi/\dif\tau$ \label{fig:tetrahedra}}
\end{figure}

Each quartic interaction vertex can be seen as a 
tetrahedron with fields $\psi^{d}$ at their vertices
and with marked (coloured) faces, being $\psi^d$
opposite to the face with colour $d$ for each $d=0,\ldots,3$.
Thus, for the complex Gur\u au-Witten model, 
the Wick's contraction of $\dot \psi^d$ with $\bar \psi^d$ is, 
as in Figure \ref{fig:tetrahedra}, gluing the face coloured $d$ 
(opposite to the vertex $\psi^d$ in that figure), and any Feynman 
diagram results in certain simplicial complex, which, in case of 
having external legs, has a boundary consisting of coloured triangles.
If the number of triangles of 
colour $A\in\{0,\ldots,3,\bar 0,\ldots,\bar 3\}$ 
at the boundary is denoted by $\kappa_A$,
a boundary graph consists then of $2k$ triangles
and is specified not only by $(\kappa_0,\ldots, \kappa_3,\kappa_{\bar 0}\ldots,\kappa_{\bar 3})\in \Z_{\geq 0}^8$
such that $2k=\sum_A \kappa_A$, but also by momentum transmission. 
To wit, 
\textit{connected} boundary graphs with $2k$ external legs have the following properties:
\begin{itemize}
 \item the vertex set is octo-partite in colours in the set 
 $\mtb O=\{0,1,2,3,\bar 0,\bar 1, \bar 2, \bar 3\}$. 
 One writes $p: \mathbb O\to \{0,1,2,3\}$, $p(A)=a$ if either $A=a$ or $A=\bar a$. 
 We impose on the `bar' operation $\bar{\cdot}:\mathbb O\to\mathbb O$ the property 
 $\bar {\bar A}=A$ for each $A\in\mtb{O}$.
 \item there are $\kappa_a$ vertices of colour $a$ and these
 numbers satisfy $  2k=\sum_{A\in \mtb  O} \kappa_A $ 
 \item the edge set is bicoloured, that is every 
 edge is labeled by one of the six elements in $S$,
 the set of unordered pairs of different colours in $\{0,1,2,3\}$
%  \[
% S= (\{0,1,2,3\}\times \{0,1,2,3\} \setminus \Delta )/ (ij) \sim (ji)    
% \quad (\Delta =\mbox{diagonal of } \{0,\ldots,3\}^2 )
% \]  
\item Let $A,B \in \mathbb O$ label two vertices of a graph. Either there is 
no edge joining them or they are connected by an
edge which, according to previous point, bears a bicolouration $\ell\in S$ and one
constrained to the following possibilities:
\begin{itemize}
 \item[(i)]  $p(A)=p(B)=:a$ and $A=\bar B$, an in this case $\ell\in S$ can be any of 
 $\{(ai) \,|\, i\neq a\}$ 
 \[\raisebox{4.6ex}{\tikz[scale=.4,
    baseline=4ex,shorten >=.1pt,node distance=28mm,%on grid,
    semithick,auto,
    every state/.style={draw=black,inner sep=.54mm,text=black,minimum size=0}]{ 
 \node[state,circle,draw=white,minimum size=0,fill=gray,text=white] (11) at (0,0) {\footnotesize $\bar a \vphantom j $} ;
 \node[state,circle,draw=white,minimum size=0,fill=gray,text=white] (12) at (3,0) {\footnotesize $a \vphantom j$} ;
 \path (11) edge node {\small $ia$} (12);
 }}\]
 \item[(ii)] $a:=p(A)\neq p(B)=:b$. In this case, $\ell=(ab)$, that is
 \[
 \raisebox{4.5ex}{\tikz[scale=.4,
    baseline=4ex,shorten >=.1pt,node distance=28mm,%on grid,
    semithick,auto,
    every state/.style={draw=black,inner sep=.54mm,text=black,minimum size=0}]{ 
 \node[state,circle,draw=white,minimum size=0,fill=gray,text=white] (11) at (0,0) {\footnotesize $b \vphantom j$} ;
 \node[state,circle,draw=white,minimum size=0,fill=gray,text=white] (12) at (3,0) {\footnotesize $\bar a \vphantom j$} ;
 \path (11) edge node {\footnotesize $ab$} (12);
 }} \mbox{ or }
 \raisebox{4.5ex}{\tikz[scale=.4,
    baseline=4ex,shorten >=.1pt,node distance=28mm,%on grid,
    semithick,auto,
    every state/.style={draw=black,inner sep=.54mm,text=black,minimum size=0}]{ 
 \node[state,circle,draw=white,minimum size=0,fill=gray,text=white] (11) at (0,0) {\footnotesize $b \vphantom j$} ;
 \node[state,circle,draw=white,minimum size=0,fill=gray,text=white] (12) at (3,0) {\footnotesize $a \vphantom j$} ;
 \path (11) edge node {\footnotesize $ab$} (12);
 }} 
 \] 
\end{itemize}
  \end{itemize}

Disconnected boundary graphs have components that enjoy from 
all these properties. 
 
This gives the following classification of 
correlation functions. 
\begin{itemize}
 \item two-point functions: for $2k=2$, the only possibilities is 
 $\kappa_a=\kappa_{\bar{a}}=1$ and $\kappa_B=0$ for $B\in\mtb O\setminus\{a,\bar a\}$. Thus there are 
 four 2-point functions:
 \[
 \raisebox{4ex}{\tikz[scale=.4,
    baseline=4ex,shorten >=.1pt,node distance=28mm,%on grid,
    semithick,auto,
    every state/.style={draw=black,inner sep=.54mm,text=black,minimum size=0}]{ 
 \node[state,circle,draw=white,minimum size=0,fill=gray,text=white] (11) at (0,0) {\footnotesize $a \vphantom j$} ;
 \node[state,circle,draw=white,minimum size=0,fill=gray,text=white] (12) at (3,0) {\footnotesize $\bar a \vphantom j$} ;
 \path (11) edge[bend left=40] node {\footnotesize $ab$} (12);
  \path (11) edge[bend left=-10] node {\footnotesize $ac$} (12);
  \path (11) edge[bend left=-40] node[swap] {\footnotesize  $ad$} (12);
 }}
 \]
 \item four-point functions. One can have:
 \begin{itemize} 
 \item $\kappa_a=\kappa_{\bar a}=2$ for 
 certain $a$, otherwise $\kappa_B=0$, $B\in \mtb O\setminus p\inv(a)$; 
 \[ 
 \raisebox{4ex}{\tikz[scale=.4,
    baseline=4ex,shorten >=.1pt,node distance=28mm,%on grid,
    semithick,auto,
    every state/.style={draw=black,inner sep=.54mm,text=black,minimum size=0}]{ 
 \node[state,circle,draw=white,minimum size=0,fill=gray,text=white] (11) at (0,0) {\footnotesize$a \vphantom j$} ;
 \node[state,circle,draw=white,minimum size=0,fill=gray,text=white] (12) at (3,0) {\footnotesize$\bar a \vphantom j$} ;
 \node at (4.5,0)
  {$\sqcup$};
 \node[state,circle,draw=white,minimum size=0,fill=gray,text=white] (13) at (6,0) {\footnotesize$a \vphantom j$} ;
 \node[state,circle,draw=white,minimum size=0,fill=gray,text=white] (14) at (9,0) {\footnotesize$\bar a \vphantom j$} ;
 \path (11) edge[bend left=40] node {\footnotesize $ab$} (12);
  \path (11) edge[bend left=-10] node {\footnotesize $ac$} (12);
  \path (11) edge[bend left=-40] node[swap] {\footnotesize $ad$} (12);
   \path (13) edge[bend left=40] node {\footnotesize $ab$} (14);
  \path (13) edge[bend left=-10] node {\footnotesize $ac$} (14);
  \path (13) edge[bend left=-40] node[swap] {\footnotesize $ad$} (14);
 }}, \qquad
  \raisebox{1ex}{\tikz[scale=.4,
    baseline=4ex,shorten >=.1pt,node distance=28mm,%on grid,
    semithick,auto,
    every state/.style={draw=black,inner sep=.54mm,text=black,minimum size=0}]{ 
 \node[state,circle,draw=white,minimum size=0,fill=gray,text=white] (11) at (0,0) {\footnotesize$a \vphantom j$} ;
 \node[state,circle,draw=white,minimum size=0,fill=gray,text=white] (12) at (3,0) {\footnotesize $\bar a \vphantom j$} ;
%  \node at (3.75,0)
%   {$\sqcup$}
;
 \node[state,circle,draw=white,minimum size=0,fill=gray,text=white] (13) at (0,3.24) {\footnotesize $a \vphantom j$} ;
 \node[state,circle,draw=white,minimum size=0,fill=gray,text=white] (14) at (3,3.24) {\footnotesize $\bar a \vphantom j$} ;
 \path (11) edge[bend left=30] node {\footnotesize$ab$} (12);
  \path (11) edge[bend left=-30] node[swap]  {\footnotesize$ac$} (12);
  \path (11) edge[bend left=-10] node[] {\footnotesize$ad$} (13);
   \path (13) edge[bend left=30] node {\footnotesize$ab$} (14);
  \path (13) edge[bend left=-30] node[swap] {\footnotesize$ac$} (14);
  \path (14) edge[bend left=-10] node[] {\footnotesize$ad$} (12);
 }} \quad (d\neq a), 
 \]
 which have been called `broken' and `unbroken' in \cite{GurauSYKinvN}.
 \item or   $\kappa_a=\kappa_c=\kappa_{\bar a}=\kappa_{\bar c}=1$
 and $\kappa_B=0$ $B\in (\mtb O \setminus p\inv\{a,c\})$,
  \[ 
 \raisebox{4ex}{\tikz[scale=.4,
    baseline=4ex,shorten >=.1pt,node distance=28mm,%on grid,
    semithick,auto,
    every state/.style={draw=black,inner sep=.54mm,text=black,minimum size=0}]{ 
 \node[state,circle,draw=white,minimum size=0,fill=gray,text=white] (11) at (0,0) {\footnotesize$a \vphantom j$} ;
 \node[state,circle,draw=white,minimum size=0,fill=gray,text=white] (12) at (3,0) {\footnotesize$\bar a \vphantom j$} ;
 \node at (4.5,0)
  {$\sqcup$};
 \node[state,circle,draw=white,minimum size=0,fill=gray,text=white] (13) at (6,0) {\footnotesize$a \vphantom j$} ;
 \node[state,circle,draw=white,minimum size=0,fill=gray,text=white] (14) at (9,0) {\footnotesize$\bar a \vphantom j$} ;
 \path (11) edge[bend left=40] node {\footnotesize $ab$} (12);
  \path (11) edge[bend left=-10] node {\footnotesize $ac$} (12);
  \path (11) edge[bend left=-40] node[swap] {\footnotesize $ad$} (12);
   \path (13) edge[bend left=40] node {\footnotesize $ab$} (14);
  \path (13) edge[bend left=-10] node {\footnotesize $ac$} (14);
  \path (13) edge[bend left=-40] node[swap] {\footnotesize $ad$} (14);
 }}, \qquad
  \raisebox{1ex}{\tikz[scale=.4,
    baseline=4ex,shorten >=.1pt,node distance=28mm,%on grid,
    semithick,auto,
    every state/.style={draw=black,inner sep=.54mm,text=black,minimum size=0}]{ 
 \node[state,circle,draw=white,minimum size=0,fill=gray,text=white] (11) at (0,0) {\footnotesize$a \vphantom j$} ;
 \node[state,circle,draw=white,minimum size=0,fill=gray,text=white] (12) at (3,0) {\footnotesize $\bar a \vphantom j$} ;
%  \node at (3.75,0)
%   {$\sqcup$}
;
 \node[state,circle,draw=white,minimum size=0,fill=gray,text=white] (13) at (0,3.24) {\footnotesize $c \vphantom j$} ;
 \node[state,circle,draw=white,minimum size=0,fill=gray,text=white] (14) at (3,3.24) {\footnotesize $\bar c \vphantom j$} ;
 \path (11) edge[bend left=30] node {\footnotesize$ab$} (12);
  \path (11) edge[bend left=-30] node[swap] {\footnotesize$ad$} (12);
  \path (11) edge[bend left=-10] node[] {\footnotesize$ac$} (13);
   \path (13) edge[bend left=30] node {\footnotesize$ab$} (14);
  \path (13) edge[bend left=-30] node[swap] {\footnotesize$ad$} (14);
  \path (14) edge[bend left=-10] node[] {\footnotesize$ac$} (12);
 }}
 \]
 
 \item $\kappa_0=\ldots=\kappa_3=1$   or  $\kappa_{\bar 0}=\ldots=\kappa_{\bar 3}=1$, 
 \[  \raisebox{1ex}{\tikz[scale=.34,
    baseline=4ex,shorten >=.1pt,node distance=28mm,%on grid,
    semithick,auto,
    every state/.style={draw=black,inner sep=.54mm,text=black,minimum size=0}]{ 
 \node[state,circle,draw=white,minimum size=0,fill=gray,text=white] (11) at (0,0) {\footnotesize$1 \vphantom j$} ;
 \node[state,circle,draw=white,minimum size=0,fill=gray,text=white] (12) at (3,0) {\footnotesize $  2 \vphantom j$} ;
%  \node at (3.75,0)
%   {$\sqcup$}
;
 \node[state,circle,draw=white,minimum size=0,fill=gray,text=white] (13) at (0,3.24) {\footnotesize $3 \vphantom j \vphantom{\bar 1} $} ;
 \node[state,circle,draw=white,minimum size=0,fill=gray,text=white] (14) at (3,3.24) {\footnotesize $ 0 \vphantom j\vphantom{\bar 1}$} ;
 \path[-]
  (14) edge[bend left=0,draw=white,double=black,double distance=\pgflinewidth, thick] (11)
  (14) edge[bend left=10] (12)
  (14) edge[bend left=-10] (13)
  (12) edge[bend left=10] (11)
  (13) edge[bend left=-10] (11)
  (13) edge[bend left=0,draw=white,double=black,double distance=\pgflinewidth, thick] (12)
; 
 }}
\mbox{ or }
 \raisebox{1ex}{\tikz[scale=.34,
    baseline=4ex,shorten >=.1pt,node distance=28mm,%on grid,
    semithick,auto,
    every state/.style={draw=black,inner sep=.54mm,text=black,minimum size=0}]{ 
 \node[state,circle,draw=white,minimum size=0,fill=gray,text=white] (11) at (0,0) {\footnotesize$ \bar 1 \vphantom j$} ;
 \node[state,circle,draw=white,minimum size=0,fill=gray,text=white] (12) at (3,0) {\footnotesize $ \bar 2 \vphantom j$} ;
%  \node at (3.75,0)
%   {$\sqcup$}
;
 \node[state,circle,draw=white,minimum size=0,fill=gray,text=white] (13) at (0,3.24) {\footnotesize $\bar 3 \vphantom j$} ;
 \node[state,circle,draw=white,minimum size=0,fill=gray,text=white] (14) at (3,3.24) {\footnotesize $\bar 0 \vphantom j$} ;
 \path[-]
  (14) edge[ draw=white,double=black,double distance=\pgflinewidth, thick] (11)
  (14) edge[bend left=10] (12)
  (14) edge[bend left=-10] (13)
  (12) edge[bend left=10] (11)
  (13) edge[bend left=-10] (11)
  (13) edge[ draw=white,double=black,double distance=\pgflinewidth, thick] (12)
; 
 }},  \mbox{ respectively, called `exceptional' in \cite{GurauSYKinvN}} 
 \]
 \end{itemize}
 
This matches for $k=1,2$ the description given in \cite{GurauSYKinvN} for the 
real case, in a somehow more different notation than that of Fig. 3 there. 
\item six-point functions will be either classified by a disconnected 
boundary with components in the graphs that classify the 2-point and/or 4-point functions
adding up to $6$-vertices or they will be connected. This latter case needs a more
complicated analysis to be fully classified. Examples of connected boundary
graphs are the known ones for complex tensor models with all the possible 
vertex-octo-colourations of the \textit{melonic} graphs in six edges, but 
also many new graphs are possible, e.g.
\[
% \logo{6}{QGW}{7}{3} \naranja{\leftthreetimes --forces\,a=b ?} \quad 
\logo{6}{NewGW}{19}{3}\] 	
\end{itemize}
A program to extend the present methods to the Gur\u au-Witten model
begins
\begin{itemize}
\item  with the association of a cycle of sources $J^c, \bar J^c$
for each boundary graph  
\item to expand the free energy of the model, $\log (Z^\C_{\mtr{Gu.Wi.}} [\{J^b,\bar J^b\}])$
in cycles of sources for each graph  $\B\in \im \,\partial_{\mtr{Gu.Wi.}}^\C$
in the boundary sector of this model. Developing a graph-calculus and classify, modulo colour-orbits,
multi-point functions 
\item the triviality of the Ward-Identity might be overcome by choosing a
propagator that does depend on the group-variables,  since this is 
a first order derivative, making `time' $\tau$ related to momenta $\boldsymbol \mu$ (e.g. 
replacing the propagator by a momentum-dependent propagator $\int \dif \tau
\sum_a\sum_{\boldsymbol{\mu}}\bar\psi^a(\tau)\dot\psi(\tau)\delta(\tau-f(\boldsymbol \mu)) $,
for which the WTI would be non-trivial)
\item extend the SDE-techniques to interaction vertices 
that are beyond the class treated here. This class allows 
only interaction vertices whose (graph-)vertices lie on 
a subgraph of the form
\[
\raisebox{-.4\height}{\includegraphics[width=4cm]{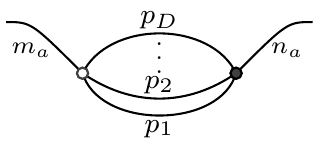}} \mbox{ for some colour $a$.}
\]
This condition was essential for the use of the WTI
when deriving the SDEs.
\end{itemize}

\begin{figure}
\raisebox{-2cm}{
\includegraphics[width=12.9cm]{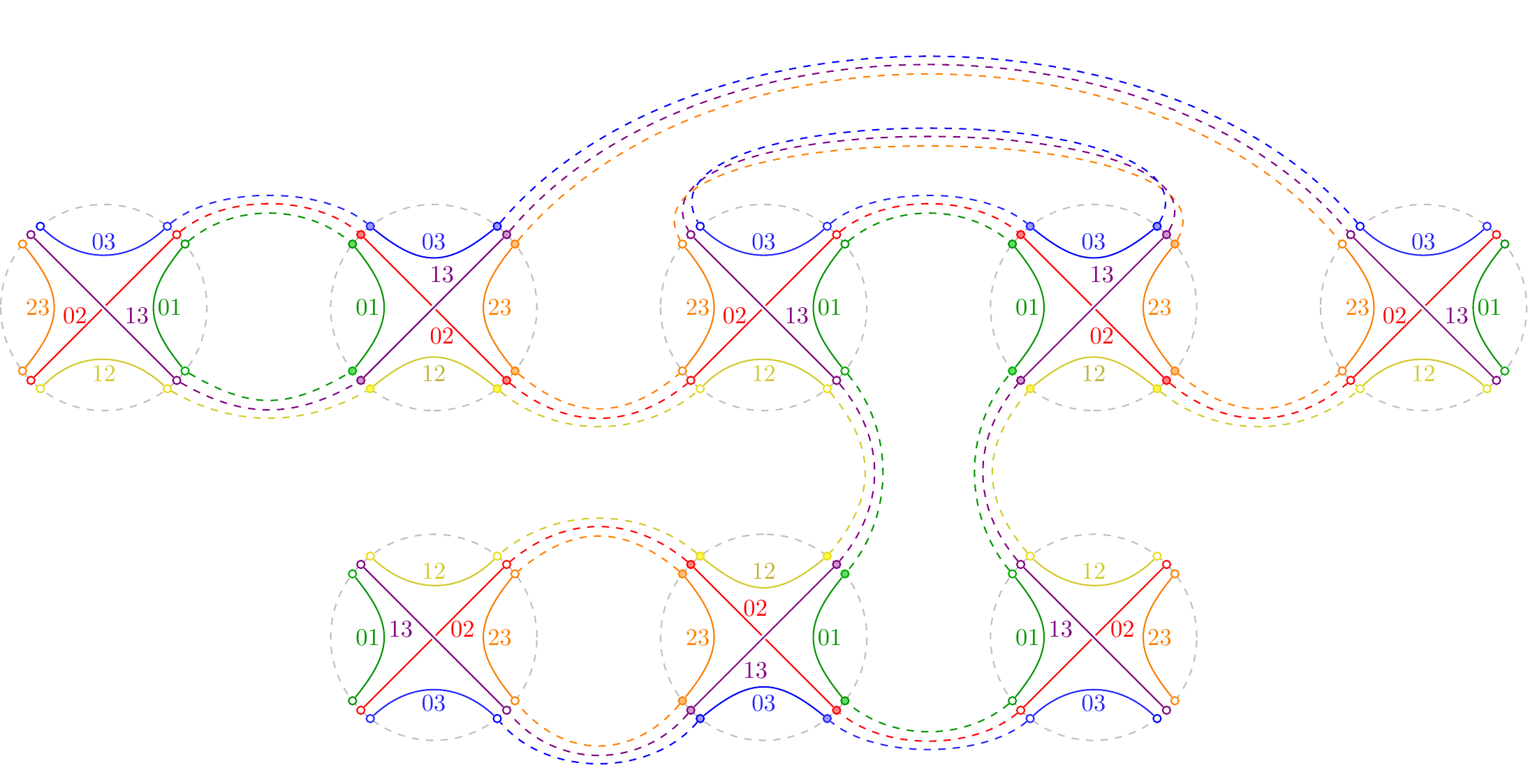}}\hspace{-1cm}
$\mapsto $ \raisebox{-1cm}{
\includegraphics[width=2cm]{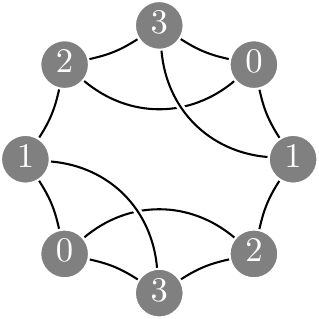}
}
 \caption{\label{fig:8pt}
 Example of a Feynman graph of the (complex) Gur\u au-Witten model that contributes to the correlation function 
 indexed by the graph in the right, where any edge from $a$ to $b$ coloured vertices is 
 $ab$-bicoloured (not always the case, as e.g. for $a=\bar b$ in another graphs). 
 }
\end{figure}

\allowdisplaybreaks 
\section{Conclusions} \label{sec:conclusions}
We studied the correlation functions of complex tensor models and, mainly, presented a collection of
generating functionals that allowed to derive the exact Schwinger-Dyson equations 
for CTMs of rank $3$ and $4$  
(and in Appendix \ref{sec:rankfive}, rank-$5$ theories). 
The symmetry of the 
colours should be exploited in order to obtain a simplified 
version of them, which shall lead to a solution. 
The path towards closed equations, i.e. equations where 
a single unknown correlation function appears, is the analysis
of Gur\u au degree $\omega$ sectors:
\[
G\hp{2k}_\B=  G\hp{2k}_{\B,\mtr{mel}}+ \sum_{\omega=1}^\infty G\hp{2k,\omega}_\B\,, \qquad\qquad (G\hp{2k}_{\B,\mtr{mel}}=G\hp{2k,0}_{\B})\,.
\]
The idea is proceed recursively: knowing the solution at leading order of 
a certain correlation function, it can be inserted 
in the equation at subleading order (called `cherry trees' 
in \cite{DartoisGurauRivasseau}), and so on. The pivotal equation 
should imply a closed equation for the 2-point function alone in the melonic
approximation, similar to the one obtained in \cite{us}. 
This requires a rather deep combinatorial and topological analysis 
of the contractions of the external lines for tensor 
models, similar to the one undertaken in \cite{GWpowercounting},
and condensed in \cite[Prop 3.3]{GW12} for matrix models. This 
would prove the claims in remark \ref{rem:melonicSDEs} and will allow 
to find the analogous equations for sectors of
any higher value of Gur\u{a}u's degree $\omega>1$. That result
would provide insight on the solvability of the $\vuno$-model. 
We address
both problems in a next paper. 
\par
Additionally, in Section \ref{sec:GW}, the 
boundary sector of the complex Gur\u au-Witten model has been characterized.
This allowed us to write down there a plan that, we hope, would be 
useful to implement non-perturbative techniques for holographic tensor models. 
\par 
It would be interesting to extend the graph calculus developed in \cite{fullward} 
(and further elaborated in this article) to the recently introduced 
2PI formalism of tensor models \cite{2PI}. 
In particular, the tensor-models-compatible connected sum defined in \cite{cips}
imply that the 2PI functional of rank-$3$ models is the generating functional of 
prime $3$-manifolds.

  \let\cleardoublepage\clearpage 
   \appendix
%   \nopartblankpage
% \appendixpage 
 
%   \let\cleardoublepage\clearpage 
\section{Proof of Lemma \ref{thm:ordersix}} \label{sec:proofoflemma}
\begin{proof}%[of expansion $\ysa$]
The proof is long but straightforward. We compute some terms as a matter of example
\begin{subequations}  \label{eq:laaaaange}
\begin{align}
 \frac1{2!\cdot 2^2} 
\langle\!\langle G\hp 8 _{|\raisebox{-.34\height}{\includegraphics[height=1.7ex]{graphs/3/Item4_Va.pdf}}|
\raisebox{-.34\height}{\includegraphics[height=1.7ex]{graphs/3/Item4_Va.pdf}}|} 
\,, & \,
\raisebox{-.4\height}{\includegraphics[height=3.5ex]{graphs/3/Logo4_Va.pdf}}
\sqcup 
\raisebox{-.4\height}{\includegraphics[height=3.5ex]{graphs/3/Logo4_Va.pdf}}\,
) \rangle\!\rangle_{s_a} \\ \nonumber 
{}={} &  \frac18 \bigg( 
\suml_{r=1,2} \Dsa{r} \Goaa +
\suml_{r=3,4} (123)^*(\Dsa{r} \Goaa) 
\bigg) \star \mathbb J ( \melon \sqcup \,\logo{4}{Va}{3.5}{38} \, ) 
\\
 \frac1{2!\cdot 2^2} 
\suml_{i\neq a}
\edgebracket{
G\hp 8 _{|\raisebox{-.4\height}{\includegraphics[height=1.7ex]{graphs/3/Item4_Vi.pdf}}|
\raisebox{-.4\height}{\includegraphics[height=1.7ex]{graphs/3/Item4_Vi.pdf}}|} 
\,, & \,  
\raisebox{-.4\height}{\includegraphics[height=4ex]{graphs/3/Logo4_Vi.pdf}}
\sqcup 
\raisebox{-.4\height}{\includegraphics[height=4ex]{graphs/3/Logo4_Vi.pdf}}\,
} 
\\ \nonumber  {}={} &
\frac{1}{8}\suml_{i\neq a}\big(
\suml_{r=1,2} \Dsa{r} \Goii
+\suml_{r=3,4} (123)^*\Dsa{r} \Goii \big)
\star \J ( \melon \sqcup \,\logo{4}{Vi}{3.5}{38} )
% & after the comma; and 
\\
\frac1{2^2}\suml_{i\neq a} \edgebracket{
 G\hp 8 _{|\raisebox{-.4\height}{\includegraphics[height=1.7ex]{graphs/3/Item4_Vi.pdf}}|
\raisebox{-.4\height}{\includegraphics[height=1.7ex]{graphs/3/Item4_Va.pdf}}|} 
\,, & \,
\raisebox{-.4\height}{\includegraphics[height=4ex]{graphs/3/Logo4_Vi.pdf}}
\sqcup 
\raisebox{-.4\height}{\includegraphics[height=4ex]{graphs/3/Logo4_Va.pdf}}\,}
\\ \nonumber   {}={} &
\frac14 \suml_{i\neq a} 
\Big((
\sumud \Dsa{r} \Goia )
\star \J (\melon \sqcup \logo{4}{Va}{3.5}{38} )
\\ & +
(\sumtc (123)^*\Dsa{p}\Goia )
\star \J (\melon \sqcup \logo{4}{Vi}{3.5}{38} )
\Big) \nonumber
\\ 
\frac1{2^2}  \edgebracket{
 G\hp 8 _{|\raisebox{-.4\height}{\includegraphics[height=1.7ex]{graphs/3/Item4_Vb.pdf}}|
\raisebox{-.4\height}{\includegraphics[height=1.7ex]{graphs/3/Item4_Vc.pdf}}|} 
 \, , & \, 
\raisebox{-.4\height}{\includegraphics[height=4ex]{graphs/3/Logo4_Vb.pdf}}
\sqcup 
\raisebox{-.4\height}{\includegraphics[height=4ex]{graphs/3/Logo4_Vc.pdf}}\,}
\\ \nonumber   {}={} &
\frac14  
\Big((
\sumud \Dsa{r} \Gobc )
\star \J (\melon \sqcup \logo{4}{Vc}{3.5}{38} )
\\ & +
(\sumtc (123)^*\Dsa{p}\Gobc )
\star \J (\melon \sqcup \logo{4}{Vb}{3.5}{38} )
\Big) \nonumber 
\\
\frac1{3}  \edgebrackett{ G\hp8 _{ 
|\raisebox{-.33\height}{\includegraphics[height=2ex]{graphs/3/Item2_Melon.pdf}}|
\raisebox{-.34\height}{\includegraphics[width=.36cm]{graphs/3/Item6_K33.pdf}}|} 
 \, &, \,
\raisebox{-.3\height}{\includegraphics[height=3ex]{graphs/3/Logo2_Melon.pdf}}
\sqcup
\raisebox{-.3\height}{\includegraphics[width=.7cm]{graphs/3/Logo6_K33}} }
\\
\nonumber 
{} = {}  &  
\frac13 \Dsa{1} \Gomkthree \star \J ( \logo{6}{K33}{3.4}{3}) 
+ \frac13\suml_{q=2,3,4} \Dsa{q} \Gomkthree  \star \J (\melon \sqcup \logo{4}{Va}{3.5}{38} )
\\
\frac1{3}
 \suml_{i\neq a}
 \Edgebracket{ \Gomqi
\,  &, \,
\raisebox{-.3\height}{\includegraphics[height=3ex]{graphs/3/Logo2_Melon.pdf}}
\sqcup
\raisebox{-.34\height}{\includegraphics[height=6ex]{graphs/3/Logo6_Qi.pdf}} }
\\ 
\nonumber 
{} = {}  & 
\frac13 
 \suml_{i\neq a}
 \Big\{ 
 \Dsa{1} \Gomqi \star \J ( \logo{6}{Qi}{4.6}{253}) 
 +\suml_{q=2,3,4}
 \Dsa{q} \Gomqi \star \J (\melon \sqcup  \logo{4}{Va}{3.5}{38} )
  \Big\}
  \\
\frac1{3} 
 \Edgebracket{ \Gomqa 
\,  &, \,
\raisebox{-.3\height}{\includegraphics[height=3ex]{graphs/3/Logo2_Melon.pdf}}
\sqcup
\raisebox{-.34\height}{\includegraphics[height=6ex]{graphs/3/Logo6_Qa.pdf}} }
\\ 
\nonumber 
{} = {}  & 
\frac13   \Big\{
 \Dsa{1} \Gomqa \star \J ( \logo{6}{Qa}{4.6}{253}) 
 +\suml_{q=2,3,4}
 \Dsa{q} \Gomqa\Big\}\star \J (\melon \sqcup  \logo{4}{Va}{3.5}{38} )
  \\  \suml_{i\neq a} 
 \Edgebracket{ 
G\hp 8_{
|\raisebox{-.24\height}{\includegraphics[height=1.9ex]{graphs/3/Item2_Melon.pdf}}|
\raisebox{-.25\height}{\includegraphics[height=2.7ex]{graphs/3/Item6_Ei.pdf}}|} 
\, & , \,
\raisebox{-.3\height}{\includegraphics[height=3ex]{graphs/3/Logo2_Melon.pdf}}
\sqcup
\raisebox{-.34\height}{\includegraphics[height=5.5ex]{graphs/3/Logo6_Eis.pdf}}  }
\\ 
\nonumber 
{} = {}  & 
\Dsa{1} \Gomec \star \J ( \logo{6}{Ecab}{5}{28}) +
\suml_{h=2,3} \Dsa{h} \Gomec \star \J (\melon\sqcup  \logo{4}{Vb}{3.5}{38}  ) 
\\ 
\nonumber 
{} + {}  & 
\Dsa{1} \Gomeb \star \J ( \logo{6}{Ebac}{5}{28}) +
\suml_{h=2,3} \Dsa{h} \Gomeb \star \J (\melon\sqcup  \logo{4}{Vc}{3.5}{38}  ) 
\\ 
\nonumber 
{} + {}  & 
(\Dsa{4} \Gomec +  
\Dsa{4} \Gomeb ) \star \J (\melon\sqcup  \logo{4}{Va}{3.5}{38}  ) 
%  \end{align}
% \begin{align}
\\
\edgebracket{ 
\Gomea
\,, & \,
\raisebox{-.3\height}{\includegraphics[height=3ex]{graphs/3/Logo2_Melon.pdf}}
\sqcup
\raisebox{-.3\height}{\includegraphics[height=5.5ex]{graphs/3/Logo6_Eabc.pdf}}  
}
\\
 \nonumber
 {} = {}  &  
\Dsa{1} \Gomea \star \J (\logo{6}{Eabc}{5}{3})
 +
 \Dsa{2} \Gomea \star \J ( \melon \sqcup \logo{4}{Vc}{3.5}{38} )
 \\
 {}  + {} & \Dsa{3} \Gomea \star \J ( \melon \sqcup \melon\sqcup \melon)
 + \Dsa{4} \Gomea \star \J ( \melon \sqcup \logo{4}{Vb}{3.5}{38} ) \nonumber \\
 \frac1{2\cdot 2!}
\suml_{i\neq a} \edgebracket{
 G\hp 8 _{|
\raisebox{-.33\height}{\includegraphics[height=2ex]{graphs/3/Item2_Melon.pdf}}|
\raisebox{-.33\height}{\includegraphics[height=2ex]{graphs/3/Item2_Melon.pdf}}|
\raisebox{-.34\height}{\includegraphics[height=1.7ex]{graphs/3/Item4_Vi.pdf}}|} 
\, & , \, 
 \raisebox{-.3\height}{\includegraphics[height=3ex]{graphs/3/Logo2_Melon.pdf}}
\sqcup 
\raisebox{-.3\height}{\includegraphics[height=3ex]{graphs/3/Logo2_Melon.pdf}}
\sqcup 
\raisebox{-.4\height}{\includegraphics[height=4ex]{graphs/3/Logo4_Vi.pdf}}}
\\
\nonumber
  {} = {}  & 
  \frac14 \big\{ \sum_{i\neq a}
  \sumud \Dsa{r} \Gommi 
\star \J (\melon\sqcup \logo{4}{Vi}{3.5}{38}) \\
\nonumber
  {} \phantom = {}  & 
+ \sumtc \Dsa{p} \Gommi \star \J (\melon^{\sqcup 3} )
\big\}
\\
 \frac1{2\cdot 2!}
\suml_{i\neq a} \edgebracket{
 G\hp 8 _{|
\raisebox{-.33\height}{\includegraphics[height=2ex]{graphs/3/Item2_Melon.pdf}}|
\raisebox{-.33\height}{\includegraphics[height=2ex]{graphs/3/Item2_Melon.pdf}}|
\raisebox{-.34\height}{\includegraphics[height=1.7ex]{graphs/3/Item4_Va.pdf}}|} 
\, & , \, 
 \raisebox{-.3\height}{\includegraphics[height=3ex]{graphs/3/Logo2_Melon.pdf}}
\sqcup 
\raisebox{-.3\height}{\includegraphics[height=3ex]{graphs/3/Logo2_Melon.pdf}}
\sqcup 
\raisebox{-.4\height}{\includegraphics[height=4ex]{graphs/3/Logo4_Va.pdf}}}
\\
\nonumber
  {} = {}  & 
  \frac14  \big\{
  \sumud \Dsa{r} \Gomma 
\star \J (\melon\sqcup \logo{4}{Va}{3.5}{38}) 
\\
\nonumber
  {} \phantom{=} {}  & 
\,\,+ \sumtc \Dsa{p} \Gomma \star \J (\melon \sqcup \melon\sqcup \melon) \big\}
\\
\frac{1}{4!} 
\edgebracket{
\Gommmm
\, ,  & \,
\melon\sqcup\melon\sqcup\melon\sqcup\melon
}\\
{ } = { } &  \nonumber
\frac{1}{4!}
 \suml_{u=1}^4 (\Dsa{u} \Gommmm) \star \J (\melon\sqcup\melon\sqcup\melon) 
 \\
 \edgebracket{ 
\GoWa
  & , \!\!
\raisebox{-.42\height}{\includegraphics[width=17ex]{graphs/3/Logo8_Wa.pdf}}  
} 
 \\  \nonumber  {} = {} & 
\Dsa{1} \GoWa \star \J (\logo{6}{Ebac}{5}{28}) +\Dsa{2} \GoWa \star \J (\logo{6}{Eabc}{5}{28})
  \\  \nonumber  {} + {} &   
  \Dsa{3} \GoWa \star \J (\logo{6}{Eabc}{5}{28}) +  \Dsa{4} \GoWa \star \J (\logo{6}{Ecba}{5}{28})
 \\  \nonumber  {} =  {} & 
\Dsa{1} \GoWa \star \J (\logo{6}{Ebac}{5}{28}) + \Dsa{2} \GoWa \star \J (\logo{6}{Eabc}{5}{28})
  \\  \nonumber  {} + {} &   
  \Dsa{3} \GoWa \star \J (\logo{6}{Eabc}{5}{28}) +(13)^*( \Dsa{4} \GoWa) \star \J (\logo{6}{Ecab}{5}{28})
   \\
 \edgebracket{ 
\GoWb
  & , \!\!
\raisebox{-.42\height}{\includegraphics[width=17ex]{graphs/3/Logo8_Wb.pdf}}  
} 
 \\  \nonumber  {} = {} & 
(13)^*(\Dsa{1} \GoWb) \star \J (\logo{6}{Ecab}{5}{28}) 
+ \Dsa{2} \GoWb  \star \J (\melon \sqcup \logo{4}{Va}{3.5}{38})
  \\  \nonumber  {} + {} &   
  [(13)^*(\Dsa{3} \GoWb) +
  (13)^*(\Dsa{4} \GoWb)] \star \J (\logo{6}{Eabc}{5}{28})
     \\ \nonumber
 \edgebracket{ 
\GoWc
  & , \!\!
\raisebox{-.42\height}{\includegraphics[width=17ex]{graphs/3/Logo8_Wc.pdf}}  
} 
 \\  \nonumber  {} = {} & 
[(13)^*(\Dsa{1} \GoWc) + (13)^*(\Dsa{2} \GoWc)] \star \J (\logo{6}{Eabc}{5}{28}) 
  \\  \nonumber  {} + {} &   
(123)^*(\Dsa{3} \GoWc) \star \J (\melon \sqcup \logo{4}{Va}{3.5}{38})
+
 \Dsa{4} \GoWc \star \J (\logo{6}{Ebac}{5}{28})
  \\ \frac12 \! \suml_{\substack{i=b,c \\ (a\neq j\neq i)}}\edgebracket{\GoRij
 \, & ,  
\!\!
\raisebox{-.42\height}{\includegraphics[width=17ex]{graphs/3/Logo8_Rij.pdf}} }
\\  \nonumber  {} = {} & 
\frac12 \!\suml_{\substack{i=b,c \\ (a\neq j\neq i)}}  \Big[ \sumud \Dsa{r} \GoRij \star \J (\logo{6}{Eija}{5}{26}) +
\sumtc \Dsa{p}\GoRij \star \J (\logo{6}{Eiaj}{5}{26})
\Big]
  \\  \nonumber  {} = {} & 
  \frac12 \!\suml_{\substack{i=b,c \\ (a\neq j\neq i)}}  \Big[ \big(\sumud (13)^*(\Dsa{r} \GoRij)  +
\sumtc \Dsa{p}\GoRij  \big)\star \J (\logo{6}{Eiaj}{5}{26}) \Big]
\\
 \nonumber  {} = {} &  
 \frac12  \Big[ \big( \sumud (13)^*(\Dsa{r} \GoRbc)  +
\sumtc \Dsa{p}\GoRbc  \big)\star \J (\logo{6}{Ebac}{5}{26}) 
\\
\nonumber  {}  {} & \,+
\big( \sumud (13)^*(\Dsa{r} \GoRcb)  +
\sumtc \Dsa{p}\GoRcb  \big)\star \J (\logo{6}{Ecab}{5}{26}) 
\Big]  
% \\ \edgebracket{ }
%  \\  \nonumber  {} = {} &  
%   \\  \nonumber  {} + {} &  
% \end{align}
% \begin{align} 
\\
\frac12\suml_{\substack{i=b,c \\ (a\neq j\neq i)}}
\edgebracket{ 
\GoRai
\,, & \!\!
\raisebox{-.42\height}{\includegraphics[width=17ex]{graphs/3/Logo8_Rai.pdf}}
\!\!
}
\\  \nonumber  {} = {} & 
  \frac12 \!\suml_{\substack{i=b,c \\ (a\neq j\neq i)}}  \Big[ 
  \big(  \Dsa{1} \GoRai \star \J (\logo{6}{Eaij}{5}{26})  
  + \Dsa{2} \GoRai \star \J (\melon \sqcup \logo{4}{Vj}{3.5}{38})
  \\  \nonumber  &  \qquad \quad +
   \Dsa{3} \GoRai \star \J (\logo{4}{Vj}{3.5}{38} \sqcup \melon)  
  + \Dsa{4} \GoRai \star \J ( \logo{6}{Eaji}{5}{26}) \Big] 
  \\  \nonumber  {} = {} &   
  \frac12  \Big[ 
  \big(  \Dsa{1} \GoRab \star \J (\logo{6}{Eabc}{5}{26})  
  + \Dsa{2} \GoRab \star \J (\melon \sqcup \logo{4}{Vc}{3.5}{38})
  \\  \nonumber  &  \qquad \quad +
   \Dsa{3} \GoRab \star \J (\logo{4}{Vc}{3.5}{38} \sqcup \melon)  
  + \Dsa{4} \GoRab \star \J ( \logo{6}{Eacb}{5}{26}) \Big]  
   \\ \nonumber  {} + {} &   
  \frac12  \Big[ 
  \big(  \Dsa{1} \GoRac \star \J (\logo{6}{Eacb}{5}{26})  
  + \Dsa{2} \GoRac \star \J (\melon \sqcup \logo{4}{Vb}{3.5}{38}) 
  \\  \nonumber  &   \quad +
   \Dsa{3} \GoRac \star \J (\logo{4}{Vb}{3.5}{38} \sqcup \melon)  
  + \Dsa{4} \GoRac \star \J ( \logo{6}{Eabc}{5}{26}) \Big]  
   \\  \nonumber  {} = {} &     
  \frac12  \bigg[ 
  \Big\{  \big(\Dsa{1} \GoRab + (13)^* (\Dsa{1} \GoRac) 
  \\ &  \nonumber
   \qquad \qquad  +
  (13)^*( \Dsa{4} \GoRab) +\Dsa{4} \GoRac  \big\}  \star \J (\logo{6}{Eabc}{5}{26})  
  \\   \nonumber    
    & \qquad + \big(\Dsa{2} \GoRab +(123)^*(\Dsa{3} \GoRab) \big)\star \J (\melon \sqcup \logo{4}{Vc}{3.5}{38})
   \\ \nonumber  {}   {} &   
   \qquad  +\big( \Dsa{2} \GoRac  + 
  (123)^* \big( \Dsa{3} \GoRac) \big) \star \J (\melon \sqcup \logo{4}{Vb}{3.5}{38}) \bigg] 
  \\
   \frac12\suml_{\substack{i=b,c \\ (a\neq j\neq i)}}   
\edgebracket{ 
\GoRia
\,, & \!\!
\raisebox{-.42\height}{\includegraphics[width=17ex]{graphs/3/Logo8_Ria.pdf}}
\!\!
} 
 \\  \nonumber  {} = {} &  
\frac{1}{2} 
\bigg[
\big\{
\Dsa{1}\GoRba +(13)^* \Dsa{4} \GoRba 
\big\} \star \J (\logo{6}{Ebac}{5}{28})
 \\  \nonumber  {} + {} & 
 \big\{
\Dsa{1}\GoRca +(13)^* \Dsa{4} \GoRca 
\big\} \star \J (\logo{6}{Ecab}{5}{28})
\\  \nonumber  {} + {} &  
\suml_{h=2,3} 
\big( \Dsa{h} \GoRba \star \J (\logo{6}{Qc}{4.6}{25})
+ \Dsa{h} \GoRca \star \J (\logo{6}{Qb}{4.6}{25})\big) 
\bigg]
\\
\frac14 \suml_{i\neq a} 
 \edgebracket{\GoAi
\, & ,\! 
\raisebox{-.4\height}{\includegraphics[width=7ex]{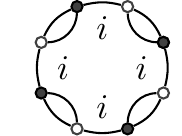}} } +
\frac14  
 \edgebracket{\GoAa \,,\,
\raisebox{-.4\height}{\includegraphics[width=7ex]{graphs/3/Logo8_Aa.pdf}} }
  \\  \nonumber  {} = {} &  \frac14\suml_{i\neq a}\suml_{r=1}^4 \Dsa{r}  \GoAi \star \J (\logo{6}{Qi}{4.6}{25})
+   \frac14 \suml_{r=1}^4 \Dsa{r}    \GoAa \star \J (\logo{6}{Qa}{4.6}{25}) 
\\ 
\suml_{i\neq a } \Edgebracket{  \GoPai
\,  , & \,
\raisebox{-.4\height}{\includegraphics[width=10ex]{graphs/3/Logo8_Pai.pdf}}  }
 \\  \nonumber  {} = {} &  
 \Dsa{1} \GoPab \star \J (\logo{6}{Qc}{4.6}{25}) +\Dsa{2} \GoPab \star \J (\melon \sqcup \logo{4}{Vc}{3.5}{38}) 
  \\  \nonumber  {} + {} & 
  \Dsa{1} \GoPac \star \J (\logo{6}{Qb}{4.6}{25}) +\Dsa{2} \GoPac \star \J (\melon \sqcup \logo{4}{Vb}{3.5}{38}) 
  \\  \nonumber  {} + {} & 
  \sumtc \big( \Dsa{p} \GoPab + (13)^* \Dsa{p} \GoPac  \big) \star \J ( \logo{6}{Eabc}{4.6}{25} )
%  \end{align}
%  \end{subequations}
%  \begin{subequations}
%   \begin{align}
% \tag{\ref{eq:laaaaange}w}\label{eq:laaaaangeaa}
\\
\suml_{i\neq a } \Edgebracket{  \GoPia
\,  , & \,
\raisebox{-.4\height}{\includegraphics[width=10ex]{graphs/3/Logo8_Pia.pdf}}  }
 \\  \nonumber  {} = {} &  
 \sumud \Dsa{r} \GoPba \star \J (\logo{6}{Qc}{4.6}{25}) +  \sumtc   \Dsa{p} \GoPba   \star \J ( \logo{6}{Ebac}{4.6}{25} ) 
  \\  \nonumber  {} + {} & 
 \sumud \Dsa{r} \GoPca \star \J (\logo{6}{Qb}{4.6}{25}) +  \sumtc  \Dsa{p} \GoPca   \star \J ( \logo{6}{Ecab}{4.6}{25} ) 
 \\ 
%  \tag{\ref{eq:laaaaange}x}\label{eq:laaaaangeaa}
 \sum_{\substack{i,j\neq a ; i\neq j}} \Edgebracket{    
 \GoPij
 \, &, \,
\raisebox{-.4\height}{\includegraphics[width=10ex]{graphs/3/Logo8_Pij.pdf}}
 }
 \\  \nonumber  {} = {} &  
 ( \Dsa{1} \GoPbc + \Dsa{1} \GoPcb ) \star \J (\logo{6}{Qa}{4.6}{25}) 
 \\  \nonumber  {} + {} &  
  \suml_{q=2,3,4}  \big\{  (13)^* \Dsa{q} \GoPbc   \star \J (\logo{6}{Ebac}{4.6}{25} ) 
 + (13)^* \Dsa{q} \GoPcb   \star \J (\logo{6}{Ecab}{4.6}{25} )  \big\}
 \\ 
%  \tag{\ref{eq:laaaaange}y}\label{eq:laaaaangeaa}
 \suml_{i\neq a }
 \Edgebracket{ 
 \GoXi
 \, & ,
\!\!\!\!\!
\raisebox{-.42\height}{\includegraphics[width=12ex]{graphs/3/Logo8_Xi.pdf}}
 }
 \\  \nonumber  {} = {} & 
 (\Dsa{1} \GoXb +  \Dsa{1} \GoXc ) \star \J(\logo{6}{K33}{3.6}{3})
   \\  \nonumber  {} + {} &  \sum_{q=2,3,4} (13)^*\Dsa{q} \GoXb \star \J (\logo{6}{Ecab}{4.6}{25}) 
   + (13)^*\Dsa{q} \GoXc  \star \J (\logo{6}{Ebac}{4.6}{25})
    \\  
%     \tag{\ref{eq:laaaaange}z}\label{eq:laaaaangeaa}
 \Edgebracket{ 
 \GoXabc
 \, & ,
\!\!\!\!\!
\raisebox{-.42\height}{\includegraphics[width=12ex]{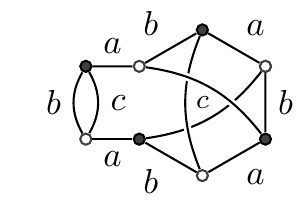}} 
 }\qquad \qquad \qquad \qquad \qquad \qquad \qquad \quad\mbox{ (symmetric in $b,c$)} \nonumber 
 \\  {} = {} & 
  \sumud \Dsa{r} \GoXabc \star \J(\logo{6}{K33}{3.6}{3}) + \sumtc \Dsa{p} \GoXabc  \star \J ( \logo{6}{Qa}{4.6}{25} )   
\\
%  \tag{\ref{eq:laaaaange}aa}\label{eq:laaaaangeaa}
\edgebrackett{  
\GoS
 \, & ,\,
\!\!
\raisebox{-.4\height}{\includegraphics[width=9ex]{graphs/3/Logo8_S.pdf}}
}
\\ 
\nonumber  {} = {} &  \sumud \Dsa{r} \GoS \star \J  \Big(\!\!\!\!\!\!\raisebox{-3.30ex}{\includegraphics[height=8ex]{graphs/3/Logo4S} }\!\!\Big)
\\ 
\nonumber  {} \phantom{=} {} & 
+
\Dsa{3} \GoS \star \J (\logo{6}{Ebac}{4.6}{25})
+
 \Dsa{4} \GoS \star \J  \Big(\!\!\! \raisebox{-3.30ex}{\includegraphics[height=8ex]{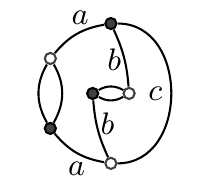} }\!\!\Big)
\\  \nonumber  {} = {} &  \sumud (13)^* \big (\Dsa{r} \GoS \big)
\star \J  ( \logo{6}{Eabc}{4.6}{25} ) 
+
\Dsa{3} \GoS \star \J (\logo{6}{Ebac}{4.6}{25})
\\ & \nonumber
+
\sumtc  (13)^* \big( \Dsa{p} \GoS \big) \star \J (\logo{6}{Ecab}{4.6}{25})
% \\  \nonumber  {} + {} & 
%\\ \edgebracket{   }
%  \\  \nonumber  {} = {} &  
%   \\  \nonumber  {} + {} &  
\\ 
\frac14
\edgebracket{ 
\GoCubo
 \, & ,\,
\raisebox{-.4\height}{\includegraphics[width=9ex]{graphs/3/Logo8_cubo.pdf}}
} 
% \tag{\ref{eq:laaaaange}ab}\label{eq:laaaaangeab}
\\  \nonumber  {} = {} &  
\frac14
\big\{(23)^*\Dsa{1} \GoCubo 
+ 
\Dsa{2} \GoCubo 
+(13)^* \Dsa{3} \GoCubo \\ 
\nonumber  {} \phantom{=} {} & 
+(123)^* \Dsa{4} \GoCubo \big\}
\star \J(\logo{6}{Eabc}{4.6}{25})
\\ 
%  \tag{\ref{eq:laaaaange}ac}\label{eq:laaaaangeac}
\frac14\suml_{\substack{ i=b,c; \\ (i\neq j\neq a)}} \edgebracket{ 
\GoYi
 \, & ,\,
\! 
\raisebox{-.44\height}{\includegraphics[width=11ex]{graphs/3/Logo8_Yi.pdf}}
}
  \\  \nonumber  = {}  {}  
  \frac14 \Big[
  \big\{  &
   \Dsa{1} \GoYb 
   +(123)^* \Dsa{2}\GoYb + (23)^*\Dsa{3}\GoYb  
   \\ 
\nonumber  {} \phantom{=} {} &\qquad\qquad \,\,\,+ (13)^*\Dsa{4}\GoYb 
  \big\} \star \J (\logo{6}{Ecab}{4.6}{25})
  \\  \nonumber  {}    {}   & \hspace{-1cm}\,\,+\big\{  \Dsa{1} \GoYc + (123)^*\Dsa{2}\GoYc + (23)^*\Dsa{3}\GoYc   
  \\ 
\nonumber  {} \phantom{=} {} & \qquad\qquad \,\,+ (13)^*\Dsa{4}\GoYc\big\}  \star \J (\logo{6}{Ebac}{4.6}{25}) \Big]
   \\
    \tag{\ref{eq:laaaaange}aa} %\label{eq:laaaaangeaa}
  \frac14  \edgebracket{ 
\GoYabc
 \, & ,\,
\!  
\raisebox{-.44\height}{\includegraphics[width=11ex]{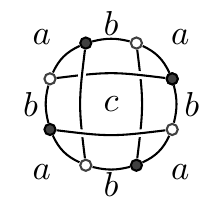}}} = 
\frac14 \suml_{r=1}^4 \GoYabc \star \J (\logo{6}{K33}{3.6}{3})\,.
   %   \\  \nonumber  {} + {} &  
 \end{align}
 \end{subequations}
One adds all the previous equations and associates by $\J(\B)$, for $\B$ one of the 11 graphs with $6$ vertices. 
\end{proof}

\section{Rank-five quartic theories}\label{sec:rankfive}
The generating function that enumerates the rank-$5$ connected 
boundary graphs (and interaction vertices) is the 
\textsf{OEIS} \textbf{A057007}:
\[
Z_{\mathrm{conn., 5}}(x)=  x + 15x^2+ 235x^3 + 14120x^4 + 1712845x^5+ 371515454x^6+\ldots
\]
We will not classify the 235 connected graphs with six 
vertices, but, aiming 
only at obtaining the $2$-point function's equation,
we will compute the free energy up to $\mathcal O(J^3,\bJ^3)$:

\begin{align*}
% \label{expansion_congraficas4}
% correct factors
W_{{}_{D=5}}[J,\bar J] & = 
G\hp 2 _{\includegraphics[height=2ex]{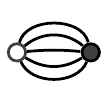}} \star \mathbb{J}\big(
\raisebox{-.34\height}{\includegraphics[height=3.77ex]{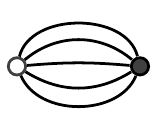}}\big) + 
\frac1{2!}
G\hp 4 _{|\raisebox{-.34\height}{\includegraphics[height=2ex]{graphs/5/Icono5_melon.pdf}}|
\raisebox{-.34\height}{\includegraphics[height=2ex]{graphs/5/Icono5_melon.pdf}}|} \star \mathbb{J}\big(
\raisebox{-.34\height}{\includegraphics[height=3.75ex]{graphs/5/Logo5_melon.pdf}}|
\raisebox{-.34\height}{\includegraphics[height=3.75ex]{graphs/5/Logo5_melon.pdf}}\big)
\\
&\quad
+
\sum\limits_{j=1}^5 \frac1{2}
G\hp4_{\raisebox{-.4\height}{\includegraphics[height=2.5ex]{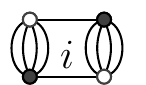}}} \star \mathbb{J}
\big(  
\raisebox{-.34\height}{\includegraphics[height=3ex]{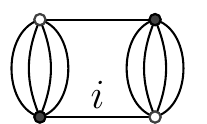}} \! \big)
+\sum\limits_{i<j} \frac1{2}G\hp4_{\raisebox{-.4\height}{\includegraphics[height=3.5ex]{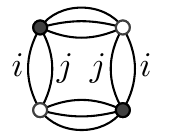}}} 
\! 
\star \mathbb{J}
\bigg(\!\!
\raisebox{-.4\height}{\includegraphics[height=6.9ex]{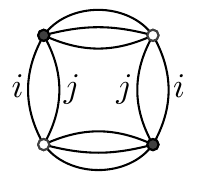}}\!\!\!\bigg)+
O(J^3,\bar J^3) \\
 & =  
G\hp 2 _{\includegraphics[height=2ex]{graphs/5/Icono5_melon.pdf}} \star \mathbb{J}\big(
\raisebox{-.34\height}{\includegraphics[height=3.77ex]{graphs/5/Logo5_melon.pdf}}\big) + 
\frac1{2!}
G\hp 4 _{|\raisebox{-.34\height}{\includegraphics[height=2ex]{graphs/5/Icono5_melon.pdf}}|
\raisebox{-.34\height}{\includegraphics[height=2ex]{graphs/5/Icono5_melon.pdf}}|} \star \mathbb{J}\big(
\raisebox{-.34\height}{\includegraphics[height=3.75ex]{graphs/5/Logo5_melon.pdf}}|
\raisebox{-.34\height}{\includegraphics[height=3.75ex]{graphs/5/Logo5_melon.pdf}}\big) 
+
 \frac1{2}G\hp4_{\raisebox{-.4\height}{\includegraphics[height=2.5ex]{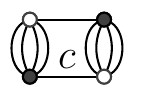}}} \star \mathbb{J}
\big(  
\raisebox{-.34\height}{\includegraphics[height=3ex]{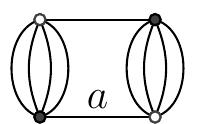}} \! \big)
\\
&\quad+
\sum\limits_{c\neq a} \frac1{2}G\hp4_{\raisebox{-.4\height}{\includegraphics[height=2.5ex]{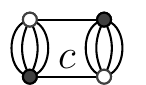}}} \star \mathbb{J}
\big(  
\raisebox{-.34\height}{\includegraphics[height=3ex]{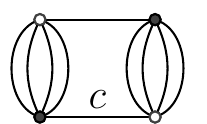}} \! \big)
+\sum\limits_{c\neq a} \frac1{2}G\hp4_{\raisebox{-.4\height}{\includegraphics[height=3.5ex]{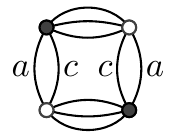}}} 
\! 
\star \mathbb{J}
\bigg(\!\!
\raisebox{-.4\height}{\includegraphics[height=6.9ex]{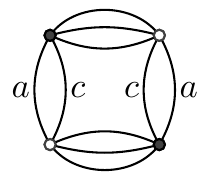}}\!\!\!\bigg)
\\
& 
\quad +\sum\limits_{\substack{c,d\neq a \\ c<d}} 
\frac1{2}G\hp4_{\raisebox{-.4\height}{\includegraphics[height=3.5ex]{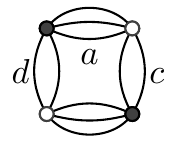}}} 
\! 
\star \mathbb{J}
\bigg(\!\!
\raisebox{-.4\height}{\includegraphics[height=6.9ex]{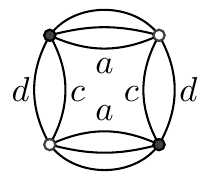}}\!\!\!\bigg)
+
O(J^3,\bar J^3)\,.
\end{align*}
For an arbitrary model in rank $5$ one has,
up to $\mathcal O(J^2,\bar J^2)$-terms,

\begin{align*}
Y\hp a_{m_a}[J,\bar J]\,\,
\sim  & \! \sum\limits_{q_{i_1},\ldots,q_{i_4}}
G\hp 2 _{\includegraphics[height=2ex]{graphs/5/Icono5_melon.pdf}} (m_a,q_{i_1},q_{i_2},q_{i_3},q_{i_4})
+
\frac{1}{2}
\bigg\{
\sum\limits_{s=1,2}
\Dma{s}
G\hp 4 _{|\raisebox{-.34\height}{\includegraphics[height=2ex]{graphs/5/Icono5_melon.pdf}}|
\raisebox{-.34\height}{\includegraphics[height=2ex]{graphs/5/Icono5_melon.pdf}}|}
+
\sum\limits_{c\neq a} \sum\limits_{s=1,2}
\Dma{s}G\hp4_{\raisebox{-.4\height}{\includegraphics[height=2.5ex]{graphs/5/Icono5_Vc.pdf}}} 
\\
& 
+\sum\limits_{\substack{c,d\neq a \\ c<d}} 
\sum\limits_{s=1,2}
\Dma{s}
G\hp4_{\raisebox{-.4\height}{\includegraphics[height=3.6ex]{graphs/5/Icono5_Np.pdf}}} 
+
\sum\limits_{c\neq a}
\sum\limits_{s=1,2}
\Dma{s}
G\hp4_{\raisebox{-.4\height}{\includegraphics[height=3.6ex]{graphs/5/Icono5_Nac.pdf}}}
+
\sum\limits_{s=1,2}
\Dma{s}
G\hp4_{\raisebox{-.4\height}{\includegraphics[height=2.5ex]{graphs/5/Icono5_Va.pdf}}} 
\bigg\} \star \mathbb{J}(\raisebox{-.3\height}{\includegraphics[height=3ex]{graphs/5/Logo5_melon.pdf}})\,.
\end{align*}
One straightforwardly gets

\begin{align*}
G\hp 2_{\includegraphics[height=2.5ex]{graphs/5/Icono5_melon.pdf}}(\mathbf x)
& = 
\frac{1}{E_{\mathbf x}} +  \frac{(-\lambda  )}{E_{\mathbf x}} 
 \Bigg\{
  \sum\limits_{a=1}^5 \bigg[
  2 \cdot G\hp 2_{\includegraphics[height=2.5ex]{graphs/5/Icono5_melon.pdf}}(\mathbf x)
  \cdot
  \big(\sum\limits_{q_{i_1(a)} }
  \sum\limits_{q_{i_2(a)} }
  \sum\limits_{q_{i_3(a)}} 
  \sum\limits_{q_{i_4(a)}} 
G\hp 2 _{\includegraphics[height=2ex]{graphs/5/Icono5_melon.pdf}} (x_a,q_{i_1(a)},q_{i_2(a)},q_{i_3(a)})\big)
  \\
  &
  \quad
+\sum\limits_{q_{i_1(a)}} 
\sum\limits_{q_{i_2(a)}}
\sum\limits_{q_{i_3(a)}}
\sum\limits_{q_{i_4(a)}} 
\Big(
G\hp 4 _{|\raisebox{-.34\height}{\includegraphics[height=2ex]{graphs/5/Icono5_melon.pdf}}|
\raisebox{-.34\height}{\includegraphics[height=2ex]{graphs/5/Icono5_melon.pdf}}|}
(x_a,q_{i_1(a)},q_{i_2(a)},q_{i_3(a)},q_{i_4(a)}; \mathbf{x}) \\
& \qquad\qquad\qquad\qquad\qquad\,\,
+
G\hp 4 _{|\raisebox{-.34\height}{\includegraphics[height=2ex]{graphs/5/Icono5_melon.pdf}}|
\raisebox{-.34\height}{\includegraphics[height=2ex]{graphs/5/Icono5_melon.pdf}}|}
( \mathbf{x};x_a,q_{i_1(a)},q_{i_2(a)},q_{i_3(a)},q_{i_4(a)}) 
\Big)
\\& 
\quad 
+
\sum\limits_{c\neq a}
\sum\limits_{q_{b(a,c)} }
\sum\limits_{q_{d(a,c)}}
\sum\limits_{q_{e(a,c)}}
\Big(
G\hp4_{\raisebox{-.4\height}{\includegraphics[height=2.5ex]{graphs/5/Icono5_Vc.pdf}}}(x_a,x_c,q_b,q_d,q_e;\mathbf{x})
+
G\hp4_{\raisebox{-.4\height}{\includegraphics[height=2.5ex]{graphs/5/Icono5_Vc.pdf}}}(\mathbf{x};x_a,x_c,q_b,q_d,q_e)
\Big)
\\&
\quad+
\sum\limits_{\substack{d,e\neq a \\ d<e}}
\sum\limits_{q_c}
\sum\limits_{q_b}
\Big(
G\hp4_{\raisebox{-.4\height}{\includegraphics[height=3.7ex]{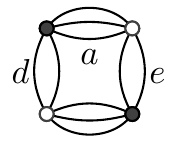}}} (x_a,x_b,x_c,q_d,q_e;\mathbf x)
+
G\hp4_{\raisebox{-.4\height}{\includegraphics[height=3.7ex]{graphs/5/Icono5_Npe.pdf}}} (\mathbf x;x_a,x_b,x_c,q_d,q_e)
\Big)\\&
\quad+
\sum\limits_{c\neq a}
\sum\limits_{q_c}
\Big(
G\hp4_{\raisebox{-.4\height}{\includegraphics[height=3.7ex]{graphs/5/Icono5_Nac.pdf}}} (x_a,x_b,q_c,x_d,x_e;\mathbf x)
+
G\hp4_{\raisebox{-.4\height}{\includegraphics[height=3.7ex]{graphs/5/Icono5_Nac.pdf}}} (\mathbf x;x_a,x_b,q_c,x_d,x_e)
\Big)+
2 G\hp4_{\raisebox{-.4\height}{\includegraphics[height=2.5ex]{graphs/5/Icono5_Va.pdf}}}(\mathbf{x};\mathbf{x})
  \\ \nonumber 
  & \quad + \sum\limits_{y_a} \frac{2}{|x_a|^2-|y_a|^2} \big( 
   G\hp 2_{\includegraphics[height=2.5ex]{graphs/5/Icono5_melon.pdf}}(\mathbf x)
   -
    G\hp 2_{\includegraphics[height=2.5ex]{graphs/5/Icono5_melon.pdf}}(y_a, 
    x_{i_1(a)},x_{i_2(a)},x_{i_3(a)},x_{i_4(a)})
  \big)
 \bigg]
 \Bigg\} \,\,.
\end{align*}

%  
%  \bibliography{paperSDE}
 \bibliographystyle{plain}

\end{document}